\documentclass[12pt]{article}
\usepackage{amsmath}
\usepackage{graphicx,psfrag,epsf}
\usepackage{enumerate}
\usepackage{natbib}
\usepackage{amsthm}
\usepackage{amssymb}
\usepackage{amsfonts}

\usepackage[colorlinks,citecolor=blue,urlcolor=blue,filecolor=blue,backref=page]{hyperref}
\usepackage{bbm}
\usepackage{bm}
\usepackage{dcolumn}
\newcolumntype{d}[1]{D{.}{.}{#1}}

\usepackage{booktabs,siunitx}
\usepackage{multirow}
\usepackage{lscape}
\usepackage{enumitem}
\usepackage{xcolor}

\newif\ifshowrev
\showrevfalse
\newcommand{\rev}[1]{\ifshowrev{\color{blue}#1}\else{#1}\fi}
\newenvironment{revblock}{\ifshowrev\color{blue}\fi}{}

\newcommand*\diff{\mathop{}\!\mathrm{d}}

\newcommand\indep{\protect\mathpalette{\protect\independenT}{\perp}}
\def\independenT#1#2{\mathrel{\rlap{$#1#2$}\mkern2mu{#1#2}}}

\newtheorem{assumption}{Assumption}
\newtheorem{proposition}{Proposition}
\newtheorem{lemma}{Lemma}

\theoremstyle{remark}
\newtheorem{remark}{Remark}

\newcommand{\blind}{0}

\begin{document}

\def\spacingset#1{\renewcommand{\baselinestretch}%
{#1}\small\normalsize} \spacingset{1}


\newcommand{\tit}{\bf Identification and Estimation of Treatment Effects from Network Data with Confounded Social Links}

\if0\blind
{
  \title{\tit}
  \author{Licheng Liu\thanks{
     Assistant Professor, Department of Political Science, University of Michigan.
     \\Email:
     \href{mailto:lichengl@umich.edu}{lichengl@umich.edu}. I thank Robert Franzese, Danny Hidalgo, 
      Yue Hou, In Song Kim, Jialu Li, Xun Pang, Lily Tsai, Yiqing Xu, Teppei Yamamoto, Han Zhang, 
      members of the Kim Research Group at MIT, and participants at seminars at Stanford University, 
      University of Michigan, University of Rochester, Hong Kong University of Science and 
      Technology, and Annual Meeting of the Society for Political Methodology in 2023 and 2024 for 
      their helpful comments. All remaining errors are my own.
     }
  }
  \date{}
  \maketitle
} \fi

\if1\blind
{
   \title{\tit}
   \date{This draft: October 10, 2026 \\ First draft: August 30, 2023}
   \maketitle
} \fi

\bigskip
\begin{abstract}
\begin{revblock}
In observational network data, the adoption of a treatment and the formation of social ties often 
depend on the same characteristics that affect outcomes. We propose a generalized propensity score 
for each unit's effective treatment, built from a model for treatment assignment and a dyadic 
model for network formation, which may include latent node heterogeneity. The score is 
marginalized over the network given covariates rather than conditioned on the realized network, 
and the estimand is a population effect: the average effect of an exposure level for a unit drawn 
from the population of units and networks. We relate this estimand to those of methods that 
condition on the realized network and show what marginalizing costs and buys. It requires 
network formation to be unconfounded given observed covariates, a stronger assumption than 
unconfoundedness of treatment given the network. In exchange, it defines effects at exposure 
levels that poorly connected units cannot reach, and at levels every unit can reach it gives an 
estimator whose asymptotic variance is no larger than the network-conditional one. We combine 
the score with an augmented inverse probability weighting estimator, prove asymptotic normality 
as the number of networks grows with bounded network size, discuss what carries over to a single 
large sparse network, and apply the method to the effect of a political communication technology 
on political participation in Uganda.
\end{revblock}
\end{abstract}

\noindent%
{\it Keywords:}  Causal Inference, Doubly Robust Estimation, Interference, Network Formation, Social Network Data
\vfill

\newpage
\spacingset{1.45} 
\section{Introduction}

When the units in a study are connected, a treatment given to one of them can affect the outcomes 
of others, and the Stable Unit Treatment Value Assumption (SUTVA) fails. The causal inference 
literature calls this \textit{interference} (e.g., \citealp[]{bowers2013reasoning}). It arises in 
behavioral studies, where individuals form family or friendship ties with people of similar age or 
education (e.g., \citealp[]{sinclair2012detecting}), and in cross-country analyses, where nations 
are connected by geography, trade or joint membership in international organizations 
(e.g., \citealp[]{simmons2004globalization}). The usual remedy is a regression that includes the 
unit's own treatment, a measure of interference such as the average treatment of its neighbors, 
and unit-level controls (e.g., \citealp[]{arpino2016assessing}), with treatment effects read off 
the coefficients. This works when the connections are unconfounded, as geographic proximity 
usually is.\footnote{That is, when no variable affects both the outcome and the formation of 
connections between units.} Social ties are different. People choose them on the basis of 
characteristics that may also affect the outcome \citep{goldsmith2013social} and the adoption of 
the treatment, and a regression that adjusts for individual characteristics alone leaves the 
confounding that enters through the ties.

\rev{By \textit{confounded social links} we mean that the ties, like the treatment, depend on 
covariates that also affect potential outcomes. The network is then not a confounder in its own 
right; it is a component of the assignment, and that component is confounded. Our identifying 
assumption is that the covariates that drive ties, adoption and outcomes are observed 
(Assumption~\ref{asm:ig} below).}

We propose a generalized propensity score (GPS) approach for data that consist of many disjoint 
groups, with ties formed only within groups. Treatment effects are differences in average 
potential outcomes across levels of an ``effective treatment'' \citep{manski2013identification}, a 
function of the treatment vector and the network that \citet{aronow2017estimating} call an 
``exposure mapping.'' Its form encodes how the researcher believes units interfere; we take it to 
be flexible but correctly specified. Given a model for treatment assignment and a model for 
network formation \citep{toulis2021estimating}, the score of a unit is the probability of its 
exposure level given the covariates of everyone in its group. For simple exposure mappings the 
score has a closed form; otherwise it can be simulated from the fitted models 
(e.g., \citealp[]{aronow2017estimating}). With the scores in hand, we estimate average potential 
outcomes and treatment effects by augmented inverse probability weighting (AIPW), which is doubly 
robust \citep{glynn2010introduction}, and use a cluster-robust \rev{variance estimator} for the 
dependence among units in the same group.

\begin{revblock}
What is new is what the score conditions on. Existing methods condition on the realized network, 
either by treating it as fixed or by placing it in the propensity score. Our score averages over 
the networks that the formation process could have produced given covariates. We think of the 
treatment vector and the network as a joint assignment with two components, and then the choice 
between conditioning on the network component and averaging over it is a choice of estimand, of 
the same kind as the choice between an average treatment effect and an effect on the treated. 
Averaging costs more: it requires tie formation to be unconfounded given observed covariates, 
which methods that condition on the network do not need. It buys three things. The estimand is a 
population effect whose uncertainty refers to the population of networks, which is the relevant 
frame when the policy question concerns villages, schools or markets whose ties have not been 
observed, and which stays well defined when ties respond to treatment. Effects are defined at 
exposure levels that poorly connected units cannot reach, such as having two treated friends for a 
unit with one tie, and are recovered from units with similar covariates who did reach them. And 
at levels that every unit can reach, the network affects exposure but not outcomes given 
covariates, so leaving it out of the score lowers the asymptotic variance, as leaving an 
instrument out of a propensity score does. Section~\ref{sec:cond} states these results and 
compares the three estimands involved.

The cost is real, and the latent-homophily literature is right that selection on observables is 
a strong requirement for ties \citep{shalizi2011homophily, mcfowland2023estimating}. Conditioning 
on the realized network is robust to unobserved characteristics that shape ties and outcomes but 
not adoption, because the network then serves as a proxy for them. A middle course fits the 
framework: latent-variable network models \citep{hoff2002latent, hoff2021additive} recover 
node-level latent positions from the observed ties, and the score can condition on these 
positions while averaging over the ties themselves. The application uses this specification. Its 
justification is a plug-in argument for large groups, which we state but do not prove.
\end{revblock}

The approach is related to design-based causal inference with network data 
\citep{ogburn2024causal}. \citet{aronow2017estimating} propose inverse probability weighting for 
network experiments. For observational data, \citet{forastiere2021identification} factorize the 
exposure into a direct treatment and an interference component and propose a semiparametric 
weighting estimator. \citet{sanchez2022network} develops a network propensity score for data with 
many disjoint groups that allows for group-level unobserved heterogeneity but takes the ties as 
randomly assigned \textit{a priori}, and under tighter assumptions suggests parametric models that 
estimate the score directly. \citet{jackson2020adjusting} introduce a peer-influenced propensity 
score and \citet{leung2022unconfoundedness} a propensity score based on graph neural networks; 
both allow treatment adoption to diffuse through the network. \rev{All of these condition on the 
realized network. Random-network frameworks are rarer: \citet{li2022random} derive asymptotics 
for randomized experiments when the interference graph is drawn from a graphon, and 
\citet{clark2024causal} model ties and covariates jointly in a Bayesian framework.} The setting is 
also close to the literature on causal peer effects or \textit{contagion}, the influence of 
neighbors' outcomes on a unit's own outcome, where the central difficulty is confounding by 
homophily \citep{shalizi2011homophily} and researchers often turn to longitudinal data 
(e.g., \citealp[]{egami2024identification, uppala2023contagion}). We estimate the overall effect of 
the treatment vector and do not try to isolate the part of a spillover that works through 
contagion.

\rev{The main development assumes that treatment assignment and network formation are 
conditionally independent given covariates, a mechanism \citet{franzese2008contagion} call 
\textit{common exposure}, with dyadic tie formation that may depend on observed and latent node 
attributes. Two extensions follow. When treatment diffuses along ties or ties respond to 
treatment, the score is obtained by simulation from the joint model. For a single large sparse 
network, we show that identification, positivity and the rate of the estimated score carry over 
and say what is known about inference.}

Section~\ref{me} introduces the motivating example, Section~\ref{sec:setup} the framework and 
estimands, and Section~\ref{sec:gps} the score, the comparison with network-conditional 
approaches (Section~\ref{sec:cond}), and estimation and inference. \rev{Section~\ref{sec:ext} 
covers the extensions.} Sections~\ref{SMC} and~\ref{ea} report Monte Carlo evidence and the 
application, and Section~\ref{sec:conc} concludes.

\section{A Motivating Example}\label{me} 

Whether information and communication technology (ICT) promotes political engagement in developing 
countries is still debated. In theory, ICT lowers the cost of communicating with politicians and so 
raises participation, especially among marginalized groups who otherwise have little access to them 
\citep{grossman2020information}. The empirical evidence is mixed.

We study the effect of U-Bridge, a political communication technology, on political participation in 
16 Ugandan villages, using the replication data of \citet{ferrali2020takes} and 
\citet{eubank2021viral}. U-Bridge was introduced in a district of northwestern Uganda and lets 
residents contact local officials by anonymous text message. It was rolled out as a field experiment 
in 131 randomly selected villages, where treatment villages received periodic community meetings 
about its use. To study adoption, \citet{ferrali2020takes} surveyed residents of 16 treatment 
villages in person in 2016, two years after the launch, collecting individual data on demographics, 
meeting attendance, adoption of U-Bridge, and social ties.

Our analysis uses the 3,018 respondents with complete data, of whom 135 adopted U-Bridge. The 
treatment is a binary indicator of adoption; the outcome is a continuous index of political 
participation that summarizes political actions over the previous year. The covariates are age, 
gender, income, secondary education, a formal leadership role, phone ownership, a behavioral proxy 
for community engagement, and an indicator for attendance at a U-Bridge meeting.\footnote{The 
variables are described in \citet{ferrali2020takes}.}

Residents are linked by several kinds of tie: family, friendship, lending, and problem solving 
\citep{sinclair2012detecting}. Such ties can carry spillovers of adoption through conversation. They 
also form on the basis of individual characteristics that influence participation, which is the 
source of confounding. We build one network per village by linking residents who share any of the 
four ties, treat the links as undirected, and drop ties across villages, which leaves 16 separate 
components.\footnote{This follows \citet{ferrali2020takes} and \citet{sanchez2022network}.}

The effect of U-Bridge on participation is hard to estimate for three reasons. Adoption is not 
random and depends on characteristics such as education. The network creates interference, so 
SUTVA fails. And the ties themselves form on characteristics that also affect participation: in 
these villages, leaders have nearly twice as many ties as other residents (Figure~\ref{ns0} shows 
one village, with leaders in red), so ignoring how the network formed invites bias. \rev{The same 
fact shows why conditioning on the realized network changes the question. Given the realized 
ties, the effect of having two adopter friends, two ties to residents who adopted, can only be 
identified among residents with at least two ties. Those residents are older and more often 
leaders than the rest, so the effect refers to a sub-population selected through the network.} We 
address these issues with a potential outcomes framework that includes the network and with a 
propensity-score-based estimation strategy.

 \begin{figure}[!ht]
  \begin{center}
  \begin{tabular}{c}
   \raisebox{-0.5\height}{\includegraphics[scale=0.35]{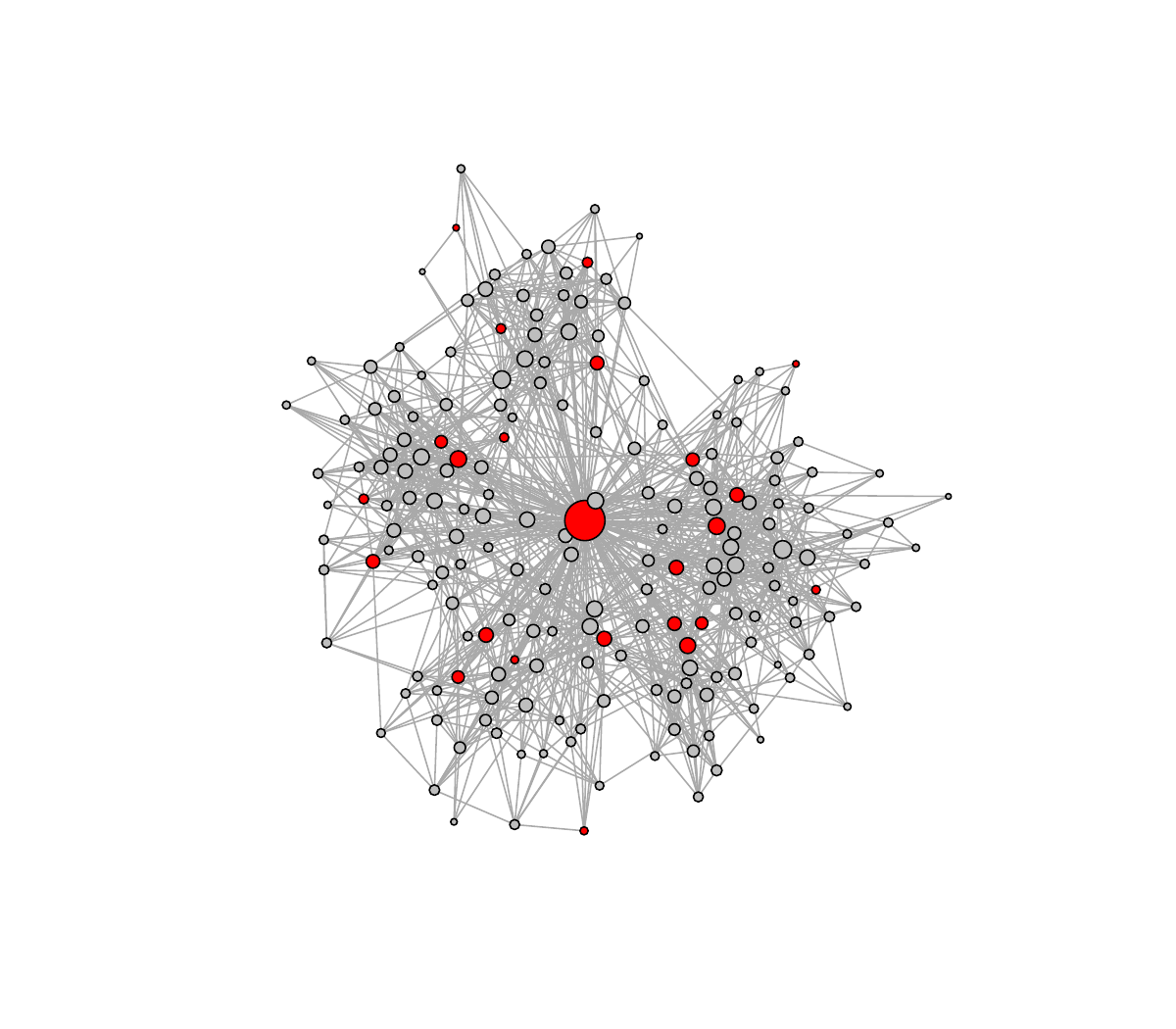}} 	
   \end{tabular}
   \end{center}
   \caption{Network Visualization for a Village}\label{ns0}
   \footnotesize\textbf{Note:} The size of each vertex is proportional to the square root of its 
   degree. Red vertices are respondents who hold a formal leadership role in the village.
   \end{figure}

\section{Setup}\label{sec:setup}

  \subsection{Notation}
    
  Consider a cross-sectional dataset with $G$ distinct groups; in the motivating example each village 
  is a group. Let $\mathcal{G} = \{1, \ldots, G\}$ be the set of groups, $N_g$ the number of units in 
  group $g$, and $N = \sum_{g=1}^{G} N_g$ the total number of units. Let $\bm{W}_g$ be the observed 
  network of group $g$, a random ($N_g \times N_g$) adjacency matrix whose entry $W_{ijg}$ records the 
  connection between units $\{ig\}$ and $\{jg\}$.\footnote{The index pair $\{ig\}$ denotes the $i$th 
  unit of group $g$. Bold uppercase letters denote random vectors and matrices and bold lowercase 
  letters their realizations.} In the language of network analysis, $\bm{W}_g$ is a random graph 
  $\Gamma_g = (\mathcal{N}_g, \mathcal{E}_g)$, where $\mathcal{N}_g$ is the set of units of group $g$ (its 
  ``nodes'' or ``vertices''), with cardinality $N_g$, and $\mathcal{E}_g$ is the set of edges between 
  pairs of units in $g$, so that $(\{ig\},\{jg\}) \in \mathcal{E}_g$ if $W_{ijg} > 0$. If the graph is 
  undirected, $W_{ijg} \equiv W_{jig}$ for all $\{ig\}, \{jg\} \in \mathcal{N}_g$; otherwise 
  \rev{$\bm{W}_g$ is asymmetric and} $W_{ijg}$ and $W_{jig}$ may differ. Throughout we take the network 
  to be unweighted, $W_{ijg} \in \{0, 1\}$ for all $g$, and we rule out self-links, $W_{iig} \equiv 0$. 
  Since $\bm{W}_g$ is random, let $\mathcal{W}_g$ be its sample space. \rev{We write 
  $\deg_{ig}(\bm{w}_g) = \sum_{\{jg\} \neq \{ig\}} w_{ijg}$ for the degree of unit $\{ig\}$ in network 
  $\bm{w}_g$, and $\mathcal{N} = \cup_g \mathcal{N}_g$ for the set of all $N$ units.}
  
   Let $D_{ig} \in \{0, 1\}$ be a binary treatment assigned to unit $\{ig\}$ and $Y_{ig}$ its observed 
   outcome. In vector form, $\bm{D}_g = (D_{1g}, \ldots, D_{N_gg})$ is the ($N_g \times 1$) vector of 
   treatments and $\bm{Y}_g = (Y_{1g}, \ldots, Y_{N_gg})$ the vector of observed outcomes for group 
   $g$, and $\mathcal{D}_g = \{0, 1\}^{N_g}$ is the sample space of $\bm{D}_g$.
   
   For each unit we also observe a ($p \times 1$) vector of pre-treatment covariates $X_{ig}$. With 
   network data it is useful to think of two subsets: confounders $X_{ig}^D$ that influence treatment 
   adoption, and covariates $X_{ig}^W$ that generate homophily and thereby affect tie formation. The 
   two subsets need not be disjoint; a covariate may affect both adoption and tie formation. 
   $X_{ig}^W$ may also contain dyadic covariates, such as the distance or the volume of bilateral 
   trade between two countries. In matrix form, $\bm{X}_g = (X_{1g}, \ldots, X_{N_g g})'$ is the 
   \rev{($N_g \times p$)} matrix of covariates for group $g$, and $\bm{X}_g^D$ and $\bm{X}_g^W$ are 
   the matrices for the two subsets.

   \subsection{Potential outcomes in social network data}
 
   In most social network data the observed network $\bm{W}_g$ is random, and its formation is 
   confounded in the same way as the assignment of treatment, because the covariates that drive 
   tie formation may also affect potential outcomes. The network also creates interference, since the 
   treatment of one unit may affect the potential outcomes of others. To handle both, we extend the 
   potential outcome framework \citep{rubin1974estimating} and write each unit's potential outcome as 
   a function of the treatment vector and the network: $Y_{ig}(\bm{d}_g, \bm{w}_g)$ for 
   $\{ig\} \in \mathcal{N}_g$, $\bm{d}_g \in \mathcal{D}_g$ and $\bm{w}_g \in \mathcal{W}_g$. Under 
   consistency, that is, no multiple versions of treatment \citep{rubin1986ifs}, the observed outcome 
   is $Y_{ig} = Y_{ig}(\bm{d}_g, \bm{w}_g)$ when $\bm{D}_g = \bm{d}_g$ and $\bm{W}_g = \bm{w}_g$. 
   
   In this framework a change in the network can change potential outcomes even when the treatment 
   vector stays the same. In Figure~\ref{po}, the network $\bm{w}_g$ in the left panel is denser than 
   the network $\bm{w}'_g$ in the right panel while every unit's treatment status is the same in both 
   \rev{panels}; $Y_{ig}(\bm{d}_g, \bm{w}_g)$ may still differ from $Y_{ig}(\bm{d}_g, \bm{w}'_g)$, 
   because unit $\{ig\}$ is connected to more treated units in the denser network. We therefore treat 
   $\bm{D}_g$ and $\bm{W}_g$ as a ``joint'' assignment of treatment and network. \rev{On this view the 
   network is a component of the assignment rather than a confounder: the pair 
   $(\bm{D}_g, \bm{W}_g)$ plays the role that the treatment indicator plays in the classical setting, 
   and the covariates $\bm{X}_g$ confound both components.}

   With this notation, the data generating process for network formation, treatment assignment, and 
   outcomes can be written as a non-parametric structural equation model 
   (NP-SEM; \citealp[]{pearl2009causality}):\footnote{We assume that the observed covariates fully 
   account for network formation and treatment adoption, which rules out unobserved confounders. 
   Latent homophily can be inferred from the network with latent space models 
   (e.g., \citealp[]{hoff2021additive,mcfowland2023estimating})\rev{; Section~\ref{sec:latent} shows 
   how such latent node attributes enter the generalized propensity score}.}
   \begin{equation}\label{all}
       \begin{gathered}
           W_{ijg} = f_w(\{ig\}, \{jg\}, \bm{X}_g, \epsilon^w_{ijg}), \\
           D_{ig} = f_d(\{ig\}, \bm{X}_g, \epsilon^d_{ig}), \\
           Y_{ig}(\bm{d}_g, \bm{w}_g) = f_y(\{ig\}, \bm{d}_g, \bm{w}_g, \bm{X}_g, \epsilon^y_{ig}), 
       \end{gathered}
   \end{equation}  
   where $\epsilon^w_{ijg}$, $\epsilon^d_{ig}$ and $\epsilon^y_{ig}$ are random errors. Correlation 
   between $\epsilon^w_{ijg}$ and $\epsilon^d_{ig}$ captures interaction between network formation 
   and treatment assignment. The framework is not restricted to undirected networks, but for an 
   undirected network in which a tie depends only on the covariates of the pair, the formation model 
   simplifies to $W_{ijg} = f_w(X_{ig}, X_{jg}, \epsilon^w_{ijg})$. 

   In social network data the covariates of all units in a group, $\bm{X}_g$, may affect each 
   unit's adoption and outcome. To make this explicit, \rev{write} $\bm{X}_g$ in NP-SEM~(\ref{all}) as $(X_{ig}, \bm{X}_{-ig})$, where 
   $\bm{X}_{-ig}$ is what the network literature calls a \textit{contextual variable} 
   \citep{jackson2008social, jackson2020adjusting}. Including it in $f_d(\cdot)$ and $f_y(\cdot)$ 
   allows \textit{social norms} to influence adoption and outcomes. A researcher with a strong prior 
   that no such influence exists can drop the contextual variable and work with 
   $D_{ig} = f_d(X_{ig}, \epsilon^d_{ig})$ and 
   $Y_{ig}(\bm{d}_g, \bm{w}_g) = f_y(\bm{d}_g, \bm{w}_g, X_{ig}, \epsilon^y_{ig})$.

   Because the treatment vector and the network jointly determine potential outcomes, $f_y(\cdot)$ 
   has high-dimensional arguments, which makes identification and estimation hard. Following the 
   literature, we assume that a low-dimensional, possibly vector-valued, function of the treatment 
   vector and the network determines the potential outcomes: 

   \begin{assumption}[Exposure mapping] \label{asm:em}
   For each group $g$, there exists a known function 
   $ e: \mathcal{N}_g \times \mathcal{D}_g \times \mathcal{W}_g \rightarrow \mathcal{T}_g $, such that 
   $$ 
   Y_{ig}(\bm{d}_g, \bm{w}_g) = Y_{ig}(t) 
   $$
   if $ e(\{ig\}, \bm{d}_g, \bm{w}_g) =  t $.
   \end{assumption}  
   Under Assumption~\ref{asm:em}, $ Y_{ig}(\bm{d}_g, \bm{w}_g) = Y_{ig}(\bm{d}'_g, \bm{w}'_g) $ 
   whenever $ e(\{ig\}, \bm{d}_g, \bm{w}_g) = e(\{ig\}, \bm{d}'_g, \bm{w}'_g) $. In the literature on 
   interference, $ T_{ig} = e(\{ig\}, \bm{D}_g, \bm{W}_g) $ is the ``effective treatment'' 
   \citep{manski2013identification} and $e(\cdot)$ the ``exposure mapping'' \citep{aronow2017estimating}. 
   
   Exposure mappings are used both in methodological work on interference 
   (e.g., \citealp[]{forastiere2021identification}) and in empirical work on spillovers 
   (e.g., \citealp[]{arpino2016assessing}). That literature usually takes the network as fixed and 
   builds it into the exposure mapping. We instead treat the network as an input on the same footing 
   as the treatment vector: $\bm{W}_g$ and $\bm{D}_g$ are common inputs for all units in a group, 
   and the output, the effective treatment, is unit specific. 
   
   Specifying the exposure mapping requires substantive knowledge of the treatment. Without 
   interference, the effective treatment is the unit's own treatment, 
   $T_{ig} = e(\{ig\}, \bm{D}_g, \bm{W}_g) = D_{ig}$. When the network creates interference and the 
   potential outcome depends on own treatment and on the number of treated neighbors, the 
   mapping is vector valued:\footnote{$Z_{ig}$ can be defined more generally as any function of the 
   network and the treatment vector of the other units, such as the number of connections or the 
   proportion of treated units among a unit's connections.}
   \begin{equation}\label{epp}
      T_{ig} = e(\{ig\}, \bm{D}_g, \bm{W}_g) = (D_{ig}, \sum_{\{jg\} \in \mathcal{N}_g} W_{ijg} D_{jg}) 
      = (D_{ig}, Z_{ig}).
   \end{equation}
   The mapping can be flexible, but we assume the researcher knows it and specifies it 
   correctly.\footnote{\citet{savje2023causal} studies misspecified exposure mappings\rev{: exposure 
   effects remain well-defined contrasts of average potential outcomes even when the mapping misses 
   part of the interference structure, but interpreting them as the effect of a specific intervention 
   requires correct specification. We return to the choice of mapping in the application 
   (Section~\ref{sec:expo}).}} Under Assumption~\ref{asm:em} the outcome equation of 
   NP-SEM~(\ref{all}) becomes $Y_{ig}(t) = f_y(\{ig\}, t, \bm{X}_g, \epsilon^y_{ig})$.

  \begin{figure}[!ht]
  \begin{center}
  \begin{tabular}{cc}
   \raisebox{-0.5\height}{\includegraphics[scale=0.3]{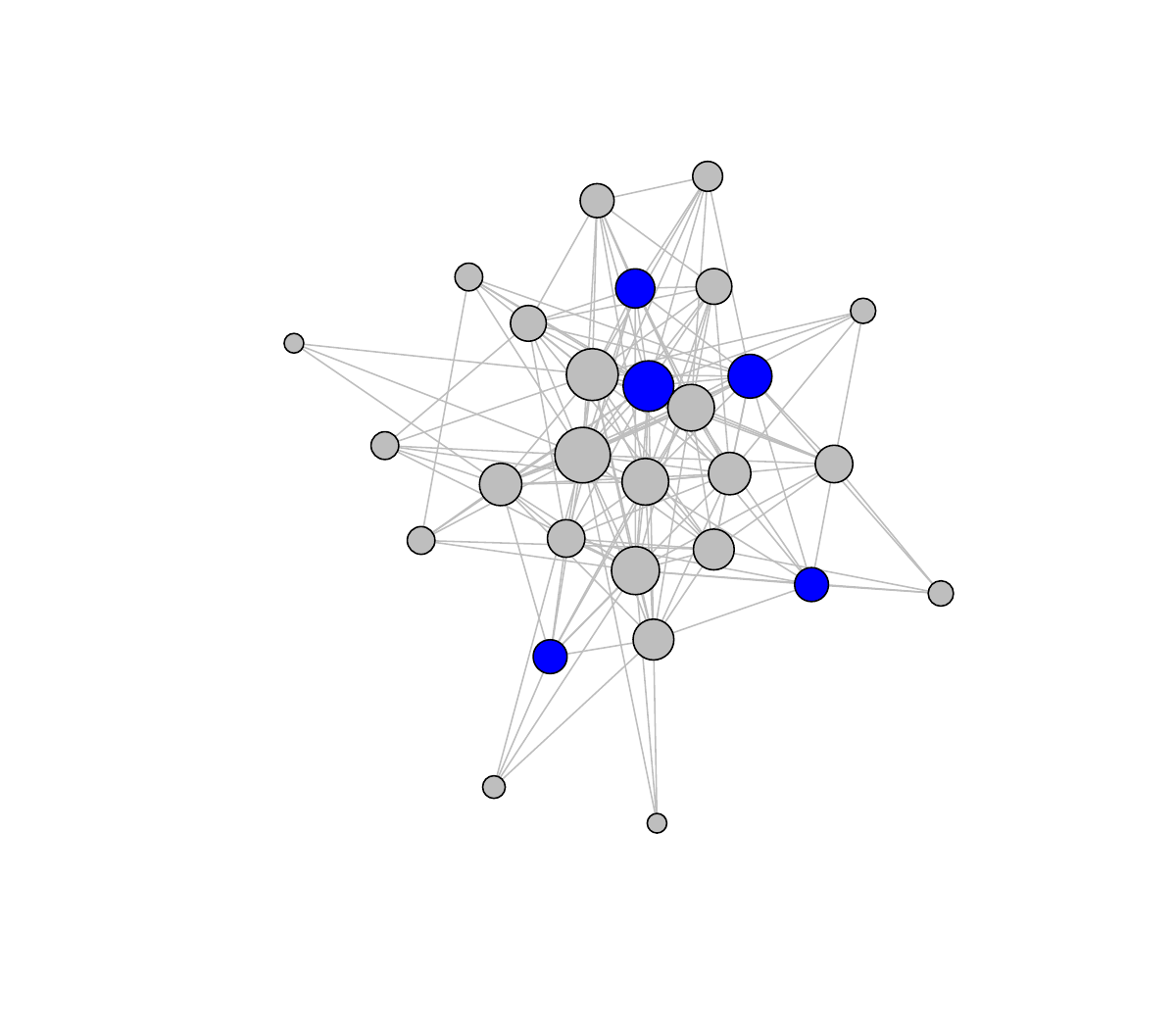}} &
   \raisebox{-0.5\height}{\includegraphics[scale=0.3]{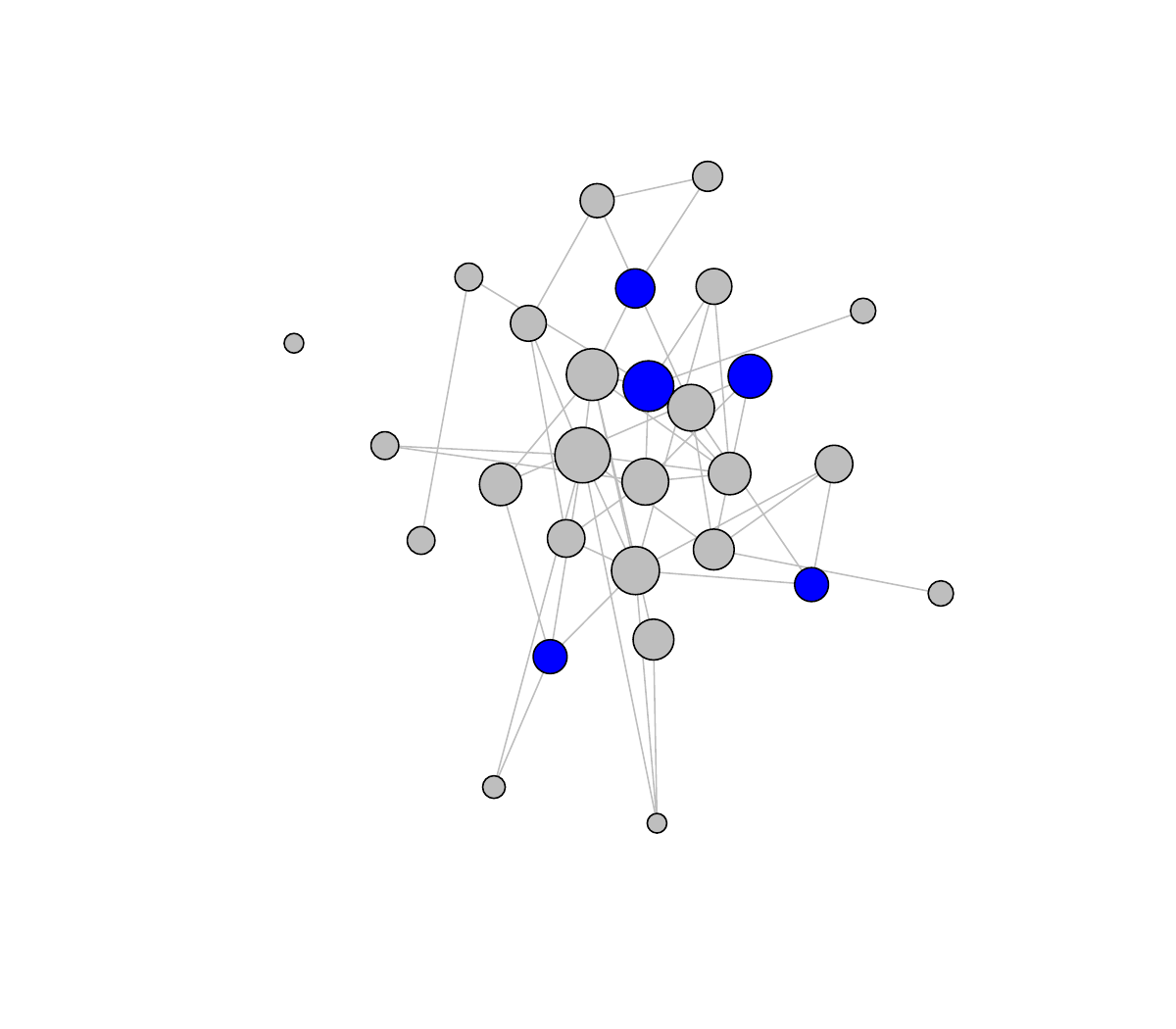}} 	
   \end{tabular}
   \end{center}
   \caption{Potential Outcomes with Different Networks for a Group}\label{po}
   \footnotesize\textbf{Note:} Blue nodes are treated units and gray nodes are control units.
   \end{figure}

   \subsection{Causal quantities of interest} 

   With a binary treatment each unit has two potential outcomes and the estimands are familiar: the 
   average treatment effect $ATE = \mathbbm{E}[Y_{ig}(1) - Y_{ig}(0)]$, the difference in average 
   potential outcomes under treatment and control, and the average effects on the treated and on the 
   controls. With network data the effective treatment is usually multi-valued. Since treatment 
   effects are differences in average potential outcomes, we start from the average potential outcome 
   at a level of the effective treatment, $\mu(t) = \mathbbm{E}[Y_{ig}(t)]$, also called the 
   \textit{average dose-response function} (ADRF), a term from medical statistics that the causal 
   inference literature has adopted for \rev{multi-valued and} continuous treatments. The effect 
   of effective treatment $t$ relative to $t'$ is
   \begin{equation}\label{ate}
   \tau(t, t') = \mathbbm{E}[Y_{ig}(t) - Y_{ig}(t')] = \mu(t) - \mu(t').    
   \end{equation}
   \rev{The expectation runs over all the randomness in NP-SEM~(\ref{all}): covariates, 
   treatment assignment, network and potential outcomes. Group sizes may vary, so the weighting 
   has to be stated. We take $\mu(t)$ to be the average potential outcome of a unit drawn at 
   random from all units, so that larger groups count more:
   \begin{equation}\label{mudef}
   \mu(t) = \frac{\mathbbm{E}\big[ \sum_{\{ig\} \in \mathcal{N}_g} Y_{ig}(t) \big]}{\mathbbm{E}[N_g]} 
   = \frac{\mathbbm{E}\big[ \sum_{\{ig\} \in \mathcal{N}_g} \mu(\{ig\}, t, \bm{X}_g) \big]}{\mathbbm{E}[N_g]}, 
   \qquad \mu(\{ig\}, t, \bm{X}_g) = \mathbbm{E}[Y_{ig}(t) \mid \bm{X}_g].
   \end{equation}
   This equals $\mathbbm{E}[Y_{ig}(t)]$ for any unit when group size is fixed, and it equals the 
   group-weighted mean $\mathbbm{E}[N_g^{-1} \sum_i Y_{ig}(t)]$ when group size is independent of the 
   conditional means. The conditional mean $\mu(\{ig\}, t, \bm{X}_g)$ may depend on the unit through 
   its own covariates and through contextual variables; where no confusion can arise we abbreviate it 
   to $m_t(\bm{X}_g)$. We call~(\ref{ate}) the \textit{population} effect, to distinguish it 
   from the network-conditional estimands of Section~\ref{sec:cond}.}

   Simple as it looks, the effect~(\ref{ate}) can mean many things depending on the exposure mapping. 
   Our interest is in direct and spillover effects of the treatment vector, and we use the 
   mapping~(\ref{epp}) to illustrate. The average potential outcome is then 
   $\mu(d, z) = \mathbbm{E}[Y_{ig}(d, z)]$, and the \textit{conditional direct effect} is
   \begin{equation}\label{ate_dir}
      \tau(z) = \mathbbm{E}[Y_{ig}(1, z) - Y_{ig}(0, z)] = \mu(1, z) - \mu(0, z), 
   \end{equation}
   the effect of own adoption $d$ with the level of interference $z$ held fixed. The 
   \textit{average direct effect} averages $\tau(z)$ over the distribution of $z$: 
   $\tau = \rev{\sum_z} \tau(z) f_Z(z)$\rev{, where $f_Z$ is the probability mass function of 
   $Z_{ig}$}. The \textit{conditional spillover effect} is
   \begin{equation}\label{ate_ind}
     \delta(z, z', d) = \mathbbm{E}[Y_{ig}(d, z) - Y_{ig}(d, z')] = \mu(d, z) - \mu(d, z'),  
   \end{equation}
   the difference in average potential outcomes between two levels of interference $z$ and $z'$ with 
   own treatment $d$ held fixed. When $z'$ is a reference level such as $z' = 0$, we drop it and write 
   $\delta(z, d) = \mu(d, z) - \mu(d, 0)$. The \textit{average spillover effect} is 
   $\delta(d) = \rev{\sum_z} \delta(z, d) f_Z(z)$.

   \subsection{Unconfoundedness of treatment assignment and network formation}

   When SUTVA holds, identification from observational data usually rests on conditional 
   unconfoundedness (ignorability), $\{Y_{ig}(0), Y_{ig}(1)\} \indep D_{ig} \mid X_{ig}$: given 
   individual confounders, potential outcomes are independent of treatment assignment. With network 
   data the potential outcomes depend on the effective treatment, a function of the treatment vector 
   and the network, so the assumption has to be restated. The following version allows the same 
   covariates to confound treatment assignment, network formation and potential outcomes:

   \begin{assumption}[Unconfoundedness of treatment assignment and network formation] \label{asm:ig}
   For each group, given the observed covariate matrix at group level, the potential outcomes of all 
   units in the group are jointly independent of the treatment assignment vector and the network. 
   I.e., $ \forall g \in \mathcal{G}$,
   \begin{equation}\label{ass2}
       \{Y_{ig} (t) : \{ig\} \in \mathcal{N}_g, t \in \mathcal{T}_g\} \indep \{\bm{D}_g, \bm{W}_g\} \mid \bm{X}_g.
   \end{equation}
   \end{assumption} 
   \rev{Assumption~\ref{asm:ig} is stated for the collection of potential outcomes in the group, 
   which is what NP-SEM~(\ref{all}) delivers when the outcome errors are independent of the 
   assignment errors. Identification (Propositions~\ref{prop:b}, \ref{prop:a} and~\ref{prop:idpos}) 
   uses only its unit-wise implication $Y_{ig}(t) \indep \{\bm{D}_g, \bm{W}_g\} \mid \bm{X}_g$ 
   (independence of the whole collection implies independence of each member); the joint version 
   is used in 
   Proposition~\ref{prop:eff}, in the proof of Proposition~\ref{ttt}, and in Appendix~\ref{app:single}.}
   Assumption~\ref{asm:ig} differs from the standard unconfoundedness assumption in that it 
   conditions on the covariates of all units in the group rather than on the unit's own covariates, 
   because $\bm{X}_g$ drives both the treatment vector and the network. Conditioning on $X_{ig}$ 
   alone does not remove the confounding between $\bm{D}_g$, $\bm{W}_g$ and the potential outcomes. 
   Since the effective treatment $T_{ig} = e(\{ig\}, \bm{D}_g, \bm{W}_g)$ is a function of $\bm{D}_g$ 
   and $\bm{W}_g$, an equivalent statement is $Y_{ig}(t) \indep T_{ig} \mid \bm{X}_g$ for 
   $t \in \mathcal{T}_g$. 
   \rev{Assumption~\ref{asm:ig} rules out unobserved variables that affect potential outcomes and 
   either treatment adoption or tie formation. It is a selection-on-observables assumption for the 
   \textit{joint} assignment. It is strictly stronger than the assumption maintained by approaches 
   that condition on the realized network, which we state next for comparison and discuss in 
   Section~\ref{sec:cond}.}

\begin{revblock}
   \begin{assumption}[Unconfoundedness of treatment given the realized network] \label{asm:ig2}
   $ \forall g \in \mathcal{G}, \{ig\} \in \mathcal{N}_g$ and $t \in \mathcal{T}_g$, 
   \begin{equation}\label{ass2b}
       Y_{ig} (t) \indep \bm{D}_g \mid \bm{X}_g, \bm{W}_g .
   \end{equation}
   \end{assumption} 
   Assumption~\ref{asm:ig2} is the condition used, in one form or another, by 
   \citet{forastiere2021identification}, \citet{jackson2020adjusting}, \citet{leung2022unconfoundedness} 
   and \citet{ogburn2024causal}. Assumption~\ref{asm:ig} implies Assumption~\ref{asm:ig2}: if 
   potential outcomes are independent of the pair $(\bm{D}_g, \bm{W}_g)$ given $\bm{X}_g$, they 
   are independent of $\bm{D}_g$ given $\bm{X}_g$ and $\bm{W}_g$ (the weak union property). The 
   converse fails: 
   Assumption~\ref{asm:ig2} lets unobserved node attributes shape both ties and potential outcomes, 
   provided they do not shape treatment adoption given $(\bm{X}_g, \bm{W}_g)$, because the realized 
   network then stands in for them. Section~\ref{sec:cond} spells out what one gives up, and what one 
   gains, by working under the stronger Assumption~\ref{asm:ig}.
\end{revblock}

   \section{A Generalized Propensity Score Based Approach}\label{sec:gps}

   The propensity score is the conditional probability of receiving a given treatment level given 
   observed confounders \citep{rosenbaum1983central}. When the treatment is multi-valued or 
   continuous, the same probability is called the generalized propensity score 
   (GPS, \citealp[]{hirano2004propensity}).

   This section uses propensity score methods to identify and estimate treatment effects from 
   social network data under Assumption~\ref{asm:ig}. Unlike existing approaches (e.g., 
   \citealp[]{sanchez2022network,forastiere2021identification}), which assume conditional 
   independence or exchangeability of treatment assignments, Assumption~\ref{asm:ig} allows 
   treatment assignments to be interdependent, because the exposure mapping takes as inputs the 
   treatment vector and the network, which all units in a group share.\footnote{Two things stop us 
   from modeling the assignment of the effective treatment directly: the covariate matrix 
   $\bm{X}_g$ it depends on is high dimensional, and the effective treatments of units in the same 
   group are interdependent because they share the network. We therefore model network formation 
   and treatment assignment and build the score from them (Section~\ref{sec:joint}).}
   
   \subsection{The generalized propensity score}

   The generalized propensity score of unit $\{ig\}$ for effective treatment $T_{ig} = t$, given the 
   covariate matrix, is
   \begin{equation}\label{gps}
      r(\{ig\}, t; \bm{X}_g) = Pr( T_{ig} = t | \bm{X}_g) = 
      Pr\left( e(\{ig\}, \bm{D}_g, \bm{W}_g) = t | \bm{X}_g\right). 
   \end{equation}
   We refer to it as the score $r$. It extends the usual propensity score, the probability that a 
   unit receives a given level of a binary treatment. In the binary case the score depends only on 
   the unit's own covariates, and units differ because their covariates differ. Here the score 
   depends on $\bm{X}_g$, which all units in group $g$ share, and unit-level heterogeneity enters 
   through the index $\{ig\}$. \rev{Approaches that condition on the realized network use instead 
   the \textit{network-conditional score} 
   $r^{c}(\{ig\}, t; \bm{X}_g, \bm{W}_g) = Pr(T_{ig} = t \mid \bm{X}_g, \bm{W}_g)$. The score 
   $r$ averages $r^{c}$ over the network given covariates, and we call it the \textit{marginal 
   score} when the contrast matters; Section~\ref{sec:cond} compares the two.}

   We impose the usual overlap condition.

   \begin{assumption}[Positivity of the generalized propensity score] \label{asm:p}
   $ \forall g \in \mathcal{G}, \{ig\} \in \mathcal{N}_g, 
   t \in \mathcal{T}_g,$ 
   \begin{equation}
   0 < r(\{ig\}, t; \bm{X}_g) < 1.  
   \end{equation}
   \end{assumption} 
 
Some levels of the effective treatment exist in one group but not in another: a unit in a group 
of four can have at most three treated neighbors, a unit in a group of ten up to nine. To keep 
groups comparable we restrict attention to levels $t$ with $t \in \mathcal{T}_g$ for all 
$g \in \mathcal{G}$\rev{, that is, levels that exist for every group size}.\footnote{An alternative 
is to restrict the analysis to the units for which positivity holds, as in 
\citet{forastiere2021identification}. \rev{Proposition~\ref{prop:idpos}(c) shows that under the 
marginal score this restriction depends on group size only, whereas under the network-conditional 
score it depends on each unit's realized degree.}}

   Given the score~(\ref{gps}) and Assumptions~\ref{asm:em} and~\ref{asm:ig}, we have the 
   following two results.

   \begin{proposition}[Balancing property of generalized propensity score]\label{prop:b}
   $ \forall g \in \mathcal{G}, \{ig\} \in \mathcal{N}_g, t \in \mathcal{T}_g,$
   \begin{equation}
   Pr\left(T_{ig} = t|\bm{X}_g, r(\{ig\}, t; \bm{X}_g)\right) = 
    Pr\left(T_{ig} = t| r(\{ig\}, t; \bm{X}_g)\right).     
   \end{equation}
   \end{proposition}    
   Proposition~\ref{prop:b} says that conditioning on the score is enough to remove the 
   confounding of both treatment assignment and network formation. The balancing property 
   implies unconfoundedness given the score,
   \begin{equation}
   Y_{ig} (t) \indep \mathbbm{1}\{T_{ig}=t\} | r(\{ig\}, t; \bm{X}_g),  
   \end{equation}
   and, together with Assumption~\ref{asm:p}, identification of the average potential outcome.

   \begin{proposition}[Identification of the average potential outcome]\label{prop:a}
   $ \forall g \in \mathcal{G}, \{ig\} \in \mathcal{N}_g, t \in \mathcal{T}_g,$
   \begin{equation}\label{ind}
   \begin{gathered}
   \mathbbm{E}\Big\{ \frac{Y_{ig}\mathbbm{1}\{ T_{ig} = t \} }{ r(\{ig\}, t; \bm{X}_g)} \,\Big|\, \bm{X}_g \Big\} = 
   \mathbbm{E}\{Y_{ig}(t) \mid \bm{X}_g\}, 
   \qquad \text{and hence} \\
   \frac{\mathbbm{E}\big\{ \sum_{\{ig\} \in \mathcal{N}_g} Y_{ig}\mathbbm{1}\{ T_{ig} = t \} / r(\{ig\}, t; \bm{X}_g) \big\}}{\mathbbm{E}[N_g]} = \mu(t).    
   \end{gathered}
   \end{equation}
   \end{proposition}

\begin{revblock}
   \subsection{To condition on the network or not: estimands, assumptions, and trade-offs}\label{sec:cond}

   The literature on interference in observational studies defines effects conditional on the 
   realized network; our estimand~(\ref{ate}) does not. This subsection makes the difference 
   precise, because it is the source of both the benefits and the costs of the proposed approach.

   \paragraph{Three estimands.} Think of $(\bm{D}_g, \bm{W}_g)$ as a joint assignment with two 
   components and of $T_{ig}$ as the dose it delivers. With a binary treatment, the average 
   treatment effect leaves the assignment unconditioned and the effect on the treated conditions 
   on its realization. With a two-component assignment there is an intermediate option: condition 
   on the network component only. We distinguish
   \begin{align}
     \tau(t, t') &= \mathbbm{E}[Y_{ig}(t) - Y_{ig}(t')] 
        && \text{(population effect)}, \label{est:pop} \\
     \tau^{S}_g(t, t') &= \frac{1}{N_g} \sum_{\{ig\} \in \mathcal{N}_g} 
        \mathbbm{E}[Y_{ig}(t) - Y_{ig}(t') \mid \bm{X}_g, \bm{W}_g] 
        && \text{(realized-sample effect)}, \label{est:sample}\\
     \tau(t, t' \mid \mathcal{A}) &= \mathbbm{E}[Y_{ig}(t) - Y_{ig}(t') \mid \mathcal{A}_{ig}(t, t')] 
        && \text{(feasible-population effect)}, \label{est:net}
   \end{align}
   where $\mathcal{A}_{ig}(t, t') = \{ r^{c}(\{ig\}, t; \bm{X}_g, \bm{W}_g) \, r^{c}(\{ig\}, t'; \bm{X}_g, \bm{W}_g) > 0 \}$ 
   is the event that unit $\{ig\}$ can reach both exposure levels given its realized ties (we 
   suppress its dependence on $(i, g, t, t')$ in the notation $\tau(t, t' \mid \mathcal{A})$). For 
   the exposure mapping~(\ref{epp}), 
   $\mathcal{A}_{ig}((d, z), (d', z')) = \{\deg_{ig}(\bm{W}_g) \geq \max(z, z')\}$ when treatment 
   probabilities lie in $(0, 1)$: a unit with one tie cannot have two treated neighbors.

   The realized-sample effect~(\ref{est:sample}) is what fixed-network approaches target. They 
   take the units, their covariates and their ties as given 
   \citep{forastiere2021identification, leung2022causal}, and design-based approaches such as 
   \citet{aronow2017estimating} fix the potential outcomes as well, targeting 
   $N_g^{-1} \sum_i [Y_{ig}(t) - Y_{ig}(t')]$. The realized-sample effect is the sample-average 
   analogue of the population effect~(\ref{est:pop}). The feasible-population effect~(\ref{est:net}) 
   is what those approaches can identify when some units cannot reach an exposure level, as in 
   \citet{forastiere2021identification}, who restrict attention to units for which positivity 
   holds. It conditions on a realized component of the assignment, having enough ties, and in this 
   sense it is the analogue of the effect on the treated.

   One can go one step further and condition on the realized exposure itself. We call
   \begin{equation}\label{est:exp}
   \tau^{e}(t, t') = \mathbbm{E}[Y_{ig}(t) - Y_{ig}(t') \mid T_{ig} = t]
   \end{equation}
   the \textit{effect on the exposed}; it is the exact analogue of the effect on the treated. Under 
   Assumption~\ref{asm:ig} it is identified by weighting units with $T_{ig} = t'$ by the odds 
   $r(\{ig\}, t; \bm{X}_g) / r(\{ig\}, t'; \bm{X}_g)$, as the effect on the treated is identified by 
   odds weights. It is the estimand we use for the direct effects in the application, where 
   adoption is rare. A fourth object, 
   $\mathbbm{E}[Y_{ig}(t) - Y_{ig}(t') \mid \bm{W}_g = \bm{w}_g]$, conditions on the network but 
   integrates over covariates. No method targets it, but it helps intuition, since it is the 
   effect among units ``as connected in $\bm{w}_g$,'' and part (c) of the next proposition covers it.

   \paragraph{How the estimands are related.} Let 
   $\tau(t, t' \mid \bm{X}_g) = \mathbbm{E}[Y_{ig}(t) - Y_{ig}(t') \mid \bm{X}_g] = 
   \mu(\{ig\}, t, \bm{X}_g) - \mu(\{ig\}, t', \bm{X}_g)$ be the conditional average effect given 
   covariates, which may vary across the units of a group.

   \begin{proposition}[Realized-sample, feasible-population and population effects]\label{prop:dec}
   Under Assumptions~\ref{asm:em} and~\ref{asm:ig}, for all $g$, $\{ig\} \in \mathcal{N}_g$ and 
   $t, t' \in \mathcal{T}_g$,
   \begin{enumerate}[label=(\alph*)]
   \item $\tau^{S}_g(t, t') = N_g^{-1} \sum_{\{ig\} \in \mathcal{N}_g} \tau(t, t' \mid \bm{X}_g)$. It does 
   not depend on $\bm{W}_g$, and $\mathbbm{E}[N_g \tau^{S}_g(t, t')] / \mathbbm{E}[N_g] = \tau(t, t')$; 
   in particular $\mathbbm{E}[\tau^{S}_g(t, t')] = \tau(t, t')$ when group size is constant or 
   independent of $\tau^{S}_g(t, t')$.
   \item $\tau(t, t' \mid \mathcal{A}) = \mathbbm{E}[\tau(t, t' \mid \bm{X}_g) \mid \mathcal{A}_{ig}(t, t')]$ 
   whenever $Pr(\mathcal{A}_{ig}(t, t')) > 0$. Hence $\tau(t, t' \mid \mathcal{A}) = \tau(t, t')$ if 
   either (i) $\mathbbm{1}\{\mathcal{A}_{ig}(t, t')\} \indep \bm{X}_g$, or (ii) $\tau(t, t' \mid \bm{X}_g)$ 
   does not depend on $\bm{X}_g$.
   \item $\mathbbm{E}[Y_{ig}(t) - Y_{ig}(t') \mid \bm{W}_g = \bm{w}_g] = 
   \mathbbm{E}[\tau(t, t' \mid \bm{X}_g) \mid \bm{W}_g = \bm{w}_g]$ for every $\bm{w}_g$ with positive 
   probability; it equals $\tau(t, t')$ for all $\bm{w}_g$ if $\bm{W}_g \indep \bm{X}_g$ or under (ii).
   \end{enumerate}
   \end{proposition}

   Proposition~\ref{prop:dec}, proved in Appendix~\ref{app:dec}, makes two points that are easily 
   run together. By part (a), under Assumption~\ref{asm:ig} the realized-sample effect 
   differs from the population effect only as a sample average differs from a population 
   average: once covariates are fixed, the realized ties carry no information about potential 
   outcomes, and homophily plays no role. Homophily enters in part (b). When the 
   network-conditional approach has to restrict attention to units with enough ties, the 
   estimand it identifies is an average of covariate-specific effects over the well-connected 
   sub-population, just as the effect on the treated averages over the covariate distribution of 
   the treated. 

   In the motivating example, residents with at least two ties are older, more often male and 
   more often leaders than the others. If the effect of adoption varies with these 
   characteristics, the feasible-population direct effect at $z = 2$ is an effect for that 
   sub-population and differs from~(\ref{est:pop}) by a composition term. Part (c) is the 
   familiar fact that conditioning on a variable that predicts the assignment but not the outcome 
   reweights the average effect by that variable; it coincides with the population effect only 
   when ties are independent of covariates, that is, without observed homophily. Which of the 
   population and the feasible-population effects is the quantity of interest depends on the 
   question: what the program did among the residents who could have been exposed, or what it 
   would do for a resident drawn from villages whose ties have not been surveyed 
   \citep{li2022random, leung2022causal}.

   \paragraph{Identification and positivity.} Conditioning on the network makes the score zero 
   for every unit whose realized degree is below the exposure level; averaging over the network 
   does not. The next proposition records the consequences for what each score identifies.

   \begin{proposition}[Identification with and without conditioning on the network]\label{prop:idpos}
   Suppose Assumption~\ref{asm:em} holds.
   \begin{enumerate}[label=(\alph*)]
   \item Under Assumption~\ref{asm:ig2}, $\mathbbm{E}[Y_{ig}(t) \mid \bm{X}_g, \bm{W}_g] = 
   \mathbbm{E}[Y_{ig} \mid T_{ig} = t, \bm{X}_g, \bm{W}_g]$ whenever 
   $r^{c}(\{ig\}, t; \bm{X}_g, \bm{W}_g) > 0$. Hence $\tau^{S}_g(t, t')$ is identified if every unit in 
   the group has $\mathcal{A}_{ig}(t, t')$; otherwise what is identified is the feasible-population 
   effect $\tau(t, t' \mid \mathcal{A})$, and more generally $\mathbbm{E}[Y_{ig}(t) \mid r^{c}(\{ig\}, t; \bm{X}_g, \bm{W}_g) > 0]$, 
   through $\mathbbm{E}[Y_{ig} \mathbbm{1}\{T_{ig} = t\} / r^{c} \mid r^{c} > 0]$.
   \item Under Assumption~\ref{asm:ig}, which implies Assumption~\ref{asm:ig2}, the marginal score 
   identifies the population average potential outcome, 
   $\mathbbm{E}[Y_{ig} \mathbbm{1}\{T_{ig} = t\} / r(\{ig\}, t; \bm{X}_g)] = \mu(t)$ 
   (Proposition~\ref{prop:a}), and the network-conditional score identifies 
   $\mathbbm{E}[Y_{ig}(t) \mid r^{c} > 0]$, which in general equals $\mu(t)$ only if $r^{c} > 0$ almost surely.
   \item For the exposure mapping~(\ref{epp}), 
   $r^{c}(\{ig\}, (d, z); \bm{X}_g, \bm{W}_g) = 0$ whenever $z > \deg_{ig}(\bm{W}_g)$, so 
   $r^{c} > 0$ almost surely fails for every $z \geq 1$ under any dyadic model with link probabilities 
   below one. In contrast, if the joint assignment has full support given covariates, i.e., 
   $Pr(\bm{D}_g = \bm{d}_g, \bm{W}_g = \bm{w}_g \mid \bm{X}_g) > 0$ for all $(\bm{d}_g, \bm{w}_g)$, 
   then $r(\{ig\}, (d, z); \bm{X}_g) > 0$ for every $d \in \{0, 1\}$ and $z \in \{0, \ldots, N_g - 1\}$.
   \end{enumerate}
   \end{proposition}

   Part (a) is the identification result of \citet{forastiere2021identification}; identification is 
   meant in the population-of-groups sense, since with continuous covariates the conditional 
   expectations cannot be recovered from a single group. Part (b) records that under 
   Assumption~\ref{asm:ig} both scores are valid weights, but for different estimands unless the 
   network-conditional score is positive for everyone. 

   Part (c) is where the two routes part ways in practice. Under the network-conditional approach, 
   the effect of having $z$ treated neighbors is undefined for units with fewer than $z$ 
   neighbors. The estimand is restricted by realized degree, a post-assignment variable, and the 
   restriction grows with $z$. Under the marginal approach the requirement is that every covariate 
   profile reach $(d, z)$ with positive probability (Assumption~\ref{asm:p}), a statement about the 
   joint assignment given covariates and not about realized ties. Concretely, under the dyadic 
   model of Section~\ref{sec:joint} both scores are probabilities that $z$ of a unit's neighbors 
   are treated; the network-conditional score counts over the realized neighbors, the marginal 
   score over all other units in the group weighted by their link probabilities. 

   How does the marginal approach handle a unit whose realized degree is below $z$? That unit 
   contributes nothing to the weighting part of the AIPW estimator at level $(d, z)$, since its 
   exposure indicator is zero. Its potential outcome at that level is recovered from units with 
   similar covariates who did reach the level, through the outcome model and those units' weights. 
   The formation model weights the units that reached the level correctly; it does not 
   extrapolate. The extrapolation runs through covariates and the outcome model, as in any 
   weighting estimator for a unit whose realized exposure is uncommon given its covariates, and 
   Assumption~\ref{asm:ig} licenses it: the ties a unit did not form carry no information about 
   its potential outcomes given covariates. Large weights belong to units that reached an exposure 
   level that is rare given their covariates, the usual overlap problem. Assumption~\ref{asm:mom}(b) 
   below and the trimming in Section~\ref{ea} address it, and trimming brings back a 
   sub-population restriction of a different kind, by covariate profile rather than by realized 
   degree \citep{crump2009dealing}. Full support holds, for instance, under Assumption~\ref{asm:ce} 
   with link and treatment probabilities in $(0, 1)$.

   \paragraph{Efficiency.} Under Assumption~\ref{asm:ig}, $\bm{W}_g$ affects the exposure $T_{ig}$ but 
   not the potential outcomes once $\bm{X}_g$ is held fixed. The propensity score literature has long 
   known that including such a variable in the score inflates the variance of weighting estimators 
   \citep{hahn1998role, hahn2004functional, brookhart2006variable, bhattacharya2007instrumental}. 
   The same logic applies to the network, provided the two estimators target the same estimand, 
   which by Proposition~\ref{prop:idpos} requires exposure levels that every unit can reach. Let 
   $m_t(\bm{X}_g) = \mathbbm{E}[Y_{ig}(t) \mid \bm{X}_g]$ and 
   $\sigma^2_t(\bm{X}_g) = \text{var}(Y_{ig}(t) \mid \bm{X}_g)$. The two augmented inverse probability 
   weighting estimators of $\tau(t, t')$ introduced in Section~\ref{sec:est} differ only in which 
   score sits in the denominator of their influence functions,
   \begin{equation}\label{inf2}
   \begin{aligned}
   \varphi^{M}_{ig} &= \Big[\frac{\mathbbm{1}\{T_{ig} = t\}(Y_{ig} - m_t(\bm{X}_g))}{r(\{ig\}, t; \bm{X}_g)} + m_t(\bm{X}_g)\Big] 
   - \Big[\frac{\mathbbm{1}\{T_{ig} = t'\}(Y_{ig} - m_{t'}(\bm{X}_g))}{r(\{ig\}, t'; \bm{X}_g)} + m_{t'}(\bm{X}_g)\Big] - \tau(t,t'), \\
   \varphi^{C}_{ig} &= \text{the same with } r^{c}(\{ig\}, \cdot; \bm{X}_g, \bm{W}_g) \text{ in place of } r(\{ig\}, \cdot; \bm{X}_g).
   \end{aligned}
   \end{equation}

   \begin{proposition}[Efficiency of marginalizing over the network]\label{prop:eff}
   Suppose Assumptions~\ref{asm:em} and~\ref{asm:ig} hold, potential outcomes are independent across 
   units within a group given $\bm{X}_g$, i.e., $Y_{ig}(\cdot) \indep Y_{jg}(\cdot) \mid \bm{X}_g$ for 
   $i \neq j$, and for $s \in \{t, t'\}$ both $r(\{ig\}, s; \bm{X}_g)$ and 
   $r^{c}(\{ig\}, s; \bm{X}_g, \bm{W}_g)$ are bounded away from zero almost surely. Then 
   \begin{equation}\label{eff}
   \begin{aligned}
   &\text{var}\Big( \sum_{\{ig\} \in \mathcal{N}_g} \varphi^{C}_{ig} \Big) - 
   \text{var}\Big( \sum_{\{ig\} \in \mathcal{N}_g} \varphi^{M}_{ig} \Big) \\
   &\qquad = \sum_{s \in \{t, t'\}} \sum_{\{ig\} \in \mathcal{N}_g} 
   \mathbbm{E}\left[ \sigma_s^2(\bm{X}_g) \left\{ 
   \mathbbm{E}\left[ \frac{1}{r^{c}(\{ig\}, s; \bm{X}_g, \bm{W}_g)} \,\Big|\, \bm{X}_g \right] 
   - \frac{1}{r(\{ig\}, s; \bm{X}_g)} \right\} \right] \geq 0,
   \end{aligned}
   \end{equation}
   with strict inequality whenever $\sigma_s^2(\bm{X}_g) > 0$ and $r^{c}(\{ig\}, s; \bm{X}_g, \bm{W}_g)$ 
   varies with $\bm{W}_g$ given $\bm{X}_g$ on a set of positive probability.
   \end{proposition}

   The proof (Appendix~\ref{app:eff}) applies Jensen's inequality to $\mathbbm{E}[r^{c} \mid \bm{X}_g] = r$. 
   The result comes with qualifications. The hypothesis that $r^{c}$ is bounded away from zero 
   excludes, for the exposure mapping~(\ref{epp}), every contrast involving $z \geq 1$ 
   (Proposition~\ref{prop:idpos}(c)). The proposition applies to contrasts at $z = 0$, such as 
   the direct effect among units with no treated neighbor, and to mappings under which every 
   unit can reach every level, for instance an indicator of ``any treated neighbor'' among units 
   with at least one tie, or a share. For levels that some units cannot reach, the two estimators 
   target different estimands and an efficiency comparison is meaningless. The gain comes from 
   the exclusion restriction that $\bm{W}_g$ does not enter the outcome given $\bm{X}_g$, which is 
   what Assumption~\ref{asm:ig} adds to Assumption~\ref{asm:ig2}; an analyst unwilling to make 
   that assumption gives up the gain along with the population estimand. Finally, the comparison 
   is between two AIPW estimators with the same outcome model and says nothing about regression 
   estimators.

   \paragraph{Summary of the trade-off.} Table~\ref{tab:tradeoff} collects the comparison. 
   Conditioning on the realized network needs less and identifies a realized-sample effect at 
   exposure levels that every unit can reach, or a feasible-population effect otherwise. 
   Averaging over the network needs Assumption~\ref{asm:ig} and a formation model. In exchange 
   it delivers a population estimand, effects at levels that poorly connected units cannot reach, 
   and, at levels every unit can reach, a weakly smaller variance. The two gains do not add up: 
   where the first one bites, the second does not apply. The population estimand also remains 
   well defined when ties respond to treatment, a case in which conditioning on $\bm{W}_g$ would 
   amount to post-treatment adjustment and the score has to be computed under the 
   endogenous-network factorization of Section~\ref{sec:interdep} 
   \citep{comola2021treatment, gao2024endogenous}. Latent-variable network models offer a middle 
   course: node attributes $\bm{U}_g$ are estimated from the observed ties, Assumption~\ref{asm:ig} 
   is imposed given $(\bm{X}_g, \bm{U}_g)$, and the score is averaged over ties given 
   $(\bm{X}_g, \bm{U}_g)$ (Section~\ref{sec:latent}). Section~\ref{SMC} and Appendix~\ref{abc} 
   compare the two routes in simulations, and Section~\ref{ea} compares them in the application.

   \begin{table}[h]
   \centering
   \caption{Conditioning on the realized network versus averaging over it}\label{tab:tradeoff}
   \footnotesize
   \begin{tabular}{p{3.6cm} p{5.3cm} p{5.3cm}}
   \hline\hline
    & Condition on $\bm{W}_g$ & Average over $\bm{W}_g$ \\
   \hline
   Score & network-conditional $r^{c}$ & marginal $r$ \\
   Identifying assumption & Assumption~\ref{asm:ig2} (treatment unconfounded given $\bm{X}_g, \bm{W}_g$) & Assumption~\ref{asm:ig} (treatment and ties unconfounded given $\bm{X}_g$) \\
   Model of tie formation & not needed & needed \\
   Latent homophily that does not affect adoption & tolerated & biases the score unless latent attributes are estimated (Section~\ref{sec:latent}) \\
   Estimand & realized-sample effect, or feasible-population effect when some units cannot reach a level (Proposition~\ref{prop:dec}) & population effect \\
   Levels poorly connected units cannot reach & undefined for them & identified by extrapolation from units with similar covariates (Proposition~\ref{prop:idpos}) \\
   Variance at levels every unit can reach & larger & weakly smaller (Proposition~\ref{prop:eff}) \\
   Ties that respond to treatment & post-treatment adjustment & estimand well defined; score from Section~\ref{sec:interdep} \\
   \hline\hline
   \end{tabular}
   \end{table}
\end{revblock}

  \subsection{Joint modeling of treatment assignment and network formation under common exposure}\label{sec:joint}

  For a binary treatment the propensity score is usually estimated with a logit or probit model for 
  treatment assignment. Doing the same for the generalized propensity score~(\ref{gps}) would mean 
  fitting a model with the high-dimensional input $\bm{X}_g$. We take a different route and model 
  the assignment of treatment and the formation of the network jointly. Under 
  Assumption~\ref{asm:ig} the two processes can be represented by a probabilistic model 
  $f_{\bm{D}_g, \bm{W}_g|\bm{X}_g} (\bm{d}_g, \bm{w}_g)$ given the covariate matrix $\bm{X}_g$, and 
  the score~(\ref{gps}) can be rewritten as
  \begin{equation}\label{gps2}
  \begin{aligned}
   r(\{ig\}, t; \bm{X}_g) &= Pr( e(\{ig\}, \bm{D}_g, \bm{W}_g) = t | \bm{X}_g) \\
  &= \int_{\bm{d}_g \in \mathcal{D}_g} \int_{\bm{w}_g \in \mathcal{W}_g} 
  \mathbbm{1} \{ e(\{ig\}, \bm{d}_g, \bm{w}_g) = t \}  f_{\bm{D}_g, \bm{W}_g|\bm{X}_g} (\bm{d}_g, \bm{w}_g) 
  \diff \bm{d}_g \diff \bm{w}_g .  
  \end{aligned}
  \end{equation}
  With a model for the joint distribution of the treatment vector and the network and a given 
  exposure mapping, the score follows from~(\ref{gps2}). 

   To model $f_{\bm{D}_g, \bm{W}_g \mid \bm{X}_g} (\bm{d}_g, \bm{w}_g)$ we factorize the treatment 
   assignment mechanism and the network formation process under further assumptions. \rev{The main 
   development maintains the following model, in which both processes are driven by covariates but 
   neither responds to the other, a mechanism \citet{franzese2008contagion} call ``common 
   exposure.'' Interdependent processes are treated in Section~\ref{sec:interdep}.}

\begin{revblock}
   \begin{assumption}[Common exposure with dyadic link formation]\label{asm:ce}
   For all $g \in \mathcal{G}$, 
   \begin{enumerate}[label=(\alph*)]
   \item (Common exposure) $\bm{D}_g \indep \bm{W}_g \mid \bm{X}_g$;
   \item (Individualized treatment) $f_{\bm{D}_g|\bm{X}_g}(\bm{d}_g) = 
   \prod_{\{ig\} \in \mathcal{N}_g} Pr(D_{ig} = d_{ig} \mid \bm{X}_g)$, where 
   $Pr(D_{ig} = 1 \mid \bm{X}_g)$ may depend on $X^D_{ig}$ and on contextual variables;
   \item (Dyadic link formation) there exist node-level variables $\bm{U}_g = (U_{1g}, \ldots, U_{N_g g})$, 
   possibly latent, with $\bm{U}_g \indep \bm{D}_g \mid \bm{X}_g$, such that conditional on 
   $(\bm{X}_g, \bm{U}_g)$ the links $\{W_{ijg}\}_{i < j}$ (all ordered pairs $i \neq j$ for a directed 
   network) are mutually independent with 
   $Pr(W_{ijg} = 1 \mid \bm{X}_g, \bm{U}_g) = p_{ijg}(\bm{X}_g, \bm{U}_g)$. When $\bm{U}_g$ is empty, 
   $p_{ijg}$ depends on $(X^W_{ig}, X^W_{jg})$ only.
   \end{enumerate}
   \end{assumption}
\end{revblock}
   \rev{Under Assumption~\ref{asm:ce}(a)} the joint distribution factorizes as
   \begin{equation}\label{cee}
       f_{\bm{D}_g, \bm{W}_g|\bm{X}_g}(\bm{d}_g, \bm{w}_g) = 
       f_{\bm{D}_g|\bm{X}_g}(\bm{d}_g) f_{\bm{W}_g|\bm{X}_g}(\bm{w}_g)
       = f_{\bm{D}_g|\bm{X}^D_g}(\bm{d}_g) f_{\bm{W}_g|\bm{X}^W_g}(\bm{w}_g),
   \end{equation}
   \rev{where the second equality holds when treatment depends on covariates only through 
   $\bm{X}^D_g$ and ties only through $\bm{X}^W_g$; only the first equality is needed.}

  Once the joint distribution is modeled and the exposure mapping is fixed, the generalized 
  propensity score of each unit follows. \rev{Under Assumption~\ref{asm:ce}(b) and (c) without latent 
  variables}, the factorization in~(\ref{cee}) becomes:\footnote{The dyadic model is the simplest 
  illustration of covariates $\bm{X}^W_g$ as confounders; \citet{shalizi2011homophily} use it in 
  simulations of homophily and contagion. Reciprocity, transitivity or clustering can be built in 
  with exponential random graph models \citep{cranmer2011inferential} or other models of 
  interdependent link formation (e.g., \citealp[]{chyzh2019local})\rev{, with the score obtained by 
  simulation as in Section~\ref{sec:interdep}}.}
   $$
   \begin{gathered}
    f_{\bm{D}_g|\bm{X}^D_g}(\bm{d}_g) = \prod_{\{ig\} \in \mathcal{N}_g}Pr (D_{ig} = d_{ig}|X_{ig}^D), \\
    f_{\bm{W}_g|\bm{X}^W_g}(\bm{w}_g) = \prod_{\substack{\{ig\}, \{jg\} \in \mathcal{N}_g  \\ i < j}} 
    Pr (W_{ijg} = w_{ijg}|X_{ig}^W, X_{jg}^W)
   \end{gathered}
   $$ 
   \rev{for an undirected network (for a directed network the product runs over ordered pairs).} 
   Take the exposure mapping~(\ref{epp}), under which the potential outcome depends on own treatment 
   and on the number of treated neighbors. For an unweighted network, 
   (\ref{gps2}) becomes
   \begin{equation}\label{fact}
   \begin{aligned}
   &r (\{ig\}, (d, z); \bm{X}_g) \\  = &Pr\left[e(\{ig\}, \bm{D}_g, \bm{W}_g) = (d, z) | \bm{X}_g\right] \\
                      = &Pr\left[D_{ig} = d, \sum_{\{jg\} \in \mathcal{N}_g } W_{ijg}D_{jg} = z 
                      | \bm{X}_g\right] \\
                      = &Pr(D_{ig} = d|X_{ig}^D) \times \\ 
                     & \sum_{\substack{d_{jg} \in \{0, 1\}, w_{ijg} \in \{0, 1\} \\ \{jg\} \neq \{ig\}}} 
                     \mathbbm{1}\{ \sum_{\substack{\{jg\} \in \mathcal{N}_g \\ \{jg\} \neq \{ig\}}} 
                      w_{ijg}d_{jg} = z \}  
                      \prod_{\substack{\{jg\} \in \mathcal{N}_g \\ \{jg\} \neq \{ig\}}}
                      Pr(D_{jg} = d_{jg}|X_{jg}^D)Pr (W_{ijg} = w_{ijg}|X_{ig}^W, X_{jg}^W),
   \end{aligned}    
   \end{equation}
   which combines the model for treatment assignment with the model for network formation. 
   \rev{The second factor in~(\ref{fact}) is the probability mass function of a Poisson-binomial 
   variable: given $\bm{X}_g$, $Z_{ig}$ is a sum of $N_g - 1$ independent Bernoulli variables 
   $W_{ijg}D_{jg}$ with success probabilities 
   $q_{ijg} = Pr(W_{ijg} = 1 \mid X^W_{ig}, X^W_{jg}) Pr(D_{jg} = 1 \mid X^D_{jg})$. It can be computed 
   exactly by convolution, and it depends on the network only through the link 
   \textit{probabilities}, never through the realized ties.}

\begin{revblock}
   \paragraph{Operational contrast with the network-conditional score.} Under 
   Assumption~\ref{asm:ce}, the network-conditional score of Section~\ref{sec:cond} factorizes as 
   $r^{c}(\{ig\}, (d, z); \bm{X}_g, \bm{W}_g) = Pr(D_{ig} = d \mid X^D_{ig}) \, 
   Pr(Z_{ig} = z \mid \bm{X}_g, \bm{W}_g)$, where the second factor is the probability that exactly 
   $z$ of the \textit{realized} neighbors of $\{ig\}$ are treated, a Poisson-binomial with success 
   probabilities $\{Pr(D_{jg} = 1 \mid X^D_{jg}) : W_{ijg} = 1\}$ 
   \citep{forastiere2021identification}. The two scores differ in exactly one place: whether the 
   realized links $W_{ijg}$ or their probabilities enter the success probabilities of the neighbors' 
   treatment indicators. This is also where positivity differs (Proposition~\ref{prop:idpos}(c)): 
   the network-conditional score sums over at most $\deg_{ig}(\bm{W}_g)$ Bernoulli variables, the marginal 
   score over $N_g - 1$. In practice the network-conditional score is often estimated with a parametric 
   model for $Z_{ig}$ given covariates and network summaries \citep{forastiere2021identification}.

   \subsubsection{Latent node heterogeneity}\label{sec:latent}

   Dyadic independence given observed covariates alone is restrictive. Real networks show degree 
   heterogeneity, transitivity and clustering that observed homophily terms do not reproduce, and 
   Assumption~\ref{asm:ce}(c) accommodates them through the node-level variables $U_{ig}$. In the 
   additive and multiplicative effects (AME) model \citep{hoff2021additive} used in 
   Section~\ref{ea}, for example, 
   \begin{equation}\label{ame}
   Pr(W_{ijg} = 1 \mid \bm{X}_g, \bm{U}_g) = 
   \Phi\left( |X_{ig} - X_{jg}|' \beta_{w1, g} + X_{ijg}' \beta_{w2, g} + \alpha_{ig} + \alpha_{jg} + 
   \tilde{U}_{ig}' \tilde{U}_{jg} \right), 
   \qquad U_{ig} = (\alpha_{ig}, \tilde{U}_{ig}),
   \end{equation}
   where $|X_{ig} - X_{jg}|$ collects absolute differences in node covariates (observed homophily), 
   $X_{ijg}$ collects dyadic covariates such as indicators that at least one of the two nodes has a 
   given attribute, $(\beta_{w1, g}, \beta_{w2, g})$ are group-specific coefficients, $\alpha_{ig}$ 
   captures degree heterogeneity \citep{graham2017econometric}, and the inner product of the latent 
   positions $\tilde{U}_{ig}$ captures latent homophily and clustering \citep{hoff2002latent}. Given 
   $(\bm{X}_g, \bm{U}_g)$ the links are independent, so the score is~(\ref{fact}) with 
   $p_{ijg}(\bm{X}_g, \bm{U}_g)$ in place of $Pr(W_{ijg} = 1 \mid X^W_{ig}, X^W_{jg})$, integrated over 
   the distribution of $\bm{U}_g$:
   \begin{equation}\label{gpslat}
   r(\{ig\}, t; \bm{X}_g) = \int r(\{ig\}, t; \bm{X}_g, \bm{u}_g) \, \diff F_{\bm{U}_g \mid \bm{X}_g}(\bm{u}_g).
   \end{equation}
   How $\bm{U}_g$ is handled depends on the size of the groups, and only the first of the two 
   regimes below is covered by the results of this paper. 

   With many small groups, the asymptotic regime of Section~\ref{sec:est}, node-level parameters 
   cannot be estimated consistently, so $\bm{U}_g$ has to be treated as random effects whose 
   distribution $F_{\bm{U}_g \mid \bm{X}_g}$ has finitely many parameters. Then (\ref{gpslat}) 
   integrates over that distribution, the exact integral is a smooth function of the parameters 
   (its Monte Carlo approximation is discussed in Section~\ref{sec:interdep}), and 
   Lemma~\ref{lem:gps} and Proposition~\ref{ttt} apply. In this regime the score is a function of 
   $\bm{X}_g$ alone, so it is a valid weight only if Assumption~\ref{asm:ig} holds given 
   $\bm{X}_g$. Latent random effects make the formation model more realistic, but they do not 
   protect against latent attributes that also affect outcomes. 

   With large groups in which degrees grow with group size, as in the application (villages of 
   102 to 282 residents, one of 28), latent positions can be estimated node by node with uniform 
   consistency \citep{li2023statistical, mcfowland2023estimating}. One can then use the plug-in 
   score $r(\{ig\}, t; \bm{X}_g, \hat{\bm{U}}_g)$, which integrates over ties given the estimated 
   attributes instead of over $F_{\bm{U}_g \mid \bm{X}_g}$. The application uses this plug-in. 
   Because $\hat{\bm{U}}_g$ is a function of the realized network, the plug-in score is a 
   network-conditional object. Its target is $Pr(T_{ig} = t \mid \bm{X}_g, \bm{U}_g)$, which is a 
   valid weight for $\mu(t)$ when Assumption~\ref{asm:ig} holds with $(\bm{X}_g, \bm{U}_g)$ in 
   place of $\bm{X}_g$. If $\bm{U}_g$ affects only ties, conditioning on it is valid but, by the 
   logic of Proposition~\ref{prop:eff}, weakly less efficient than integrating over 
   $F_{\bm{U}_g \mid \bm{X}_g}$. If $\bm{U}_g$ also affects outcomes, conditioning on it is 
   required, and $\hat{\bm{U}}_g$ should enter the outcome model and, if relevant, the treatment 
   model. Replacing $\bm{U}_g$ by $\hat{\bm{U}}_g$ requires $\hat{\bm{U}}_g$ to be uniformly 
   consistent, which latent space models deliver only when degrees grow with $N_g$, a different 
   regime from the sparse one of Section~\ref{sec:single}. We do not prove a plug-in theorem.

   This plug-in is the middle course mentioned in the introduction. It conditions on the part of 
   the realized network that is informative about latent node attributes and averages over the 
   remaining, idiosyncratic variation in ties. It thus keeps the robustness to latent homophily 
   that motivates conditioning on the network, together with the population interpretation of 
   averaging over ties, subject to the qualifications of Section~\ref{sec:cond}.

   \paragraph{Specifying the models in practice.} The treatment model is a standard binary-response 
   model for $D_{ig}$ given the unit's own covariates, contextual variables such as neighborhood 
   averages of covariates when social norms are thought to matter, and group indicators or 
   group-level covariates. The network model is a dyadic model for $W_{ijg}$ whose systematic 
   part should contain
   \begin{itemize}
   \item homophily terms, typically absolute differences or indicators of agreement in the 
   covariates that also predict the outcome;
   \item node-level attributes that shift the propensity to form ties, such as leadership and 
   phone ownership in the application;
   \item dyadic covariates such as distance; and
   \item latent node effects as in~(\ref{ame}) when the observed terms fail to reproduce degree 
   heterogeneity or clustering.
   \end{itemize}
   Several diagnostics are available. Since $\mathbbm{E}[r(\{ig\}, t; \bm{X}_g)] = Pr(T_{ig} = t)$, 
   the estimated score averaged over units should reproduce the marginal distribution of the 
   effective treatment; Figure~\ref{eme} in Appendix~\ref{C} reports this check for the 
   application. Posterior predictive or parametric bootstrap checks of the degree distribution 
   and of the number of triangles show whether latent effects are needed. Covariate balance 
   across exposure levels after weighting by $1 / \hat{r}$, summarized by standardized mean 
   differences, checks the score as a whole. The distribution of $\hat{r}$ should also be 
   inspected for extreme values, which we trim in the application.
\end{revblock}
  
  \subsection{Estimation and inference}\label{sec:est}

  So far we have defined the generalized propensity score $r(\{ig\}, t; \bm{X}_g)$, shown its 
  balancing property, and shown that average potential outcomes, and therefore treatment effects, 
  are identified. We now turn to estimation and inference. We work with conditional effects; 
  the average effects are weighted averages of them. 

  Once the models for treatment assignment and network formation are estimated, the score 
  $\hat{r}(\{ig\}, t; \bm{X}_g)$ follows from the exposure mapping, and the average potential outcome 
  $\hat{\mu}(t)$ can be estimated by inverse probability weighting 
  \citep{horvitz1952generalization}. For efficiency we use the augmented inverse probability 
  weighting (AIPW) estimator with an outcome model $\mu(\{ig\}, T, \bm{X}_g)$:
   \begin{equation}\label{dr}
      \hat{\mu}(t)^{AIPW} = \frac{1}{N}\sum_{g \in \mathcal{G}}\sum_{\{ig\} \in \mathcal{N}_g} 
      \mathbbm{1}\{ T_{ig} = t \} 
      \frac{Y_{ig} - \hat{\mu}(\{ig\}, t, \bm{X}_g) }{\hat{r}(\{ig\}, t; \bm{X}_g)} + 
      \hat{\mu}(\{ig\}, t, \bm{X}_g).   
   \end{equation}
   
   The AIPW estimator~(\ref{dr}) is doubly robust: it is consistent if either the generalized 
   propensity score or the outcome model is correctly specified \citep{glynn2010introduction}. 
   Because $\bm{X}_g$ is high dimensional, $\hat{\mu}(\{ig\}, t, \bm{X}_g)$ can be fitted with a 
   machine learning method such as random forests, or with a linear model in which the unit's own 
   covariates replace $\bm{X}_g$. The AIPW estimator of the effect of level $t$ relative to level 
   $t'$ is
   \begin{equation}\label{eatt}
   \hat{\tau}(t, t') = \hat{\mu}(t)^{AIPW} - \hat{\mu}(t')^{AIPW},
   \end{equation}
   which we use throughout the paper. With the exposure mapping~(\ref{epp}), conditional direct 
   and spillover effects correspond to specific levels: for the conditional spillover effect 
   $\delta(z, d) = \mu(d, z) - \mu(d, 0)$, set $t = (d, z)$ and $t' = (d, 0)$, so that 
   $\hat\delta(z, d) = \hat{\tau}(t, t')$. The average spillover effect is estimated as 
   $\hat\delta(d) = \sum_z \hat\delta(z, d) \hat f_Z(z)$, where $\hat f_Z(z)$ is the empirical 
   distribution of $Z_{ig}$ \citep{forastiere2021identification}. For inference we derive the 
   asymptotic distribution of~(\ref{eatt}). Define
   $$
   \begin{aligned}
       \tau_{ig}(t, t') = &[\mathbbm{1}\{ T_{ig} = t \} 
     \frac{Y_{ig} - \mu(\{ig\}, t, \bm{X}_g) }{r(\{ig\}, t; \bm{X}_g)} + \mu(\{ig\}, t, \bm{X}_g)]\\
    & - [\mathbbm{1}\{ T_{ig} = t' \} 
     \frac{Y_{ig} - \mu(\{ig\}, t', \bm{X}_g) }{r(\{ig\}, t'; \bm{X}_g)} + \mu(\{ig\}, t', \bm{X}_g)]
   \end{aligned}
   $$
   and $ \varphi_{ig}(t, t') = \tau_{ig}(t, t') - \tau(t, t') $\rev{, with 
   $\mu(\{ig\}, t, \bm{X}_g) = \mathbbm{E}[Y_{ig}(t) \mid \bm{X}_g]$, and let 
   $S_g(t, t') = \sum_{\{ig\} \in \mathcal{N}_g} \varphi_{ig}(t, t')$ be the group-level sum of the 
   influence functions. In the appendix the two bracketed terms of $\tau_{ig}(t, t')$ are written 
   $\psi_{t, ig}$ and $\psi_{t', ig}$.}

\begin{revblock}
   We study the estimator in a ``many-networks'' regime in which the number of groups grows while 
   group sizes stay bounded. This is the regime of the Monte Carlo studies, and it describes data 
   such as classrooms, villages or firms. We maintain the following conditions.

   \begin{assumption}[Sampling of groups]\label{asm:samp}
   (a) The group-level data $\{(\bm{X}_g, \bm{D}_g, \bm{W}_g, \{Y_{ig}(\cdot)\}_{\{ig\} \in \mathcal{N}_g})\}_{g=1}^G$ 
   are independent and identically distributed draws from NP-SEM~(\ref{all}); (b) 
   $1 \leq N_g \leq \bar{N} < \infty$ almost surely.
   \end{assumption}
   Asymptotics are taken as $G \to \infty$. Assumption~\ref{asm:samp}(a) allows arbitrary dependence 
   among the potential outcomes of units in the same group, for instance through a group-level 
   shock. The within-group conditional independence used in Proposition~\ref{prop:eff} is not 
   needed below.

   \begin{assumption}[Moments and overlap]\label{asm:mom}
   (a) There is $M < \infty$ such that $\mathbbm{E}[Y_{ig}(t)^4 \mid \bm{X}_g] \leq M$ almost surely 
   for all $t$, and $N^{-1} \sum_{\{ig\}} \hat{\mu}(\{ig\}, t, \bm{X}_g)^4 = O_p(1)$; (b) there is 
   $[\underline{\pi}, \bar{\pi}] \subset (0, 1)$ such that $r(\{ig\}, t; \bm{X}_g)$ and 
   $\hat{r}(\{ig\}, t; \bm{X}_g)$ lie in $[\underline{\pi}, \bar{\pi}]$ almost surely for all $t$ of 
   interest; (c) $\mathbbm{E}[S_g(t, t')^2] > 0$.
   \end{assumption}
   Part (b) is demanding when treatment is rare: for $t = (1, z)$ under~(\ref{epp}), 
   $r(\{ig\}, t; \bm{X}_g) \leq Pr(D_{ig} = 1 \mid \bm{X}_g)$, so a lower bound on the score is a lower 
   bound on the adoption probability. In practice the bound is enforced by trimming, which changes 
   the estimand to one for a trimmed population \citep{crump2009dealing}. Section~\ref{ea} returns 
   to this.

   \begin{assumption}[Nuisance estimators]\label{asm:rate}
   The estimators $\hat{r}$ and $\hat{\mu}$ satisfy
   \begin{equation}\label{rates}
   \begin{gathered}
   \|\hat{r} - r\|_N^2 = \frac{1}{N} \sum_{\{ig\} \in \mathcal{N}} \big(\hat{r}(\{ig\}, t; \bm{X}_g) - r(\{ig\}, t; \bm{X}_g)\big)^2 = o_p(1), \\
   \|\hat{\mu} - \mu\|_N^2 = \frac{1}{N} \sum_{\{ig\} \in \mathcal{N}} \big(\hat{\mu}(\{ig\}, t, \bm{X}_g) - \mu(\{ig\}, t, \bm{X}_g)\big)^2 = o_p(1), \\
   \|\hat{r} - r\|_N \, \|\hat{\mu} - \mu\|_N = o_p(N^{-1/2}),
   \end{gathered}
   \end{equation}
   for $t$ and $t'$, and in addition either 
   \begin{enumerate}[label=(\alph*)]
   \item (cross-fitting) they are obtained by $K$-fold cross-fitting across groups, each group's 
   nuisance values being computed from a fit that excludes the group; or
   \item (parametric nuisances) $\hat{r}(\{ig\}, t; \bm{X}_g) = r(\{ig\}, t; \bm{X}_g; \hat{\theta})$ and 
   $\hat{\mu}(\{ig\}, t, \bm{X}_g) = \mu(\{ig\}, t, \bm{X}_g; \hat{\beta})$ for finite-dimensional 
   parameters with $\hat{\theta} - \theta = O_p(N^{-1/2})$ and $\hat{\beta} - \beta = O_p(N^{-1/2})$, 
   where $r(\cdot; \theta)$ and $\mu(\cdot; \beta)$ are smooth in the parameters in the sense 
   stated in Appendix~\ref{app:a2}. In this case the bound on $\hat{r}$ in 
   Assumption~\ref{asm:mom}(b) is required with probability approaching one instead of almost 
   surely.
   \end{enumerate}
   \end{assumption}

   Assumption~\ref{asm:samp}(b) is what the simulations and the many-networks description require. 
   It implies that the $\sqrt{N}$ and $\sqrt{G}$ rates differ only by the factor 
   $(N/G)^{1/2} \to \mathbbm{E}[N_g]^{1/2}$; Remark~\ref{rem:growth} discusses growing group sizes. 
   The rate condition in Assumption~\ref{asm:rate} is the usual product-rate condition for doubly 
   robust estimators \citep{farrell2015robust, chernozhukov2018double}, and it holds with room to 
   spare when both nuisance models are parametric, since each mean-square error is then 
   $O_p(N^{-1})$. Part (a) removes the need for Donsker-type restrictions when flexible estimators 
   are used. Part (b) is the case of the Monte Carlo studies and the application, where the 
   nuisances are fitted once on the full sample. The closed-form score~(\ref{fact}) built from 
   generalized linear models with bounded covariates satisfies it, because the score is a 
   polynomial in the link and treatment probabilities; with unbounded covariates the 
   corresponding moment conditions are needed. The next lemma shows that a score built from parametric treatment and 
   link models inherits their parametric rate.

   \begin{lemma}[Rate of the estimated generalized propensity score]\label{lem:gps}
   Suppose Assumptions~\ref{asm:ce} and~\ref{asm:samp}(b) hold, the treatment and link probabilities 
   are $Pr(D_{ig} = 1 \mid \bm{X}_g) = F_d(\{ig\}, \bm{X}_g; \theta_d)$ and 
   $p_{ijg} = F_w(\{ig\}, \{jg\}, \bm{X}_g; \theta_w)$ (or, with latent random effects, 
   $f_{\bm{W}_g \mid \bm{X}_g}$ is the mixture of the conditional law given $\bm{U}_g$ over 
   $\bm{U}_g \sim F_U(\cdot; \theta_w)$), where $F_d$ and $F_w$ are Lipschitz in 
   $\theta = (\theta_d, \theta_w)$ uniformly in their other arguments, $F_U(\cdot; \theta_w)$ is 
   Lipschitz in $\theta_w$ in total variation, the parameter $\theta$ is common to all groups, and 
   $\hat{\theta} - \theta = O_p(G^{-1/2})$. Then for any exposure mapping with finitely many values, 
   $\max_{g, \{ig\}, t} |\hat{r}(\{ig\}, t; \bm{X}_g) - r(\{ig\}, t; \bm{X}_g)| = O_p(G^{-1/2})$, and 
   hence $N^{-1} \sum_{\{ig\}} (\hat{r} - r)^2 = O_p(N^{-1})$. For the exposure 
   mapping~(\ref{epp}) the bound holds with the sharper constant 
   $|\hat{r} - r| \leq |\hat{P}^d_{ig} - P^d_{ig}| + \sum_{\{jg\} \neq \{ig\}} |\hat{q}_{ijg} - q_{ijg}|$.
   \end{lemma}

   The proof (Appendix~\ref{app:gps}) uses the fact that, with bounded group size and finitely many 
   exposure levels, the score is a polynomial of bounded degree in finitely many link and treatment 
   probabilities; the sharper bound for~(\ref{epp}) is a total-variation bound for Poisson-binomial 
   laws. Pointwise consistency of the link probabilities is not enough when group sizes grow, which 
   is why Assumption~\ref{asm:samp}(b) is imposed here and why the single-network case in 
   Section~\ref{sec:single} needs a sparsity condition instead. The lemma also requires a parameter 
   that is common across groups: group-specific parameters, such as the village-specific network 
   models of Section~\ref{ea}, raise an incidental-parameter problem in this regime and belong to 
   the large-group regime of Section~\ref{sec:latent}.

   \begin{proposition}[Asymptotic normality and variance estimation]\label{ttt}
   Suppose Assumptions~\ref{asm:em}, \ref{asm:ig} and~\ref{asm:samp}--\ref{asm:rate} hold 
   and both nuisance models are correctly specified. Then
   \begin{equation}\label{asy}
     \sigma_{t, t'}^{-1} \sqrt{N} (\hat{\tau}(t, t') - \tau(t, t')) \xrightarrow{d} \mathcal{N}(0,1), 
     \qquad 
     \sigma^2_{t, t'} = \frac{\mathbbm{E}[S_g(t, t')^2]}{\mathbbm{E}[N_g]}, 
   \end{equation}
   and the cluster-robust variance estimator
   \begin{equation}\label{hac}
      \hat{\sigma}^2_{t, t'} = \frac{1}{N} \sum_{g \in \mathcal{G}} 
      \Big( \sum_{\{ig\} \in \mathcal{N}_g} \hat{\varphi}_{ig}(t, t') \Big)^2, 
      \qquad 
      \hat{\varphi}_{ig}(t, t') = \hat{\tau}_{ig}(t, t') - \hat{\tau}(t, t'),
   \end{equation}
   where $\hat{\tau}_{ig}$ replaces $(r, \mu)$ by $(\hat{r}, \hat{\mu})$, satisfies 
   $\hat{\sigma}^2_{t, t'} - \sigma^2_{t, t'} \xrightarrow{p} 0$. If only one of the two nuisance 
   models is correctly specified, $\hat{\tau}(t, t') \xrightarrow{p} \tau(t, t')$ still holds, but 
   the asymptotic variance then contains an additional term from the estimation of the correctly 
   specified nuisance, and $\hat{\sigma}^2_{t, t'}$ is not consistent for it in general.
   \end{proposition}

   The proof (Appendix~\ref{app:a2}) does not need limit theorems for dependent data on a single 
   network. Under Assumption~\ref{asm:samp} the group sums $S_g(t, t')$ are independent and 
   identically distributed with finite variance, so the central limit theorem applies to them 
   directly. The remainder terms from estimating $r$ and $\mu$ are controlled by 
   Assumption~\ref{asm:rate} and by the orthogonality of the AIPW moment condition, which makes 
   first-stage estimation error second order \citep{chernozhukov2018double}. The following 
   remarks concern what the variance contains, growing group sizes, and few groups.

   \begin{remark}[What the variance contains]\label{rem:var}
   $\sigma^2_{t, t'}$ is the unconditional variance of the group sum, taken over covariates, 
   treatments and the network. When both nuisance models are correct and potential outcomes are 
   conditionally independent within groups, the within-group cross terms of the weighting parts of 
   the influence functions vanish (Fact 2 in Appendix~\ref{app:eff}), the network enters 
   $\sigma^2_{t, t'}$ only through the score, via $\mathbbm{E}[\sigma^2_t(\bm{X}_g) / r]$, and the 
   remaining within-group covariance comes from shared covariates. 

   Group-level outcome shocks add within-group covariance to the influence functions, and because 
   groups are independent draws the cluster-robust estimator~(\ref{hac}) picks it up. 

   A misspecified outcome model is a different matter. The AIPW moment is then no longer 
   orthogonal to the score parameters, 
   $\sqrt{N}(\hat{\tau} - \tau)$ acquires a first-order term from $\sqrt{N}(\hat{\theta} - \theta)$, 
   and (\ref{hac}), computed from $\hat{\varphi}_{ig}$ alone, ignores it. The standard remedy is to 
   stack the score's estimating equations with the AIPW moment and use the sandwich variance of the 
   joint system. We do not pursue it; the Monte Carlo evidence in Section~\ref{SMC}, where the outcome 
   model is deliberately misspecified, suggests the omitted term is small in that design. Unlike 
   a single network (Section~\ref{sec:single}), the many-networks design does let us estimate 
   whatever within-group covariance the influence functions have, including the part induced by 
   the shared realization of $(\bm{D}_g, \bm{W}_g)$, because each group is an independent draw.
   \end{remark}

   \begin{remark}[Growing group sizes]\label{rem:growth}
   If group sizes grow, the argument changes in two places. The Lyapunov condition for the central limit theorem over 
   groups holds, for bounded outcomes, if $\max_g N_g = o(N^{1/2 - \epsilon})$ for some $\epsilon > 0$ 
   (Appendix~\ref{app:a2}), and the score rate in Lemma~\ref{lem:gps} has to be replaced by a bound 
   whose constant does not grow with $N_g$, which requires a sparse network. Both are available for 
   the exposure mapping~(\ref{epp}) and are developed in Section~\ref{sec:single}.
   \end{remark}

   \begin{remark}[Few groups]\label{rem:few}
   With few groups, as in the application with 16 villages, the normal approximation 
   in~(\ref{asy}) is coarse and the cluster-robust estimator~(\ref{hac}) is biased downward. The 
   standard remedies apply: critical values from a $t$ distribution with $G - 1$ degrees of 
   freedom, or a wild cluster bootstrap that re-estimates the nuisance models under perturbed 
   residuals \citep{cameron2015practitioner}. A bootstrap that only flips the signs of the group 
   sums of the estimated influence function reproduces the cluster-robust variance exactly and 
   checks nothing. When the groups are large and sparse, the within-group asymptotics of 
   Section~\ref{sec:single} suggest a within-group variance estimator. No result in this paper 
   covers the application's specification, with village-specific formation models, node-level 
   latent effects and growing degrees, so there that estimator is a heuristic.
   \end{remark}
\end{revblock}

   \begin{remark}[Regression-based estimation]\label{rem:reg}
   Linear models are common in empirical work on direct and spillover effects 
   (e.g., \citealp[]{cai2015social}). A typical specification is
   \begin{equation}\label{reg_ms}
   Y_{ig}=\alpha_g+\tau D_{ig}+\gamma \sum_{\{jg\} \in \mathcal{N}_g} W_{ijg} D_{jg} + 
   X_{ig}^{\prime} \beta + \varepsilon_{ig},
   \end{equation}
   where $\alpha_g$ is a group fixed effect, and the causal quantities are read off $\tau$ and 
   $\gamma$. Model~(\ref{reg_ms}) is easy to fit and to interpret, but when network formation is 
   confounded the individual covariates $X_{ig}$ may not be enough to adjust for it, and when the 
   outcome model is misspecified the estimated effects are biased. \rev{A fairer regression 
   benchmark estimates $\mu(t)$ by regression imputation (g-computation) from a model that is 
   flexible in $(D_{ig}, Z_{ig})$ and in their interactions with covariates; the AIPW estimator uses 
   such a model as its outcome component and adds the weighting correction. Section~\ref{SMC} 
   compares the two.} 
   \end{remark}

\section{Extensions}\label{sec:ext}

\subsection{Interdependent treatment assignment and network formation}\label{sec:interdep}

The conditional independence in \rev{Assumption~\ref{asm:ce}(a)} can be replaced by assumptions 
under which treatment and network formation interact. The first line of~(\ref{cee2}) is a factorization 
for treatment diffusion, in which a unit's adoption depends on the treatment status of its peers. If 
instead the network is endogenous to treatment, with ties forming in response to treatment status, 
the second line of~(\ref{cee2}) is the appropriate one.\footnote{Appendix~\ref{D} gives graphical 
representations of these factorizations and discusses them further.} Either way the two processes 
have to be modeled sequentially, since one depends on the other, and the score no longer has a 
closed form like~(\ref{fact}). It can be computed by simulation instead (e.g., 
\citealp[]{aronow2017estimating, toulis2021estimating}).
  \begin{equation}\label{cee2}
       f_{\bm{D}_g, \bm{W}_g|\bm{X}_g}(\bm{d}_g, \bm{w}_g)
       = 
       \begin{cases}
       f_{\bm{D}_g|\bm{X}^D_g, \bm{W}_g}(\bm{d}_g) f_{\bm{W}_g|\bm{X}^W_g}(\bm{w}_g)
       \quad (\text{treatment diffusion}) \\
       f_{\bm{D}_g|\bm{X}^D_g}(\bm{d}_g) f_{\bm{W}_g|\bm{X}^W_g, \bm{D}_g}(\bm{w}_g)
       \quad (\text{endogenous network}) 
       \end{cases}
       \end{equation}

\begin{revblock}
Identification is unchanged, since Assumption~\ref{asm:ig} and Propositions~\ref{prop:b}--\ref{prop:a} 
do not use the factorization. Only the computation of $r(\{ig\}, t; \bm{X}_g)$ changes, 
because~(\ref{gps2}) calls for the \textit{joint} law of $(\bm{D}_g, \bm{W}_g)$ given $\bm{X}_g$. 

Under treatment diffusion one fits a network model $f_{\bm{W}_g \mid \bm{X}_g}$ and a model for 
$f_{\bm{D}_g \mid \bm{X}_g, \bm{W}_g}$. A logistic model for $D_{ig}$ with the number or share of 
treated neighbors as a regressor specifies the full conditionals 
$Pr(D_{ig} \mid \bm{D}_{-ig}, \bm{W}_g, \bm{X}_g)$, which in general are compatible with no joint law. A 
coherent specification is an autologistic model, whose joint law has an intractable normalizing 
constant but can be sampled by Gibbs sampling, or a sequential-adoption model when adoption times 
are observed, as in the peer-influenced propensity score of \citet{jackson2020adjusting}, a 
conditional object defined along the adoption sequence. 

Under an endogenous network one fits the treatment model first and then a link model with $(D_{ig}, D_{jg})$ as dyadic covariates 
\citep{comola2021treatment}. This is coherent, and Assumption~\ref{asm:em} should then be read with 
the potential network under each treatment vector in place of the realized one. 

In both cases the score is estimated by Monte Carlo from the joint law: for $h = 1, \ldots, H$, draw 
$(\bm{D}^{(h)}_g, \bm{W}^{(h)}_g)$ from the fitted models (for the autologistic model, with a Gibbs 
sampler run to stationarity), compute $T^{(h)}_{ig} = e(\{ig\}, \bm{D}^{(h)}_g, \bm{W}^{(h)}_g)$, and 
set $\hat{r}(\{ig\}, t; \bm{X}_g) = H^{-1} \sum_h \mathbbm{1}\{T^{(h)}_{ig} = t\}$. The simulation 
error is $O_p(\sqrt{\log N / H})$ uniformly over units and levels (a union bound over $N$ binomial 
proportions) and $O_p(H^{-1})$ in mean square, which is negligible next to the parametric rate 
$O_p(N^{-1})$ of Lemma~\ref{lem:gps} when $H / N \to \infty$. The simulated score is not smooth in 
the fitted parameters, so Assumption~\ref{asm:rate}(b) does not apply to it directly; with $H$ 
large it inherits the properties of the exact score it approximates, and we treat it that way. The 
application uses this device for the treatment-diffusion specification (Section~\ref{ea}) and the 
closed form under common exposure. 

Under treatment diffusion treatments are no longer individualized, and the positivity argument of 
Proposition~\ref{prop:idpos}(c) requires the joint law to give every configuration positive 
probability, which the autologistic model does. When the network forms after treatment, the 
network-conditional estimands of Section~\ref{sec:cond} condition on a post-treatment variable, 
whereas the population estimand does not. The population estimand stays well defined, but its 
interpretation changes, because part of a unit's exposure is then caused by its own treatment 
through the ties it forms, and estimating it requires the endogenous-network factorization in 
place of Assumption~\ref{asm:ce}(a).
\end{revblock}

\begin{revblock}
\subsection{A single large network}\label{sec:single}

The many-networks regime of Section~\ref{sec:est} covers data organized in groups. Many 
applications observe one large network instead. This subsection goes through identification, the 
score and inference for that setting and says what remains open; details are in 
Appendix~\ref{app:single}. We drop the group index: there are $N$ units, 
$\bm{X} = (X_1, \ldots, X_N)$, treatment $\bm{D}$ and network $\bm{W}$, and we maintain 
Assumptions~\ref{asm:em}, \ref{asm:ig} and~\ref{asm:ce} with the exposure mapping~(\ref{epp}).

\paragraph{Identification and the score.} Propositions~\ref{prop:dec}--\ref{prop:a} are stated for 
one group and hold as written. The score is 
$r(i, (d, z); \bm{X}) = Pr(D_i = d \mid \bm{X}) \, Pr(Z_i = z \mid \bm{X})$, where $Z_i \mid \bm{X}$ is 
Poisson-binomial with success probabilities $q_{ij} = p_{ij}(\bm{X}) Pr(D_j = 1 \mid \bm{X})$, 
$j \neq i$. It can be computed exactly by convolution in $O(N^2 \bar{z})$ operations for levels up 
to $\bar{z}$, or approximated by a Poisson law when the network is sparse. The sparsity condition we 
use is the one \citet{graham2024sparse} uses for dyadic logistic regression: 
$p_{ij} = \rho_N \, h(X_i, X_j)$ with $N \rho_N \to \lambda \in (0, \infty)$ and $h$ bounded between 
positive constants, so that expected degrees stay bounded as $N$ grows.

\begin{proposition}[Positivity under sparse network formation]\label{prop:sparsepos}
Under the sparsity condition, for every fixed $z \in \{0, 1, 2, \ldots\}$ and $d \in \{0, 1\}$, 
$\liminf_{N \to \infty} \min_i r(i, (d, z); \bm{X}) > 0$ almost surely, whereas 
$r^{c}(i, (d, z); \bm{X}, \bm{W}) = 0$ for every unit with $\deg_i(\bm{W}) < z$, a set whose share of 
the population converges to a positive limit for every $z \geq 1$.
\end{proposition}

Proposition~\ref{prop:sparsepos} (Appendix~\ref{app:single}) is the single-network version of 
Proposition~\ref{prop:idpos}(c), and sharper. In a sparse network a fixed fraction of units 
have fewer than $z$ ties, so the network-conditional approach can only ever identify $\mu(d, z)$ on 
a sub-population whose size falls with $z$, while the marginal approach identifies it for everyone. 
Le Cam's inequality gives the rate at which the Poisson-binomial law approaches the Poisson law, 
$O(\sum_j q_{ij}^2) = O(N^{-1})$.

\paragraph{Estimation of the score.} The total-variation bound in Lemma~\ref{lem:gps} does not 
involve group size, so it applies here: 
$|\hat{r}(i, (d,z)) - r(i, (d, z))| \leq |\hat{P}^d_i - P^d_i| + \sum_{j \neq i} |\hat{q}_{ij} - q_{ij}|$. 
Under sparsity $\sum_{j} p_{ij} = O(1)$, and for the dyadic logit with $\rho_N$ absorbed in the 
intercept, \citet{graham2024sparse} shows that the maximum likelihood estimator of $\theta_w$ is 
asymptotically normal at the rate of the square root of the number of links, which here is 
$\sqrt{N}$. Lemma~\ref{lem:single} in Appendix~\ref{app:single} concludes that 
$\max_i |\hat{r}(i, t) - r(i, t)| = O_p(N^{-1/2})$, so the score part of the rate condition in 
Assumption~\ref{asm:rate} holds. Node-level latent positions cannot be estimated consistently in 
this bounded-degree regime, where each node has $O_p(1)$ ties; the plug-in of 
Section~\ref{sec:latent} belongs to the regime in which degrees grow with $N$.

\paragraph{Inference.} Here the single network differs from the many-networks design. Conditional 
on $(\bm{X}, \bm{W})$, the influence functions $\varphi_i$ depend on one another only through 
shared treatment indicators, so units $i$ and $j$ are conditionally dependent only if they are 
neighbors or share a neighbor. A central limit theorem for local dependence 
\citep{baldi1989normal, chen2004normal, ross2011fundamentals} then gives conditional asymptotic 
normality when the maximum degree grows slowly, as it does in sparse inhomogeneous random graphs. 

The unconditional distribution adds the variability of the conditional mean 
$\mathbbm{E}[\hat{\tau} \mid \bm{X}, \bm{W}]$ across network realizations. 
Proposition~\ref{prop:singleclt} in Appendix~\ref{app:single} covers the case in which the outcome 
model is correctly specified and potential outcomes are conditionally independent across units. 
Then the conditional mean does not depend on $\bm{W}$, the two sources of variability are 
orthogonal, and the within-network dependence contributes no covariance terms at all, because 
the weighting parts of the influence functions are conditionally uncorrelated across units 
(Fact 2 of Appendix~\ref{app:eff}). A consistent estimator of the unconditional variance is then 
simply
\begin{equation}\label{hac2}
\begin{gathered}
\hat{\sigma}^2 = \frac{1}{N} \sum_{i} \hat{a}_i^2 + \frac{1}{N} \sum_i \hat{b}_i^2, 
\qquad 
\hat{b}_i = \hat{\mu}(i, t, \bm{X}) - \hat{\mu}(i, t', \bm{X}) - \hat{\tau}(t, t'), \\
\hat{a}_i = \frac{\mathbbm{1}\{T_i = t\}(Y_i - \hat{\mu}(i, t, \bm{X}))}{\hat{r}(i, t; \bm{X})} 
- \frac{\mathbbm{1}\{T_i = t'\}(Y_i - \hat{\mu}(i, t', \bm{X}))}{\hat{r}(i, t'; \bm{X})}, 
\end{gathered}
\end{equation}
with the outcome model depending on $\bm{X}$ only through $X_i$. The local dependence matters for 
the central limit theorem but leaves the variance unchanged. 

Dependence does enter the variance in two cases that Proposition~\ref{prop:singleclt} excludes. 
When potential outcomes are correlated across 
neighbors, a network HAC estimator \citep{kojevnikov2021limit, leung2023network} applied to 
$\hat{a}_i$ over pairs within distance two is the natural extension. When the outcome model is 
misspecified and only the score is correct, the conditional mean of $\hat{\tau}$ varies with 
$\bm{W}$ through $r^{c} / r$, and the resulting between-network component cannot be recovered from 
a single realization without further structure. A model-based bootstrap that redraws 
$(\bm{D}, \bm{W})$ from the fitted assignment models is the natural device for that case; we leave 
its validity for future work. This is the main reason to prefer the many-networks design when it is 
available: there Remark~\ref{rem:var} applies and the issue does not arise. 

For data with a few large groups, estimator~(\ref{hac2}) can be applied within each group and the group-level terms 
summed with weights $N_g / N$, provided the groups are sparse and the nuisances parametric. The 
16 villages in Section~\ref{ea}, with node-level latent effects and a mean degree of about 
fifteen, do not meet these conditions, so there the estimator is a heuristic. The result in 
Appendix~\ref{app:single} is a proof sketch for the oracle estimator with a heuristic extension to 
estimated nuisances, and it is not at the level of rigor of Proposition~\ref{ttt}.
\end{revblock}

  \section{Monte Carlo Studies}\label{SMC}

  This section reports Monte Carlo studies of the finite sample properties of the proposed method. 
  We focus on the AIPW estimator~(\ref{eatt}) and report bias, root mean squared error (RMSE), the 
  coverage rate of 95\% confidence intervals, and the difference between the average estimated 
  standard error (SE) and the sampling standard deviation (SD). 

  Each simulated dataset consists of disjoint groups. Units in the same group may form links; 
  there are no links across groups. All groups have the same size, the treatment is binary, and 
  the networks are undirected and unweighted. The data generating process (DGP) for network 
  formation and treatment assignment is 
  \begin{equation}\label{dgp1}
  \begin{gathered}
  W_{ijg} = 
      \mathbbm{1}\{ 0.1 - 0.25 |X_{ig} - X_{jg}| + \nu_{ijg} \geq 0 \}, \\
    D_{ig} = \mathbbm{1}\{ 0.1 + 0.5 X_{ig} + e_{ig} \geq 0 \}, 
  \end{gathered}
  \end{equation} 
  where $X_{ig} \sim N(0, 1)$ is an individual covariate and $\nu_{ijg} \sim N(0, 1)$ and 
  $e_{ig} \sim N(0, 1)$ are independent errors. \rev{The DGP satisfies Assumption~\ref{asm:ce} 
  without latent variables, and Assumption~\ref{asm:samp} with $N_g = 4$.} Groups have four units 
  and are formed by index: units 1 to 4 form group 1, units 5 to 8 group 2, and so on. With the 
  exposure mapping~(\ref{epp}), $Z_{ig} \in \{0, 1, 2, 3\}$ and the effective treatment $T_{ig}$ 
  takes eight values, $(0, 0)$, $(0, 1)$, $(0, 2)$, $(0, 3)$, $(1, 0)$, $(1, 1)$, $(1, 2)$ and 
  $(1, 3)$. The potential outcome model is
  \begin{equation}\label{dgp2}
  Y_{ig}(d, z) = 1 + X_{ig} + \frac{\sum_{\{jg\} \in \mathcal{N}_g, \{jg\} \neq \{ig\}} X_{jg}}{3} + d + 0.5z + d\cdot z + \varepsilon_{ig},  
  \end{equation}
  where $\varepsilon_{ig} \sim N(0, 1)$ and the average of the other units' covariates is a 
  contextual effect. 
  
  With eight levels of the effective treatment there are ${8 \choose 2} = 28$ possible effects. 
  We focus on a direct effect, $\tau((1, 1), (0, 1)) = \mu(1, 1) - \mu(0, 1)$, and a spillover 
  effect, $\tau((1, 1), (1, 0)) = \mu(1, 1) - \mu(1, 0)$. Since group size is fixed, the sample 
  size is varied through the number of groups, \rev{$G$} $= 50$, $100$, $200$, $500$ and $1000$, 
  with 2,000 replications for each. 

  We first estimate the treatment and network models, \rev{a probit for $D_{ig}$ on $X_{ig}$ and a 
  probit for $W_{ijg}$ on $|X_{ig} - X_{jg}|$ pooled across groups,} both correctly specified as 
  in~(\ref{dgp1}), and build the generalized propensity score from the closed form~(\ref{fact}). 
  We then apply the AIPW estimator~(\ref{eatt})\rev{, with the linear model~(\ref{reg_ms}) in 
  $D_{ig}$, $Z_{ig}$ and $X_{ig}$ (without group effects) as the outcome component. This outcome 
  model is deliberately misspecified: it omits the interaction $D_{ig} Z_{ig}$ and the contextual 
  term of~(\ref{dgp2}). The score is correct, so the design illustrates double robustness rather 
  than the case of Proposition~\ref{ttt}, in which both nuisances are correct. The Gaussian 
  covariate is unbounded, so Assumption~\ref{asm:mom}(b) holds only approximately; the regularity 
  conditions of Section~\ref{sec:est} are sufficient, not sharp. Table~\ref{tab:bench} adds the case 
  with a correctly specified outcome model, the case Proposition~\ref{ttt} covers.} Confidence 
  intervals use \rev{the cluster-robust variance estimator~(\ref{hac})}. Tables~\ref{dir1} 
  and~\ref{spi1} report the results for the two effects. 

  \begin{table}[h]
  \centering
  \caption{Finite Sample Properties: $\tau((1,1),(0,1))$ \label{dir1}}
  \begin{tabular}{l *{5}{S[table-format=-1.3]}}
  \hline\hline  
  {\rev{$G$}} & {Bias} & {RMSE} & {Coverage Rate} & {SE} & {SE - SD}  \\ 
  \hline
 50   &-0.012 &0.298    &0.936 &0.287  &-0.010 \\
 100   &0.001 &0.212    &0.941 &0.205  &-0.007 \\
 200   &0.003 &0.149    &0.939 &0.147  &-0.002 \\
 500   &0.001 &0.094    &0.953 &0.093  &-0.001 \\
 1000 &-0.002 &0.065    &0.960 &0.066   &0.001 \\
  \hline\hline  
 \end{tabular}
\end{table}

  \begin{table}[h]
  \centering
  \caption{Finite Sample Properties: $\tau((1,1),(1,0))$ \label{spi1}}
  \begin{tabular}{l *{5}{S[table-format=-1.3]}}
  \hline\hline  
  {\rev{$G$}} & {Bias} & {RMSE} & {Coverage Rate} & {SE} & {SE - SD}  \\ 
  \hline
 50   &-0.003 &0.277    &0.943 &0.268  &-0.009 \\
 100  &-0.002 &0.196    &0.950 &0.193  &-0.003 \\ 
 200  &-0.002 &0.137    &0.948 &0.137   &0.000 \\ 
 500   &0.000 &0.087    &0.948 &0.087   &0.000 \\
 1000  &0.002 &0.061    &0.956 &0.062   &0.001 \\
  \hline\hline  
 \end{tabular}
\end{table}

For both effects the AIPW estimator has small bias even with few groups, and its RMSE falls as the 
number of groups grows. The coverage of the 95\% intervals is close to nominal, and the gap 
between the average estimated standard error and the sampling standard deviation shrinks with the 
number of groups, as it should if the variance estimator is consistent.

For comparison we also fit the regression model~(\ref{reg_ms}), with standard errors clustered by 
group, and read the direct effect off $\hat{\tau}$ and the spillover effect off $\hat{\gamma}$. 
Model~(\ref{reg_ms}) is misspecified: it omits the interaction $D_{ig} \cdot Z_{ig}$ and the 
contextual effect of the covariate. Figure~\ref{sam} shows the sampling distributions of the two 
estimators of the spillover effect and Figure~\ref{sam1} their bias and RMSE. The AIPW estimator 
has bias close to zero even with few groups, and its RMSE falls quickly with the number of groups; 
the RMSE of the regression estimator barely moves because of its bias. \rev{The comparison shows 
the cost of a misspecified outcome model rather than the merit of weighting as such: a regression 
with the interaction and the contextual term would recover the effects in this design. The AIPW 
estimator uses the same misspecified outcome model and stays unbiased because its score is 
correct, and its coverage is close to nominal even though the variance estimator~(\ref{hac}) 
ignores the first-stage term of Remark~\ref{rem:var}.} \begin{revblock}

\paragraph{Benchmarks.} Table~\ref{tab:bench} compares seven estimators of the two effects under 
the same DGP. The labels are used again in later tables.
\begin{itemize}
\item AIPW-JM, mis.: the AIPW estimator with the marginal score and the misspecified outcome 
model, the estimator of Tables~\ref{dir1}--\ref{spi1}. ``JM'' stands for joint modeling of 
treatment and network.
\item AIPW-JM: the same with the correctly specified outcome model 
$y \sim d \cdot z + x + \bar{x}_{-i}$, where $\bar{x}_{-i}$, the mean covariate of the other 
units in the group, is the contextual term of~(\ref{dgp2}).
\item IPW-JM: H\'{a}jek inverse probability weighting with the marginal score and no outcome model.
\item OR, correct and OR, mis.: outcome regression, that is, regression imputation from the 
correct and from the misspecified outcome model.
\item Regression, interacted: the regression 
$y \sim d \cdot z + x + \bar{x}_{-i} + d \cdot (x - \bar{x})$ with effects read from the 
coefficients.
\item AIPW-SEP: the AIPW estimator whose score conditions on own covariates only, discussed 
below.
\end{itemize}
When both nuisance models are correct, the AIPW estimator is unbiased with nominal coverage, and 
its RMSE is about a fifth smaller than with the misspecified outcome model. The correctly 
specified regression and regression imputation are also unbiased with nominal coverage, and they 
are more efficient than AIPW, with RMSE about thirty percent lower. When the outcome model is 
known to be correct, weighting costs precision. The case for AIPW is robustness, as the 
misspecified-regression rows show, with bias of $-0.30$ and coverage falling to zero. Inverse probability weighting alone is unbiased but the least efficient, with RMSE sixty 
to eighty-five percent above that of AIPW with the correct outcome model, and its H\'{a}jek 
variance estimator is somewhat conservative at large $G$ (coverage $0.97$ at $G = 1000$). The full 
grid of $G$ and the replication code are in the supplementary folder.

\begin{table}[h]
\centering
\caption{Benchmarks under DGP~(\ref{dgp1})--(\ref{dgp2}): bias, RMSE and coverage of 95\% intervals}\label{tab:bench}
\scriptsize
\setlength{\tabcolsep}{4pt}
\begin{tabular}{l *{9}{S[table-format=-1.3]}}
\hline\hline
 & \multicolumn{3}{c}{$G = 50$} & \multicolumn{3}{c}{$G = 200$} & \multicolumn{3}{c}{$G = 1000$} \\
\cline{2-4} \cline{5-7} \cline{8-10}
 & {Bias} & {RMSE} & {Cov.} & {Bias} & {RMSE} & {Cov.} & {Bias} & {RMSE} & {Cov.} \\
\hline
\multicolumn{10}{l}{\textit{Direct effect $\tau((1,1),(0,1))$}} \\
AIPW-JM, mis. outcome model & -0.012 & 0.298 & 0.936 &  0.003 & 0.149 & 0.939 & -0.002 & 0.065 & 0.960 \\
AIPW-JM                     & -0.010 & 0.242 & 0.946 &  0.006 & 0.124 & 0.947 & -0.001 & 0.054 & 0.957 \\
IPW-JM                      &  0.026 & 0.391 & 0.946 &  0.012 & 0.198 & 0.961 & -0.001 & 0.087 & 0.972 \\
OR, correct                 & -0.004 & 0.167 & 0.938 &  0.003 & 0.080 & 0.952 & -0.001 & 0.036 & 0.949 \\
OR, mis.                    & -0.309 & 0.361 & 0.620 & -0.298 & 0.312 & 0.107 & -0.301 & 0.304 & 0.000 \\
Regression, interacted      & -0.004 & 0.168 & 0.935 &  0.003 & 0.080 & 0.954 & -0.001 & 0.036 & 0.949 \\
AIPW-SEP                    & -0.013 & 0.287 & 0.941 &  0.002 & 0.143 & 0.942 & -0.002 & 0.062 & 0.960 \\
\hline
\multicolumn{10}{l}{\textit{Spillover effect $\tau((1,1),(1,0))$}} \\
AIPW-JM, mis. outcome model & -0.003 & 0.277 & 0.943 & -0.002 & 0.137 & 0.948 &  0.002 & 0.061 & 0.956 \\
AIPW-JM                     &  0.002 & 0.231 & 0.949 &  0.001 & 0.112 & 0.956 &  0.002 & 0.050 & 0.957 \\
IPW-JM                      &  0.010 & 0.408 & 0.932 &  0.003 & 0.206 & 0.941 &  0.001 & 0.092 & 0.952 \\
OR, correct                 &  0.001 & 0.137 & 0.925 & -0.001 & 0.067 & 0.946 &  0.000 & 0.030 & 0.949 \\
OR, mis.                    & -0.342 & 0.371 & 0.303 & -0.335 & 0.342 & 0.003 & -0.333 & 0.334 & 0.000 \\
Regression, interacted      &  0.001 & 0.137 & 0.926 & -0.001 & 0.067 & 0.947 &  0.000 & 0.030 & 0.949 \\
AIPW-SEP                    &  0.117 & 0.303 & 0.917 &  0.124 & 0.184 & 0.846 &  0.127 & 0.141 & 0.470 \\
\hline\hline
\end{tabular}
\parbox{0.95\textwidth}{\footnotesize{\bf Notes:} 2,000 replications, the same run as Tables~\ref{dir1}--\ref{spi1}. ``mis.'' denotes the outcome 
model $y \sim d + z + x$; ``correct'' denotes $y \sim d \cdot z + x + \bar{x}_{-i}$. Intervals use 
cluster-robust variances (delta method for OR and regression).}
\end{table}
\end{revblock}

\begin{figure}[!ht]
  \begin{center}
  \begin{tabular}{ccc}
   \raisebox{-0.5\height}{\includegraphics[scale=0.2]{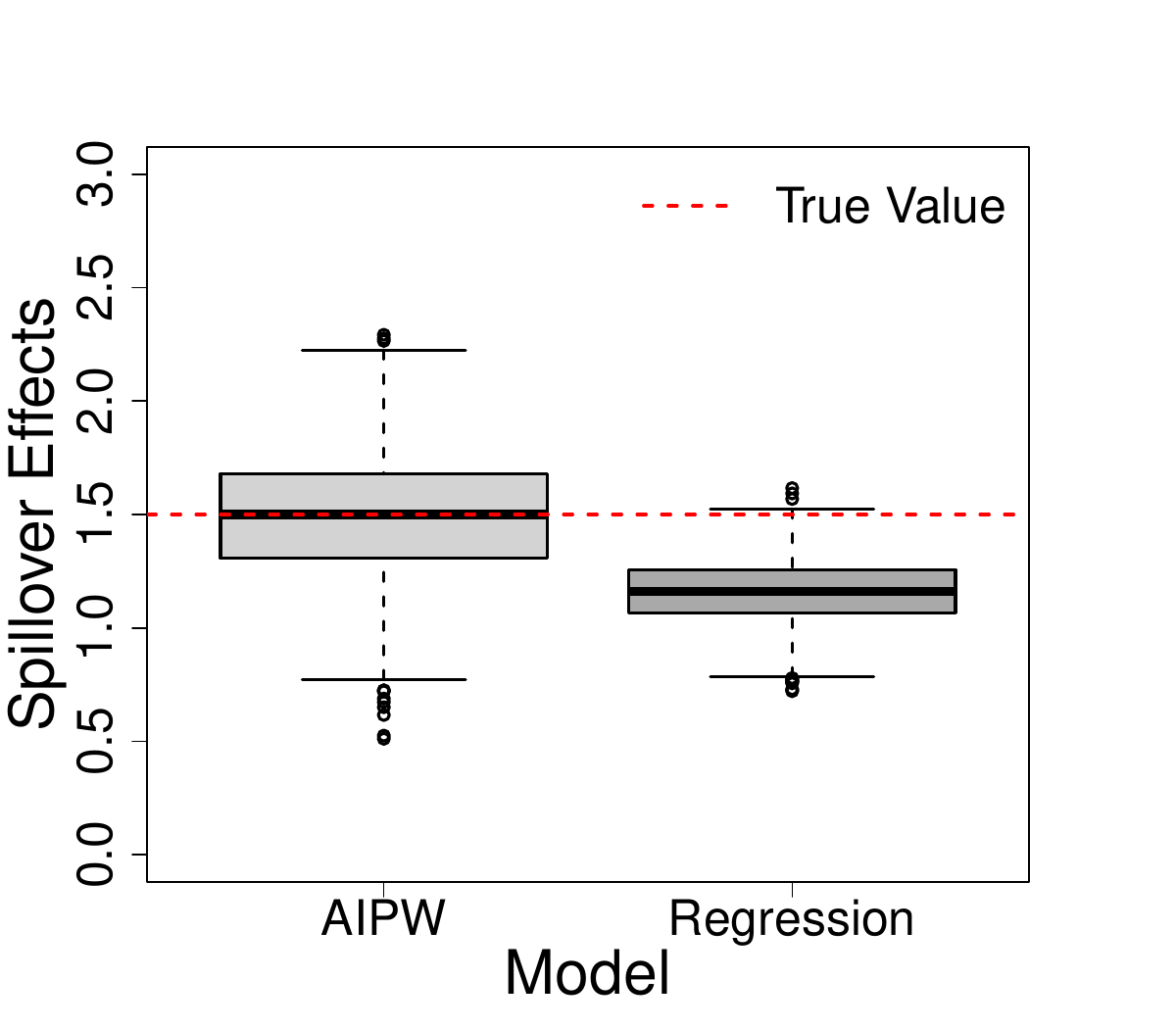}} &
   \raisebox{-0.5\height}{\includegraphics[scale=0.2]{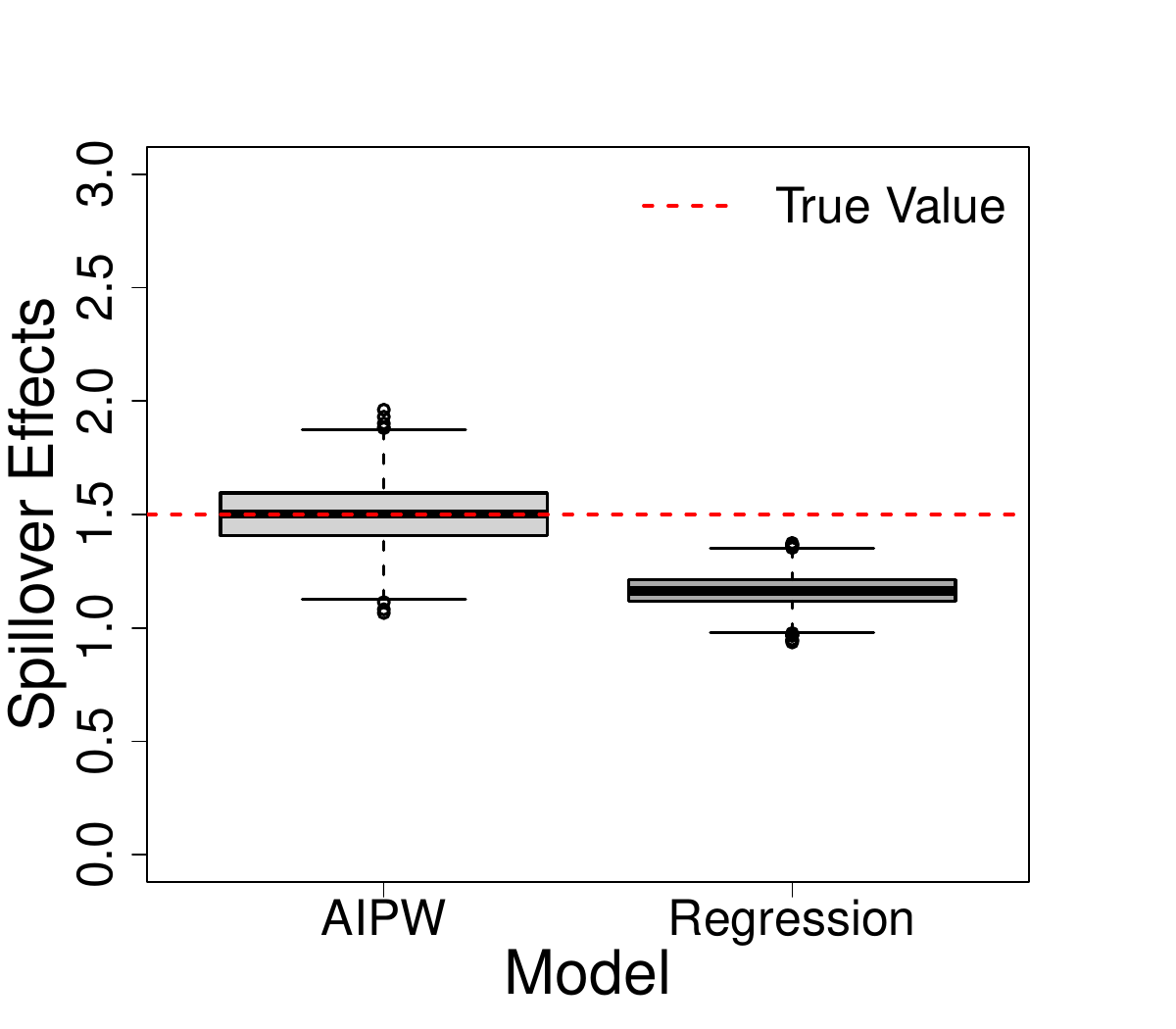}} &
   \raisebox{-0.5\height}{\includegraphics[scale=0.2]{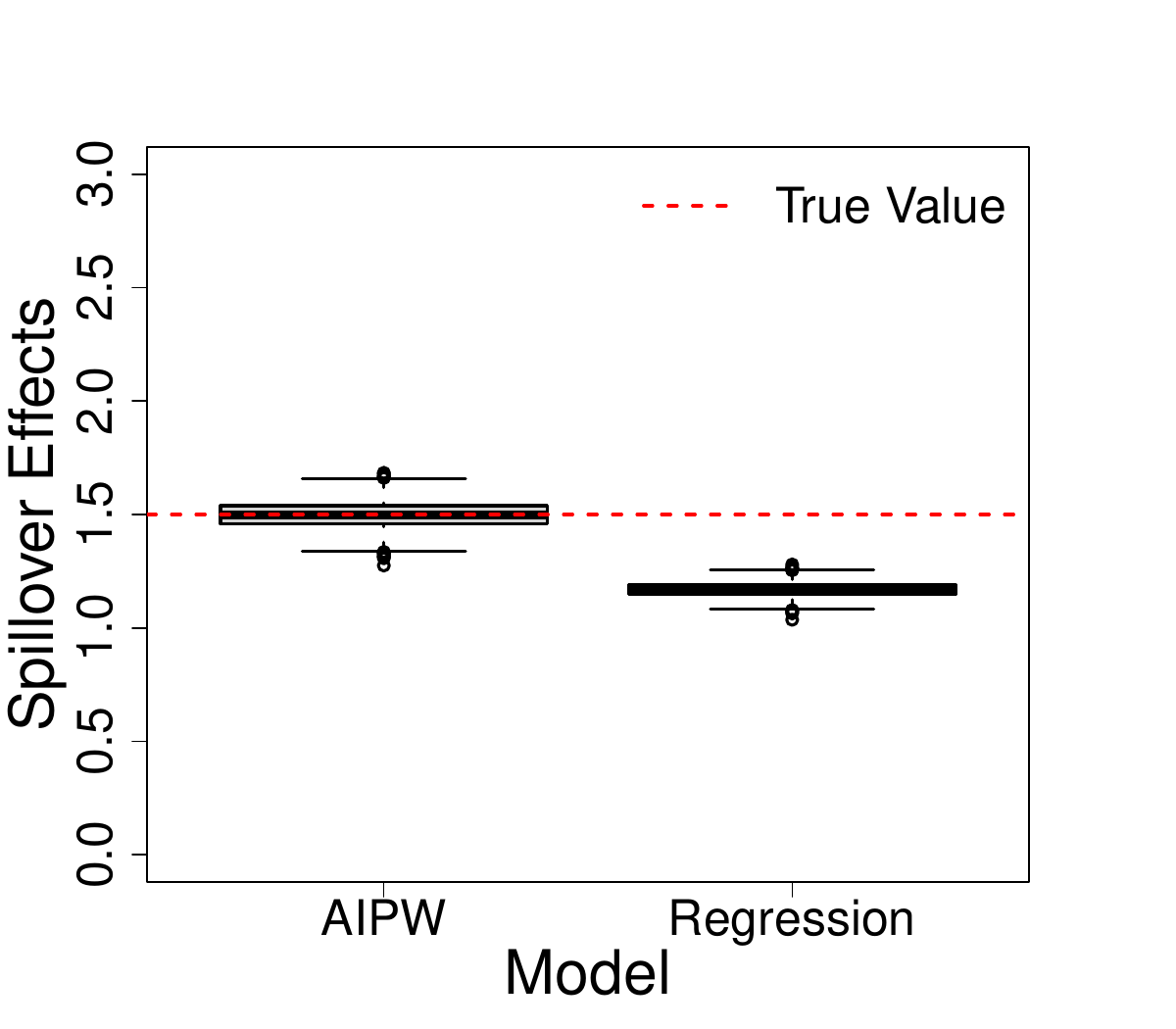}} \\
   50 Groups & 200 Groups & 1000 Groups
   \end{tabular}
   \caption{Sampling Distribution with Different Number of Groups}\label{sam}
   \end{center}
   \end{figure}

\begin{figure}[!ht]
  \begin{center}
  \begin{tabular}{cc}
   \raisebox{-0.5\height}{\includegraphics[scale=0.3]{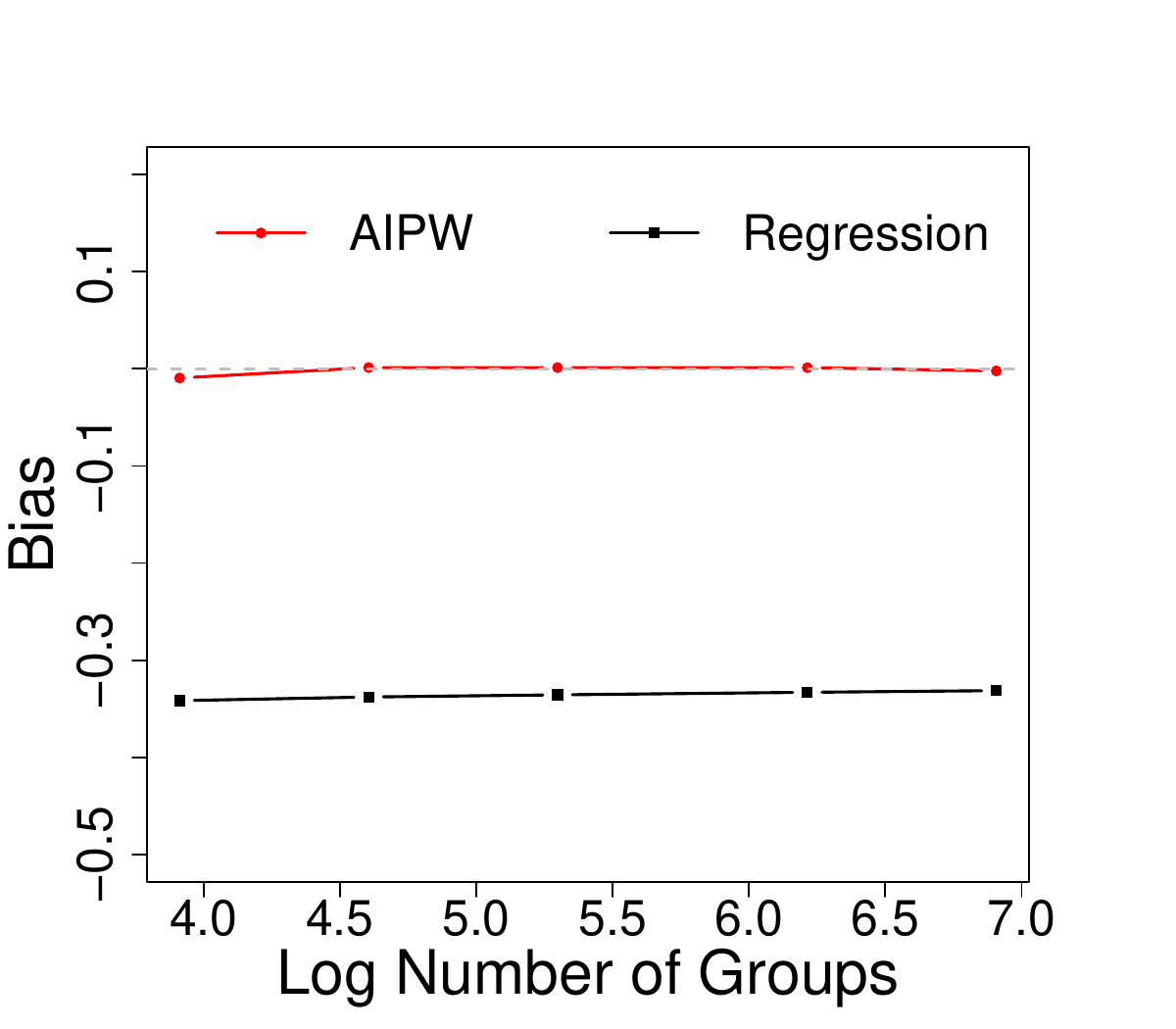}} &
   \raisebox{-0.5\height}{\includegraphics[scale=0.3]{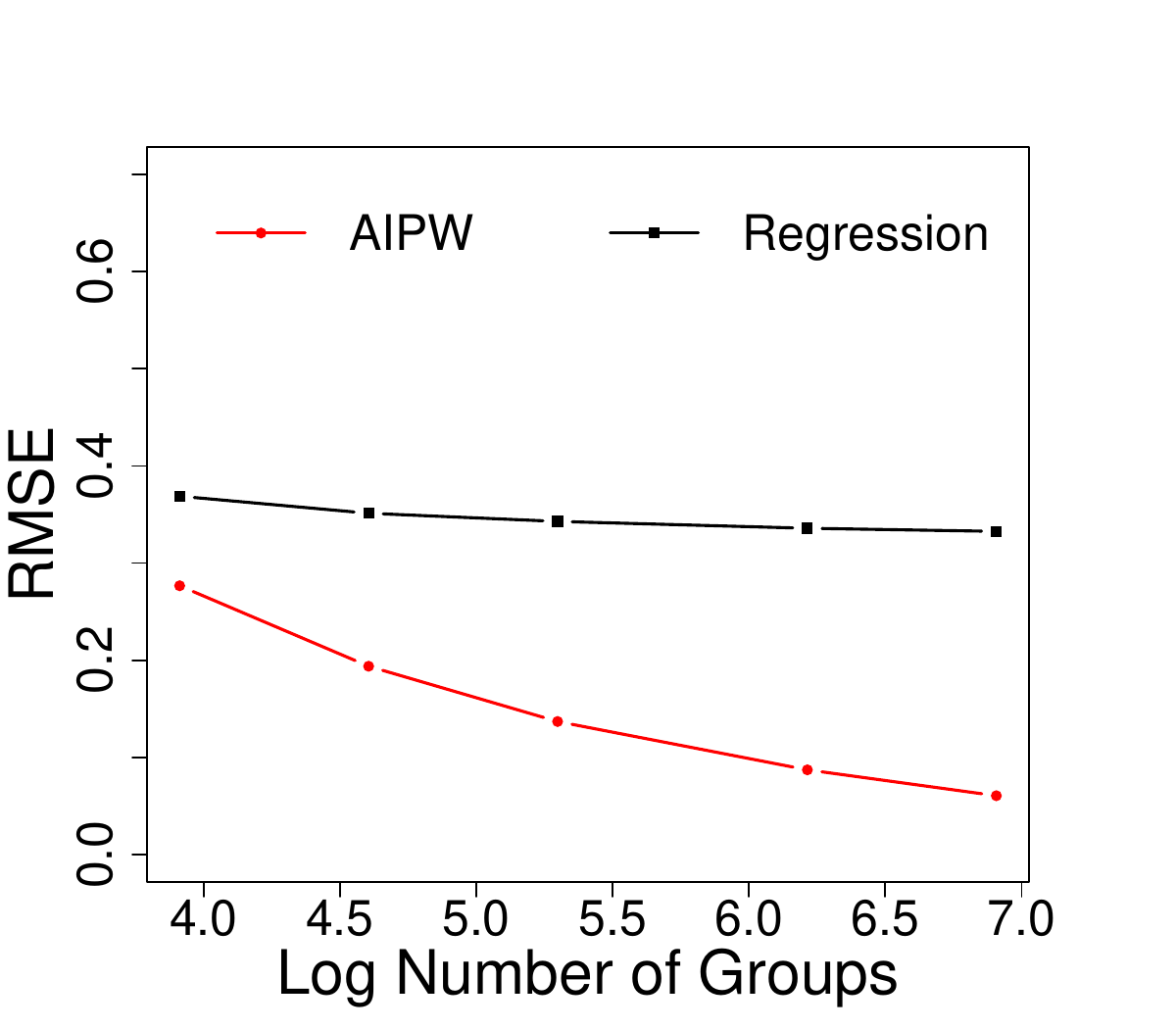}}  
   \end{tabular}
   \caption{Finite Sample Properties (Spillover Effect)}\label{sam1}
   \end{center}
   \end{figure}

\begin{revblock}
\paragraph{Marginal score with and without group covariates.} Appendix~\ref{abc} compares the 
AIPW estimator with the marginal score~(\ref{fact}) (``AIPW-JM'', for joint modeling) with an AIPW 
estimator whose score is the product of a probit model for $D_{ig}$ given $X_{ig}$ and a 
multinomial logit model for $Z_{ig}$ given $X_{ig}$ only (``AIPW-SEP'', for separate modeling of 
the two components of the exposure), with the same outcome model. AIPW-SEP uses the factorization 
$Pr(T_{ig} = t \mid X_{ig}) = Pr(D_{ig} = d \mid X_{ig}) Pr(Z_{ig} = z \mid X_{ig})$ that one would use 
if the exposure were treated as a multi-valued treatment with own covariates as the only 
confounders; it conditions neither on the realized network nor on the covariates of other units. 
Under DGP~(\ref{dgp1})--(\ref{dgp2}) both estimators are consistent for the direct effect and 
perform almost identically (Figure~\ref{sam1_d_ps} in Appendix~\ref{abc}). For the spillover 
effect $\tau((1,1),(1,0))$, AIPW-SEP has a small but persistent bias that does not vanish as $G$ 
grows, and a larger RMSE (Figures~\ref{sam1_ps}--\ref{sam_ps}). The reason is omitted confounding 
in the score. $Z_{ig}$ depends on the covariates of the other units in the group, through the 
homophily term in the link model and through their treatment probabilities, and so does the 
outcome, through the contextual term in~(\ref{dgp2}). A score that conditions on $X_{ig}$ alone 
omits $\bm{X}_{-ig}$, which Assumption~\ref{asm:ig} requires, and the outcome model omits it too, 
so neither component is correct. The marginal score~(\ref{fact}) conditions on $\bm{X}_g$ through 
the formation model, which is the practical content of modeling network formation: it is how the 
covariates of potential neighbors enter the score. This comparison is about the conditioning set, 
not about conditioning on the realized network, and it says nothing about 
Propositions~\ref{prop:dec}--\ref{prop:eff}. Table~\ref{tab:cond} reports the comparison that 
Section~\ref{sec:cond} calls for.

\paragraph{Population versus feasible-population effects.} We compare AIPW-JM with an AIPW 
estimator that uses the closed-form network-conditional score $r^{c}$ (a Poisson-binomial over 
the realized neighbors with success probabilities $\hat{P}^1_{jg}$, Section~\ref{sec:joint}), 
restricted to the common support of the two exposure levels, that is, to units with at least one 
tie (``AIPW-COND''). By Proposition~\ref{prop:idpos} it targets the feasible-population effect 
$\tau(t, t' \mid \mathcal{A})$. A third estimator (``AIPW-COND, naive'') uses $r^{c}$ but keeps all 
units at both levels, replacing a zero score by one for units that cannot reach a level; their 
exposure indicator is zero, so they contribute only the outcome-model prediction. 

We use two DGPs: the constant-effect DGP above, and a DGP with heterogeneous effects in which 
$Y_{ig}(d, z)$ gains the term $X_{ig} d z$ and the link index $0.1 - 0.25 |X_{ig} - X_{jg}|$ gains 
a sociability term $0.3 (X_{ig} + X_{jg})$, so that well-connected units have higher $X_{ig}$. 
(Under homophily on $|X_{ig} - X_{jg}|$ alone the composition term of Proposition~\ref{prop:dec}(b) 
is zero by symmetry.) The outcome model is correctly specified in both cases. The 
feasible-population effects, computed by simulation, equal the population effects under constant 
effects and exceed them by $0.15$ under heterogeneous effects, because units with at least one 
tie have $\mathbbm{E}[X_{ig} \mid \mathcal{A}] = 0.15$. The two DGPs have the same link intercept 
and almost the same density, so about a fifth of the units cannot reach level $(1, 1)$ in either. 

Table~\ref{tab:cond} shows the pattern Proposition~\ref{prop:dec} predicts. Under constant 
effects the two routes estimate the same number, both with nominal coverage, and AIPW-JM has the 
smaller RMSE, the empirical counterpart of Proposition~\ref{prop:eff}, with the caveat that 
AIPW-COND uses only the units with at least one tie. Under heterogeneous effects AIPW-COND is 
unbiased for the feasible-population effect and off by the composition term for the population 
effect, with coverage of the latter collapsing as $G$ grows, while AIPW-JM is unbiased for the 
population effect. 

The naive comparison is unbiased for the population effect here only because the correctly 
specified outcome model does the extrapolation for the units without ties. But the 
moment is then no longer orthogonal to the outcome-model parameters, since the units without 
ties contribute $\hat{m}(t) - \hat{m}(t')$ only and the variance estimator ignores the estimation 
error in that term, so its coverage falls short of nominal. The same mechanism applies in the 
application to the direct effects on adopters, where only 17 or 28 units carry a weighting term.

\begin{table}[h]
\centering
\caption{Marginal versus network-conditional score: bias and coverage for the population and 
feasible-population effects}\label{tab:cond}
\footnotesize
\begin{tabular}{l l *{6}{S[table-format=-1.3]}}
\hline\hline
 & & \multicolumn{2}{c}{$G = 50$} & \multicolumn{2}{c}{$G = 200$} & \multicolumn{2}{c}{$G = 1000$} \\
\cline{3-4} \cline{5-6} \cline{7-8}
Estimator & Truth & {Bias} & {Cov.} & {Bias} & {Cov.} & {Bias} & {Cov.} \\
\hline
\multicolumn{8}{l}{\textit{Constant effects, direct effect $\tau((1,1),(0,1))$}} \\
AIPW-JM          & population          &  0.002 & 0.936 &  0.001 & 0.951 &  0.000 & 0.948 \\
AIPW-COND        & feasible population &  0.000 & 0.931 &  0.002 & 0.948 & -0.001 & 0.944 \\
AIPW-COND, naive & population          &  0.000 & 0.899 &  0.002 & 0.922 & -0.001 & 0.923 \\
\multicolumn{8}{l}{\textit{Constant effects, spillover effect $\tau((1,1),(1,0))$}} \\
AIPW-JM          & population          &  0.001 & 0.951 &  0.000 & 0.955 &  0.000 & 0.946 \\
AIPW-COND        & feasible population & -0.005 & 0.937 & -0.003 & 0.956 &  0.001 & 0.956 \\
AIPW-COND, naive & population          & -0.003 & 0.935 & -0.002 & 0.952 &  0.001 & 0.959 \\
\hline
\multicolumn{8}{l}{\textit{Heterogeneous effects, direct effect}} \\
AIPW-JM          & population          &  0.010 & 0.934 &  0.005 & 0.941 &  0.003 & 0.949 \\
AIPW-COND        & feasible population &  0.006 & 0.933 &  0.005 & 0.938 &  0.002 & 0.946 \\
AIPW-COND        & population          &  0.159 & 0.887 &  0.157 & 0.776 &  0.155 & 0.269 \\
AIPW-COND, naive & population          &  0.007 & 0.902 &  0.005 & 0.914 &  0.003 & 0.927 \\
\multicolumn{8}{l}{\textit{Heterogeneous effects, spillover effect}} \\
AIPW-JM          & population          &  0.004 & 0.931 & -0.003 & 0.953 &  0.001 & 0.952 \\
AIPW-COND        & feasible population & -0.002 & 0.945 & -0.006 & 0.955 &  0.000 & 0.951 \\
AIPW-COND        & population          &  0.151 & 0.906 &  0.146 & 0.848 &  0.153 & 0.469 \\
AIPW-COND, naive & population          &  0.000 & 0.930 & -0.005 & 0.947 &  0.001 & 0.944 \\
\hline\hline
\end{tabular}
\parbox{0.95\textwidth}{\footnotesize{\bf Notes:} 2,000 replications; correctly specified outcome 
models; cluster-robust 95\% intervals. Population truths: $2$ and $1.5$. Feasible-population 
truths under heterogeneous effects: $2.15$ and $1.65$ (by simulation with $G = 20{,}000$). Both 
DGPs use the link index $0.1 - 0.25|X_{ig} - X_{jg}|$, to which the second adds $0.3(X_{ig} + X_{jg})$.}
\end{table}
\end{revblock}

Two further Monte Carlo studies look at heterogeneous treatment effects and at misspecification of 
the network and treatment models. In the first, the DGP for the potential outcome lets the effects 
vary with the covariate $X_{ig}$. In the second, an unobserved binary covariate affects \rev{both} 
network formation and treatment adoption. The AIPW estimator remains unbiased under effect 
heterogeneity. \rev{In the second design the unobserved covariate affects adoption and tie 
formation but not the outcome, so the score is misspecified without the covariate being a 
confounder, and the AIPW estimator stays unbiased with nominal coverage. Misspecification of the 
formation model matters only when the omitted variable also affects outcomes, and a design in which 
it does is needed to calibrate the sensitivity analysis of Section~\ref{sec:sens}. 
Appendix~\ref{app:conf} reports that design. With the unobserved covariate entering the outcome 
with coefficient $\gamma_y$, the bias of the AIPW estimator is close to $-0.12 \gamma_y$ for the 
direct effect and $-0.055 \gamma_y$ for the spillover effect, and coverage drops below $0.90$ once 
$\gamma_y$ exceeds about $0.5$, that is, once the omitted covariate shifts the outcome by half a 
standard deviation of the outcome error.} Appendix~\ref{B} gives the details of both studies.

  \section{Empirical Analysis}\label{ea}

  We now apply the method to the motivating example, the effect of U-Bridge adoption on political 
  participation in Uganda, and compare it with regression models, discussing how plausible the 
  assumptions behind each are in the setting of Section~\ref{me}.

  \subsection{Regression baselines}

  We start with regression models \rev{nested in} the second line of~(\ref{emp_mod}), where 
  $\alpha_g$ is a village fixed effect\rev{ and $\bar{X}_{-ig}$ collects the village means of the 
  other residents' covariates, the contextual variables of Section~\ref{sec:setup}}. The first 
  model ignores the network and assumes SUTVA. Since resident-level covariates may affect both 
  adoption and participation, it controls for them. Column (1) of Table~\ref{reg} shows a positive 
  and statistically significant coefficient on adoption, so adopting U-Bridge appears to raise 
  participation.
  \begin{equation}\label{emp_mod}
      \begin{gathered}
          T_{ig} = e(\{ig\}, \bm{D}_g, \bm{W}_g) = (D_{ig}, \sum_{\{jg\} \in \mathcal{N}_g} W_{ijg} D_{jg}) 
  = (D_{ig}, Z_{ig}), \\
  Y_{ig}=\alpha_g +\tau D_{ig}+\gamma Z_{ig} + \delta D_{ig} \cdot Z_{ig} + 
  X_{ig}^{\prime} \beta \rev{+ \bar{X}_{-ig}^{\prime} \beta_2} + \varepsilon_{ig}.
      \end{gathered}
  \end{equation}
  
  Residents are connected by several kinds of social tie within each village, so the next question 
  is whether one resident's adoption affects the participation of others. Such spillovers are 
  plausible because information travels through the network, and \citet{sanchez2022network} find 
  a sizable spillover from attending U-Bridge meetings on participation.\footnote{We include the 
  binary indicator of attending a U-Bridge meeting as a covariate.} To identify spillovers we 
  need an exposure mapping, and we use the first line of~(\ref{emp_mod}): participation depends 
  on a resident's own adoption and on the number of her ties who adopt, which we call her 
  adopter friends ($Z_{ig}$, the treated neighbors of Section~\ref{sec:setup}). The second model 
  adds $Z_{ig}$ as a further treatment variable, with its coefficient as the spillover effect. 
  Column (2) of Table~\ref{reg} shows a positive and sizable spillover, and the coefficient on 
  $D_{ig}$, the direct effect, falls but stays positive and significant. The third model adds the 
  interaction of $D_{ig}$ and $Z_{ig}$ to allow for effect heterogeneity. In Column (3) the 
  coefficients on $D_{ig}$ and $Z_{ig}$ both rise slightly, and the interaction is negative and 
  marginally significant, so the effect of adoption falls as $Z_{ig}$ rises. 
  \rev{The drop from Column (1) to Column (2) shows that ignoring interference loads part of the 
  spillover onto adoption, and Column (3) shows that the direct effect of adoption is not 
  constant in the number of adopter friends. These coefficients are causal only if the linear 
  model is right and adjusting for individual and contextual covariates suffices. The next 
  subsection relaxes both.}

\begin{table}[h]
  \centering
  \caption{Regression Models for the Effect of U-Bridge on Political Participation} \label{reg} 
  \begin{tabular}{l *{3}{S[table-format=-1.3, table-space-text-post={***}, table-align-text-post=false]}}
  \hline\hline
  {Outcome Variable} & \multicolumn{3}{c}{{Political Participation Index ($Y_{ig}$)}} \\
  & {Model (1)} & {Model (2)} & {Model (3)} \\ 
  \hline
  Adoption ($D_{ig}$) & 0.265***  & 0.204**  & 0.309*** \\
                  & {(0.066)}   & {(0.065)}   & {(0.084)}  \\
  Interference ($Z_{ig}$) &       & 0.044*** & 0.049*** \\
                                    &       & {(0.007)}  & {(0.007)}  \\   
  Interaction ($D_{ig} \cdot Z_{ig}$) &       &       & -0.029* \\
                                             &       &       & {(0.014)}  \\              
  Control Variable ($X_{ig}$) & $\checkmark$ & $\checkmark$ & $\checkmark$ \\
  \rev{Contextual covariates ($\bar{X}_{-ig}$)} & $\checkmark$ & $\checkmark$ & $\checkmark$ \\
  Village Fixed Effects ($\alpha_g$)    & $\checkmark$ & $\checkmark$ & $\checkmark$ \\
  Observations             & {3,018}        & {3,018}        & {3,018} \\
  \hline\hline
  \end{tabular}

  \parbox{0.8\textwidth}{
  \footnotesize{\bf Notes:} Robust standard errors clustered at the village level are in parentheses. \\
  *** $p<0.01$, ** $p<0.05$, * $p<0.1$.}
\end{table}

  \subsection{Generalized propensity score: specification and plausibility}\label{sec:expo}

Regression models are flexible and easy to fit, but their estimates can be biased by 
misspecification, and the coefficients are hard to interpret when the network itself is part of 
the exposure mapping. We therefore turn to the proposed method, \rev{which models treatment 
assignment and network formation explicitly}. The third regression specification serves as the 
outcome model of the AIPW estimator, and the network and adoption models in~(\ref{emp1}) give 
the scores.\footnote{The network effect in adoption documented by \citet{ferrali2020takes} is 
consistent with joint modeling of treatment assignment and network formation, since the causal 
chain $ D_{jg} \leftarrow X_{jg} \rightarrow W_{ijg} \leftarrow X_{ig} \rightarrow D_{ig} $ implies 
correlation between $D_{ig}$ and $D_{jg}$.}
  \begin{equation}\label{emp1}
      \begin{gathered}
          Pr(W_{ijg} = 1 \mid \bm{X}_g, \bm{U}_g) = \Phi(|X_{ig} - X_{jg}|'\beta_{w1,g} + X_{ijg}'\beta_{w2,g} + 
          \alpha_{ig} + \alpha_{jg} + \tilde{U}_{ig}'\tilde{U}_{jg} ), \\
          Pr(D_{ig} = 1 \mid \bm{X}_g) = \text{logit}^{-1}(\xi_{g} + X_{ig}'\beta_{d}).
      \end{gathered}
  \end{equation}

\begin{revblock}
\paragraph{Common exposure: a working model that the data push back on.} 
Assumption~\ref{asm:ce}(a) requires that, given covariates, a resident's adoption does not depend 
on her ties and her ties do not depend on her adoption. The data contain a direct check on the 
first half. Adding the number of adopter friends to the adoption model gives a coefficient of 
$0.35$ with a standard error of $0.057$ on the log-odds scale (profile-likelihood interval 
$[0.24, 0.47]$), so adoption is dependent among friends even given meeting attendance, the 
covariates and the village effects. Common exposure is therefore a working model, and we 
estimate the treatment-diffusion factorization of Section~\ref{sec:interdep}, with the number of 
adopter friends entering the adoption model, alongside it.

There are still reasons to report the common-exposure specification. The ties were elicited as 
family, friendship, lending and problem-solving relationships, types that are usually 
long-standing relative to a program introduced two years before the survey, so the second half of 
the assumption, that ties do not respond to adoption, is plausible, although the survey does not 
date the ties. The main channel through which information about U-Bridge travelled was the 
village meetings, and meeting attendance is in the adoption model, so part of the correlation 
between a resident's adoption and her friends' adoption that \citet{ferrali2020takes} document is 
what common exposure predicts through homophily. The score under common exposure has the closed 
form~(\ref{fact}), and, as Appendix~\ref{C} shows, the diffusion specification gives estimates 
that are close to it.

Appendix~\ref{C} implements the diffusion specification as follows. The adoption model with the 
number of adopter friends as a regressor defines the full conditionals of an autologistic model, 
of which the Firth logit is the maximum pseudo-likelihood estimate. The scores at every level 
are then recomputed by Gibbs sampling from the implied joint law of adoption given simulated 
networks, as in Section~\ref{sec:interdep}, with 2,000 joint draws per village. The resulting 
spillover estimates (Table~\ref{tab:rob}) are $0.025$ and $0.072$, against $0.034$ and $0.051$ 
under common exposure, and the direct effects on adopters are $0.37$ and $0.19$, against $0.36$ 
and $0.19$. The conclusions do not depend on the factorization. 

The diffusion specification does not, however, repair the misfit of the exposure distribution. 
The joint law implied by the estimated full conditionals puts $0.45$ of residents at zero 
adopter friends and $0.12$ at three or more, no closer to the observed $0.33$ and $0.19$ than 
common exposure. In several villages the adoption rate implied by the autologistic law also 
differs markedly from the observed one, a known property of pseudo-likelihood estimates of such 
models. Peer dependence in adoption given the observed covariates therefore does not explain the 
clustering of adopters among friends, which points to unobserved attributes shared by friends. 
The latent-augmented analysis of Section~\ref{sec:sens} supports this reading.

\paragraph{Unconfounded network formation.} 
Assumption~\ref{asm:ig} requires that no unobserved resident characteristic drives both ties (or 
adoption) and participation. The survey measures the characteristics that the ethnographic and 
quantitative literature on these villages identifies as the main sources of homophily and of 
participation: age, gender, income, education, formal leadership, phone ownership, and a 
behavioral measure of care for the community. We do not rely on observed homophily alone. The AME 
network model in~(\ref{emp1}) has additive node effects $\alpha_{ig}$, which absorb unobserved 
differences in sociability or status that lead leaders and other central residents to form more 
ties, and multiplicative latent positions $\tilde{U}_{ig}$, which absorb unobserved similarity 
among connected residents. In the terms of Section~\ref{sec:latent}, the score conditions on 
these estimated latent attributes and averages over the idiosyncratic component of ties. 
15 of the 16 villages have between 102 and 282 surveyed residents, enough to estimate 
node-level latent positions with reasonable precision 
\citep{mcfowland2023estimating, li2023statistical}. The 16th (village 11) has 28 residents, and 
its model is fitted with one multiplicative dimension and without the phone covariate, which 
predicts the ties in that village perfectly (complete separation; Appendix~\ref{C}). The 
assumption still excludes an unobserved attribute that affects participation and adoption or 
ties but leaves no trace in the network structure. Section~\ref{sec:sens} assesses sensitivity to 
this possibility.

\paragraph{Choice of the exposure mapping.} We take the effective treatment to be own adoption 
and the number of adopter friends, $Z_{ig}$. Alternatives are the share of adopter friends, 
exposure weighted by tie type, or the covariates of adopter friends. The count is the natural 
choice here for two reasons. Substantively, the mechanism is informational: each adopter friend 
is one more source of information about the technology and about the responsiveness of officials, 
and the marginal value of one more source need not depend on how many non-adopter friends one 
has. Statistically, each added dimension of the exposure thins the support. With 135 adopters 
among 3,018 residents, most residents have zero or one adopter friend, and exposure levels 
defined by shares or by tie type would have too few observations for credible overlap, the 
trade-off between richness and positivity that \citet{savje2023causal} emphasize. The count also 
has a transparent interpretation under Section~\ref{sec:cond}: $\mu(d, z)$ is the average 
participation that would obtain if a resident drawn from the population were assigned adoption 
status $d$ and $z$ adopter friends, through whatever ties the formation process produces. If the 
true exposure mechanism is richer, the estimated contrasts are exposure effects in the sense of 
\citet{savje2023causal}: well-defined averages over the finer exposures that the count 
aggregates, but no longer the effect of a single intervention. As a check, 
Section~\ref{sec:results} also reports estimates for the exposure mapping 
$(D_{ig}, \mathbbm{1}\{Z_{ig} \geq 1\})$ among residents with at least one tie, for which every unit 
can reach both levels and Proposition~\ref{prop:eff} applies, and which pools the 130 adopters 
with at least one adopter friend.
\end{revblock}
  
  We fit the network model separately for each village to allow for village-level 
  heterogeneity, using the additive and multiplicative effects model 
  (AME, \citealp[]{hoff2021additive}), the first line of~(\ref{emp1}), for the tie between a pair 
  of residents. Here $\alpha_{ig}$ and $\tilde{U}_{ig}$ are unobserved resident-level random 
  effects.\footnote{The AME model is estimated by Markov chain Monte Carlo with priors on these 
  terms, so they are random rather than fixed effects.} The controls are absolute differences in 
  age, gender, income, secondary education, and the proxy for care about the community, which 
  capture observed homophily, plus indicators for whether at least one of the pair owns a phone 
  and whether at least one holds a leadership role. The estimated coefficients vary across 
  villages, but in most of them the homophily coefficients are negative\rev{, with posterior 
  standard deviations small relative to the point estimates}, and holding a leadership role 
  raises tie formation. Table~\ref{nreg2} in Appendix~\ref{C} reports the estimates.
  
  For treatment assignment we fit a \rev{Firth-penalized} logit model, the second line 
  of~(\ref{emp1}), for adoption of U-Bridge with resident-level covariates and village fixed 
  effects ($\xi_{g}$) for unobserved village-level heterogeneity\rev{; the penalty guards against 
  separation with 135 adopters and 16 village effects}. Table~\ref{reg2} in Appendix~\ref{C} 
  reports the estimates. Age, gender, education and leadership have statistically significant 
  coefficients, income does not, and \rev{owning} a phone and attending a U-Bridge meeting are 
  strong predictors of adoption. 

  We then compute the score from~(\ref{fact}) with the fitted network and adoption 
  models\rev{. The latent node effects are handled by the plug-in of Section~\ref{sec:latent}. 
  The link probability of each pair is the posterior-predictive probability of the tie, 
  $\hat{p}_{ijg} = \mathbbm{E}[\Phi(\eta_{ijg}) \mid \bm{W}_g]$, where $\eta_{ijg}$ is the AME linear 
  predictor including the latent effects, computed as the posterior mean of the simulated tie 
  over the stored MCMC scans. These probabilities enter the closed-form score~(\ref{fact}) 
  directly, at every level $(d, z)$ with $z$ up to the largest observed value. The plug-in 
  integrates over the posterior of $\bm{U}_g$ given the observed network, not over the prior as 
  in~(\ref{gpslat}). Plugging the posterior mean of $\eta_{ijg}$ into $\Phi(\cdot)$ instead would 
  under-predict degrees by about fifteen percent, because the posterior mean of the 
  multiplicative effects is shrunk by their rotational invariance; the posterior-predictive 
  probabilities reproduce the observed degree of every resident almost exactly (correlation above 
  $0.99$ in every village). Through the additive effects the score therefore conditions on each 
  resident's realized degree, a point we come back to when comparing it with the 
  network-conditional score. The main specification does not include $\hat{\bm{U}}_g$ in the 
  adoption and outcome models, so it is valid if the latent effects affect ties only; 
  Section~\ref{sec:sens} relaxes this.} 

  As a diagnostic, the average of the estimated score over residents, 
  $N^{-1} \sum \hat{r}(\{ig\}, t; \bm{X}_g)$, should reproduce the marginal frequency of each 
  exposure level, since $\mathbbm{E}[r(\{ig\}, t; \bm{X}_g)] = Pr(T_{ig} = t)$. \rev{Figure~\ref{eme} in 
  Appendix~\ref{C} compares the implied distribution of the number of adopter friends with its 
  empirical frequency. Under common exposure the fitted models reproduce the share of residents 
  with exactly one adopter friend ($0.32$ in both) but put more mass at zero ($0.42$ against 
  $0.33$) and less at three or more ($0.12$ against $0.19$): adopters are more clustered among 
  friends than independent adoption given covariates implies. This is the network effect in 
  adoption that \citet{ferrali2020takes} document, and it is the evidence against 
  Assumption~\ref{asm:ce}(a) discussed above. The figure checks the fit of the models rather than 
  balance; Table~\ref{tab:bal} in Appendix~\ref{C} reports balance with the weights and samples 
  the estimators use, and Section~\ref{sec:results} discusses it with the results.}

  \subsection{Estimated direct and spillover effects}\label{sec:results}

\begin{revblock}
    \paragraph{Estimands.} We report two kinds of effect, written for a resident with $z$ 
    adopter friends:
    \begin{equation}\label{est:app}
    \tau^{s}(z) = \tau((0, z), (0, 0)), \qquad 
    \tau^{e}(z) = \mathbbm{E}[Y_{ig}(1, z) - Y_{ig}(0, z) \mid T_{ig} = (1, z)].
    \end{equation}
    The spillover effect $\tau^{s}(z)$ is the population effect, for non-adopters, of having $z$ 
    adopter friends instead of none. The direct effect $\tau^{e}(z)$ is the effect on the 
    exposed~(\ref{est:exp}): the effect of adoption on the adopters who have $z$ adopter friends. 
    We also report both for the exposure ``at least one adopter friend,'' written 
    $\tau^{s}(\text{any})$ and $\tau^{e}(\text{any})$. Estimators are labeled as in 
    Section~\ref{SMC}: AIPW-JM and IPW-JM use the marginal score, with and without the outcome 
    model; AIPW-COND uses the network-conditional score; AIPW-ATT and IPW-ATT use the odds weights 
    that identify the effect on the exposed, with and without the outcome model; and Regression is 
    the coefficient contrast from Column (3) of Table~\ref{reg}, which coincides with regression 
    imputation under that linear model. Table~\ref{tab:appest} in Appendix~\ref{C} gives every 
    estimate with its standard error and interval, and Figure~\ref{te} plots them. In the text we 
    quote estimates with village-clustered standard errors in parentheses.

    \paragraph{Why two kinds of estimand.} Overlap is ample for the spillover component and thin 
    for adoption. Non-adopters with zero, one or two adopter friends are common, so the population 
    spillover effects are the natural targets. Adoption is rare: the median estimated adoption 
    probability is $0.007$, so the score at any adopter level $(1, z)$ is below $0.05$ for more 
    than ninety-five percent of residents, and a population direct effect $\mu(1, z) - \mu(0, z)$ 
    cannot be estimated by weighting without extrapolating from a handful of adopters to everyone. 
    (Only 168 residents have both $r(\{ig\}, (1,1); \bm{X}_g)$ and $r(\{ig\}, (0,1); \bm{X}_g)$ in 
    $[0.05, 0.95]$; Table~\ref{tab:appest} reports the resulting estimate.) For the direct 
    effects we therefore target the effect on the exposed, which under Assumption~\ref{asm:ig} is 
    identified by weighting non-adopters with $z$ adopter friends by the odds 
    $r(\{ig\}, (1,z); \bm{X}_g) / r(\{ig\}, (0,z); \bm{X}_g)$. These odds are small and stable (they 
    sum to $24$ over the 934 non-adopters with one adopter friend and to $26$ over the 463 with 
    two), so no trimming is needed, and the estimator is the usual doubly robust estimator of an 
    effect on the treated. The scores are computed at every observed level, so no resident is 
    excluded because of her observed exposure.

    \paragraph{Trimming.} For the spillover effects a resident is kept if both scores 
    $r(\{ig\}, (0,z); \bm{X}_g)$ and $r(\{ig\}, (0,0); \bm{X}_g)$ lie in $[0.05, 0.95]$, a restriction 
    by covariate profile in the sense of \citet{crump2009dealing} (Table~\ref{tab:trim} in 
    Appendix~\ref{C}). For $\tau^{s}(1)$ this keeps 2,803 residents; the 215 excluded have a high 
    probability of adopting themselves (two thirds have secondary education, a third are leaders, 
    forty percent attended a meeting). For $\tau^{s}(2)$ it keeps 2,132; the 886 excluded are 
    residents for whom two adopter friends are unlikely given their expected number of adopter 
    friends, mostly residents with few ties, and their covariates are close to those of the kept 
    residents. The estimands are the population effects on these trimmed populations.

    \paragraph{Direct effects on adopters.} The right panel of Figure~\ref{te} shows the 
    picture: the direct effect is about $0.3$ on every estimator. The effect of adoption on the 
    17 adopters with one adopter friend is $0.36$ ($0.09$) by AIPW-ATT, $0.42$ ($0.08$) by IPW-ATT 
    and $0.28$ ($0.07$) by Regression. For the 28 adopters with two adopter friends the estimates 
    are $0.19$ ($0.11$), $0.23$ ($0.16$) and $0.25$ ($0.07$). For the 130 adopters with at least 
    one adopter friend they are $0.27$ ($0.07$) and $0.25$ ($0.09$) by AIPW-ATT and IPW-ATT. The 
    odds weights balance the covariates well (Table~\ref{tab:bal}). Before weighting, adopters 
    with one adopter friend differ from non-adopters with one adopter friend by $1.75$ standard 
    deviations in education and $1.40$ in meeting attendance; after weighting the differences are 
    $0.46$ and $0.55$, and all others are below $0.2$. For two adopter friends all weighted 
    differences are below $0.35$, and for ``any'' below $0.15$. The AIPW-ATT estimate splits into 
    the mean residual of the adopters' outcomes from the outcome model ($0.41$ for $z = 1$, $0.10$ 
    for $z = 2$) and a weighted correction from the comparison non-adopters ($-0.05$ and 
    $+0.08$). The weighting evidence rests on 17 and 28 adopters. It agrees with the regression, 
    but it cannot tell an effect that declines with the number of adopter friends from a 
    constant one.

    \paragraph{Spillover effects.} The left panel of Figure~\ref{te} shows a positive but small 
    spillover that the regression overstates at two adopter friends. For non-adopters, having one 
    adopter friend instead of none raises participation by $0.034$ ($0.016$) by AIPW-JM, $0.018$ 
    ($0.020$) by AIPW-COND, $0.023$ ($0.027$) by IPW-JM and $0.049$ ($0.007$) by Regression. 
    Having at least one adopter friend raises it by $0.052$ ($0.016$) by AIPW-JM. For two adopter 
    friends the AIPW-JM estimate is $0.051$ ($0.034$) and the Regression estimate $0.098$ 
    ($0.013$). The AIPW-JM estimates split into the regression contrast ($0.049$ and $0.098$) and 
    a weighting correction of $-0.015$ and $-0.047$. The weighting evidence pulls both spillovers 
    below the regression's values, and more strongly at two adopter friends, where the regression 
    imposes that the effect is exactly twice the effect of one. Weighting balances the covariates 
    across the spillover levels, with all weighted standardized differences below $0.22$. 

    Substantively, a non-adopter's participation rises by about $0.03$ to $0.05$ when a friend 
    adopts, an effect distinguishable from zero for ``at least one adopter friend'' and marginally 
    so for ``exactly one.'' The evidence that a second adopter friend adds as much again comes from 
    the linearity of the regression, not from the residents with two adopter friends. The AIPW-COND 
    estimates are close to the AIPW-JM estimates, but this says nothing about 
    Proposition~\ref{prop:dec}. Through the additive effects of the AME model, the plug-in score 
    reproduces each resident's realized degree almost exactly, so the two scores carry the same 
    information about ties, and the same outcome model dominates both estimators.

    \paragraph{Inference with 16 villages.} Table~\ref{tab:appinf} in Appendix~\ref{C} reports 
    three checks on the cluster-robust standard errors. Intervals with $t_{15}$ critical values 
    are wider by about $0.003$ at each end for the spillovers and $0.02$ for the direct effects. 
    A wild cluster bootstrap that re-estimates the outcome model under Rademacher-perturbed 
    residuals by village, with the scores held fixed, gives standard errors within ten percent 
    of the cluster-robust ones, so estimating the outcome model adds little to the uncertainty 
    once the influence function is clustered. The within-village estimator~(\ref{hac2}), a 
    heuristic here (Remark~\ref{rem:few}), gives larger standard errors for the spillovers 
    ($0.023$ for $\tau^{s}(1)$ against $0.016$), and with it the interval for $\tau^{s}(1)$ 
    includes zero. Our reading is that the spillover of one adopter friend is positive but small 
    enough that its statistical significance depends on the variance estimator, while the direct 
    effect on adopters with one adopter friend does not.
\end{revblock}

   \begin{figure}[!ht]
   \begin{center}
   \includegraphics[width=0.95\textwidth]{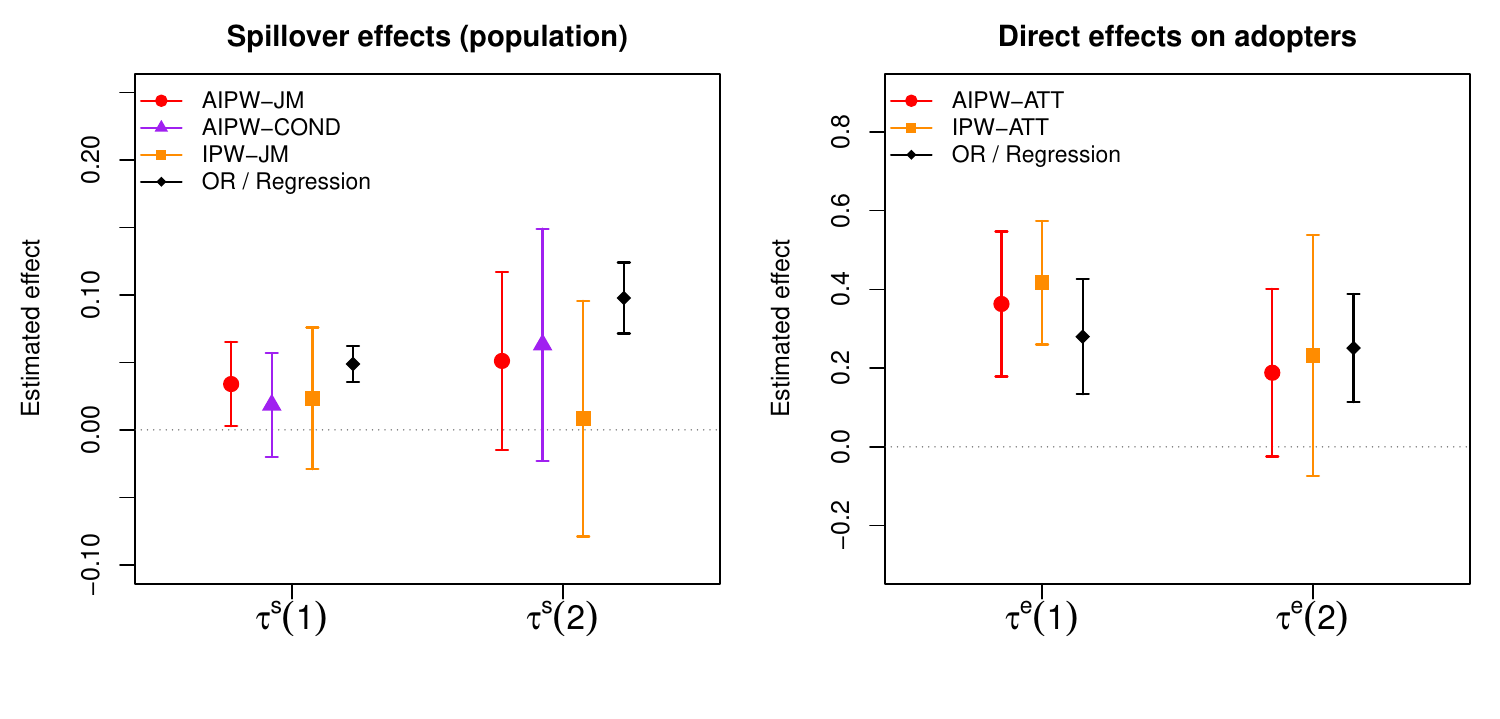}
   \end{center}
   \caption{Estimated Effects of U-Bridge Adoption on Political Participation}\label{te}
   \footnotesize\textbf{Note:} \rev{Point estimates and 95\% intervals with village-clustered 
   variances. Left: population spillover effects $\tau^{s}(1)$ and $\tau^{s}(2)$ by AIPW-JM 
   (marginal score), AIPW-COND (network-conditional score), IPW-JM, and Regression (the 
   coefficient contrast, identical to regression imputation). Right: direct effects on adopters 
   with one and two adopter friends, $\tau^{e}(1)$ and $\tau^{e}(2)$, by AIPW-ATT (odds weights 
   with the outcome model), IPW-ATT (odds weights alone), and Regression. Estimates and standard 
   errors are in Table~\ref{tab:appest}.}
   \end{figure}

\begin{revblock}
  \subsection{Sensitivity to unobserved confounding}\label{sec:sens}

  Assumption~\ref{asm:ig} cannot be tested. We assess sensitivity in two ways that use the 
  structure of the model (Tables~\ref{tab:appsens} and~\ref{tab:sensgrid} in Appendix~\ref{C}). 
  First, the latent node effects estimated by the AME model are the network's own evidence about 
  unobserved resident attributes. Re-estimating the adoption and outcome models with the posterior 
  means of $(\alpha_{ig}, \tilde{U}_{ig})$ added as covariates, and recomputing the scores with the 
  new adoption probabilities, implements the latent-augmented version of Assumption~\ref{asm:ig} 
  described in Section~\ref{sec:latent}. The direct effects on adopters move little ($0.32$ and 
  $0.17$ against $0.36$ and $0.19$). The spillover effects fall to $0.010$ ($0.022$) and $0.028$ 
  ($0.034$) from $0.034$ and $0.051$, and neither remains distinguishable from zero. Residents 
  who are more sociable than their covariates predict (large $\alpha_{ig}$) have both more 
  adopter friends and higher participation, and once this is adjusted for, less of the 
  association between adopter friends and participation remains. We regard this as the main 
  caveat on the spillover estimates.

  Second, for attributes that leave no trace in the network, we consider a hypothetical binary 
  attribute $U_{ig}$ with $Pr(U_{ig} = 1) = 1/2$ that lowers the adoption log-odds by $\gamma_d$ 
  and raises participation by $\gamma_y$. Given the fitted adoption model, the conditional 
  distribution of $U_{ig}$ given each resident's observed adoption and covariates follows from 
  Bayes' rule. We draw $U_{ig}$ from it, add $\gamma_y U_{ig}$ to the outcome, re-estimate ignoring 
  $U_{ig}$, and average the shift in the estimates over draws. The realized network and the scores 
  are held fixed, so the exercise measures confounding of adoption only, and the fitted log-odds 
  are treated as those of residents with $U_{ig} = 0$. The shift in the direct effects on adopters 
  is 
  \begin{equation}\label{sensshift}
  -\gamma_y \left[ Pr(U_{ig} = 1 \mid D_{ig} = 1, \bm{X}_g) - Pr(U_{ig} = 1 \mid D_{ig} = 0, \bm{X}_g) \right], 
  \end{equation}
  which for rare adoption is about $-0.13 \gamma_y$, $-0.22 \gamma_y$ and $-0.41 \gamma_y$ at 
  $\gamma_d = 0.5$, $1$ and $2$, and cannot exceed $-0.5 \gamma_y$ in magnitude. A confounder as 
  strong as secondary education in the adoption model ($\gamma_d = 2$) that also raises 
  participation by $0.5$ reduces $\hat{\tau}^{e}(1)$ from $0.36$ to $0.16$; bringing it to zero 
  would take $\gamma_y$ of at least $0.7$ whatever the value of $\gamma_d$. The spillover effects 
  do not move in this exercise because it leaves ties untouched. The Monte Carlo design of 
  Appendix~\ref{app:conf}, in which the omitted covariate also drives tie formation, shows that 
  such a confounder biases the spillover effects as well. Its magnitudes do not transfer to the 
  villages, because the bias depends on how strongly the confounder shifts link probabilities and 
  on the network density, both of which differ between the simulated groups of four and villages 
  with a mean degree of fifteen.
\end{revblock}

\section{Conclusion}\label{sec:conc}

This paper proposes a generalized propensity score for treatment effects under interference when 
social ties, like the treatment, depend on covariates that affect outcomes. The score is built 
from a model for treatment assignment and a model for network formation, and average potential 
outcomes and treatment effects are estimated by augmented inverse probability weighting.

\begin{revblock}
The main conceptual point is that treating the network as a component of a joint assignment 
turns the choice between conditioning on the realized network and averaging over it into a 
choice of estimand and assumptions, in the same way as the choice between effects on the treated 
and average effects. Conditioning requires less and is robust to latent homophily that does not 
affect adoption. Averaging requires network formation to be unconfounded given observed 
covariates, and in exchange delivers a population estimand, effects at exposure levels that 
poorly connected units cannot reach, and, where both routes target the same estimand, a weakly 
smaller variance. Latent-variable network models let one take a position in between. 

For practice, the paper's advice is specific. Choose the estimand that matches the question, 
whether what the program did among the units who could have been exposed or what it would do in 
a population whose ties have not been observed, and then choose the score that identifies it. 
Check the fitted assignment models against the observed exposure distribution, and re-estimate 
with the latent node effects of the network model added as covariates, which in the application 
left the direct effects intact and halved the spillovers.
\end{revblock}

Monte Carlo studies examine the finite sample properties of the estimator, and an application to 
the adoption of a political communication technology in Uganda compares it with regression. 
\rev{In the application the direct effects agree across estimators, while the spillover estimates 
differ because the regression imposes linearity in the number of adopter friends. The weighting 
evidence on direct effects is thin because adoption is rare, and the spillover estimates are 
sensitive to the latent node attributes recovered from the network.}

The method has limitations. It \rev{focuses} on multiple disjoint groups rather than a single 
network. \rev{Section~\ref{sec:single} shows that identification, positivity and estimation of the 
score extend to a single sparse network, but the inference results there are a sketch, and a 
variance estimator that is valid when only the score is correctly specified remains open.} It 
relies on parametric forms for the treatment assignment mechanism, the network formation process 
and the exposure mapping, and misspecification of any of these biases the score. \rev{In the 
many-small-groups regime it cannot accommodate latent node attributes that affect outcomes, and 
the plug-in of Section~\ref{sec:latent} handles them only heuristically in large groups.} It is 
more computationally demanding than approaches that estimate the score directly from a 
parametric model, which limits it to moderately sized datasets. And it is designed for 
cross-sectional data in which the network is random but static. Extending the framework to 
feedback among treatment assignment, network dynamics and outcomes in longitudinal network data 
\citep{snijders2017stochastic} is left for future work.

\paragraph{Acknowledgments}
We used Claude and Codex, the publicly available 2026 versions, for code assistance, copy editing, and checking mathematical proofs. We verified all outputs. The author is entirely responsible for the scientific content of the article.

\newpage
\bibliographystyle{chicago}
\bibliography{liter}

\clearpage

  \newpage
 \begin{center}
 {\large\bf SUPPLEMENTARY INFORMATION for \\ \textit{Identification and Estimation of Treatment Effects from Network Data with Confounded Social Links}}
 \end{center}

 \appendix

\section{Proofs}\label{A}

    \subsection{Proof of Proposition~\ref{prop:b} (balancing property)} 
    
    First, since the GPS $r(\{ig\}, t; \bm{X}_g)$ is a function of $\bm{X}_g$, 
\begin{equation}
    \begin{aligned}
        & Pr(T_{ig} = t|\bm{X}_g, r(\{ig\}, t; \bm{X}_g)) \\
        = &\mathbbm{E}\left[\mathbbm{1}\{e(\{ig\}, \bm{D}_g, \bm{W}_g) = t \}| \bm{X}_g,  
        r(\{ig\}, t; \bm{X}_g) \right] \\
        = &\mathbbm{E}\left[\mathbbm{1}\{e(\{ig\}, \bm{D}_g, \bm{W}_g) = t \}| \bm{X}_g \right] \\
        = & Pr(T_{ig} = t|\bm{X}_g) \\
        = & r(\{ig\}, t; \bm{X}_g).
    \end{aligned}
\end{equation}
By the law of iterated expectation, we have
\begin{equation}
    \begin{aligned}
        & Pr(T_{ig} = t|r(\{ig\}, t; \bm{X}_g)) \\
        = &\mathbbm{E}\left[\mathbbm{1}\{e(\{ig\}, \bm{D}_g, \bm{W}_g) = t \}| r(\{ig\}, t; \bm{X}_g) \right] \\
        = &\mathbbm{E}\left[\mathbbm{E}[\mathbbm{1}\{e(\{ig\}, \bm{D}_g, \bm{W}_g) = t \}| \bm{X}_g] | 
        r(\{ig\}, t; \bm{X}_g) \right] \\
        = & \mathbbm{E}\left[r(\{ig\}, t; \bm{X}_g)|r(\{ig\}, t; \bm{X}_g)\right] \\
        = & r(\{ig\}, t; \bm{X}_g).
    \end{aligned}
\end{equation}
Therefore, $Pr(T_{ig} = t|\bm{X}_g, r(\{ig\}, t; \bm{X}_g)) = Pr(T_{ig} = t|r(\{ig\}, t; \bm{X}_g)) $. 
Next, we show the balancing property of the GPS.
\begin{equation}\label{bal00}
    \begin{aligned}
        & Pr(T_{ig} = t| Y_{ig}(t),  r(\{ig\}, t; \bm{X}_g)) \\
        = &\mathbbm{E}\left[\mathbbm{1}\{e(\{ig\}, \bm{D}_g, \bm{W}_g) = t \}|Y_{ig}(t),  
        r(\{ig\}, t; \bm{X}_g) \right] \\
        = &\mathbbm{E}\left[ \mathbbm{E}
        [\mathbbm{1}\{e(\{ig\}, \bm{D}_g, \bm{W}_g)  = t \}|\bm{X}_g, Y_{ig}(t), 
        r(\{ig\}, t; \bm{X}_g)] | 
        Y_{ig}(t), r(\{ig\}, t; \bm{X}_g)) \right] \\
        = &\mathbbm{E}\left[ \mathbbm{E}
        [\mathbbm{1}\{e(\{ig\}, \bm{D}_g, \bm{W}_g)  = t \}| \bm{X}_g, Y_{ig}(t)] |  
        Y_{ig}(t), r(\{ig\}, t; \bm{X}_g) \right] \\
        = &\mathbbm{E}\left[ \mathbbm{E}
        [\mathbbm{1}\{e(\{ig\}, \bm{D}_g, \bm{W}_g)  = t \}| \bm{X}_g] | 
        Y_{ig}(t), r(\{ig\}, t; \bm{X}_g)\right] \\
         = &\mathbbm{E}\left[ r(\{ig\}, t; \bm{X}_g)| 
        Y_{ig}(t), r(\{ig\}, t; \bm{X}_g) \right] \\
        = & r(\{ig\}, t; \bm{X}_g).
    \end{aligned}
\end{equation}
We have the third line in Equation~(\ref{bal00}) by the law of iterated expectation, and then the 
fourth line because $\bm{X}_g$ contains all information about GPS 
$r(\{ig\}, t; \bm{X}_g)$. By Assumption~\ref{asm:ig}, we have the fifth line.

\subsection{Proof of Proposition~\ref{prop:a} (identification)}

\begin{equation}
    \begin{aligned}
       & \mathbbm{E}\{ \frac{Y_{ig}\mathbbm{1}\{ T_{ig} = t \} }{ r(\{ig\}, t; \bm{X}_g)} \}  \\
   = & \mathbbm{E}\{ \mathbbm{E}\{ \frac{Y_{ig}\mathbbm{1}\{ T_{ig} = t \} }{ r(\{ig\}, t; \bm{X}_g)} |\bm{X}_g \} \} \\
   = & \mathbbm{E}\{ \mathbbm{E}\{ \frac{Y_{ig}}{ r(\{ig\}, t; \bm{X}_g)} |T_{ig} = t, \bm{X}_g \} Pr(T_{ig} = t|\bm{X}_g) \} \\
   = & \mathbbm{E}\{ \mathbbm{E}\{ \frac{Y_{ig}}{ r(\{ig\}, t; \bm{X}_g)} |T_{ig} = t, \bm{X}_g \} r(\{ig\}, t; \bm{X}_g) \} \\
   = & \mathbbm{E}\{ \mathbbm{E}\{ Y_{ig} |T_{ig} = t, \bm{X}_g \}  \} \\
   = & \mathbbm{E}\{ \mathbbm{E}\{ Y_{ig}(t) | \bm{X}_g \}  \} \\
   = &  \mathbbm{E}\{ Y_{ig}(t)\} \\
   = & \mu(t).
    \end{aligned}
\end{equation}
\rev{The chain of equalities holds conditionally on $\bm{X}_g$ as well, which gives the first 
statement of Proposition~\ref{prop:a}; summing the conditional identity over the units of the group, 
taking expectations and dividing by $\mathbbm{E}[N_g]$ gives the pooled statement with $\mu(t)$ as 
defined in~(\ref{mudef}).} Therefore, average potential outcome $\mu(t)$ is identified given the GPS $r(\{ig\}, t; \bm{X}_g)$.

\begin{revblock}
\subsection{Proof of Proposition~\ref{prop:dec} (realized-sample, feasible-population and population effects)}\label{app:dec}

Throughout, Assumption~\ref{asm:ig} implies $Y_{ig}(t) \indep \bm{W}_g \mid \bm{X}_g$ (decomposition 
property), so $\mathbbm{E}[Y_{ig}(t) \mid \bm{X}_g, \bm{W}_g] = \mathbbm{E}[Y_{ig}(t) \mid \bm{X}_g] = m_t(\bm{X}_g)$ 
and $\mathbbm{E}[Y_{ig}(t) - Y_{ig}(t') \mid \bm{X}_g, \bm{W}_g] = \tau(t, t' \mid \bm{X}_g)$.

(a) Substituting into~(\ref{est:sample}) gives $\tau^{S}_g(t, t') = N_g^{-1} \sum_i \tau(t, t' \mid \bm{X}_g)$, 
which is a function of $\bm{X}_g$ alone. Multiplying by $N_g$, taking expectations and dividing by 
$\mathbbm{E}[N_g]$ gives $\mathbbm{E}[\sum_i \tau(t, t' \mid \bm{X}_g)] / \mathbbm{E}[N_g] = \tau(t, t')$ by 
the definition~(\ref{mudef}). If $N_g$ is constant, or independent of $\tau^{S}_g(t, t')$, then 
$\mathbbm{E}[N_g \tau^{S}_g] = \mathbbm{E}[N_g] \mathbbm{E}[\tau^{S}_g]$ and 
$\mathbbm{E}[\tau^{S}_g(t, t')] = \tau(t, t')$.

(b) The event $\mathcal{A}_{ig}(t, t')$ is a function of $(\bm{X}_g, \bm{W}_g)$. By iterated 
expectations, 
$\mathbbm{E}[Y_{ig}(t) - Y_{ig}(t') \mid \mathcal{A}_{ig}] = 
\mathbbm{E}[\mathbbm{E}[Y_{ig}(t) - Y_{ig}(t') \mid \bm{X}_g, \bm{W}_g] \mid \mathcal{A}_{ig}] = 
\mathbbm{E}[\tau(t, t' \mid \bm{X}_g) \mid \mathcal{A}_{ig}]$. If $\mathbbm{1}\{\mathcal{A}_{ig}\} \indep \bm{X}_g$, 
then $\mathbbm{E}[\tau(t, t' \mid \bm{X}_g) \mid \mathcal{A}_{ig}] = \mathbbm{E}[\tau(t, t' \mid \bm{X}_g)] = \tau(t, t')$; 
if $\tau(t, t' \mid \bm{X}_g) \equiv c$, both sides equal $c$.

(c) Identical to (b) with the conditioning event $\{\bm{W}_g = \bm{w}_g\}$, which has positive 
probability by assumption; if $\bm{W}_g \indep \bm{X}_g$, $\mathbbm{E}[h(\bm{X}_g) \mid \bm{W}_g = \bm{w}_g] = \mathbbm{E}[h(\bm{X}_g)]$ 
for every integrable $h$. \hfill $\square$

\subsection{Proof of Proposition~\ref{prop:idpos} (identification with and without conditioning)}

(a) On the event $\{r^{c}(\{ig\}, t; \bm{X}_g, \bm{W}_g) > 0\}$ the conditional expectation given 
$T_{ig} = t$ is well defined. By consistency and Assumption~\ref{asm:em}, 
$Y_{ig} = Y_{ig}(t)$ on $\{T_{ig} = t\}$, so 
$\mathbbm{E}[Y_{ig} \mid T_{ig} = t, \bm{X}_g, \bm{W}_g] = \mathbbm{E}[Y_{ig}(t) \mid T_{ig} = t, \bm{X}_g, \bm{W}_g]$. 
Since $T_{ig}$ is a function of $(\bm{D}_g, \bm{W}_g)$, Assumption~\ref{asm:ig2} gives 
$Y_{ig}(t) \indep T_{ig} \mid \bm{X}_g, \bm{W}_g$, so the last expression equals 
$\mathbbm{E}[Y_{ig}(t) \mid \bm{X}_g, \bm{W}_g]$. If every unit in the group has 
$\mathcal{A}_{ig}(t, t')$, then $\tau^{S}_g(t, t') = N_g^{-1} \sum_i (\mathbbm{E}[Y_{ig} \mid T_{ig} = t, \bm{X}_g, \bm{W}_g] 
- \mathbbm{E}[Y_{ig} \mid T_{ig} = t', \bm{X}_g, \bm{W}_g])$ is identified. In general, on 
$\{r^{c} > 0\}$,
$$
\mathbbm{E}\left[ \frac{Y_{ig} \mathbbm{1}\{T_{ig} = t\}}{r^{c}(\{ig\}, t; \bm{X}_g, \bm{W}_g)} \,\Big|\, \bm{X}_g, \bm{W}_g \right] 
= \frac{\mathbbm{E}[Y_{ig}(t) \mid \bm{X}_g, \bm{W}_g] \, r^{c}(\{ig\}, t; \bm{X}_g, \bm{W}_g)}{r^{c}(\{ig\}, t; \bm{X}_g, \bm{W}_g)}
= \mathbbm{E}[Y_{ig}(t) \mid \bm{X}_g, \bm{W}_g],
$$
and averaging over $(\bm{X}_g, \bm{W}_g)$ restricted to $\{r^{c} > 0\}$ identifies 
$\mathbbm{E}[Y_{ig}(t) \mid r^{c} > 0]$; averaging over $\mathcal{A}_{ig}(t, t')$ for both levels identifies 
$\tau(t, t' \mid \mathcal{A})$.

(b) Assumption~\ref{asm:ig} implies Assumption~\ref{asm:ig2} by the weak union property of 
conditional independence: $Y \indep (A, B) \mid C$ implies $Y \indep A \mid (B, C)$. The statement 
for $r$ is Proposition~\ref{prop:a}. The statement for $r^{c}$ follows from (a): the identified 
quantity is $\mathbbm{E}[Y_{ig}(t) \mid r^{c} > 0]$, which equals $\mu(t)$ when $Pr(r^{c} > 0) = 1$, and 
in general differs from it.

(c) For the exposure mapping~(\ref{epp}), $Z_{ig} = \sum_{\{jg\}} W_{ijg} D_{jg} \leq \deg_{ig}(\bm{W}_g)$ 
almost surely given $\bm{W}_g$, so $Pr(Z_{ig} = z \mid \bm{X}_g, \bm{W}_g) = 0$ whenever 
$z > \deg_{ig}(\bm{W}_g)$; under a dyadic model with link probabilities below one, 
$Pr(\deg_{ig}(\bm{W}_g) < z) > 0$ for every $z \geq 1$. For the marginal score, fix $d$ and 
$z \leq N_g - 1$ and choose any configuration $(\bm{d}^*_g, \bm{w}^*_g)$ with $d^*_{ig} = d$ and exactly 
$z$ units $\{jg\} \neq \{ig\}$ having $d^*_{jg} = 1$ and $w^*_{ijg} = 1$, all other $w^*_{ijg} = 0$. Then 
$r(\{ig\}, (d, z); \bm{X}_g) = \sum_{(\bm{d}_g, \bm{w}_g): e(\{ig\}, \bm{d}_g, \bm{w}_g) = (d, z)} 
f_{\bm{D}_g, \bm{W}_g \mid \bm{X}_g}(\bm{d}_g, \bm{w}_g) \geq 
f_{\bm{D}_g, \bm{W}_g \mid \bm{X}_g}(\bm{d}^*_g, \bm{w}^*_g) > 0$ by full support. \hfill $\square$

\subsection{Proof of Proposition~\ref{prop:eff} (efficiency)}\label{app:eff}

Write, for $s \in \{t, t'\}$ and a generic score $\varrho_{s, ig}$ equal to either 
$r(\{ig\}, s; \bm{X}_g)$ or $r^{c}(\{ig\}, s; \bm{X}_g, \bm{W}_g)$, and with 
$m_s(\bm{X}_g) = \mu(\{ig\}, s, \bm{X}_g)$,
$$
A_{s, ig} = \frac{\mathbbm{1}\{T_{ig} = s\} (Y_{ig} - m_s(\bm{X}_g))}{\varrho_{s, ig}}, 
\qquad 
B_{ig} = m_t(\bm{X}_g) - m_{t'}(\bm{X}_g) - \tau(t, t'),
$$
so that $\varphi_{ig} = A_{t, ig} - A_{t', ig} + B_{ig}$ for either estimator, and 
$\sum_i \varphi_{ig} = \sum_i (A_{t, ig} - A_{t', ig}) + \sum_i B_{ig}$. We use three facts.

\textit{Fact 1.} $\mathbbm{E}[A_{s, ig} \mid \bm{X}_g, \bm{W}_g, \bm{D}_g] = 0$ for both scores. Indeed, 
on $\{T_{ig} = s\}$ we have $Y_{ig} = Y_{ig}(s)$, and by Assumption~\ref{asm:ig}, 
$\mathbbm{E}[Y_{ig}(s) \mid \bm{X}_g, \bm{W}_g, \bm{D}_g] = m_s(\bm{X}_g)$. Since $\varrho_{s, ig}$ and 
$\mathbbm{1}\{T_{ig} = s\}$ are functions of $(\bm{X}_g, \bm{W}_g, \bm{D}_g)$, the claim follows.

\textit{Fact 2.} For $i \neq j$ and any $s, s'$, $\mathbbm{E}[A_{s, ig} A_{s', jg}] = 0$. Conditioning on 
$(\bm{X}_g, \bm{W}_g, \bm{D}_g)$, the product of the indicators and scores factors out, and 
$\mathbbm{E}[(Y_{ig}(s) - m_s(\bm{X}_g))(Y_{jg}(s') - m_{s'}(\bm{X}_g)) \mid \bm{X}_g, \bm{W}_g, \bm{D}_g] 
= \mathbbm{E}[(Y_{ig}(s) - m_s(\bm{X}_g))(Y_{jg}(s') - m_{s'}(\bm{X}_g)) \mid \bm{X}_g] = 0$ by 
Assumption~\ref{asm:ig}, which is stated for the collection of potential outcomes in the group and 
therefore applies to the pair $(Y_{ig}(s), Y_{jg}(s'))$, and by the within-group conditional 
independence of potential outcomes. For $i = j$ and $s \neq s'$ the product of indicators is zero.

\textit{Fact 3.} $\mathbbm{E}[A_{s, ig} B_{jg}] = 0$ for all $i, j$, by Fact 1, since $B_{jg}$ is a 
function of $\bm{X}_g$.

By Facts 1--3, 
$$
\text{var}\Big(\sum_i \varphi_{ig}\Big) = \sum_{s \in \{t, t'\}} \sum_i \mathbbm{E}[A^2_{s, ig}] + 
\text{var}\Big(\sum_i B_{ig}\Big),
$$
and the second term is the same for both scores. For the diagonal terms, 
$\mathbbm{E}[(Y_{ig} - m_s(\bm{X}_g))^2 \mid T_{ig} = s, \bm{X}_g, \bm{W}_g] = 
\text{var}(Y_{ig}(s) \mid \bm{X}_g) = \sigma^2_s(\bm{X}_g)$ by Assumption~\ref{asm:ig}, so
$$
\mathbbm{E}[A^2_{s, ig}] = \mathbbm{E}\left[ \sigma^2_s(\bm{X}_g) \, 
\mathbbm{E}\left[ \frac{\mathbbm{1}\{T_{ig} = s\}}{\varrho_{s, ig}^2} \,\Big|\, \bm{X}_g \right] \right].
$$
With the marginal score, $\varrho_{s, ig} = r(\{ig\}, s; \bm{X}_g)$ is $\bm{X}_g$-measurable and 
$\mathbbm{E}[\mathbbm{1}\{T_{ig} = s\} \mid \bm{X}_g] = r$, so the inner expectation is $1 / r(\{ig\}, s; \bm{X}_g)$. 
With the network-conditional score, $\mathbbm{E}[\mathbbm{1}\{T_{ig} = s\} \mid \bm{X}_g, \bm{W}_g] = r^{c}$, so 
the inner expectation is $\mathbbm{E}[1 / r^{c}(\{ig\}, s; \bm{X}_g, \bm{W}_g) \mid \bm{X}_g]$. 
Because $\mathbbm{E}[r^{c} \mid \bm{X}_g] = r$ and $x \mapsto 1/x$ is strictly convex on $(0, 1]$, 
Jensen's inequality gives $\mathbbm{E}[1 / r^{c} \mid \bm{X}_g] \geq 1 / r$, with equality if and 
only if $r^{c}$ is almost surely constant given $\bm{X}_g$. Summing over $i$ and $s$ 
gives~(\ref{eff}) and the strictness claim. \hfill $\square$

\subsection{Proof of Lemma~\ref{lem:gps} (rate of the estimated score)}\label{app:gps}

Fix a group with $N_g \leq \bar{N}$. Under Assumption~\ref{asm:ce}, 
$f_{\bm{D}_g, \bm{W}_g \mid \bm{X}_g}(\bm{d}_g, \bm{w}_g; \theta) = 
\prod_{i} P_{ig}(\theta_d)^{d_{ig}} (1 - P_{ig}(\theta_d))^{1 - d_{ig}} 
\prod_{i < j} p_{ijg}(\theta_w)^{w_{ijg}} (1 - p_{ijg}(\theta_w))^{1 - w_{ijg}}$ 
(for an undirected network; the directed case is identical with $i \neq j$), a product of 
$K_g = N_g + N_g(N_g - 1)/2 \leq K := \bar{N} + \bar{N}(\bar{N} - 1)/2$ factors in $[0, 1]$. For 
numbers $a_k, b_k \in [0, 1]$, telescoping gives 
$|\prod_{k=1}^{K_g} a_k - \prod_{k=1}^{K_g} b_k| \leq \sum_{k=1}^{K_g} |a_k - b_k|$. More useful than this pointwise bound is a total-variation bound. The law $f(\cdot; \theta)$ is a 
product measure over the $K_g$ binary coordinates, and for product measures the total variation 
distance is at most the sum of the coordinatewise distances (couple the coordinates independently): 
$d_{TV}(f(\cdot; \hat{\theta}), f(\cdot; \theta)) = \tfrac{1}{2} \sum_{(\bm{d}_g, \bm{w}_g)} |f(\bm{d}_g, \bm{w}_g; \hat{\theta}) - f(\bm{d}_g, \bm{w}_g; \theta)| \leq 
\sum_{k=1}^{K_g} |a_k(\hat{\theta}) - a_k(\theta)|$, where $a_k$ is the $k$-th treatment or link 
probability. Each $a_k$ is Lipschitz in $\theta$ with constant $L_F$ uniformly in its other 
arguments, so the total variation distance is at most $K L_F \|\hat{\theta} - \theta\|$. The score 
$r(\{ig\}, t; \bm{X}_g)$ is the probability under $f$ of the event $\{e(\{ig\}, \cdot) = t\}$, and 
$|P(E) - Q(E)| \leq d_{TV}(P, Q)$, hence 
$|\hat{r}(\{ig\}, t; \bm{X}_g) - r(\{ig\}, t; \bm{X}_g)| \leq K L_F \|\hat{\theta} - \theta\|$ 
uniformly in $g$, $\{ig\}$ and $t$. With latent random effects, 
$f(\bm{d}_g, \bm{w}_g; \theta) = \int \prod_i (\cdot) \prod_{i<j} (\cdot)(\bm{u}_g) \, \diff F_U(\bm{u}_g; \theta_w)$; 
for fixed $\bm{u}_g$ the integrand is a product measure and the total variation bound above applies, 
provided the Lipschitz constant of each link probability in $\theta_w$ is uniform in $\bm{u}_g$ (true 
for the probit link~(\ref{ame}) with bounded covariates and bounded latent positions), while for 
fixed $\theta$ in the integrand the mixture is Lipschitz in $\theta_w$ in total variation because 
$F_U$ is; the triangle inequality gives the same bound with a different constant. Since $\hat{\theta} - \theta = O_p(G^{-1/2})$ and $N \leq \bar{N} G$, 
$N^{-1} \sum_{\{ig\}} (\hat{r} - r)^2 \leq \max |\hat{r} - r|^2 = O_p(G^{-1}) = O_p(N^{-1})$.

For the exposure mapping~(\ref{epp}) under Assumption~\ref{asm:ce} without latent variables (with 
latent variables the summands $W_{ijg} D_{jg}$ are dependent across $j$ after marginalizing 
$\bm{U}_g$, and only the general bound applies), $r(\{ig\}, (d, z); \bm{X}_g) = P^d_{ig} \, \pi_{ig}(z; \bm{q}_{ig})$ 
where $\pi_{ig}(\cdot; \bm{q}_{ig})$ is the Poisson-binomial law with success probabilities 
$\bm{q}_{ig} = (q_{ijg})_{j \neq i}$, $q_{ijg} = p_{ijg} P^1_{jg}$. For two Poisson-binomial laws with 
success probabilities $\bm{q}$ and $\bm{q}'$, a coupling that draws the $j$-th Bernoulli pair with 
probability of disagreement $|q_j - q'_j|$ gives $d_{TV}(\pi(\cdot; \bm{q}), \pi(\cdot; \bm{q}')) \leq 
\sum_j |q_j - q'_j|$. Therefore 
$|\hat{r} - r| \leq |\hat{P}^d_{ig} - P^d_{ig}| \, \hat{\pi}_{ig}(z) + P^d_{ig} |\hat{\pi}_{ig}(z) - \pi_{ig}(z)| 
\leq |\hat{P}^d_{ig} - P^d_{ig}| + \sum_{j \neq i} |\hat{q}_{ijg} - q_{ijg}|$, which is the sharper bound. 
\hfill $\square$

\subsection{Proof of Proposition~\ref{ttt} (asymptotic normality)}\label{app:a2}

Throughout, $t$ and $t'$ are fixed, we write $r_{s, ig} = r(\{ig\}, s; \bm{X}_g)$ and 
$\mu_{s, ig} = \mu(\{ig\}, s, \bm{X}_g)$ for $s \in \{t, t'\}$, with hats for estimated quantities. 
Define, for $s \in \{t, t'\}$,
$$
\begin{gathered}
\psi_{s, ig} = \frac{\mathbbm{1}\{T_{ig} = s\}(Y_{ig} - \mu(\{ig\}, s, \bm{X}_g))}{r(\{ig\}, s; \bm{X}_g)} + \mu(\{ig\}, s, \bm{X}_g),
\\ 
\hat{\psi}_{s, ig} = \frac{\mathbbm{1}\{T_{ig} = s\}(Y_{ig} - \hat{\mu}(\{ig\}, s, \bm{X}_g))}{\hat{r}(\{ig\}, s; \bm{X}_g)} + \hat{\mu}(\{ig\}, s, \bm{X}_g),
\end{gathered}
$$
so that $\hat{\tau}(t, t') = N^{-1} \sum_{\{ig\}} (\hat{\psi}_{t, ig} - \hat{\psi}_{t', ig})$ and 
$\varphi_{ig}(t, t') = \psi_{t, ig} - \psi_{t', ig} - \tau(t, t')$. 

\paragraph{Step 1: decomposition.} Direct algebra gives, for each $s$,
\begin{equation}\label{decomp}
\begin{aligned}
\hat{\psi}_{s, ig} - \psi_{s, ig} = 
&\underbrace{\frac{\mathbbm{1}\{T_{ig} = s\}(Y_{ig} - \mu_{s, ig})(r_{s, ig} - \hat{r}_{s, ig})}{\hat{r}_{s, ig} r_{s, ig}}}_{\rho^{(1)}_{s, ig}}
+ \underbrace{(\hat{\mu}_{s, ig} - \mu_{s, ig})\Big(1 - \frac{\mathbbm{1}\{T_{ig} = s\}}{r_{s, ig}}\Big)}_{\rho^{(2)}_{s, ig}} \\
&+ \underbrace{\frac{\mathbbm{1}\{T_{ig} = s\}(\hat{\mu}_{s, ig} - \mu_{s, ig})(\hat{r}_{s, ig} - r_{s, ig})}{\hat{r}_{s, ig} r_{s, ig}}}_{\rho^{(3)}_{s, ig}}.
\end{aligned}
\end{equation}
Hence
$$
\sqrt{N}(\hat{\tau}(t, t') - \tau(t, t')) = \frac{1}{\sqrt{N}} \sum_{g} S_g(t, t') 
+ \sum_{k=1}^{3} \big( R^{(k)}_t - R^{(k)}_{t'} \big), 
\quad \text{where } R^{(k)}_s = \frac{1}{\sqrt{N}} \sum_{\{ig\}} \rho^{(k)}_{s, ig}.
$$

\paragraph{Step 2: the leading term.} Under Assumption~\ref{asm:samp}(a), the group sums 
$S_g(t, t') = \sum_{\{ig\} \in \mathcal{N}_g} \varphi_{ig}(t, t')$ are i.i.d. across $g$. By the 
pooled statement of Proposition~\ref{prop:a} and the definition~(\ref{mudef}), 
$\mathbbm{E}[\sum_i (\psi_{t, ig} - \psi_{t', ig})] = \mathbbm{E}[N_g] \tau(t, t')$, so $\mathbbm{E}[S_g] = 0$. By 
Assumption~\ref{asm:mom}(a)--(b) and $N_g \leq \bar{N}$, 
$\mathbbm{E}[S_g^2] \leq \bar{N} \sum_{i} \mathbbm{E}[\varphi_{ig}^2] \leq \bar{N}^2 C < \infty$, where 
$C$ depends only on $M$ and $\underline{\pi}$. The 
Lindeberg--L\'evy central limit theorem gives $G^{-1/2} \sum_g S_g \xrightarrow{d} \mathcal{N}(0, \mathbbm{E}[S_g^2])$, 
and the strong law gives $N / G = G^{-1} \sum_g N_g \to \mathbbm{E}[N_g]$ almost surely. By Slutsky's 
theorem, $N^{-1/2} \sum_g S_g = (G/N)^{1/2} G^{-1/2} \sum_g S_g \xrightarrow{d} 
\mathcal{N}(0, \mathbbm{E}[S_g^2] / \mathbbm{E}[N_g]) = \mathcal{N}(0, \sigma^2_{t, t'})$.

\paragraph{Step 3: the product remainder $R^{(3)}_s$.} By Assumption~\ref{asm:mom}(b) and the 
Cauchy--Schwarz inequality,
$$
\begin{aligned}
|R^{(3)}_s| &\leq \frac{\sqrt{N}}{\underline{\pi}^2} \, \frac{1}{N} \sum_{\{ig\}} |\hat{\mu}_{s, ig} - \mu_{s, ig}| |\hat{r}_{s, ig} - r_{s, ig}| \\
&\leq \frac{\sqrt{N}}{\underline{\pi}^2} \Big[ \frac{1}{N} \sum (\hat{\mu}_{s, ig} - \mu_{s, ig})^2 \Big]^{1/2} 
\Big[ \frac{1}{N} \sum (\hat{r}_{s, ig} - r_{s, ig})^2 \Big]^{1/2} = o_p(1)
\end{aligned}
$$
by the product-rate condition in Assumption~\ref{asm:rate}.

\paragraph{Step 4: the remainder $R^{(1)}_s$ under cross-fitting (Assumption~\ref{asm:rate}(a)).} Let the groups be partitioned into folds 
$\mathcal{I}_1, \ldots, \mathcal{I}_K$ and let $\hat{r}^{(-k)}$ denote the score estimated without 
fold $k$. Write $R^{(1)}_s = \sum_k R^{(1)}_{s, k}$ with $R^{(1)}_{s, k}$ the sum over groups in fold 
$k$. Conditional on the data outside fold $k$, $\hat{r}^{(-k)}$ is fixed, and for $g \in \mathcal{I}_k$ 
the summands $\rho^{(1)}_{s, ig}$ satisfy 
$\mathbbm{E}[\rho^{(1)}_{s, ig} \mid \bm{X}_g, \bm{W}_g, \bm{D}_g, \hat{r}^{(-k)}] = 0$, because 
$\mathbbm{E}[Y_{ig} - \mu_{s, ig} \mid T_{ig} = s, \bm{X}_g, \bm{W}_g, \bm{D}_g] = 
\mathbbm{E}[Y_{ig}(s) \mid \bm{X}_g] - \mu_{s, ig} = 0$ by Assumption~\ref{asm:ig} and correct 
specification of $\mu$. Groups are independent, so, using $(\sum_{i=1}^{N_g} a_i)^2 \leq N_g \sum_i a_i^2$ 
without any assumption on the within-group dependence of outcomes,
$$
\begin{aligned}
\mathbbm{E}\big[ (R^{(1)}_{s, k})^2 \mid \text{data outside fold } k \big] 
&= \frac{1}{N} \sum_{g \in \mathcal{I}_k} \mathbbm{E}\Big[ \Big( \sum_{\{ig\}} \rho^{(1)}_{s, ig} \Big)^2 \,\Big|\, \cdot \Big] \\
&\leq \frac{\bar{N}}{N} \sum_{g \in \mathcal{I}_k} \sum_{\{ig\}} 
\mathbbm{E}\Big[ \frac{\mathbbm{1}\{T_{ig} = s\} \sigma^2_s(\bm{X}_g) (r_{s, ig} - \hat{r}_{s, ig})^2}{\hat{r}^2_{s, ig} r^2_{s, ig}} \,\Big|\, \cdot \Big] \\
&\leq \frac{\bar{N} M^{1/2}}{\underline{\pi}^4} \, \frac{1}{N} \sum_{g \in \mathcal{I}_k} \sum_{\{ig\}} (\hat{r}_{s, ig} - r_{s, ig})^2 ,
\end{aligned}
$$
where $\sigma^2_s(\bm{X}_g) \leq M^{1/2}$ by Assumption~\ref{asm:mom}(a). The bound is $o_p(1)$ by 
Assumption~\ref{asm:rate}. Conditional Chebyshev and the fact that $K$ is fixed give $R^{(1)}_s = o_p(1)$.

\paragraph{Step 5: the remainder $R^{(2)}_s$ under cross-fitting.} With the same fold structure, conditional on the 
data outside fold $k$, $\mathbbm{E}[1 - \mathbbm{1}\{T_{ig} = s\} / r_{s, ig} \mid \bm{X}_g] = 0$ by 
correct specification of $r$, so $\mathbbm{E}[\rho^{(2)}_{s, ig} \mid \bm{X}_g, \hat{\mu}^{(-k)}] = 0$. 
Within a group, $T_{ig}$ and $T_{jg}$ are dependent, but 
$|\mathbbm{E}[\rho^{(2)}_{s, ig} \rho^{(2)}_{s, jg} \mid \bm{X}_g, \hat{\mu}^{(-k)}]| \leq 
\underline{\pi}^{-2} |\hat{\mu}_{s, ig} - \mu_{s, ig}| |\hat{\mu}_{s, jg} - \mu_{s, jg}|$. Therefore
$$
\begin{aligned}
\mathbbm{E}\big[ (R^{(2)}_{s, k})^2 \mid \cdot \big] 
&\leq \frac{1}{\underline{\pi}^2 N} \sum_{g \in \mathcal{I}_k} \Big( \sum_{\{ig\}} |\hat{\mu}_{s, ig} - \mu_{s, ig}| \Big)^2 \\
&\leq \frac{\bar{N}}{\underline{\pi}^2} \, \frac{1}{N} \sum_{g \in \mathcal{I}_k} \sum_{\{ig\}} (\hat{\mu}_{s, ig} - \mu_{s, ig})^2 = o_p(1),
\end{aligned}
$$
using $(\sum_{i=1}^{N_g} a_i)^2 \leq N_g \sum_i a_i^2$ and Assumption~\ref{asm:rate}. Hence 
$R^{(2)}_s = o_p(1)$.

\paragraph{Steps 4$'$ and 5$'$: $R^{(1)}_s$ and $R^{(2)}_s$ under parametric nuisances (Assumption~\ref{asm:rate}(b)).} 
The smoothness required of the parametric nuisances is the following: $r(\cdot; \theta)$ and 
$\mu(\cdot; \beta)$ are continuously differentiable in the parameters, with derivatives that are 
bounded and uniformly continuous in a neighborhood of the true values, uniformly in the other 
arguments; for $\mu$, the derivatives are bounded by a function of $\bm{X}_g$ with finite fourth 
moment and are uniformly continuous after division by that function. 
Write $\hat{r}_{s, ig} = r_{s, ig}(\hat{\theta})$ and $\hat{\mu}_{s, ig} = \mu_{s, ig}(\hat{\beta})$, and 
let $h_{s, ig}(\theta) = 1 / r_{s, ig}(\theta)$. From~(\ref{decomp}), 
$\rho^{(1)}_{s, ig} = \mathbbm{1}\{T_{ig} = s\}(Y_{ig} - \mu_{s, ig})[h_{s, ig}(\hat{\theta}) - h_{s, ig}(\theta)]$, 
and by the mean value theorem there are unit-specific intermediate points $\bar{\theta}_{ig}$ between 
$\hat{\theta}$ and $\theta$ with 
$h_{s, ig}(\hat{\theta}) - h_{s, ig}(\theta) = \nabla_\theta h_{s, ig}(\bar{\theta}_{ig})' (\hat{\theta} - \theta)$. 
Hence
$$
R^{(1)}_s = Q_N' (\hat{\theta} - \theta) + \Delta_N' (\hat{\theta} - \theta), 
\qquad 
Q_N = \frac{1}{\sqrt{N}} \sum_{\{ig\}} \mathbbm{1}\{T_{ig} = s\} (Y_{ig} - \mu_{s, ig}) \nabla_\theta h_{s, ig}(\theta),
$$
$$
\Delta_N = \frac{1}{\sqrt{N}} \sum_{\{ig\}} \mathbbm{1}\{T_{ig} = s\} (Y_{ig} - \mu_{s, ig}) 
\big[ \nabla_\theta h_{s, ig}(\bar{\theta}_{ig}) - \nabla_\theta h_{s, ig}(\theta) \big].
$$
The summands of $Q_N$ have conditional mean zero given $(\bm{X}_g, \bm{W}_g, \bm{D}_g)$ by 
Assumption~\ref{asm:ig} and correct specification of $\mu$, their second moments are bounded by 
Assumptions~\ref{asm:mom} and~\ref{asm:rate}(b), and groups are independent, so 
$Q_N = O_p(1)$ by Chebyshev's inequality applied to the normalized sum of independent group-level 
sums (using $(\sum_{i} a_i)^2 \leq N_g \sum_i a_i^2$ within groups). For $\Delta_N$, 
$\|\Delta_N\| \leq \sqrt{N} \cdot N^{-1} \sum_{\{ig\}} |Y_{ig} - \mu_{s, ig}| \cdot 
\sup_{\{ig\}} \|\nabla_\theta h_{s, ig}(\bar{\theta}_{ig}) - \nabla_\theta h_{s, ig}(\theta)\|$, where 
$N^{-1} \sum |Y_{ig} - \mu_{s, ig}| = O_p(1)$ and the supremum is $o_p(1)$ by the uniform continuity 
of the derivatives in Assumption~\ref{asm:rate}(b) together with 
$\max_{ig} \|\bar{\theta}_{ig} - \theta\| \leq \|\hat{\theta} - \theta\| \to_p 0$. Therefore 
$R^{(1)}_s = O_p(1) O_p(N^{-1/2}) + o_p(\sqrt{N}) O_p(N^{-1/2}) = o_p(1)$. The same argument with 
$\rho^{(2)}_{s, ig} = [\mu_{s, ig}(\hat{\beta}) - \mu_{s, ig}(\beta)](1 - \mathbbm{1}\{T_{ig} = s\} / r_{s, ig})$ 
gives $R^{(2)}_s = P_N' (\hat{\beta} - \beta) + o_p(1)$ with 
$P_N = N^{-1/2} \sum_{\{ig\}} \nabla_\beta \mu_{s, ig}(\beta) (1 - \mathbbm{1}\{T_{ig} = s\} / r_{s, ig})$, 
whose summands have conditional mean zero given $\bm{X}_g$ by correct specification of $r$ and 
second moments bounded by the dominating function with finite fourth moment, so $P_N = O_p(1)$ and 
$R^{(2)}_s = o_p(1)$; the remainder uses the uniform continuity of $\nabla_\beta \mu$ after division 
by the dominating function. This is the standard argument for parametric nuisances 
(see \citealp{farrell2015robust} for the Donsker version), and it covers the implementation 
in Sections~\ref{SMC} and~\ref{ea}, where the nuisances are fitted once on the full sample. 
Finally, under Assumption~\ref{asm:rate}(b) the bound $\hat{r} \in [\underline{\pi}, \bar{\pi}]$ holds 
with probability approaching one rather than almost surely; all statements above are on that event.

\paragraph{Step 6: conclusion and double robustness.} Combining Steps 2--5 with Slutsky's theorem 
yields~(\ref{asy}). If only $r$ is correctly specified, Steps 3 and 5 (or 5$'$) go through with 
$\mu$ replaced by the probability limit $\tilde{\mu}$ of $\hat{\mu}$, Step 4 (or 4$'$) no longer gives a 
remainder of order $o_p(1)$ but $R^{(1)}_s / \sqrt{N} = o_p(1)$ by Cauchy--Schwarz and consistency of $\hat{r}$, 
and $N^{-1} \sum_g S_g \to 0$ by the law of large numbers because 
$\mathbbm{E}[\sum_i \{\mathbbm{1}\{T_{ig} = s\}(Y_{ig} - \tilde{\mu}_{s, ig}) / r_{s, ig} + \tilde{\mu}_{s, ig}\}] = \mathbbm{E}[N_g] \mu(s)$ 
under correct $r$; the symmetric argument applies when only $\mu$ is correct. This establishes 
consistency under either model. In the first case $R^{(1)}_s$ is in general $O_p(1)$ rather than 
$o_p(1)$: it equals $Q_N'(\hat{\theta} - \theta)$ with $Q_N$ no longer centered, so the asymptotic 
variance acquires the first-stage term $\text{plim}\, N^{-1/2} Q_N \cdot \sqrt{N}(\hat{\theta} - \theta)$, 
which the estimator~(\ref{hac}) does not include. \hfill $\square$

\paragraph{Growing group sizes (Remark~\ref{rem:growth}).} Suppose Assumption~\ref{asm:samp}(b) is 
dropped, groups are independent but not identically distributed, outcomes are bounded so that 
$|S_g| \leq C N_g$, and $\sum_g \text{var}(S_g) \geq c N$ for some $c > 0$ (a non-degeneracy 
condition that must be assumed). Step 2 then uses the Lyapunov central limit theorem with exponent 
$2 + \delta$:
$$
\frac{\sum_g \mathbbm{E}|S_g|^{2 + \delta}}{(\sum_g \text{var}(S_g))^{1 + \delta/2}} 
\leq \frac{C^{2 + \delta} \max_g N_g^{1 + \delta} \sum_g N_g}{(c N)^{1 + \delta/2}} 
= C' \, \frac{\max_g N_g^{1 + \delta}}{N^{\delta/2}} ,
$$
which tends to zero if $\max_g N_g = o(N^{\delta / (2 + 2\delta)})$. Since $\delta$ can be taken as 
large as desired for bounded outcomes, $\max_g N_g = o(N^{1/2 - \epsilon})$ for some $\epsilon > 0$ 
suffices; with only fourth moments ($\delta = 2$) the condition is $\max_g N_g = o(N^{1/3})$. Steps 
3--5 go through with $\bar{N}$ replaced by $\max_g N_g$, which then requires 
$\max_g N_g \cdot N^{-1} \sum (\hat{\mu} - \mu)^2 = o_p(1)$ and the analogous condition for 
$\hat{r}$. The score rate in Lemma~\ref{lem:gps} is the binding constraint and is addressed in 
Appendix~\ref{app:single}. The variance consistency argument of Appendix~\ref{app:var} goes through 
with $\sigma^2_{t, t'}$ replaced by $N^{-1} \sum_g \text{var}(S_g)$.

\subsection{Proof of the consistency of the variance estimator in Proposition~\ref{ttt}}\label{app:var}

Write $\hat{S}_g = \sum_{\{ig\} \in \mathcal{N}_g} \hat{\varphi}_{ig}(t, t')$ so that 
$\hat{\sigma}^2_{t, t'} = N^{-1} \sum_g \hat{S}_g^2$. First, $N^{-1} \sum_g S_g^2 = (G / N) G^{-1} \sum_g S_g^2 
\to \mathbbm{E}[S_g^2] / \mathbbm{E}[N_g] = \sigma^2_{t, t'}$ almost surely by the strong law, since 
$\mathbbm{E}[S_g^2] < \infty$. Second, 
$|\hat{S}_g^2 - S_g^2| \leq (\hat{S}_g - S_g)^2 + 2 |S_g| |\hat{S}_g - S_g|$, so by Cauchy--Schwarz it 
suffices to show $N^{-1} \sum_g (\hat{S}_g - S_g)^2 = o_p(1)$. Now 
$\hat{S}_g - S_g = \sum_{\{ig\}} (\hat{\varphi}_{ig} - \varphi_{ig})$ and, by~(\ref{decomp}) together 
with $|1 / \hat{r} - 1 / r| \leq |\hat{r} - r| / \underline{\pi}^2$ and Assumption~\ref{asm:mom},
$$
|\hat{\varphi}_{ig} - \varphi_{ig}| \leq \sum_{s \in \{t, t'\}} \Big[ 
\frac{|Y_{ig}| + |\hat{\mu}_{s, ig}|}{\underline{\pi}^2} |\hat{r}_{s, ig} - r_{s, ig}| + 
\Big(1 + \frac{1}{\underline{\pi}}\Big) |\hat{\mu}_{s, ig} - \mu_{s, ig}| \Big] + |\hat{\tau}(t, t') - \tau(t, t')|.
$$
Using $(\sum_{i=1}^{N_g} a_i)^2 \leq \bar{N} \sum_i a_i^2$,
$$
\frac{1}{N} \sum_g (\hat{S}_g - S_g)^2 \leq C \bar{N} \sum_{s} \Big[ 
\frac{1}{N} \sum_{\{ig\}} (|Y_{ig}| + |\hat{\mu}_{s, ig}|)^2 (\hat{r}_{s, ig} - r_{s, ig})^2 + 
\frac{1}{N} \sum_{\{ig\}} (\hat{\mu}_{s, ig} - \mu_{s, ig})^2 \Big] + C \bar{N} (\hat{\tau} - \tau)^2 .
$$
The second and third terms are $o_p(1)$ by Assumption~\ref{asm:rate} and Proposition~\ref{ttt}. For 
the first, Cauchy--Schwarz and $|\hat{r} - r| \leq 1$ give 
$N^{-1} \sum (|Y_{ig}| + |\hat{\mu}_{s, ig}|)^2 (\hat{r} - r)^2 \leq [N^{-1} \sum (|Y_{ig}| + |\hat{\mu}_{s, ig}|)^4]^{1/2} 
[N^{-1} \sum (\hat{r} - r)^2]^{1/2} = O_p(1) \, o_p(1)$ by the fourth-moment conditions in 
Assumption~\ref{asm:mom}(a). \hfill $\square$
\end{revblock}

\section{Other Factorizations of Treatment Assignment and Network Formation}\label{D}


The paper's main case is the ``common exposure'' mechanism (upper left panel of 
Figure~\ref{factorize}), in which treatment assignment and network formation are conditionally 
independent. Since they are two separate processes, one may also precede and act on the other. 
The network may drive the diffusion of adoption (lower left panel of Figure~\ref{factorize}); 
then one models network formation first and treatment diffusion given the network, following the 
factorization of \citet{goldsmith2013social}. If treatment precedes network formation instead 
\citep{comola2021treatment}, the network is a mediator (upper right panel), and the mediation 
framework of \citet{imai2010identification} could be extended to network data to isolate the 
spillover of treatment that runs through tie formation. For a model-based approach, the 
instrumental variable strategy of \citet{comola2021treatment}, which models network formation as 
a consequence of treatment, gives the corresponding factorization. In both cases the two 
processes are modeled sequentially. Finally, treatment assignment and network formation may 
reinforce each other, a co-evolution process (lower right panel)\footnote{This case can be read 
as one in which the network and the treatment are measured at the same time and both may depend 
on their earlier values when longitudinal data are available. Longitudinal data are not required, 
though: the network and treatment measured together may be an equilibrium of co-evolution.}, 
which calls for a simultaneous equation model\footnote{One estimation strategy for such models is 
the latent variable model of \citet{han2021spatial}, built on the latent space model of 
\citet{hoff2021additive} for network formation.}. 

   \begin{figure}[!ht]
  \begin{center}
  \begin{tabular}{cc}
   \raisebox{-0.5\height}{\includegraphics[scale=0.6]{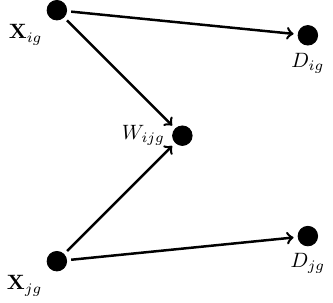}} &
   \raisebox{-0.5\height}{\includegraphics[scale=0.6]{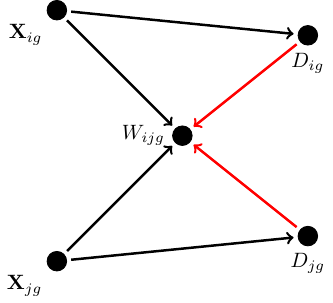}} \\
   Common Exposure & Endogenous Network \\
   \raisebox{-0.5\height}{\includegraphics[scale=0.6]{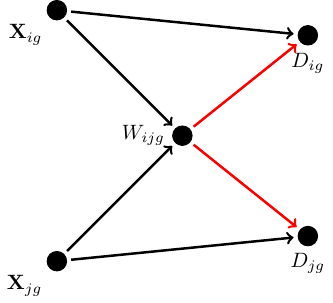}} &
   \raisebox{-0.5\height}{\includegraphics[scale=0.6]{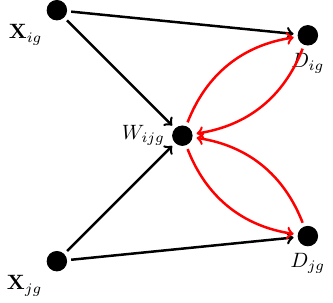}} \\
   Treatment Diffusion & Co-evolution 
   \end{tabular}
   \end{center}
   \caption{Factorization of Treatment Assignment and Network Formation}\label{factorize}
   \footnotesize\textbf{Note:} Arrows in red represent interactions between treatment assignment 
   and network formation.
   \end{figure}

\begin{revblock}
\section{Details for the Single-Network Extension}\label{app:single}

This appendix supports Section~\ref{sec:single}. There is one network of $N$ units; 
$X_1, \ldots, X_N$ are i.i.d.; Assumptions~\ref{asm:em}, \ref{asm:ig} and~\ref{asm:ce} hold with 
exposure mapping~(\ref{epp}); $P^1_i = Pr(D_i = 1 \mid \bm{X}) \in [\underline{\pi}_d, \bar{\pi}_d] \subset (0, 1)$; and 
links are independent given $\bm{X}$ with $p_{ij} = \rho_N h(X_i, X_j)$, $0 < \underline{h} \leq h \leq \bar{h} < \infty$, 
$N \rho_N \to \lambda \in (0, \infty)$ (the sparse regime of \citealp{graham2024sparse}). Let 
$q_{ij} = p_{ij} P^1_j$ and $\Lambda_i = \sum_{j \neq i} q_{ij}$.

\subsection{Proof of Proposition~\ref{prop:sparsepos}}

Given $\bm{X}$, $Z_i = \sum_{j \neq i} W_{ij} D_j$ is a sum of independent Bernoulli variables with 
success probabilities $q_{ij}$. By Le Cam's inequality (for the $\ell_1$ distance between the two laws), 
$\sum_z |Pr(Z_i = z \mid \bm{X}) - e^{-\Lambda_i} \Lambda_i^z / z!| \leq 2 \sum_{j \neq i} q_{ij}^2 
\leq 2 (N - 1) \rho_N^2 \bar{h}^2 = O(N^{-1})$. Moreover 
$\underline{\pi}_d \underline{h} (N-1) \rho_N \leq \Lambda_i \leq \bar{\pi}_d \bar{h} (N - 1) \rho_N$, so 
$\Lambda_i \in [\underline{\lambda}, \bar{\lambda}] \subset (0, \infty)$ for all $i$ and all large $N$. 
Since $\lambda \mapsto e^{-\lambda} \lambda^z / z!$ is bounded below by a positive constant on 
$[\underline{\lambda}, \bar{\lambda}]$ for every fixed $z$, 
$\liminf_N \min_i Pr(Z_i = z \mid \bm{X}) > 0$, and $r(i, (d, z); \bm{X}) = P^d_i Pr(Z_i = z \mid \bm{X})$ 
is bounded below by $\underline{\pi}_d$ times that constant. For the network-conditional score, 
$Z_i \leq \deg_i(\bm{W})$, so $r^{c}(i, (d, z); \bm{X}, \bm{W}) = 0$ when $\deg_i(\bm{W}) < z$. The 
degree $\deg_i(\bm{W}) \mid \bm{X}$ is Poisson-binomial with mean in 
$[\underline{h} (N-1) \rho_N, \bar{h} (N-1) \rho_N]$, so by the same Poisson approximation 
$Pr(\deg_i(\bm{W}) < z \mid \bm{X})$ is bounded below by a positive constant for every $z \geq 1$. 
The share of such units has conditional mean equal to the average of these probabilities, and the 
indicators $\mathbbm{1}\{\deg_i(\bm{W}) < z\}$ and $\mathbbm{1}\{\deg_j(\bm{W}) < z\}$ share only the 
link $W_{ij}$, so their covariance is $O(\rho_N)$; the variance of the share is therefore 
$O(N^{-1} + \rho_N) = O(N^{-1})$ and the share converges to the limit of its mean. \hfill $\square$

\subsection{Rate of the estimated score}

\begin{lemma}\label{lem:single}
Suppose in addition that treatment depends on own covariates only, $P^1_i = F_d(X_i; \theta_d)$, 
and $p_{ij} = F_w(X_i, X_j; \theta_w)$ with 
$F_d$ Lipschitz in $\theta_d$ uniformly in $x$, $|\partial F_w / \partial \theta_w| \leq C F_w$ 
in a neighborhood of $\theta_w$ (as for the logit link with bounded covariates and the sparsity rate 
absorbed in the intercept), $\hat{\theta}_d - \theta_d = O_p(N^{-1/2})$ and $\hat{\theta}_w - \theta_w = O_p(N^{-1/2})$. 
Then $\max_i |\hat{r}(i, (d, z); \bm{X}) - r(i, (d, z); \bm{X})| = O_p(N^{-1/2})$.
\end{lemma}

\begin{proof}
By the total-variation bound in Lemma~\ref{lem:gps}, 
$|\hat{r}_i - r_i| \leq |\hat{P}^d_i - P^d_i| + \sum_{j \neq i} |\hat{q}_{ij} - q_{ij}|$. The first term 
is $O_p(N^{-1/2})$ uniformly in $i$. For the second, 
$|\hat{q}_{ij} - q_{ij}| \leq |\hat{p}_{ij} - p_{ij}| + p_{ij} |\hat{P}^1_j - P^1_j|$, and by the mean 
value theorem $|\hat{p}_{ij} - p_{ij}| \leq C \tilde{p}_{ij} \|\hat{\theta}_w - \theta_w\|$ with 
$\tilde{p}_{ij} \leq C' p_{ij}$ for $\hat{\theta}_w$ in a neighborhood of $\theta_w$. Hence 
$\sum_{j \neq i} |\hat{q}_{ij} - q_{ij}| \leq C'' \big( \sum_{j \neq i} p_{ij} \big) 
(\|\hat{\theta}_w - \theta_w\| + \max_j |\hat{P}^1_j - P^1_j|) = O(1) \, O_p(N^{-1/2})$, 
because $\sum_{j \neq i} p_{ij} \leq \bar{h} (N - 1) \rho_N = O(1)$ under sparsity. The rate 
$\hat{\theta}_w - \theta_w = O_p(N^{-1/2})$ for the sparse dyadic logit is Theorem 1 of 
\citet{graham2024sparse}, whose rate is the inverse square root of the expected number of links, 
of order $N^2 \rho_N \asymp N$.
\end{proof}

Node-level latent positions are not covered by this lemma: in the bounded-degree regime each 
node carries $O_p(1)$ links and its latent position cannot be estimated consistently. The plug-in 
score of Section~\ref{sec:latent} requires degrees that grow with $N$, the regime in which 
\citet{li2023statistical} obtain uniform consistency of latent positions; a plug-in version of this 
lemma for that regime would replace $\|\hat{\theta}_w - \theta_w\|$ by $\max_i \|\hat{U}_i - U_i\|$ 
and is not developed here.

\subsection{Inference conditional on the network and a variance estimator}

Let $a_i = \mathbbm{1}\{T_i = t\}(Y_i - m_t(X_i)) / r(i, t; \bm{X}) - \mathbbm{1}\{T_i = t'\}(Y_i - m_{t'}(X_i)) / r(i, t'; \bm{X})$ 
and $b_i = m_t(X_i) - m_{t'}(X_i) - \tau(t, t')$, so that $\varphi_i = a_i + b_i$ and, for the 
oracle estimator with known nuisances, $\sqrt{N}(\hat{\tau} - \tau) = N^{-1/2} \sum_i a_i + N^{-1/2} \sum_i b_i$. 
We say that the outcome model is correctly specified if $m_s(X_i) = \mathbbm{E}[Y_i(s) \mid \bm{X}]$, 
that is, the conditional mean of the potential outcome depends on $\bm{X}$ only through own 
covariates. (If it depends on summaries of $\bm{X}_{-i}$ such as neighborhood averages, 
$N^{-1/2} \sum_i b_i$ is no longer a sum of i.i.d. terms and $\sigma^2_b$ below must be replaced by 
the variance of the corresponding U-statistic; we do not pursue this.)

\begin{proposition}[Conditional CLT and variance estimation on a single network]\label{prop:singleclt}
Suppose the outcome model is correctly specified, the collection of potential outcomes is 
independent of $(\bm{D}, \bm{W})$ given $\bm{X}$ (Assumption~\ref{asm:ig} for the single group) and 
potential outcomes are independent across units given $\bm{X}$, $\mathbbm{E}[Y_i(s)^4 \mid \bm{X}] \leq M$, 
the scores lie in $[\underline{\pi}, \bar{\pi}]$, and 
$\sigma^2_{a, N}(\bm{X}, \bm{W}) = N^{-1} \sum_i \mathbbm{E}[a_i^2 \mid \bm{X}, \bm{W}] \xrightarrow{p} \sigma^2_a > 0$. 
Then, for the oracle estimator,
$$
\sqrt{N}(\hat{\tau}(t, t') - \tau(t, t')) \xrightarrow{d} \mathcal{N}(0, \sigma^2_a + \sigma^2_b), 
\qquad \sigma^2_b = \text{var}(b_i).
$$
If in addition the nuisances satisfy Assumption~\ref{asm:rate}(b) (parametric nuisances) with 
$\hat{\theta} - \theta = O_p(N^{-1/2})$ for the sparse dyadic model and the treatment model, the same 
limit holds for the feasible estimator and $\hat{\sigma}^2$ in~(\ref{hac2}) is consistent for 
$\sigma^2_a + \sigma^2_b$.
\end{proposition}

\begin{proof}[Proof sketch]
\textit{Conditional mean and conditional uncorrelatedness.} By Assumption~\ref{asm:ig} and correct 
specification, $\mathbbm{E}[Y_i(s) - m_s(X_i) \mid \bm{X}, \bm{W}, \bm{D}] = 0$, so 
$\mathbbm{E}[a_i \mid \bm{X}, \bm{W}] = 0$ and $\mathbbm{E}[\hat{\tau} \mid \bm{X}, \bm{W}] - \tau = N^{-1} \sum_i b_i$ 
does not depend on $\bm{W}$. For $i \neq j$, conditioning on $(\bm{X}, \bm{W}, \bm{D})$ factors out the 
indicators and scores, and the conditional independence of potential outcomes across units gives 
$\mathbbm{E}[a_i a_j \mid \bm{X}, \bm{W}] = 0$ (Fact 2 of Appendix~\ref{app:eff}). Hence 
$\text{var}(\sum_i a_i \mid \bm{X}, \bm{W}) = \sum_i \mathbbm{E}[a_i^2 \mid \bm{X}, \bm{W}] = N \sigma^2_{a, N}$: 
within-network dependence contributes no covariance terms to the variance.

\textit{Maximum degree.} Under the sparsity condition, $\deg_i(\bm{W}) \mid \bm{X}$ is a sum of 
independent Bernoulli variables with mean at most $\bar{h} (N - 1) \rho_N = O(1)$, so by a Chernoff 
bound $Pr(\deg_i(\bm{W}) \geq c \log N \mid \bm{X}) \leq N^{-2}$ for $c$ large enough, and a union 
bound over $i$ gives $\Delta_N = \max_i \deg_i(\bm{W}) = O_p(\log N)$; this is a consequence of the 
model, not an additional hypothesis.

\textit{Dependency graph.} Given $(\bm{X}, \bm{W})$, $a_i$ is a function of 
$(D_i, \{D_j : W_{ij} = 1\}, \epsilon^y_i)$. The $D_j$ are independent given $\bm{X}$, and the 
$\epsilon^y_i$ are independent across units and of $(\bm{D}, \bm{W})$ given $\bm{X}$; this is what the 
joint Assumption~\ref{asm:ig} together with the conditional independence of potential outcomes 
across units delivers. Hence $a_i$ and $a_j$ are 
conditionally independent unless they share a treatment indicator, which happens only if 
$W_{ij} = 1$ or $i$ and $j$ have a common neighbor: the graph with edges $\{(i, j): \ell_{\bm{W}}(i, j) \leq 2\}$ 
is a dependency graph for $\{a_i\}$ given $(\bm{X}, \bm{W})$, with maximal degree 
$D_N \leq \Delta_N + \Delta_N^2$. Theorem 3.6 of \citet{ross2011fundamentals} (a consequence of 
\citealp{baldi1989normal} and \citealp{chen2004normal}) bounds the Wasserstein distance between 
$(N \sigma^2_{a, N})^{-1/2} \sum_i a_i$ and a standard normal, conditional on $(\bm{X}, \bm{W})$, by 
$D_N^2 (N \sigma^2_{a,N})^{-3/2} \sum_i \mathbbm{E}[|a_i|^3 \mid \cdot] + 
\sqrt{28} D_N^{3/2} (\sqrt{\pi} N \sigma^2_{a,N})^{-1} (\sum_i \mathbbm{E}[a_i^4 \mid \cdot])^{1/2}$, 
where the conditional third and fourth moments of $a_i$ are bounded by constants depending on $M$ 
and $\underline{\pi}$. The bound is $O_p((\log N)^4 N^{-1/2})$. Hence 
$(N \sigma^2_{a, N})^{-1/2} \sum_i a_i \mid (\bm{X}, \bm{W}) \xrightarrow{d} \mathcal{N}(0, 1)$ 
in probability, and by $\sigma^2_{a,N} \xrightarrow{p} \sigma^2_a$, 
$N^{-1/2} \sum_i a_i \mid (\bm{X}, \bm{W}) \xrightarrow{d} \mathcal{N}(0, \sigma^2_a)$ in probability.

\textit{Combining the two components.} $N^{-1/2} \sum_i b_i$ is a sum of i.i.d. functions of $X_i$ 
and converges to $\mathcal{N}(0, \sigma^2_b)$ by the classical central limit theorem. For any real $u$, 
$\mathbbm{E}[\exp\{iu N^{-1/2}\sum_i (a_i + b_i)\}] = 
\mathbbm{E}[\exp\{iu N^{-1/2}\sum_i b_i\} \, \mathbbm{E}[\exp\{iu N^{-1/2}\sum_i a_i\} \mid \bm{X}, \bm{W}]]$, 
and the inner conditional characteristic function converges in probability to $e^{-u^2 \sigma_a^2/2}$, 
a constant; dominated convergence gives the limit $e^{-u^2(\sigma_a^2 + \sigma_b^2)/2}$.

\textit{Variance estimator.} $N^{-1} \sum_i a_i^2 - \sigma^2_{a, N} \to_p 0$ because the $a_i^2$ are 
locally dependent with bounded second moments, so the conditional variance of their average is 
$O(D_N / N) = o_p(1)$; and $N^{-1} \sum_i b_i^2 \to_p \sigma^2_b$ by the law of large numbers, while 
the cross term $N^{-1} \sum_i a_i b_i$ has conditional mean zero and vanishes. Replacing $(a_i, b_i)$ 
by $(\hat{a}_i, \hat{b}_i)$ changes the average by $o_p(1)$ under the rate conditions by the argument 
of Appendix~\ref{app:var}, with $\bar{N}$ replaced by $1$ since the sum is over units.

\textit{Feasible estimator.} Under parametric nuisances, the decomposition~(\ref{decomp}) applies 
unit by unit. The product remainder $R^{(3)}_s$ is handled as in Step 3 of Appendix~\ref{app:a2}, 
using $\max_i |\hat{r}_i - r_i| = O_p(N^{-1/2})$ from Lemma~\ref{lem:single}. For $R^{(1)}_s$ and 
$R^{(2)}_s$, the mean-value expansions of Steps 4$'$ and 5$'$ apply. The summands of $Q_N$ are 
conditionally independent given $(\bm{X}, \bm{W}, \bm{D})$, since only the outcome errors are then 
random, so $Q_N = O_p(1)$ as before. The summands of $P_N$ are locally dependent given $\bm{X}$ 
through shared treatment indicators and links, with conditional mean zero; their covariances vanish 
outside the dependency graph, so $\text{var}(P_N \mid \bm{X}) = O(D_N) = O_p((\log N)^2)$ and 
$R^{(2)}_s = O_p(N^{-1/2}) O_p(\log N) = o_p(1)$.
\end{proof}

\paragraph{Correlated outcomes.} If potential outcomes of neighbors are correlated given $\bm{X}$, 
the conditional covariances $\mathbbm{E}[a_i a_j \mid \bm{X}, \bm{W}]$ for $\ell_{\bm{W}}(i, j) \leq 2$ 
no longer vanish, and the first term of~(\ref{hac2}) must be replaced by the network HAC sum 
$N^{-1} \sum_i \sum_{j: \ell_{\bm{W}}(i, j) \leq 2} \hat{a}_i \hat{a}_j$ 
\citep{kojevnikov2021limit, leung2023network}, applied to $\hat{a}_i$ only. Applying it to 
$\hat{\varphi}_i = \hat{a}_i + \hat{b}_i$ would add cross-products $b_i b_j$ over linked pairs, 
which are not part of the variance but have non-zero expectation under homophily on covariates 
that enter $m_s$.

\paragraph{What is missing when only the score is correct.} If $m_s$ is misspecified but 
$r$ is correct, then with $a_i$ defined using the misspecified $m_s$, 
$\mathbbm{E}[a_i \mid \bm{X}, \bm{W}] = (r^{c}(i, t; \bm{X}, \bm{W}) / r(i, t; \bm{X}))(\mathbbm{E}[Y_i(t) \mid \bm{X}] - m_t(X_i)) 
- (r^{c}(i, t'; \bm{X}, \bm{W}) / r(i, t'; \bm{X}))(\mathbbm{E}[Y_i(t') \mid \bm{X}] - m_{t'}(X_i))$, 
which varies with the realized network through $r^{c}$. Its variance across network realizations is of order $N^{-1}$ and involves the joint law of 
$\{r^{c}(i, \cdot)\}_i$, which depends on the formation model. A model-based bootstrap that 
redraws $(\bm{D}, \bm{W})$ from the fitted models, recomputes $r^{c}$ and the conditional mean, and 
adds the resulting variance to the first term of~(\ref{hac2}) is the natural estimator; we have not 
established its validity. In the many-networks design of Section~\ref{sec:est} this component is 
estimated without modeling, because each group is an independent draw of the network.

\begin{revblock}
\paragraph{A Monte Carlo check.} We simulated one network of $N = 2{,}000$ units with 
$X_i \sim N(0, 1)$, $p_{ij} = (4 / N) \exp(-0.25 |X_i - X_j|)$ (expected degree about four), 
$D_i = \mathbbm{1}\{0.1 + 0.5 X_i + e_i \geq 0\}$ and $Y_i(d, z) = 1 + X_i + d + 0.5 z + d z + \varepsilon_i$, 
and estimated $\tau((1,1),(0,1))$ and $\tau((1,1),(1,0))$ by AIPW with the marginal score from a 
sparse dyadic logit and the variance estimator~(\ref{hac2}), over 500 replications. About five 
percent of units have no ties, so the network-conditional score is zero for them at $z = 1$. With 
the correctly specified outcome model $y \sim d \cdot z + x$, the bias is below $0.01$ and the 
coverage of 95\% intervals is $0.95$ for the direct and $0.97$ for the spillover effect, with the 
average estimated standard error ($0.087$ and $0.094$) matching the Monte Carlo standard deviation 
($0.090$ for both). With the misspecified outcome model $y \sim d + z + x$ the estimator remains 
unbiased and coverage is $0.96$ for both effects. The check has two limits. With 
the correct outcome model the contrast $\hat{m}(1,1) - \hat{m}(0,1)$ is the same for every unit, so 
the second term of~(\ref{hac2}) is identically zero and the check cannot detect an error in it; and 
the misspecification omits only the term $d z$, so the between-network component discussed above is 
small by construction. The design thus speaks to the first term of~(\ref{hac2}) and 
is silent about the rest. The dyadic model is fitted as a logit although the DGP is exponential in 
$|x_i - x_j|$; with link probabilities of order $1/N$ the two coincide to first order.
\end{revblock}

\paragraph{Few large groups.} With $G$ fixed and $N_g \to \infty$ for each $g$, the argument above 
applies within each group, and groups are independent, so the variance of 
$\sqrt{N}(\hat{\tau} - \tau)$ is estimated by applying~(\ref{hac2}) within each group and summing 
the group-level terms with weights $N_g / N$. With group-level outcome shocks or a group-specific 
outcome mean, $b_i$ is no longer i.i.d. across units and the second term must be computed within 
groups as in~(\ref{hac}).
\end{revblock}

\section{Additional Results for Monte Carlo Studies}\label{B}

\rev{All results use 2,000 replications. The replication package documents which script 
produced each table and figure.}

\subsection{The case of constant effect: outcome regression model}
Tables~\ref{dir10} and~\ref{spi10} report the estimates from the regression model~(\ref{reg_ms}) 
under DGP~(\ref{dgp1}). The regression is biased because it is misspecified. 

  \begin{table}[h]
  \centering
  \caption{Regression: Finite Sample Properties: $\tau((1,1),(0,1))$ \label{dir10}}
  \begin{tabular}{l *{5}{S[table-format=-1.3]}}
  \hline\hline  
  {\rev{$G$}} & {Bias} & {RMSE} & {Coverage Rate} & {SE} & {SE - SD}  \\ 
  \hline
  50   &-0.304 &0.362    &0.627 &0.188  &-0.008 \\
  100  &-0.300 &0.327    &0.388 &0.133   &0.002 \\
  200  &-0.296 &0.311    &0.115 &0.094   &0.000 \\
  500  &-0.299 &0.305    &0.000 &0.059   &0.000 \\
  1000 &-0.300 &0.303    &0.000 &0.042  &-0.001 \\
  \hline\hline  
 \end{tabular}
\end{table}

  \begin{table}[h]
  \centering
  \caption{Regression: Finite Sample Properties: $\tau((1,1),(1,0))$ \label{spi10}}
  \begin{tabular}{l *{5}{S[table-format=-1.3]}}
  \hline\hline  
  {\rev{$G$}} & {Bias} & {RMSE} & {Coverage Rate} & {SE} & {SE - SD}  \\ 
  \hline
  50   &-0.341 &0.369    &0.312 &0.138  &-0.002 \\
  100  &-0.338 &0.351    &0.076 &0.097   &0.001 \\
  200  &-0.336 &0.343    &0.002 &0.069  &-0.002 \\
  500  &-0.333 &0.336    &0.000 &0.044  &-0.001 \\
  1000 &-0.331 &0.333    &0.000 &0.031   &0.000 \\
  \hline\hline  
 \end{tabular}
\end{table}

\subsection{The case of heterogeneous treatment effects}

To check that the method handles heterogeneous treatment effects (HTE), we use an alternative 
potential outcome model: 
  \begin{equation}\label{dgp3}
  Y_{ig}(d, z) = 1 + X_{ig} + \frac{\sum_{\{jg\} \in \mathcal{N}_g, \{jg\} \neq \{ig\}} X_{jg}}{3} + d + 0.5z + d\cdot z + X_{ig}\cdot d\cdot z + \varepsilon_{ig}.
  \end{equation}
Here the individual effect $ Y_{ig}(d, z) - Y_{ig}(\tilde{d}, \tilde{z}) $ depends on the covariate 
$X_{ig}$. Tables~\ref{dir2} and~\ref{spi2} report the AIPW estimates of the two effects and 
Tables~\ref{dir20} and~\ref{spi20} the regression estimates; Figure~\ref{sam_app_1} plots bias and 
RMSE for the spillover effect. The results are close to those under constant effects, although the 
sampling variation of the AIPW estimator is larger\rev{, and the coverage rate of the spillover effect drifts slightly above 
the nominal level (0.96--0.97) for large $G$: the average estimated standard error exceeds the 
sampling standard deviation by about five to ten percent in Table~\ref{spi2}, so the variance 
estimator is mildly conservative in this design}. 

  \begin{table}[h]
  \centering
  \caption{AIPW: Finite Sample Properties (HTE): $\tau((1,1),(0,1))$ \label{dir2}}
  \begin{tabular}{l *{5}{S[table-format=-1.3]}}
  \hline\hline  
  {\rev{$G$}} & {Bias} & {RMSE} & {Coverage Rate} & {SE} & {SE - SD}  \\ 
  \hline
    50  &0.006 &0.360   & 0.946 &0.349  &-0.017 \\
   100  &0.010 &0.250   & 0.950 &0.248  &-0.006 \\
   200  &0.001 &0.177   & 0.952 &0.180  &-0.001 \\
   500 &-0.002 &0.108   & 0.961 &0.114   &0.003 \\
  1000  &0.001 &0.078   & 0.956 &0.081   &0.001 \\
  \hline\hline  
 \end{tabular}
\end{table}

  \begin{table}[h]
  \centering
  \caption{Regression: Finite Sample Properties (HTE): $\tau((1,1),(0,1))$ \label{dir20}}
  \begin{tabular}{l *{5}{S[table-format=-1.3]}}
  \hline\hline  
  {\rev{$G$}} & {Bias} & {RMSE} & {Coverage Rate} & {SE} & {SE - SD}  \\ 
  \hline
   50 &-0.317 &0.377    &0.680 &0.211   &0.003 \\
   100 &-0.305 &0.340    &0.461 &0.149  &-0.006 \\
   200 &-0.306 &0.323    &0.157 &0.105  &-0.001 \\
   500 &-0.306 &0.312    &0.002 &0.066   &0.001 \\
  1000 &-0.307 &0.310    &0.000 &0.047   &0.000 \\
  \hline\hline  
 \end{tabular}
\end{table}

  \begin{table}[h]
  \centering
  \caption{AIPW: Finite Sample Properties (HTE): $\tau((1,1),(1,0))$ \label{spi2}}
  \begin{tabular}{l *{5}{S[table-format=-1.3]}}
  \hline\hline  
  {\rev{$G$}} & {Bias} & {RMSE} & {Coverage Rate} & {SE} & {SE - SD}  \\ 
  \hline
  50   &0.007 &0.310    &0.951 &0.316  &-0.001 \\
 100   &0.011 &0.215    &0.951 &0.225   &0.007 \\
 200  &-0.002 &0.152    &0.965 &0.164   &0.009 \\
 500  &-0.001 &0.094    &0.971 &0.104   &0.006 \\
 1000  &0.001 &0.070    &0.965 &0.074   &0.002 \\
  \hline\hline  
 \end{tabular}
\end{table}

  \begin{table}[h]
  \centering
  \caption{Regression: Finite Sample Properties (HTE): $\tau((1,1),(1,0))$ \label{spi20}}
  \begin{tabular}{l *{5}{S[table-format=-1.3]}}
  \hline\hline  
  {\rev{$G$}} & {Bias} & {RMSE} & {Coverage Rate} & {SE} & {SE - SD}  \\ 
  \hline
     50 &-0.172 &0.254    &0.795 &0.184  &-0.006 \\
    100 &-0.173 &0.218    &0.709 &0.131  &-0.005 \\
    200 &-0.169 &0.194    &0.552 &0.095  &-0.001 \\
    500 &-0.168 &0.178    &0.206 &0.060   &0.000 \\
   1000 &-0.169 &0.174    &0.024 &0.042   &0.001 \\
  \hline\hline  
 \end{tabular}
\end{table}

\begin{figure}[!ht]
  \begin{center}
  \begin{tabular}{cc}
   \raisebox{-0.5\height}{\includegraphics[scale=0.3]{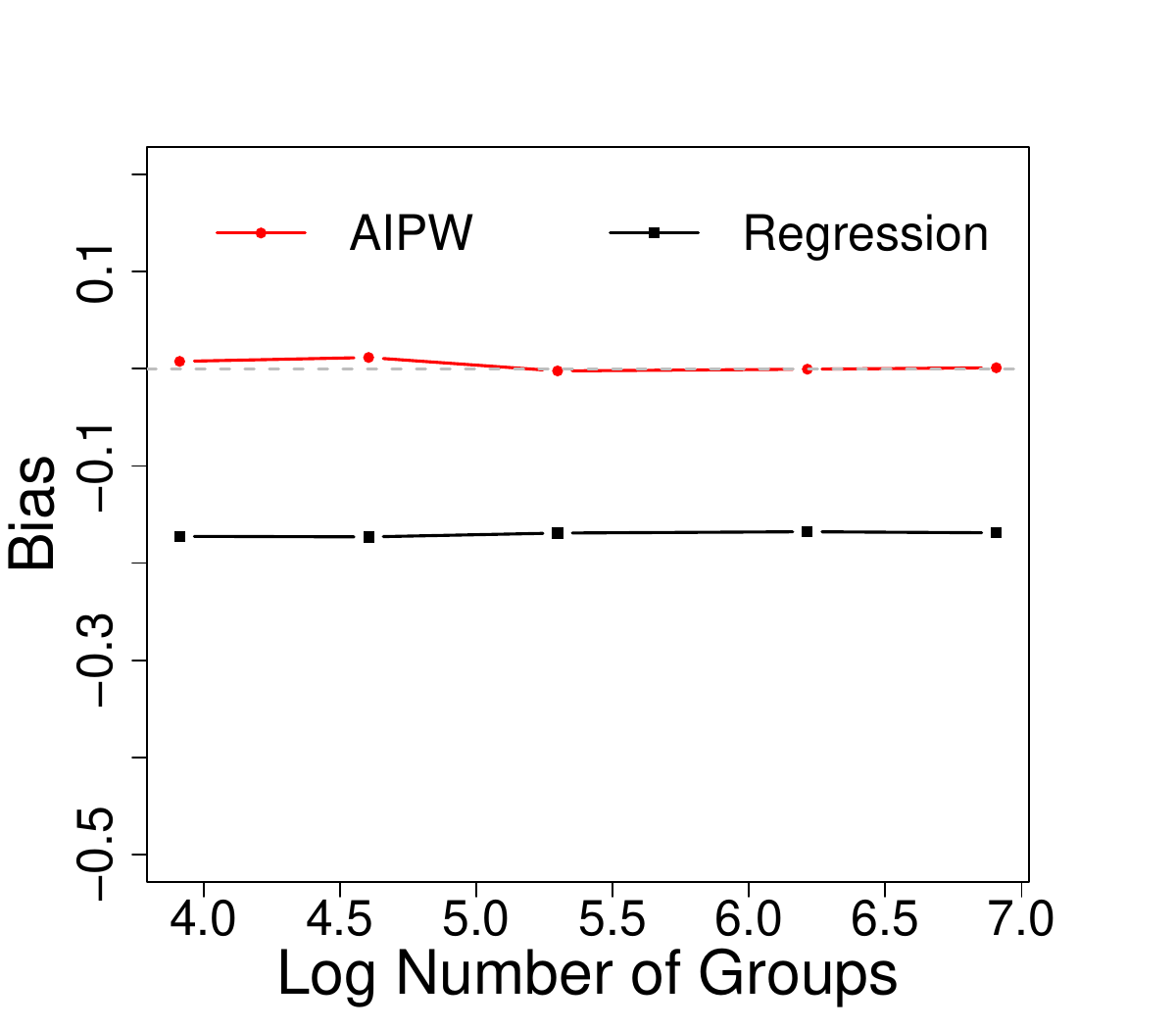}} &
   \raisebox{-0.5\height}{\includegraphics[scale=0.3]{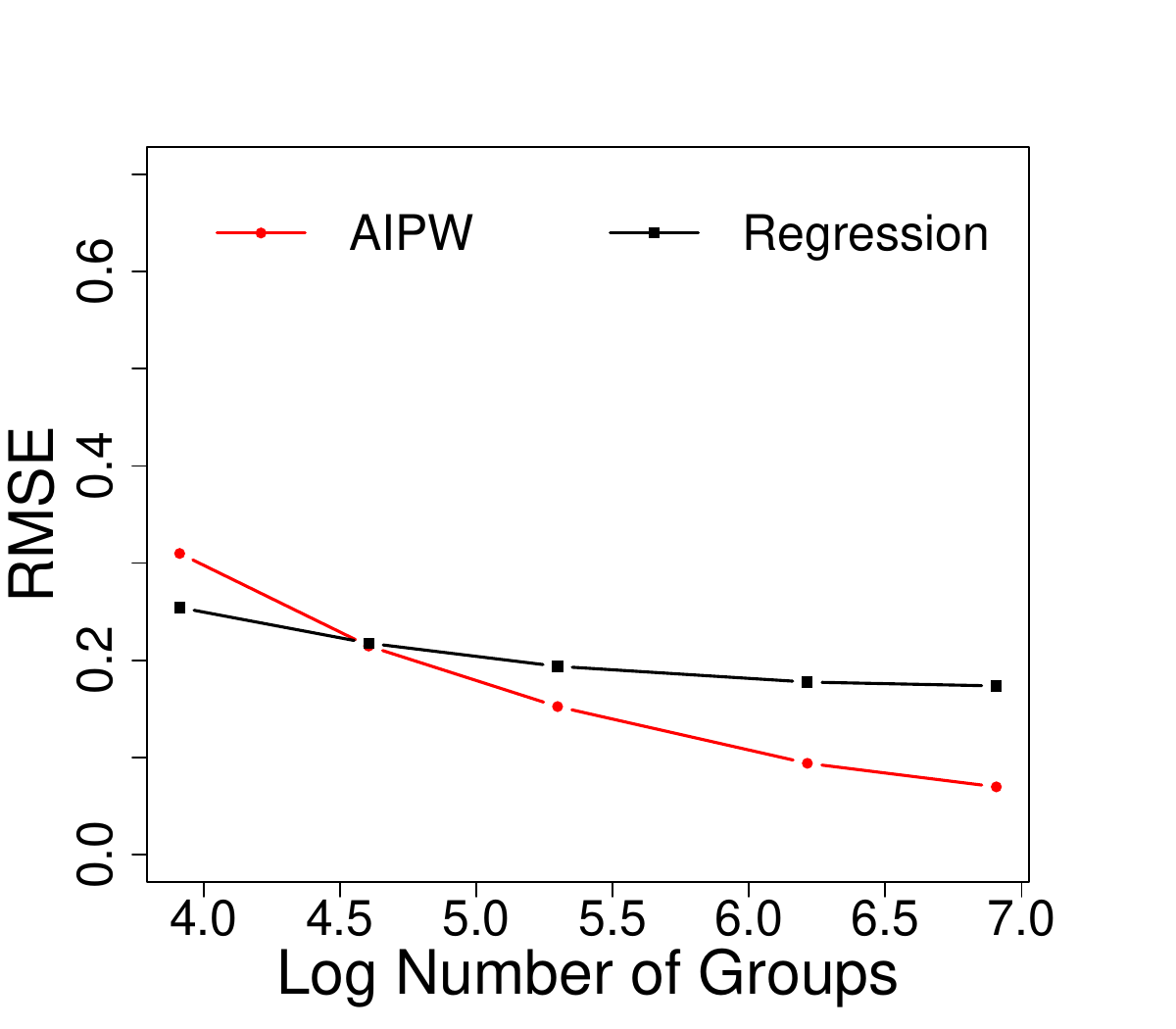}}  
   \end{tabular}
   \caption{Finite Sample Properties (Spillover Effect, HTE)}\label{sam_app_1}
   \end{center}
   \end{figure}

\subsection{The case of mis-specified models for network formation and treatment assignment}

Because the generalized propensity score is built from the network and adoption models, 
misspecifying them may bias the score and the treatment effects. To examine this, we use the 
following DGP for network formation and treatment assignment: 
 \begin{equation}\label{dgp4}
  \begin{gathered}
  W_{ijg} = 
      \mathbbm{1}\{ 0.1 - 0.25 |X_{ig} - X_{jg}| - 2|U_{ig} - U_{jg}| + \nu_{ijg} \geq 0 \}, \\
    D_{ig} = \mathbbm{1}\{ 0.1 + 0.5 X_{ig} - 0.3U_{ig} + e_{ig} \geq 0 \}, 
  \end{gathered}
  \end{equation}
  where $ U_{ig} \sim Bern(0.5) $ is a binary variable. \rev{The 
  potential outcome model is unchanged from~(\ref{dgp2}): $U_{ig}$ affects adoption and tie 
  formation but not the outcome.} 
 We fit the same network, adoption and outcome models as in Section~\ref{SMC} and treat $U_{ig}$ 
 as unobserved, so the network and adoption models are misspecified. Tables~\ref{dir3} 
 and~\ref{spi3} report the AIPW estimates and Tables~\ref{dir30} and~\ref{spi30} the regression 
 estimates; Figure~\ref{sam_app_2} plots bias and RMSE for the spillover effect. \rev{The AIPW estimator remains essentially 
 unbiased, with coverage close to nominal, even though both components of its score are 
 misspecified. This is expected: because $U_{ig}$ does not affect the outcome, 
 Assumption~\ref{asm:ig} holds given $\bm{X}_g$ alone, the marginal score given $\bm{X}_g$ is still 
 a valid weight, and the misspecified score is a different function of $\bm{X}_g$ whose error is 
 partly compensated by the outcome component. The regression estimator's bias grows relative to 
 Appendix~\ref{B}.1 because the omitted interaction and contextual terms now covary with 
 $U_{ig}$-driven exposure. The design thus shows robustness to misspecification of the 
 formation and adoption models in a direction unrelated to the outcome. It is not a test of 
 unobserved confounding, which requires $U_{ig}$ to enter the outcome as well.} \rev{Appendix~\ref{app:conf} adds the confounded design.}

  \begin{table}[h]
  \centering
  \caption{AIPW: Finite Sample Properties (Mis-specification): $\tau((1,1),(0,1))$ \label{dir3}}
  \begin{tabular}{l *{5}{S[table-format=-1.3]}}
  \hline\hline  
  {\rev{$G$}} & {Bias} & {RMSE} & {Coverage Rate} & {SE} & {SE - SD}  \\ 
  \hline
      50 &-0.006 &0.391    &0.947 &0.379  &-0.012 \\
     100 &-0.005 &0.270    &0.947 &0.270  &-0.001 \\
     200 &-0.013 &0.201    &0.951 &0.194  &-0.007 \\
     500 &-0.005 &0.129    &0.943 &0.124  &-0.004 \\
    1000 &-0.010 &0.087    &0.952 &0.088   &0.001 \\
  \hline\hline  
 \end{tabular}
\end{table}

  \begin{table}[h]
  \centering
  \caption{Regression: Finite Sample Properties (Mis-specification): $\tau((1,1),(0,1))$ \label{dir30}}
  \begin{tabular}{l *{5}{S[table-format=-1.3]}}
  \hline\hline  
  {\rev{$G$}} & {Bias} & {RMSE} & {Coverage Rate} & {SE} & {SE - SD}  \\ 
  \hline
   50 &-0.681 &0.705    &0.036 &0.182   &0.000 \\
  100 &-0.677 &0.689    &0.000 &0.128  &-0.003 \\
  200 &-0.675 &0.680    &0.000 &0.091   &0.003 \\
  500 &-0.671 &0.674    &0.000 &0.057  &-0.001 \\
 1000 &-0.673 &0.674    &0.000 &0.041   &0.001 \\
  \hline\hline  
 \end{tabular}
\end{table}

  \begin{table}[h]
  \centering
  \caption{AIPW: Finite Sample Properties (Mis-specification): $\tau((1,1),(1,0))$ \label{spi3}}
  \begin{tabular}{l *{5}{S[table-format=-1.3]}}
  \hline\hline  
  {\rev{$G$}} & {Bias} & {RMSE} & {Coverage Rate} & {SE} & {SE - SD}  \\ 
  \hline
    50  &0.001 &0.332    &0.939 &0.313  &-0.019 \\
   100 &-0.002 &0.230    &0.950 &0.223  &-0.007 \\
   200 &-0.004 &0.164    &0.940 &0.161  &-0.003 \\
   500  &0.001 &0.103    &0.948 &0.103  & 0.000 \\
  1000 &-0.003 &0.073    &0.947 &0.073  &-0.001 \\
  \hline\hline  
 \end{tabular}
\end{table}

  \begin{table}[h]
  \centering
  \caption{Regression: Finite Sample Properties (Mis-specification): $\tau((1,1),(1,0))$ \label{spi30}}
  \begin{tabular}{l *{5}{S[table-format=-1.3]}}
  \hline\hline  
  {\rev{$G$}} & {Bias} & {RMSE} & {Coverage Rate} & {SE} & {SE - SD}  \\ 
  \hline
   50 &-0.398 &0.440    &0.414 &0.182  &-0.006 \\
  100 &-0.391 &0.414    &0.167 &0.130  &-0.004 \\
  200 &-0.388 &0.399    &0.011 &0.092   &0.002 \\
  500 &-0.385 &0.389    &0.000 &0.058   &0.001 \\
 1000 &-0.386 &0.388    &0.000 &0.041   &0.001 \\
  \hline\hline  
 \end{tabular}
\end{table}

\begin{figure}[!ht]
  \begin{center}
  \begin{tabular}{cc}
   \raisebox{-0.5\height}{\includegraphics[scale=0.3]{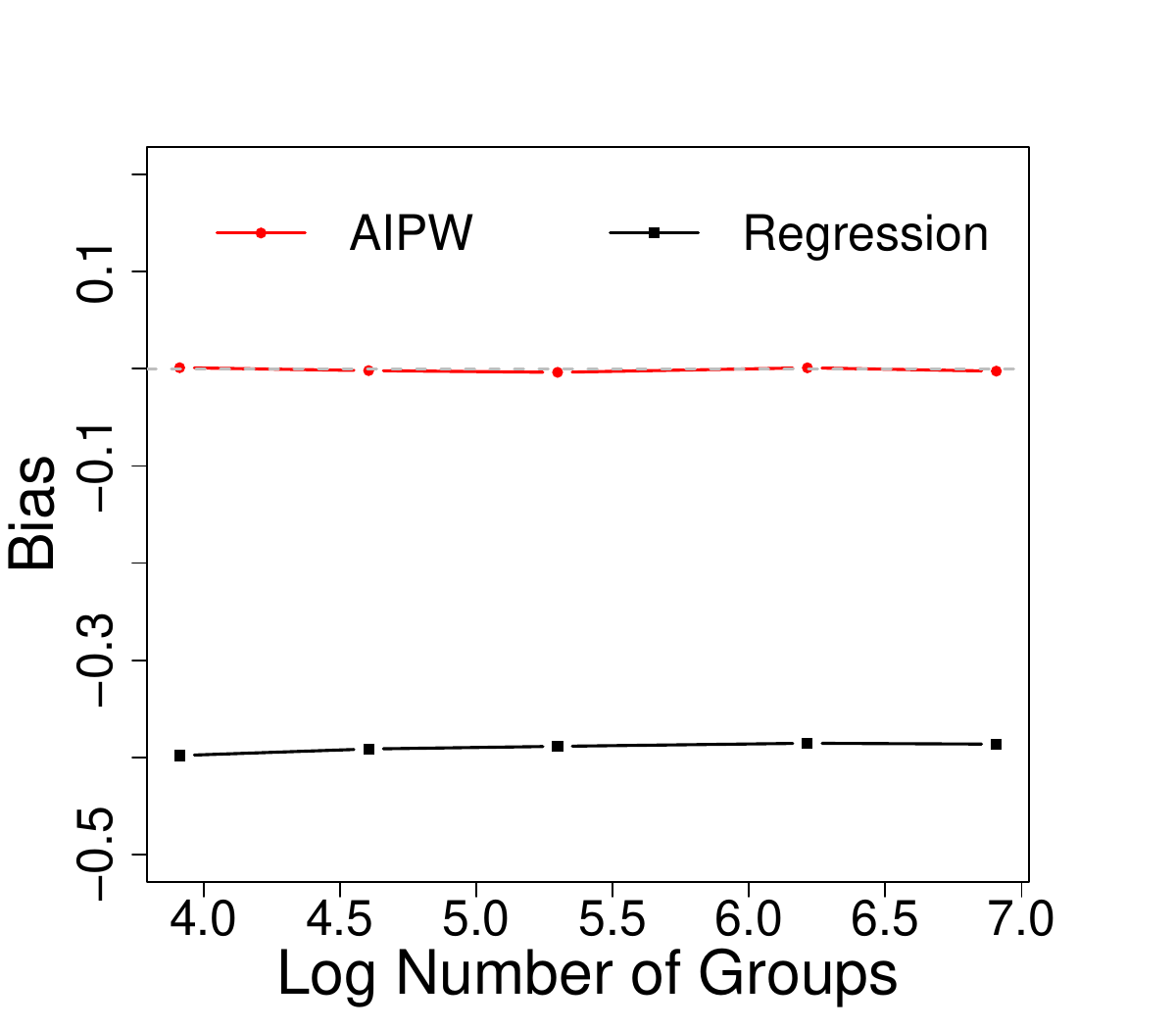}} &
   \raisebox{-0.5\height}{\includegraphics[scale=0.3]{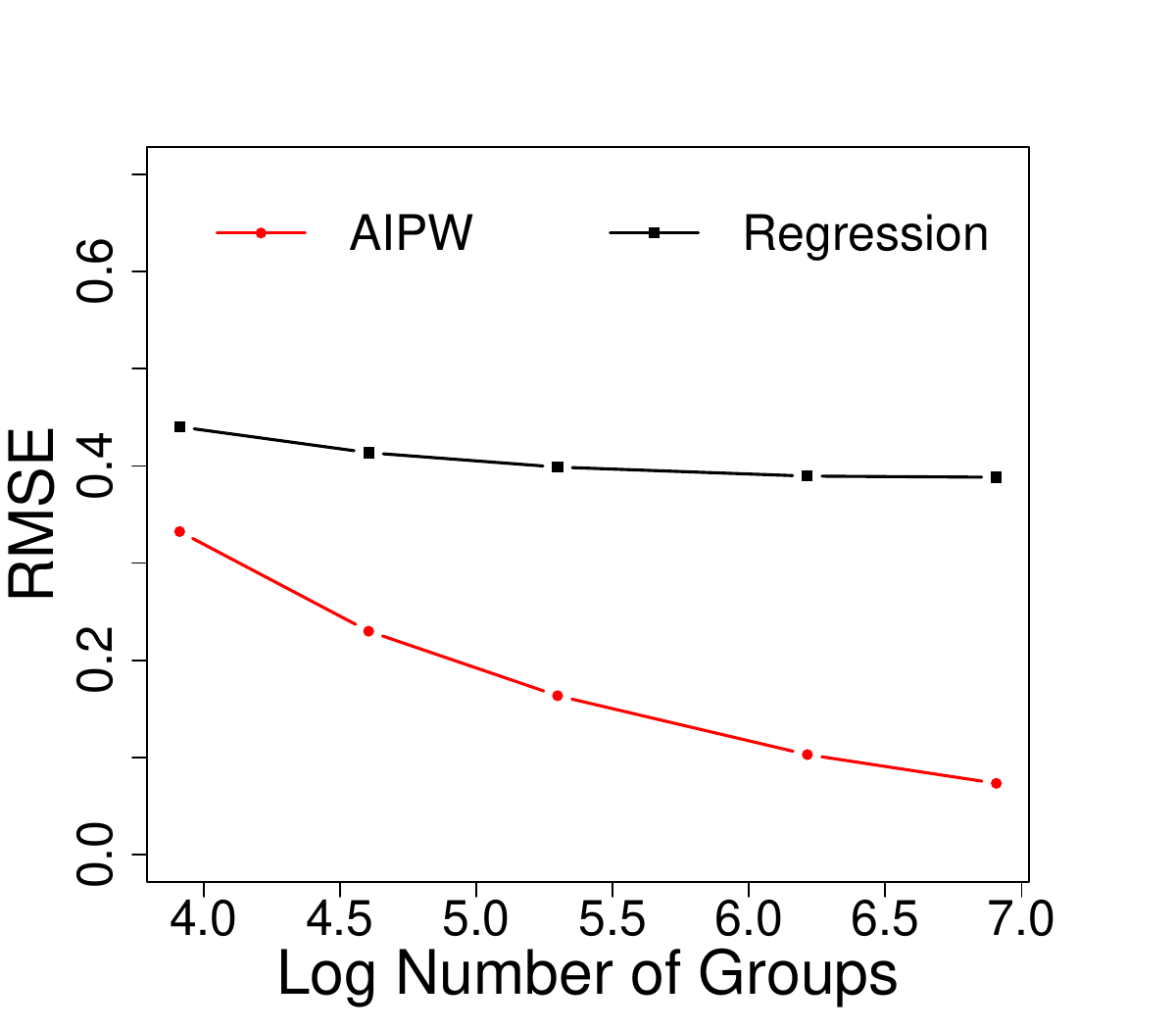}}  
   \end{tabular}
   \caption{Finite Sample Properties (Spillover Effect, Mis-specification)}\label{sam_app_2}
   \end{center}
   \end{figure}

\begin{revblock}
\subsection{A genuinely confounded unobserved covariate}\label{app:conf}

To calibrate the sensitivity analysis of Section~\ref{sec:sens}, we let the unobserved binary 
covariate $U_{ig} \sim \text{Bern}(0.5)$ of the previous subsection also enter the outcome: 
$Y_{ig}(d, z)$ in~(\ref{dgp2}) is augmented by $\gamma_y U_{ig}$, with $U_{ig}$ entering tie 
formation and adoption with the coefficients of~(\ref{dgp4}). The AIPW estimator omits $U_{ig}$ from 
all three models (``AIPW-JM''); an oracle estimator includes it (``Oracle''). Table~\ref{tab:conf} 
reports bias and coverage for $G = 500$ and 1,000 replications as $\gamma_y$ varies. The bias is 
linear in $\gamma_y$, about $-0.12 \gamma_y$ for the direct effect and $-0.055 \gamma_y$ for the 
spillover effect, and coverage of the direct effect falls to $0.90$ at $\gamma_y = 0.5$ and below it beyond. 
Since the outcome error has unit variance, $\gamma_y = 0.5$ corresponds to an omitted covariate 
that shifts the outcome by half a standard deviation of the idiosyncratic error while reducing 
the adoption index by $0.3$ and the link index by $2$ for dissimilar pairs; this is the reference 
point used in Section~\ref{sec:sens}. The oracle has a larger RMSE than AIPW-JM at 
$\gamma_y = 0$ ($0.17$ against $0.10$) because the coefficient of $2$ on $|U_{ig} - U_{jg}|$ in 
the link index makes the scores extreme once $U_{ig}$ is included, a feature of the design rather 
than an error.

\begin{table}[h]
\centering
\caption{Bias and coverage with an unobserved covariate $U_{ig}$ in ties, adoption and outcome 
($G = 500$)}\label{tab:conf}
\begin{tabular}{S[table-format=1.2] *{4}{S[table-format=-1.3]} c *{4}{S[table-format=-1.3]}}
\hline\hline
 & \multicolumn{4}{c}{AIPW-JM (omits $U_{ig}$)} & & \multicolumn{4}{c}{Oracle (includes $U_{ig}$)} \\
\cline{2-5} \cline{7-10}
 & \multicolumn{2}{c}{Direct} & \multicolumn{2}{c}{Spillover} & & \multicolumn{2}{c}{Direct} & \multicolumn{2}{c}{Spillover} \\
{$\gamma_y$} & {Bias} & {Cov.} & {Bias} & {Cov.} & & {Bias} & {Cov.} & {Bias} & {Cov.} \\
\hline
0.00 &  0.006 & 0.961 &  0.005 & 0.946 & &  0.000 & 0.951 &  0.005 & 0.956 \\
0.25 & -0.036 & 0.932 & -0.016 & 0.942 & & -0.005 & 0.956 & -0.006 & 0.940 \\
0.50 & -0.061 & 0.905 & -0.027 & 0.934 & & -0.004 & 0.947 &  0.003 & 0.952 \\
1.00 & -0.121 & 0.826 & -0.053 & 0.931 & &  0.003 & 0.966 &  0.006 & 0.961 \\
2.00 & -0.244 & 0.608 & -0.120 & 0.837 & &  0.003 & 0.953 & -0.001 & 0.957 \\
\hline\hline
\end{tabular}
\end{table}
\end{revblock}

\subsection{Comparison with a score that conditions on own covariates only}\label{abc}


\rev{This comparison, carried over from the 2025 version of the paper, is between the AIPW 
estimator with the marginal score (``AIPW-JM'') and an AIPW estimator with the same outcome model 
whose score is $\Pr(T_{ig} = t \mid X_{ig}) = \Pr(D_{ig} = d \mid X_{ig}) \Pr(Z_{ig} = z \mid X_{ig})$, 
with a probit model for the first factor and a multinomial logit for the second over the four 
levels of $Z_{ig}$ (``AIPW-SEP'', for separate modeling of the two components). AIPW-SEP conditions 
on the unit's own covariate only, neither on the realized network nor on the covariates of the 
other units in the group. It is not the network-conditional score $r^{c}$ of 
Section~\ref{sec:cond}; the comparison with $r^{c}$ is in Table~\ref{tab:cond} of the main text.}

For the direct effect $\tau((1,1),(0,1))$ the two estimators perform about equally well 
(Figure~\ref{sam1_d_ps}). For the spillover effect $\tau((1,1),(1,0))$ AIPW-JM does much better 
(Figure~\ref{sam1_ps}). The sampling distributions in Figure~\ref{sam_ps} show that it stays 
unbiased even with few groups, whereas AIPW-SEP has a small but persistent bias that does not 
shrink as the number of groups grows. \rev{As explained in Section~\ref{SMC}, the bias of AIPW-SEP 
is omitted-confounder bias: the covariates of the other units in the group enter both the 
exposure and the outcome, and AIPW-SEP's score and outcome model both omit them.}

\begin{figure}[!ht]
  \begin{center}
  \begin{tabular}{cc}
   \raisebox{-0.5\height}{\includegraphics[scale=0.3]{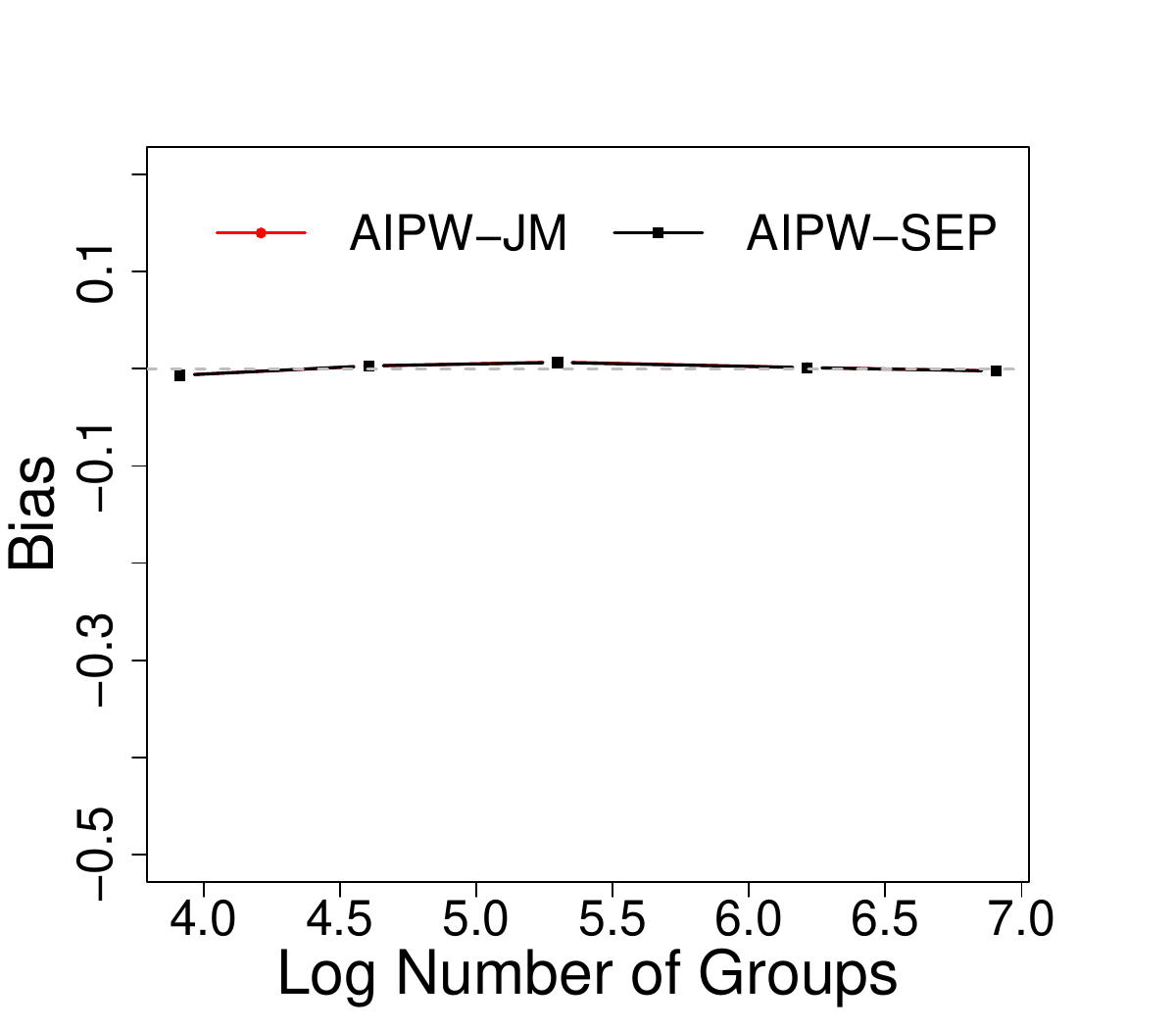}} &
   \raisebox{-0.5\height}{\includegraphics[scale=0.3]{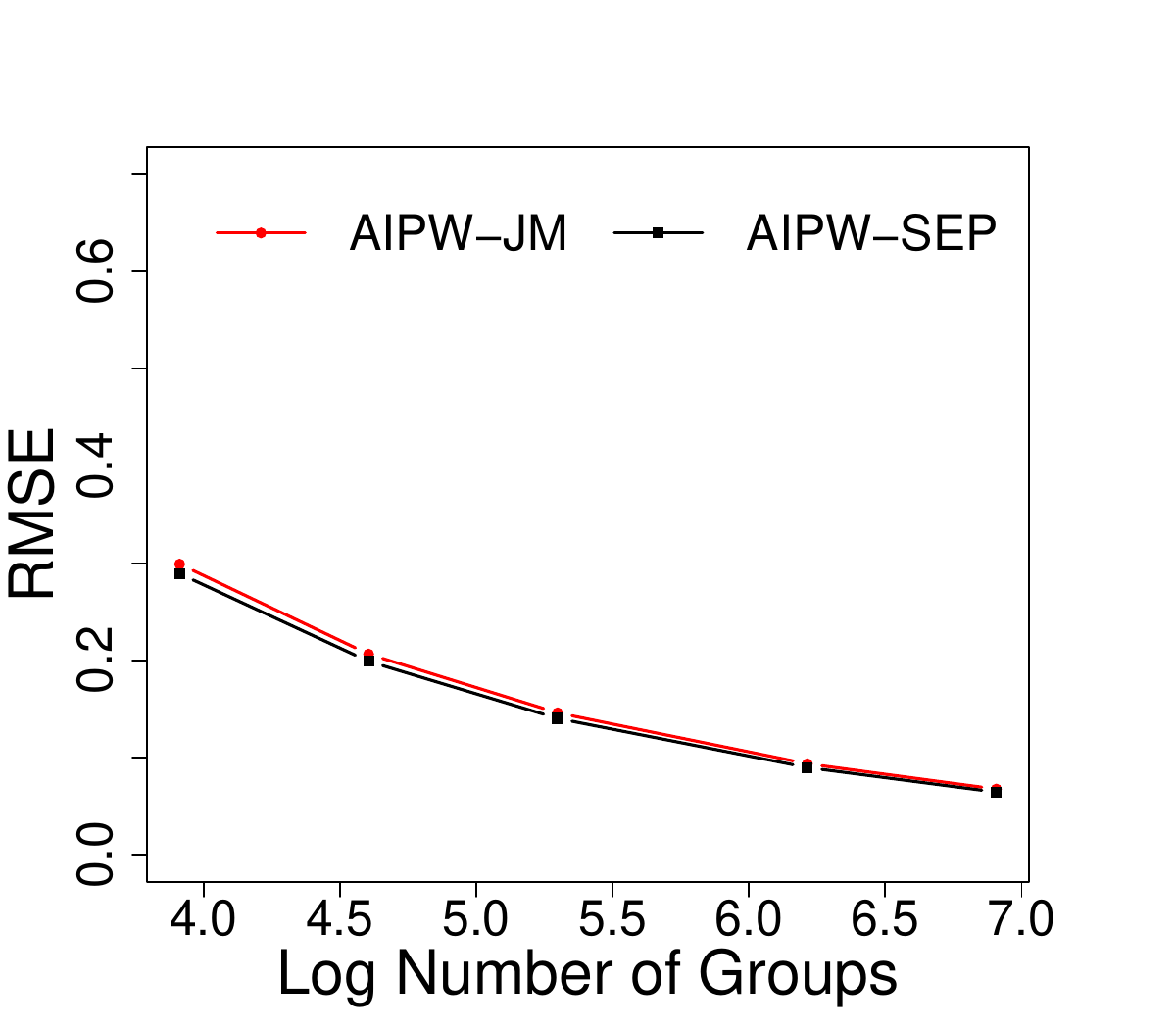}}  
   \end{tabular}
   \caption{Finite Sample Properties (Direct Effect)}\label{sam1_d_ps}
   \end{center}
   \end{figure}

\begin{figure}[!ht]
  \begin{center}
  \begin{tabular}{cc}
   \raisebox{-0.5\height}{\includegraphics[scale=0.3]{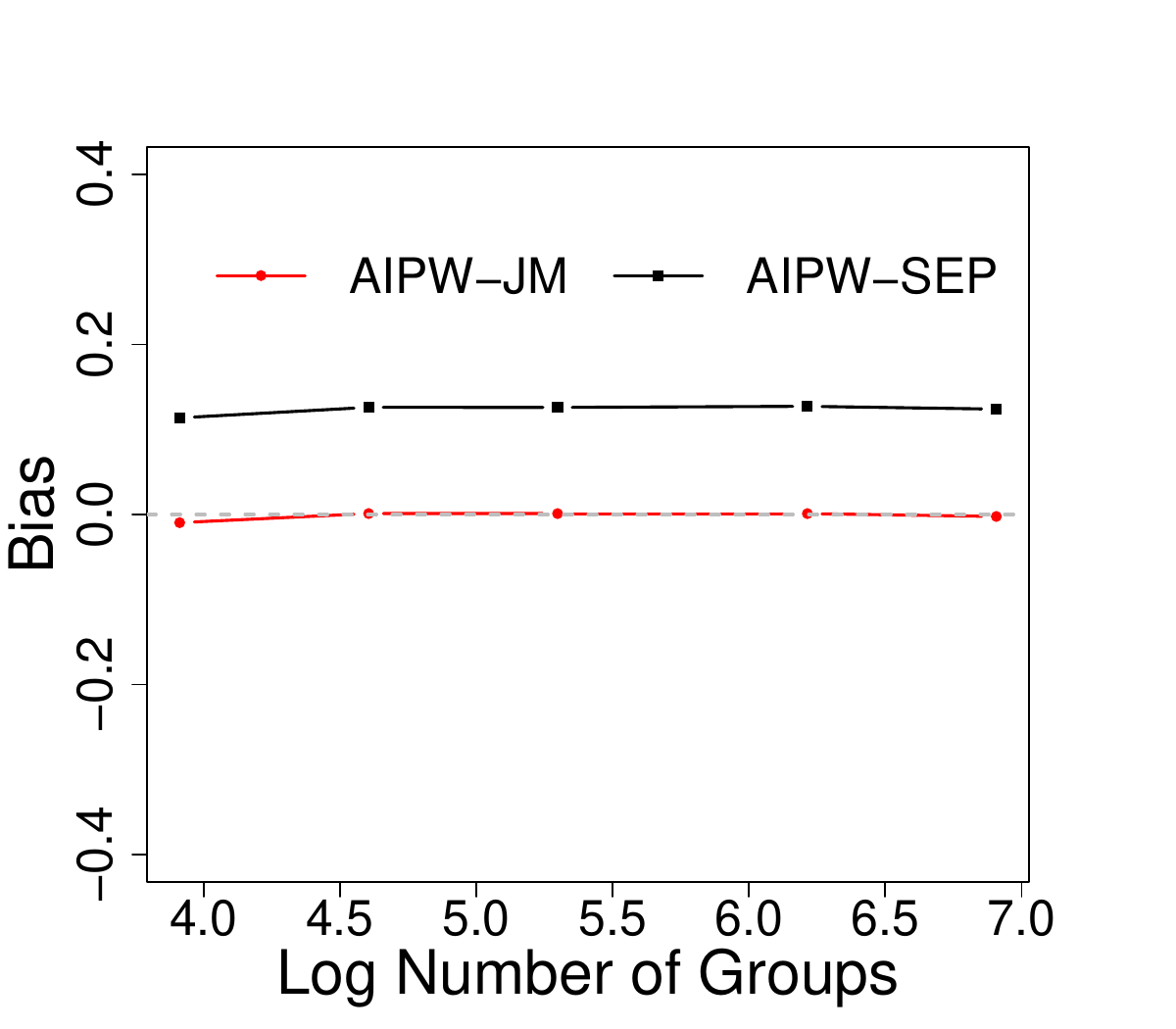}} &
   \raisebox{-0.5\height}{\includegraphics[scale=0.3]{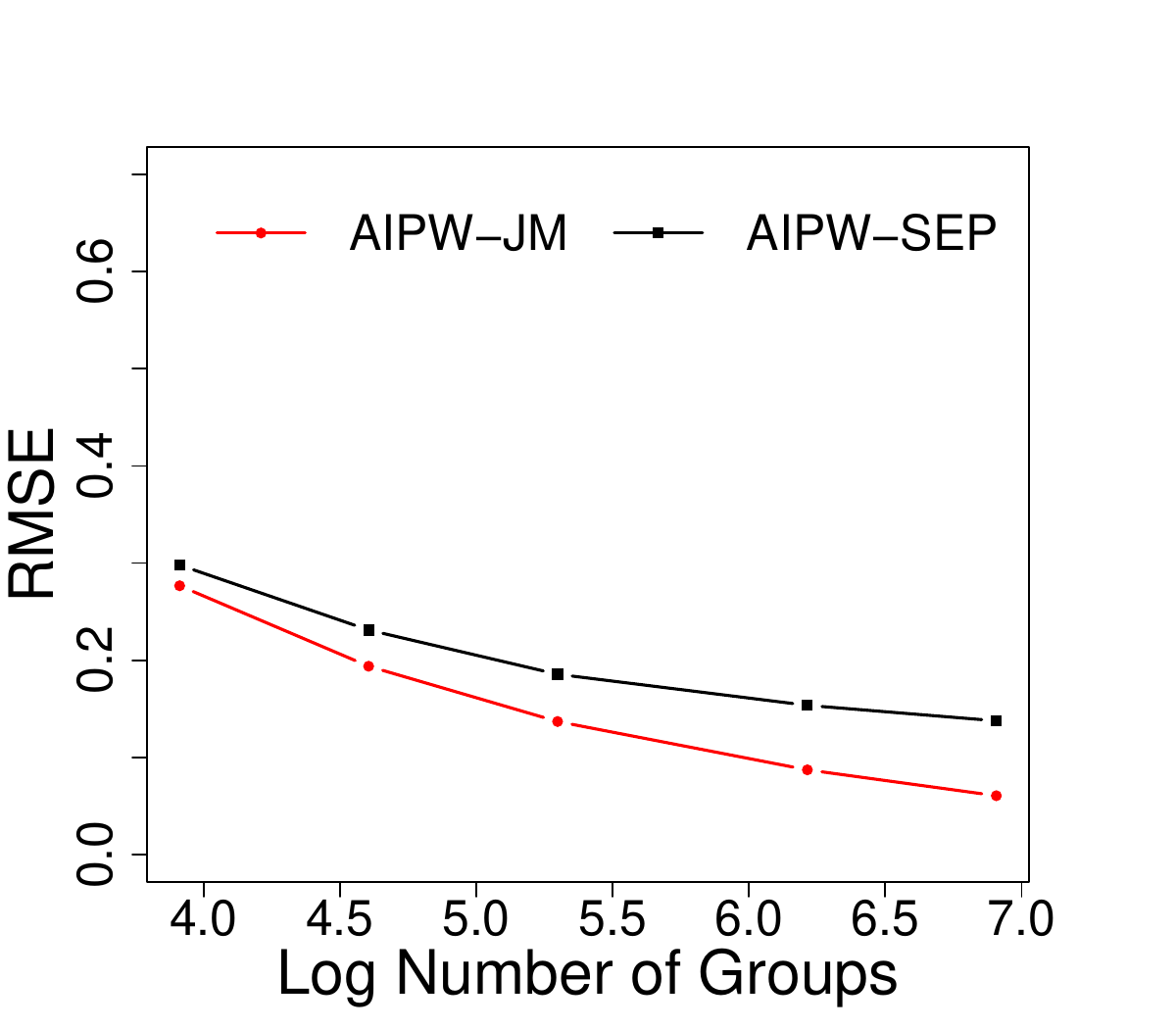}}  
   \end{tabular}
   \caption{Finite Sample Properties (Spillover Effect)}\label{sam1_ps}
   \end{center}
   \end{figure}

\begin{figure}[!ht]
  \begin{center}
  \begin{tabular}{ccc}
   \raisebox{-0.5\height}{\includegraphics[scale=0.2]{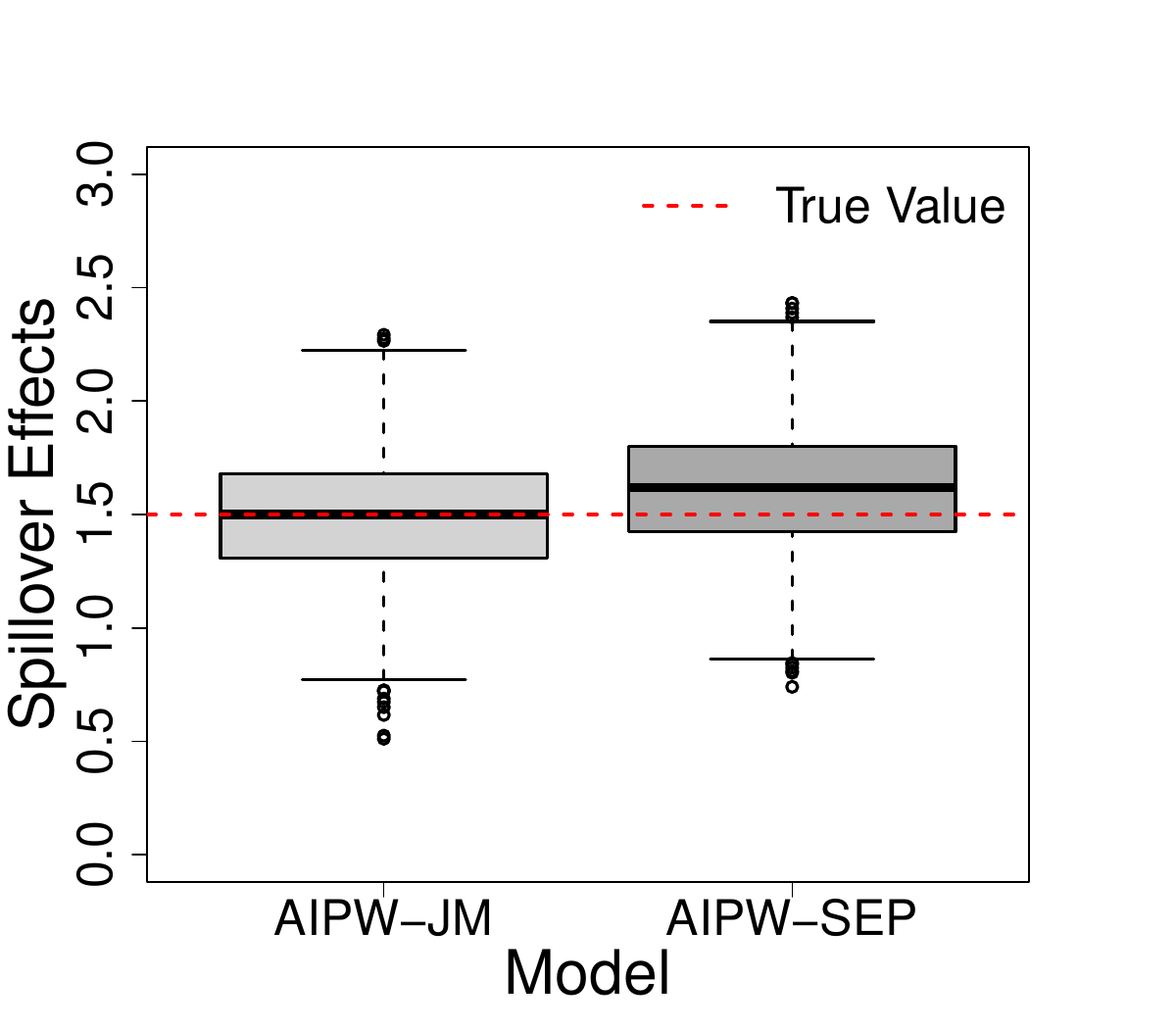}} &
   \raisebox{-0.5\height}{\includegraphics[scale=0.2]{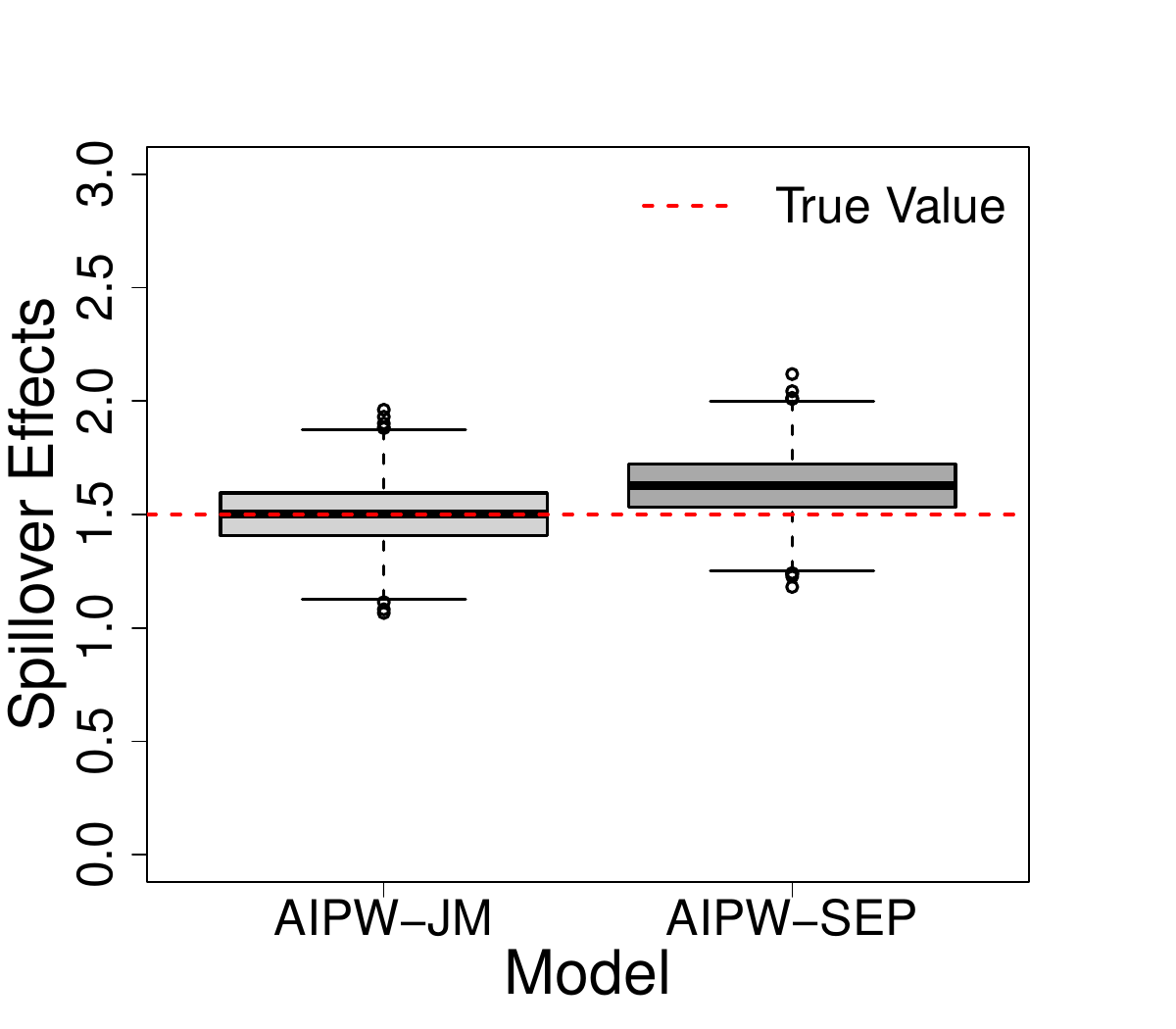}} &
   \raisebox{-0.5\height}{\includegraphics[scale=0.2]{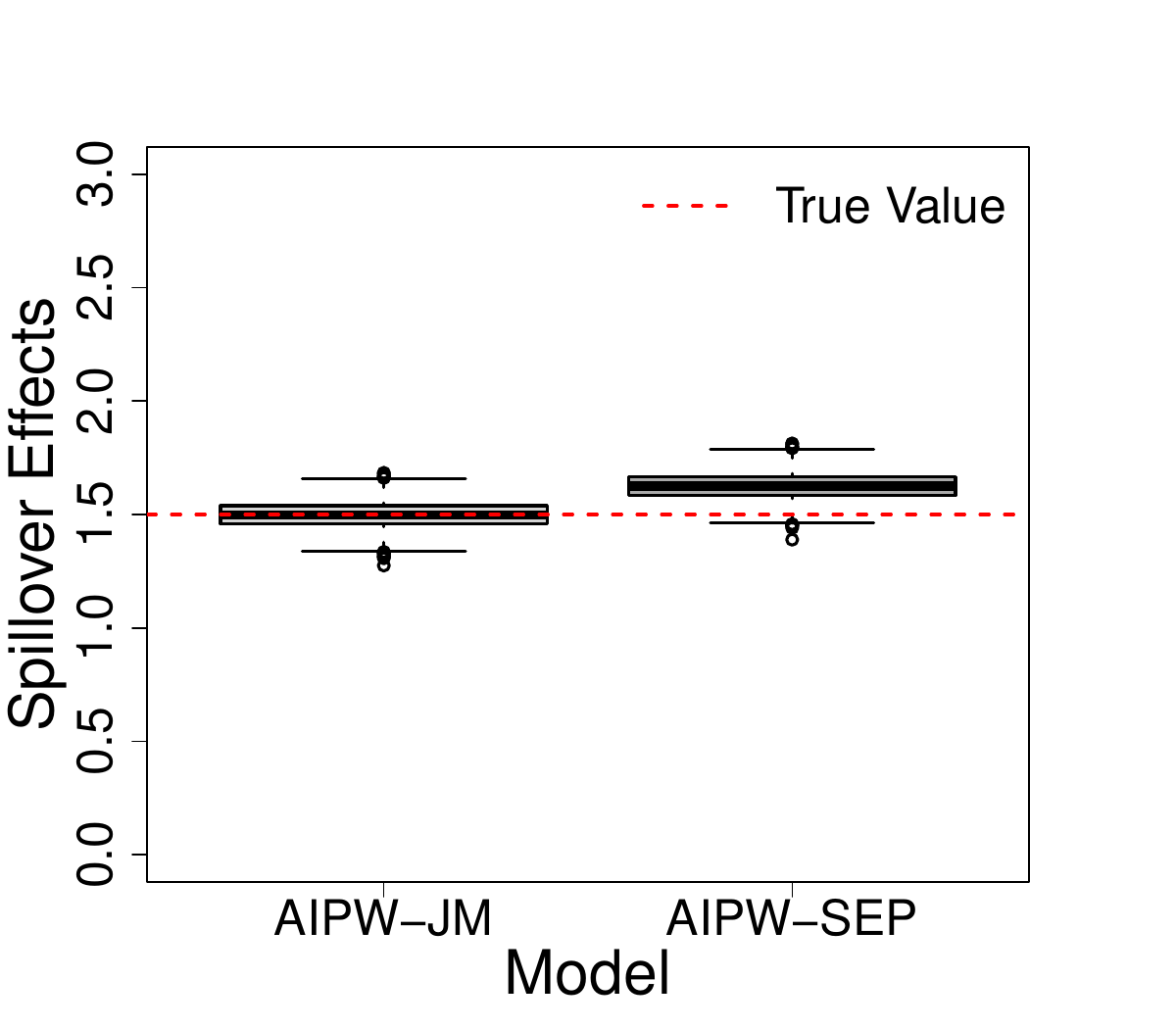}} \\
   50 Groups & 200 Groups & 1000 Groups
   \end{tabular}
   \caption{Sampling Distribution of the Spillover Effects with Different Number of Groups}\label{sam_ps}
   \end{center}
   \end{figure}

\begin{revblock}
Table~\ref{tab:sep} reports the finite-sample properties of the two estimators for the spillover 
effect, from the same run as Table~\ref{tab:bench}. AIPW-SEP has a bias of about $0.12$ that does not shrink with $G$, so its coverage falls from 
$0.92$ at $G = 50$ to $0.47$ at $G = 1000$; its sampling variation is almost identical to that of 
AIPW-JM, so the comparison is about bias, not variance. For the direct effect the two estimators are 
indistinguishable (bias below $0.015$, coverage $0.94$--$0.96$ for both).

\begin{table}[h]
\centering
\caption{Marginal score with (AIPW-JM) and without (AIPW-SEP) the covariates of other units: 
spillover effect $\tau((1,1),(1,0))$}\label{tab:sep}
\begin{tabular}{l *{3}{S[table-format=-1.3]} c *{3}{S[table-format=-1.3]}}
\hline\hline
 & \multicolumn{3}{c}{AIPW-JM} & & \multicolumn{3}{c}{AIPW-SEP} \\
\cline{2-4} \cline{6-8}
{$G$} & {Bias} & {RMSE} & {Coverage} & & {Bias} & {RMSE} & {Coverage} \\
\hline
50   & -0.003 & 0.277 & 0.943 & & 0.117 & 0.303 & 0.917 \\
100  & -0.002 & 0.196 & 0.950 & & 0.125 & 0.232 & 0.892 \\
200  & -0.002 & 0.137 & 0.948 & & 0.124 & 0.184 & 0.846 \\
500  &  0.000 & 0.087 & 0.948 & & 0.125 & 0.152 & 0.695 \\
1000 &  0.002 & 0.061 & 0.956 & & 0.127 & 0.141 & 0.470 \\
\hline\hline
\end{tabular}
\parbox{0.85\textwidth}{\footnotesize{\bf Notes:} 2,000 replications (the run of Table~\ref{tab:bench}); cluster-robust 95\% intervals. 
Both estimators use the outcome model $y \sim d + z + x$.}
\end{table}
\end{revblock}

\section{Additional Results and Plots for the Empirical Application}\label{C}

This appendix collects additional results and plots for the empirical application. 

\subsection{Visualization of network structure in each village}

Figures~\ref{ns2} and~\ref{ns3} show the networks of the 16 villages. Residents who adopted 
U-Bridge are in blue and residents who hold a formal leadership role in red.

\subsection{Estimation results for network formation and treatment assignment modeling}

Table~\ref{nreg2} reports the estimated coefficients, with 95\% intervals, of the additive and 
multiplicative effects (AME) network models for the 16 villages. Table~\ref{reg2} reports the 
logit model for adoption.

\begin{landscape}
\begin{figure}[!ht]
  \begin{center}
  \begin{tabular}{cccc}
   \raisebox{-0.5\height}{\includegraphics[scale=0.2]{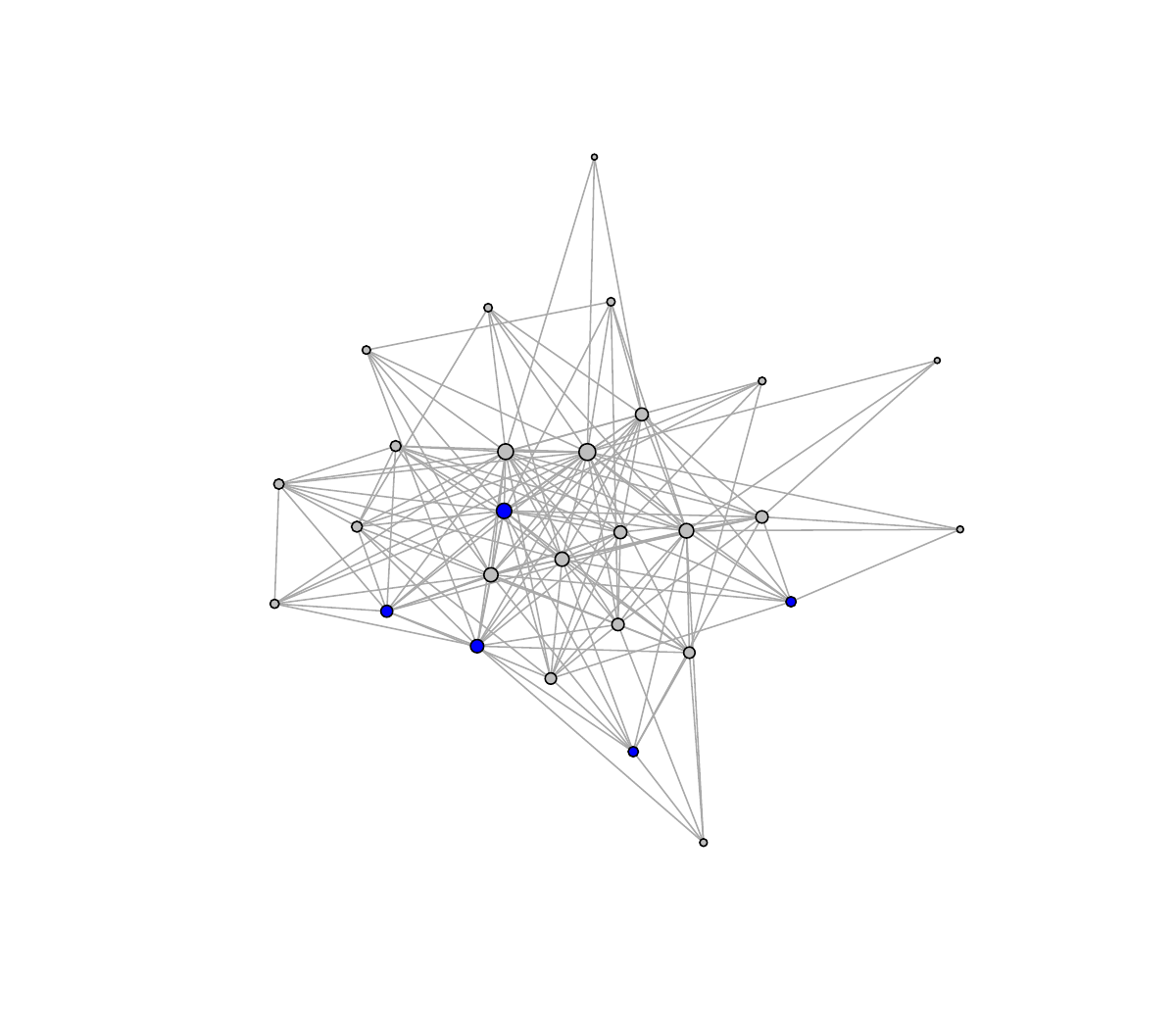}} &
   \raisebox{-0.5\height}{\includegraphics[scale=0.2]{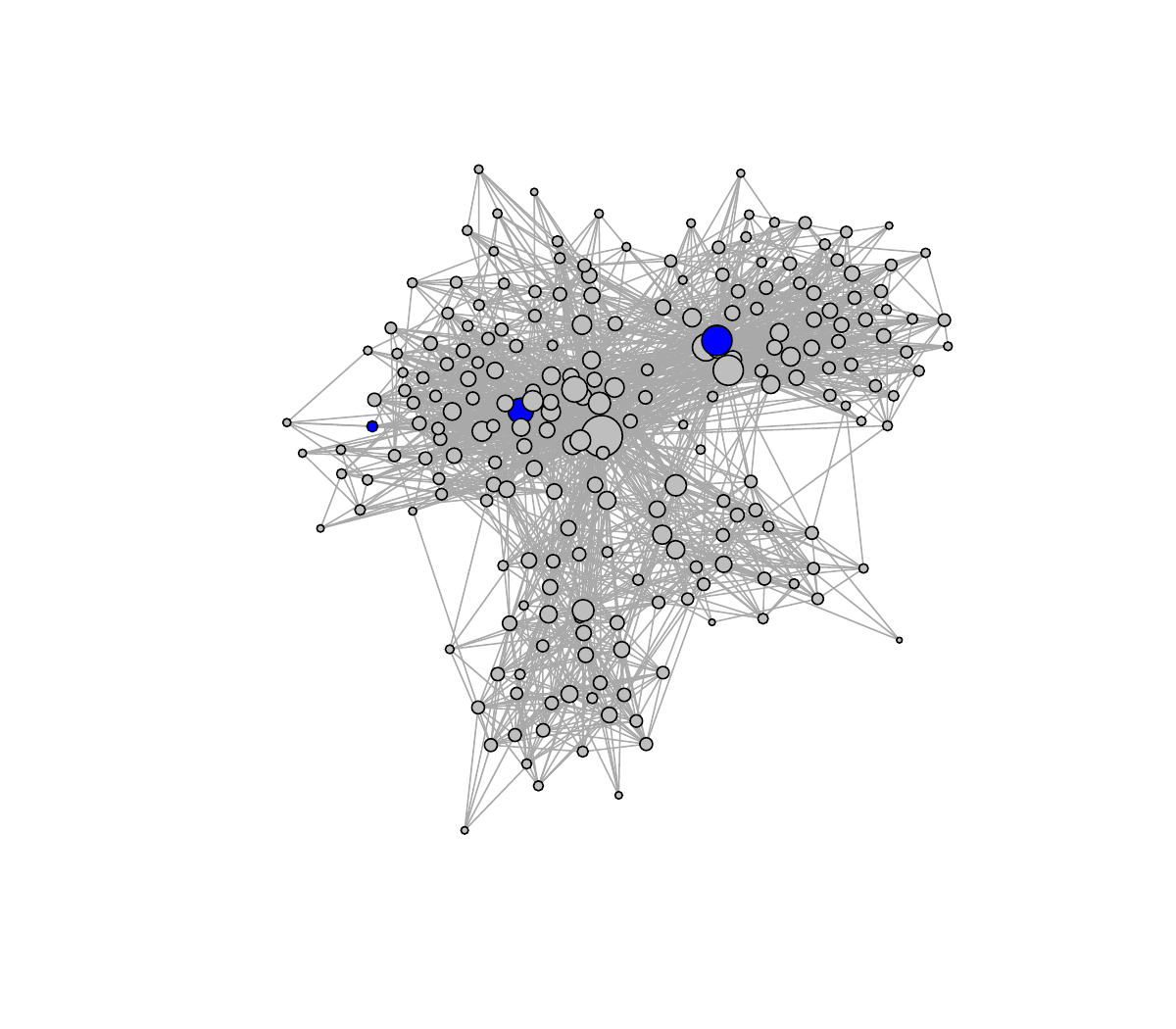}} &
   \raisebox{-0.5\height}{\includegraphics[scale=0.2]{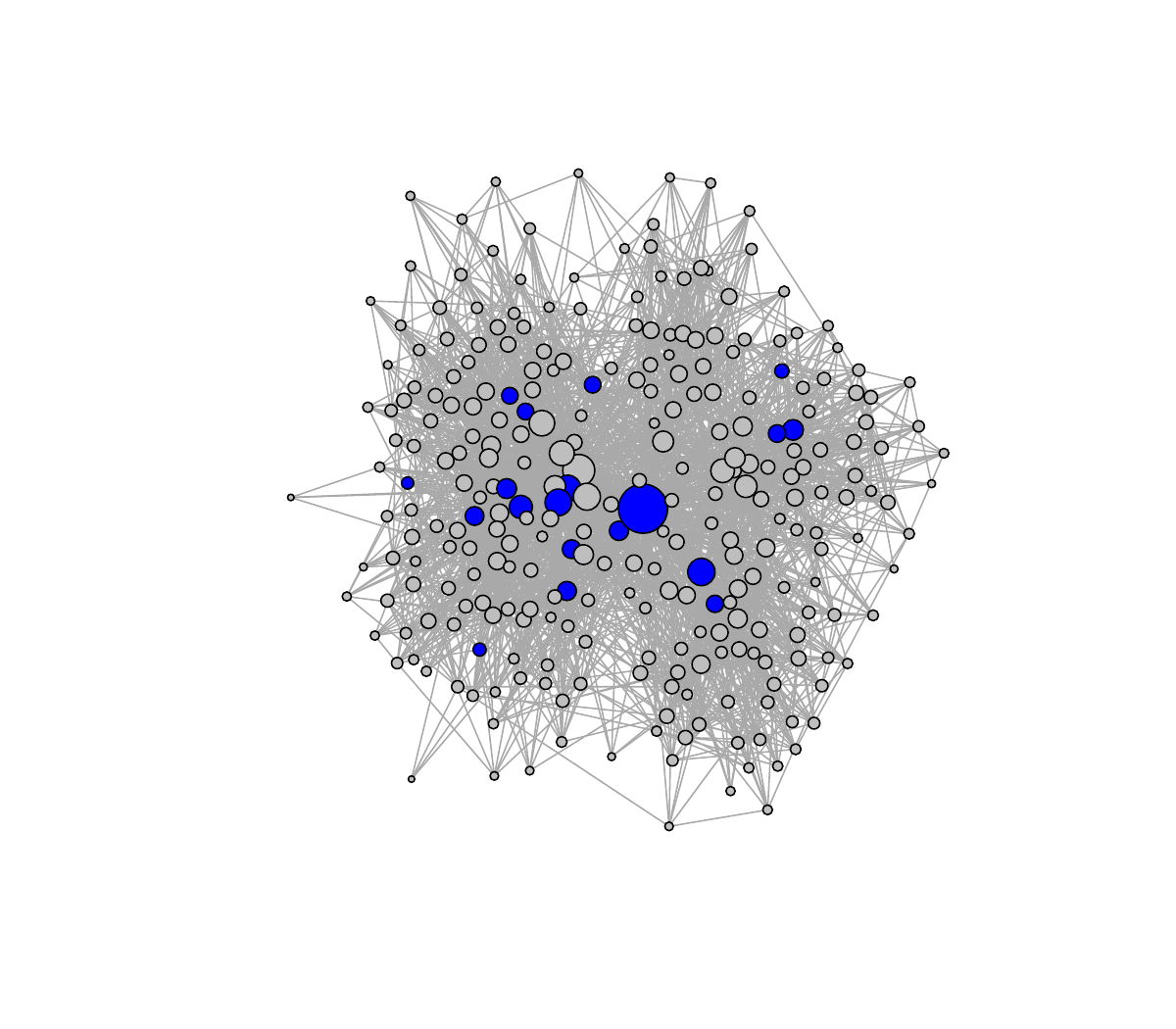}} &
   \raisebox{-0.5\height}{\includegraphics[scale=0.2]{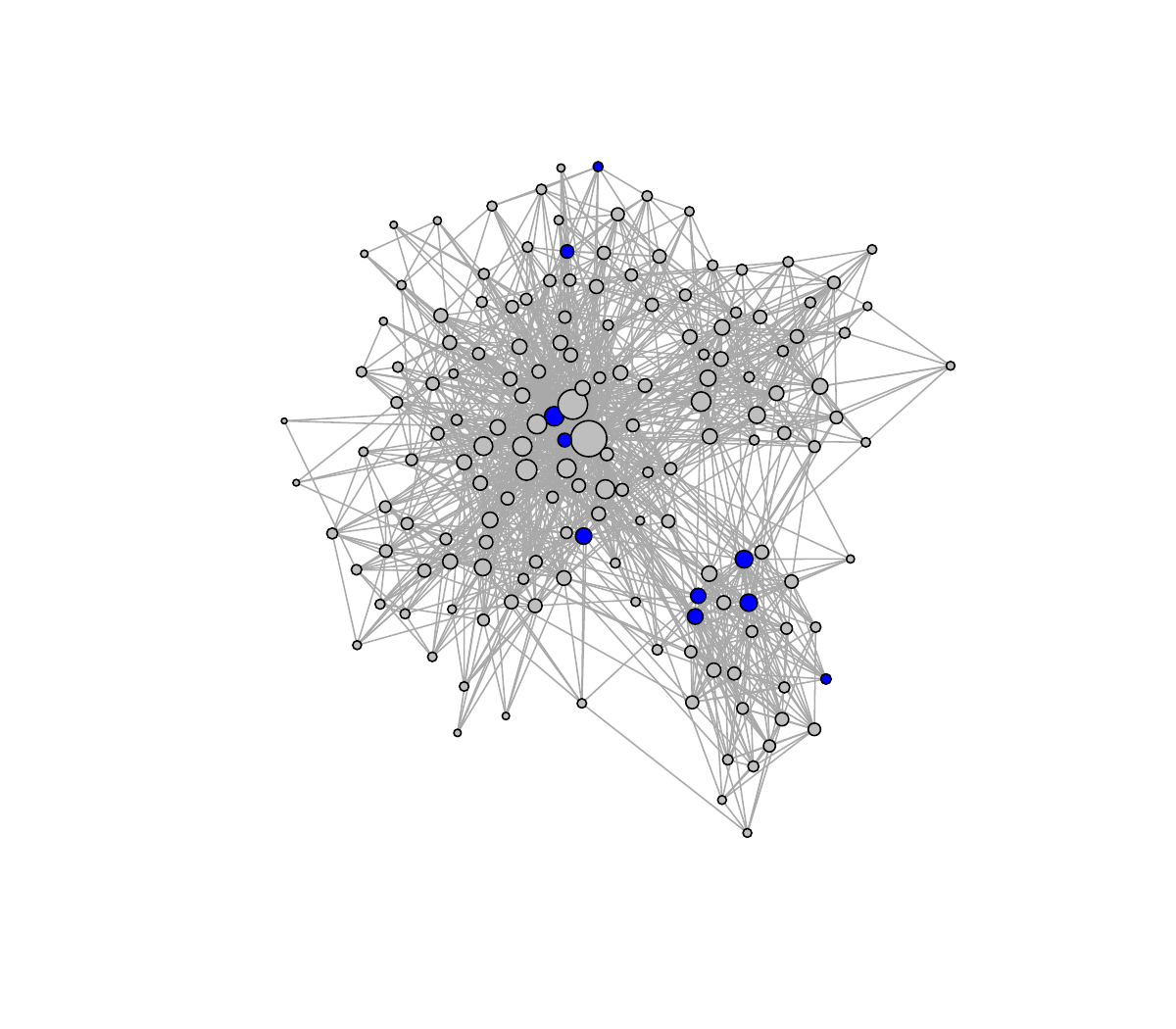}} \\
   \raisebox{-0.5\height}{\includegraphics[scale=0.2]{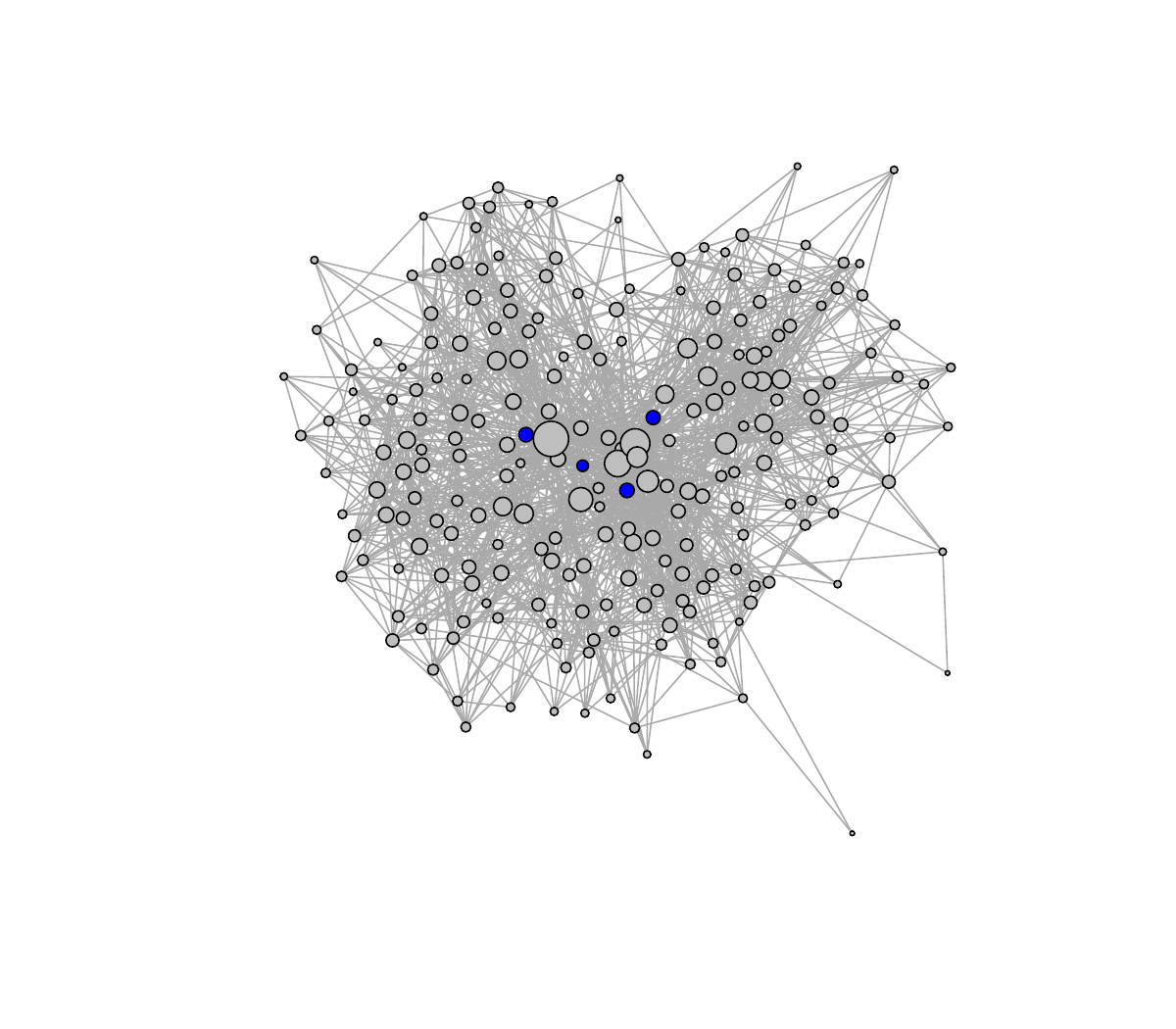}} &
   \raisebox{-0.5\height}{\includegraphics[scale=0.2]{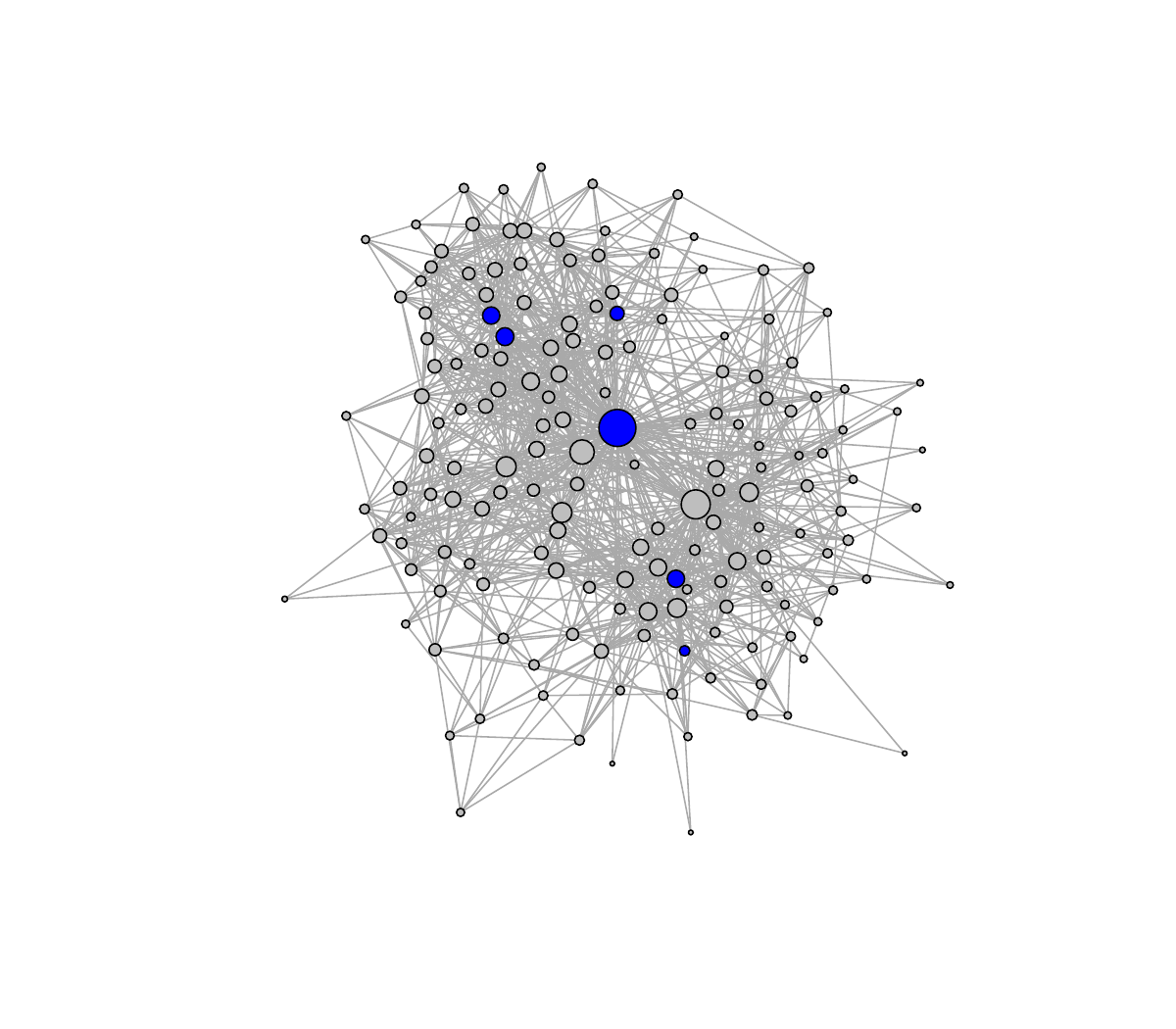}} &
   \raisebox{-0.5\height}{\includegraphics[scale=0.2]{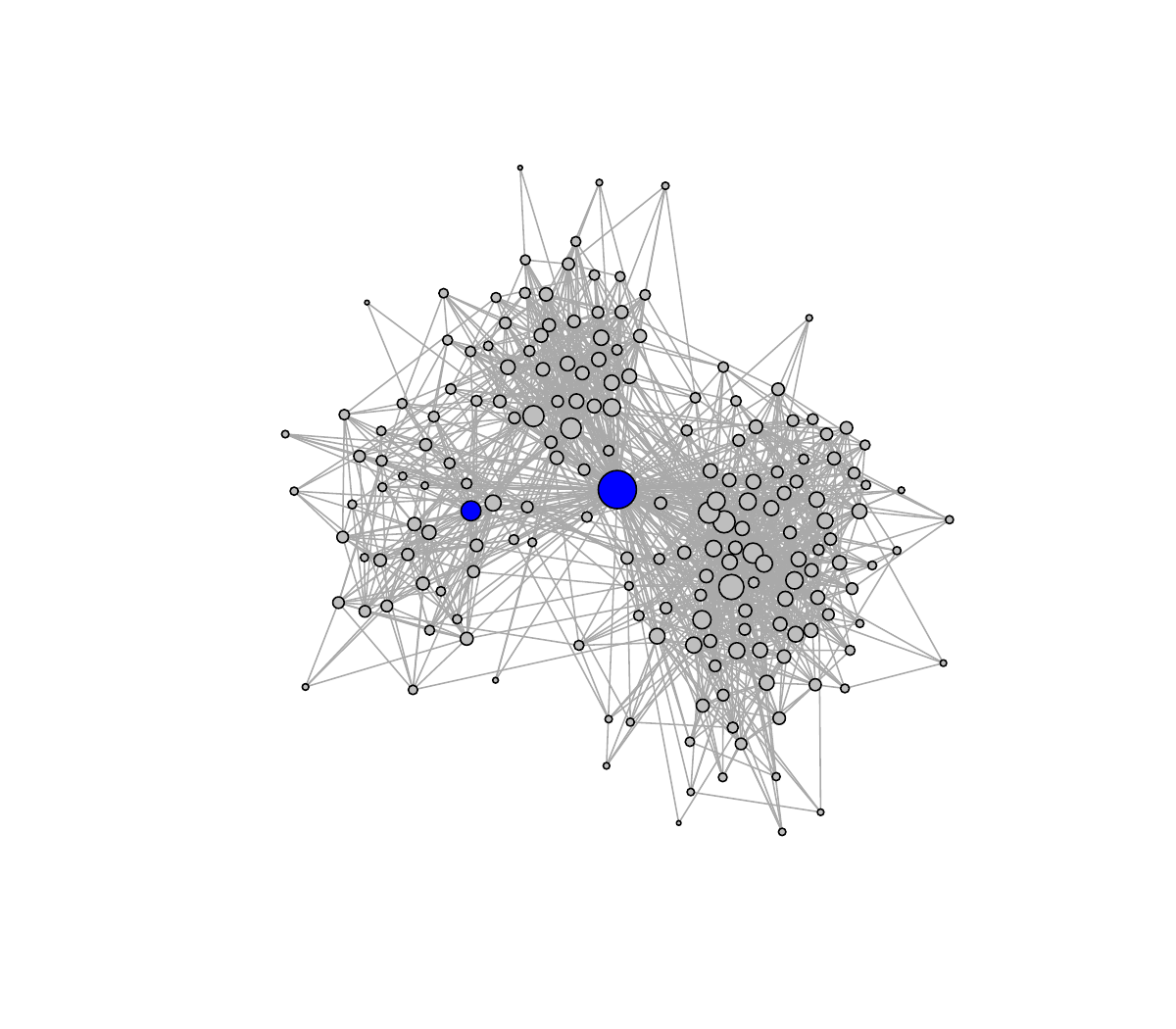}} &
   \raisebox{-0.5\height}{\includegraphics[scale=0.2]{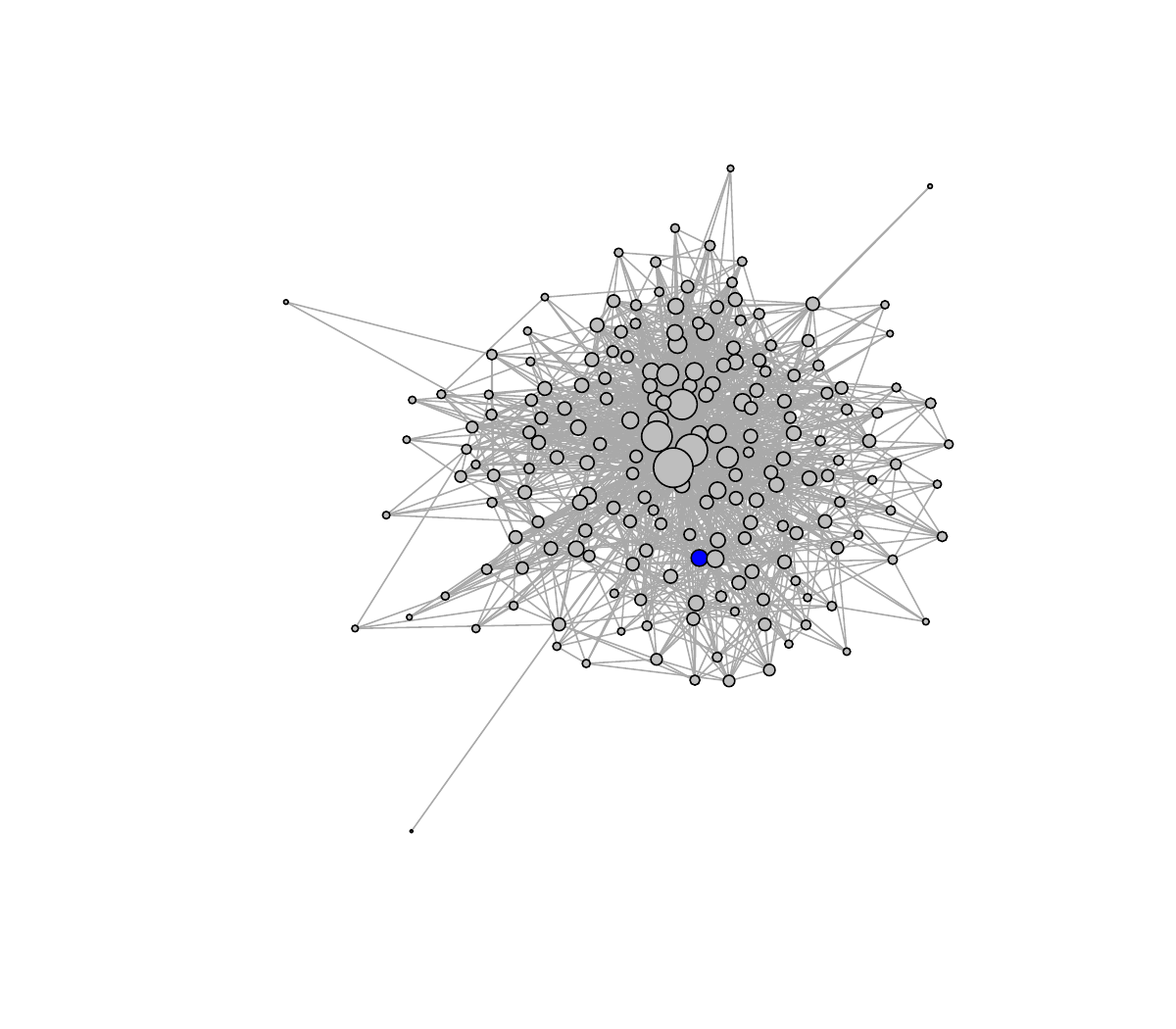}} \\
   \raisebox{-0.5\height}{\includegraphics[scale=0.2]{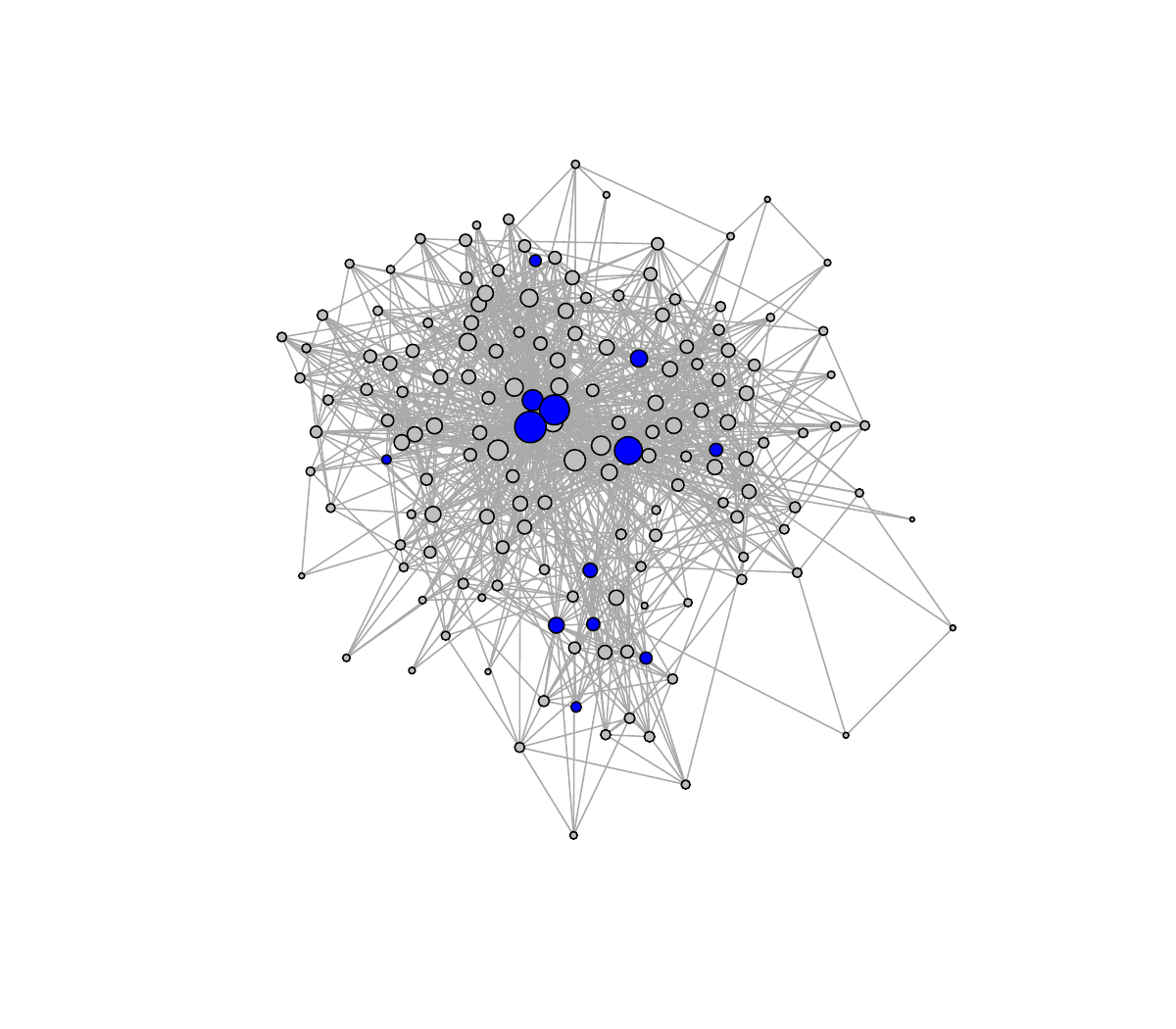}} &
   \raisebox{-0.5\height}{\includegraphics[scale=0.2]{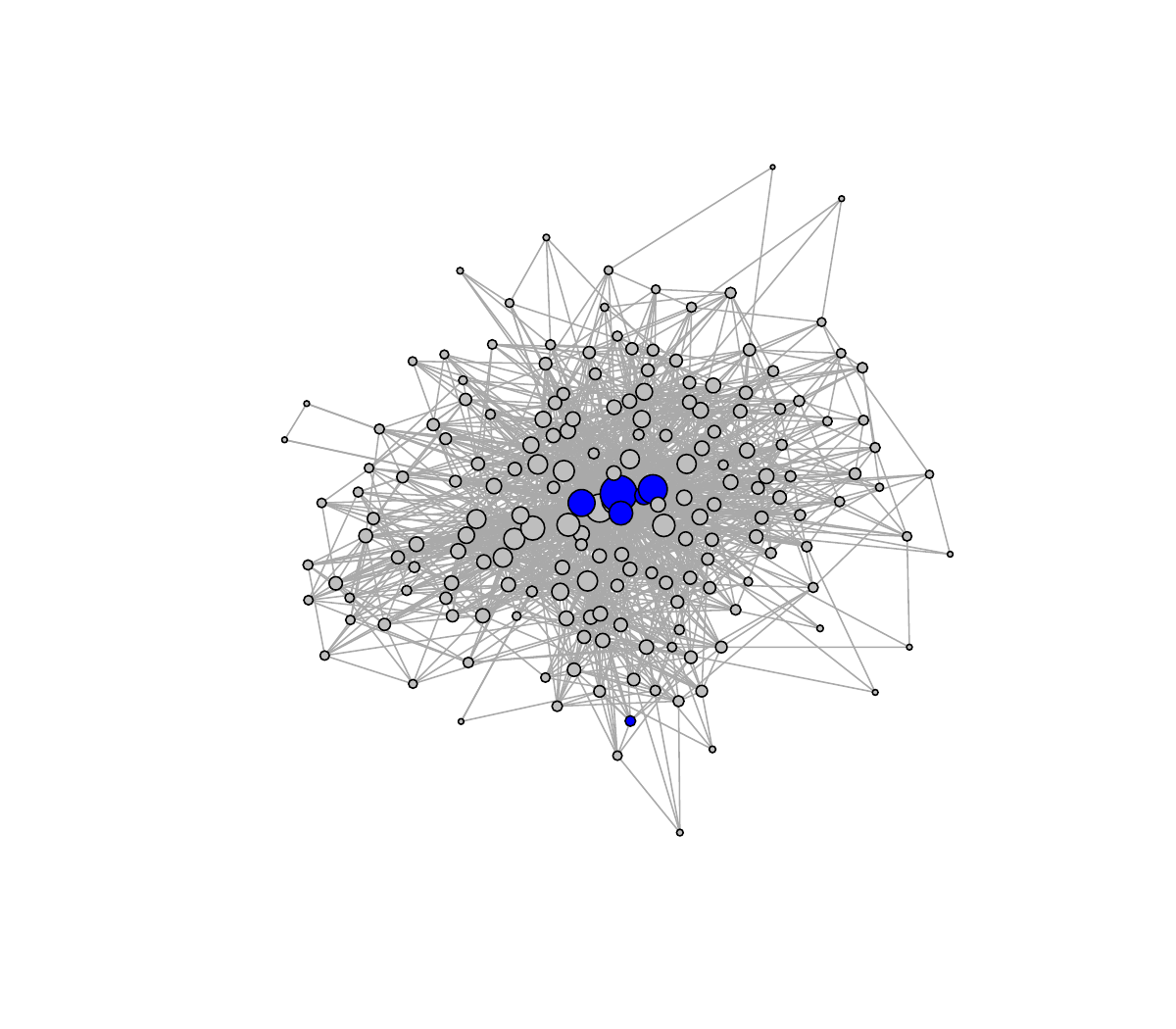}} &
   \raisebox{-0.5\height}{\includegraphics[scale=0.2]{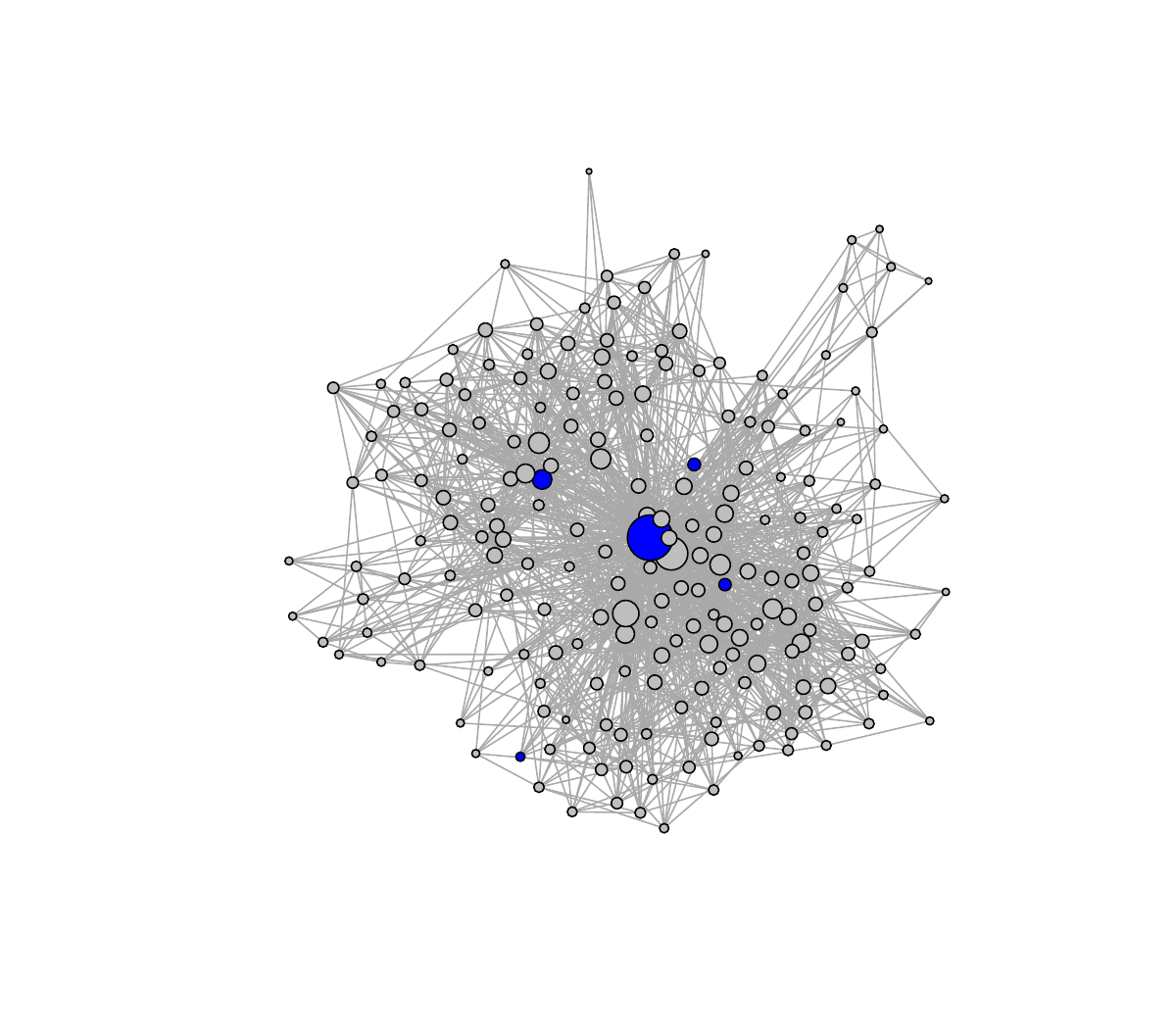}} &
   \raisebox{-0.5\height}{\includegraphics[scale=0.2]{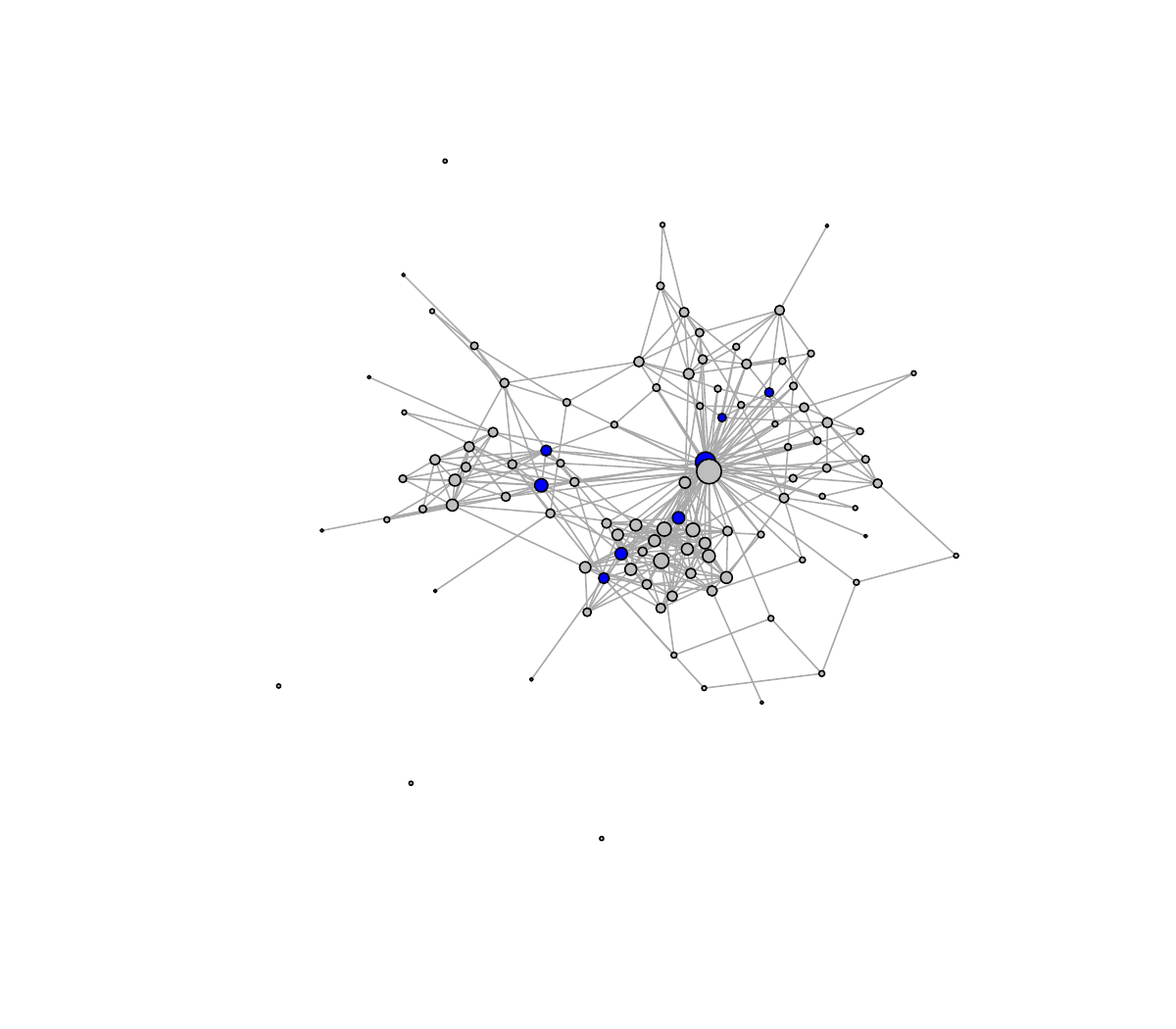}} \\
   \raisebox{-0.5\height}{\includegraphics[scale=0.2]{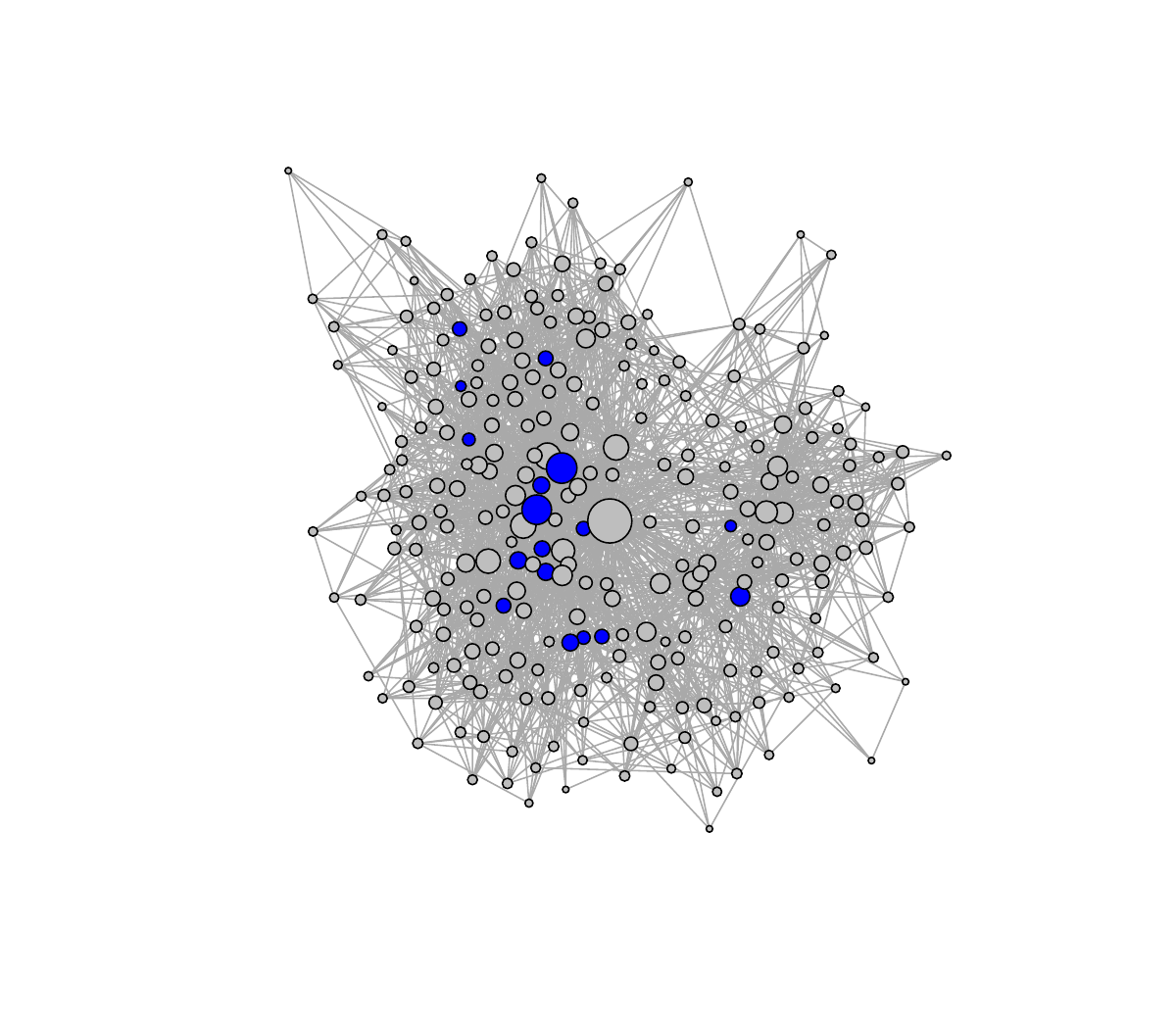}} &
   \raisebox{-0.5\height}{\includegraphics[scale=0.2]{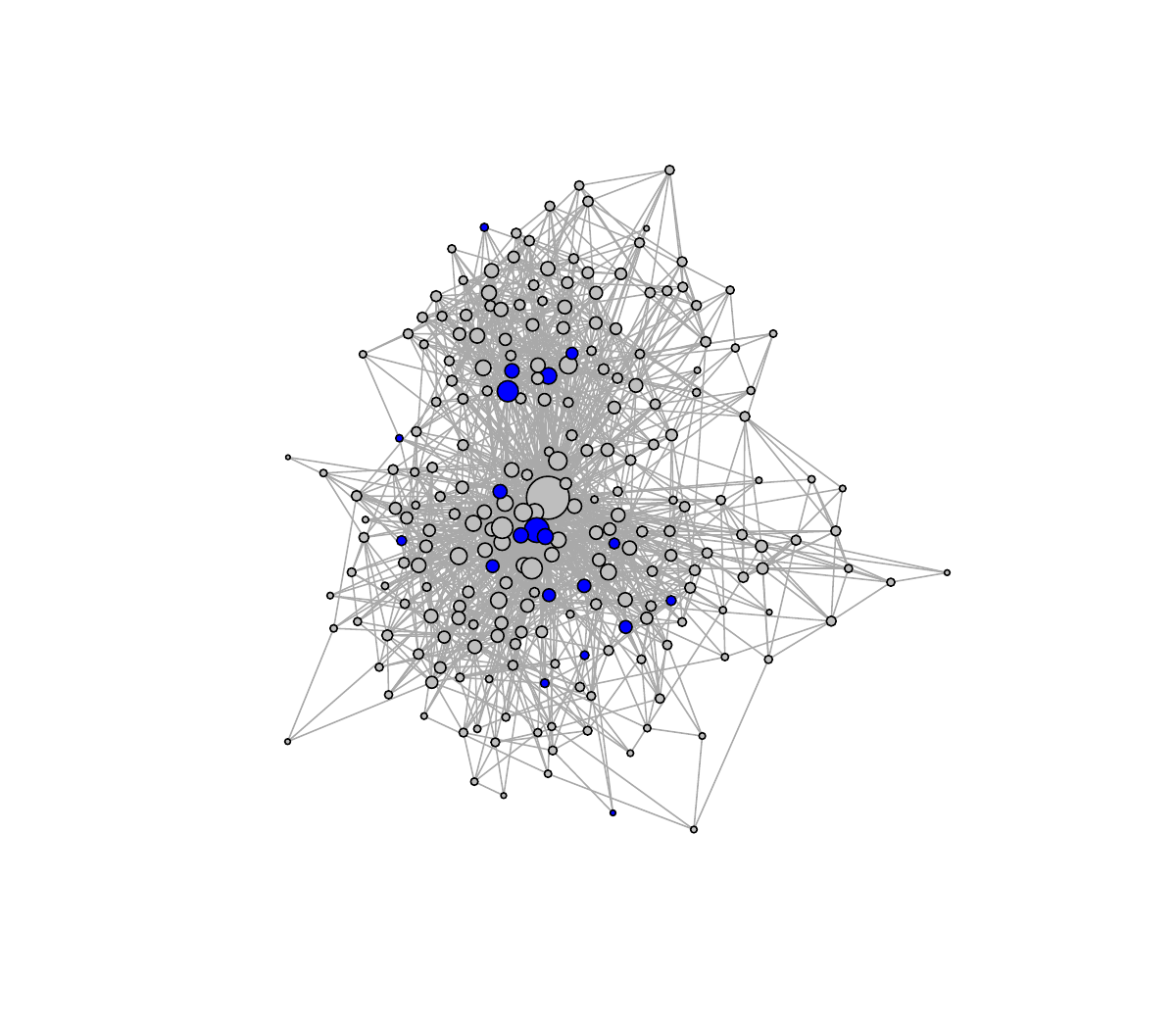}} &
   \raisebox{-0.5\height}{\includegraphics[scale=0.2]{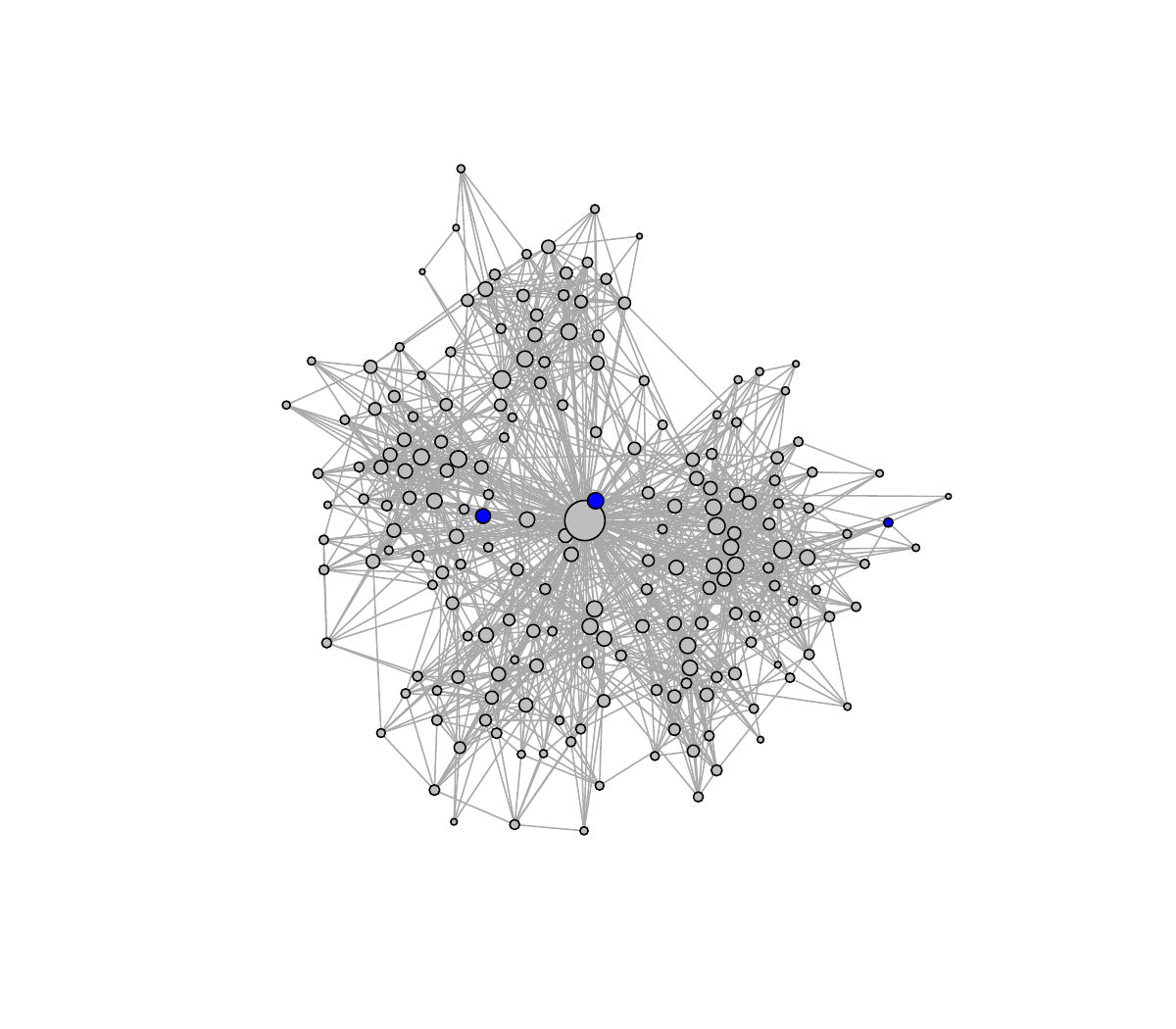}} &
   \raisebox{-0.5\height}{\includegraphics[scale=0.2]{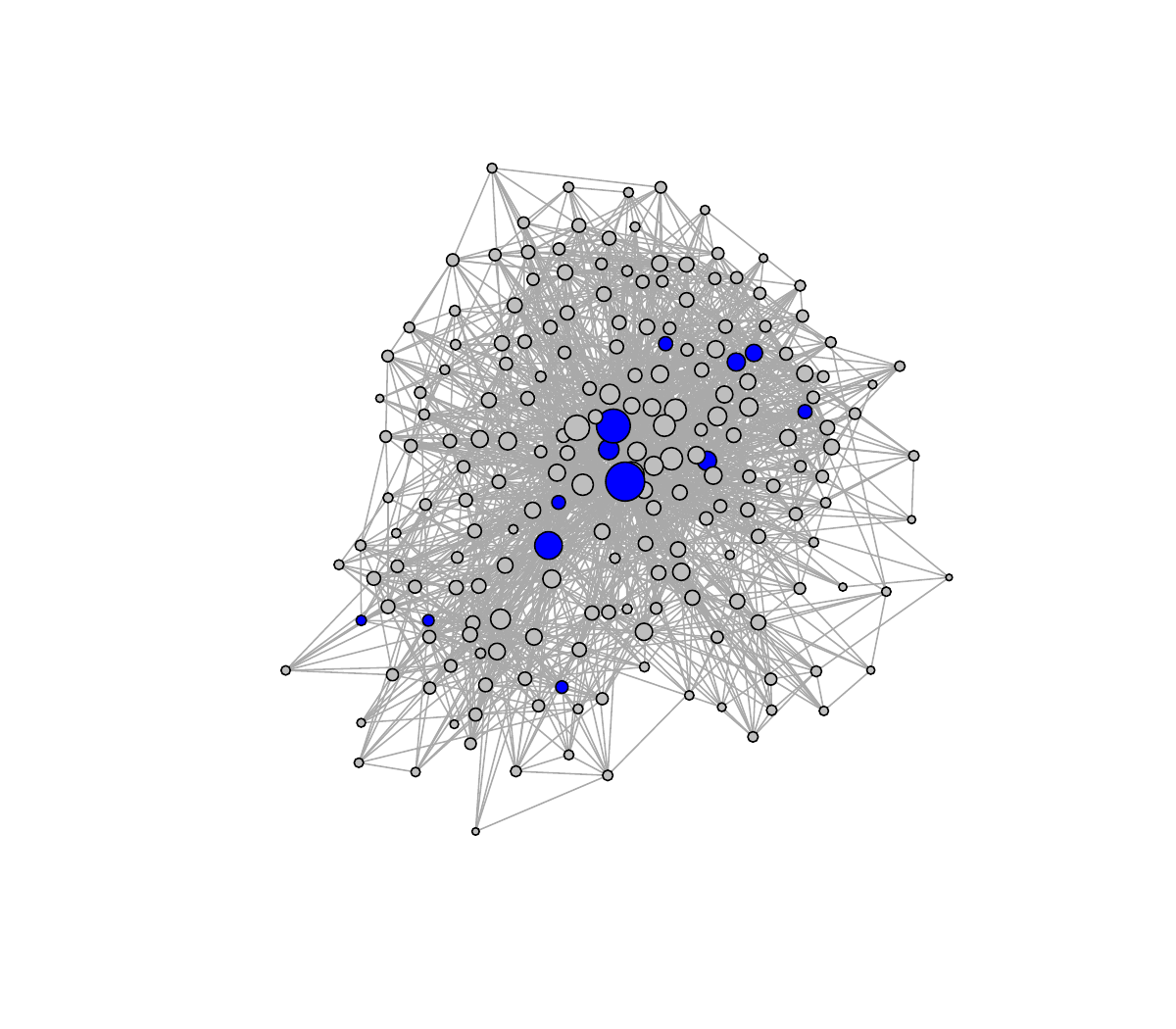}}
   \end{tabular}
   \end{center}
   \caption{Additional Network Visualizations for Villages (Treatment) }\label{ns2}
   \end{figure}
\end{landscape}

\begin{landscape}
\begin{figure}[!ht]
  \begin{center}
  \begin{tabular}{cccc}
   \raisebox{-0.5\height}{\includegraphics[scale=0.2]{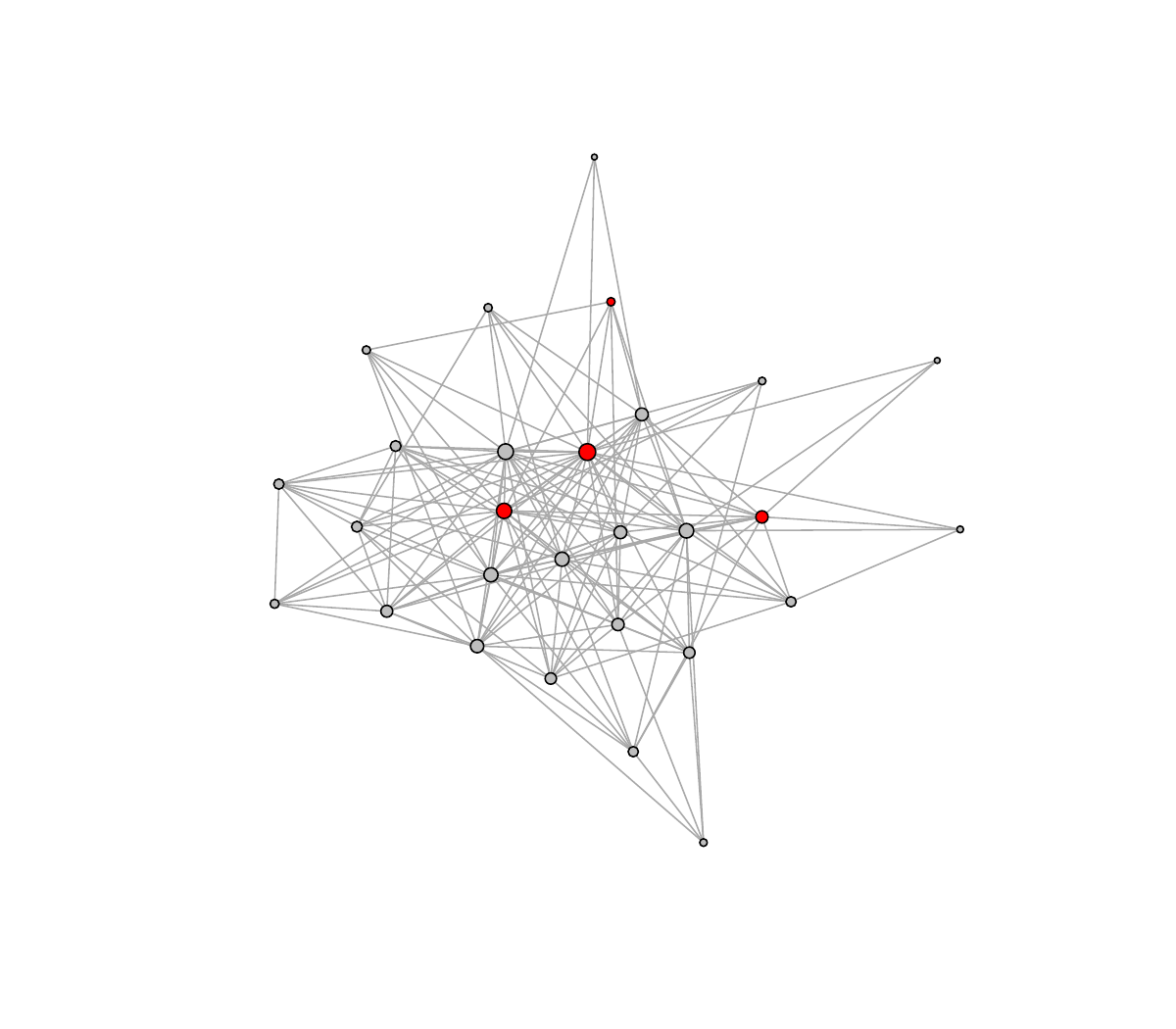}} &
   \raisebox{-0.5\height}{\includegraphics[scale=0.2]{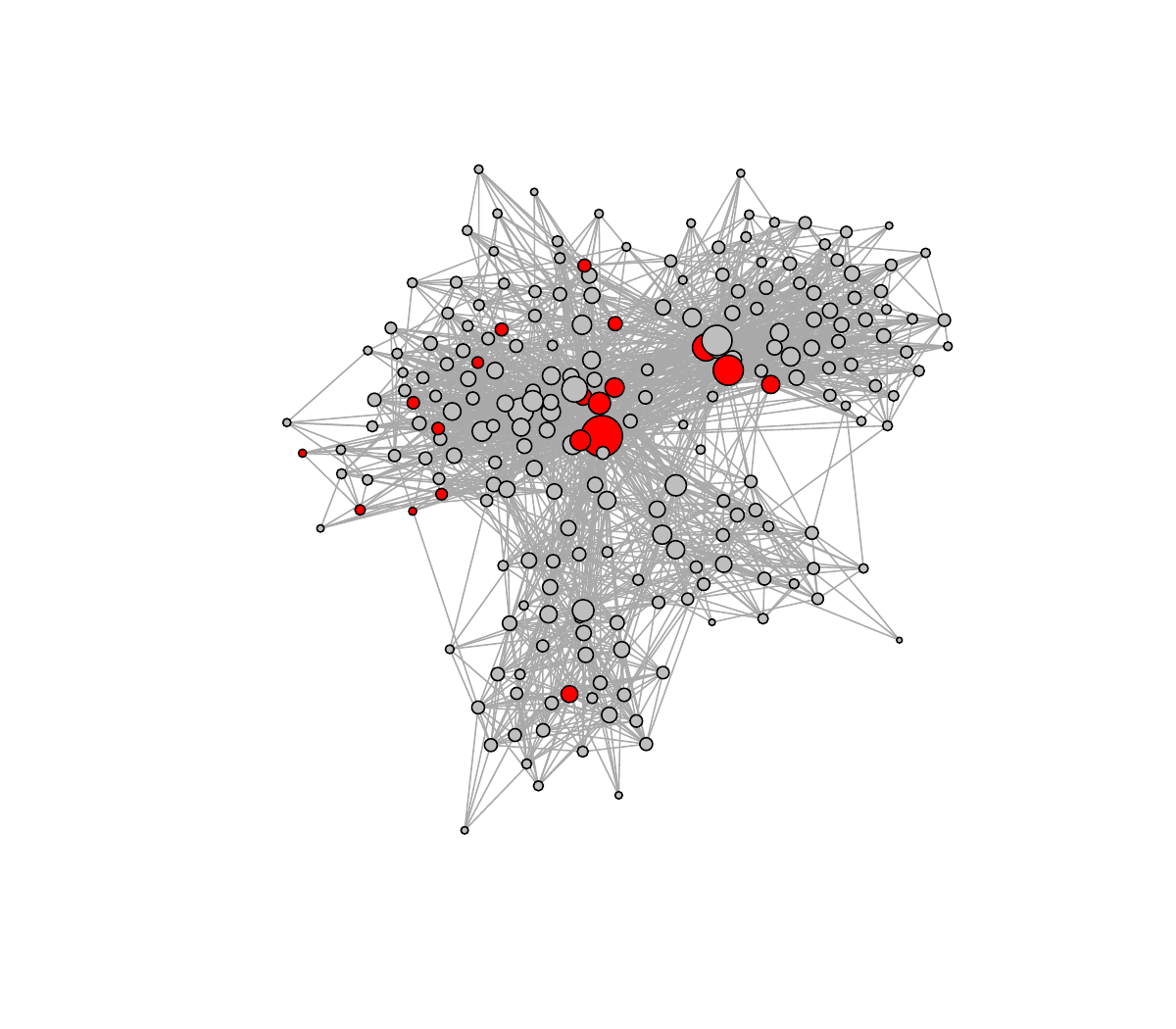}} &
   \raisebox{-0.5\height}{\includegraphics[scale=0.2]{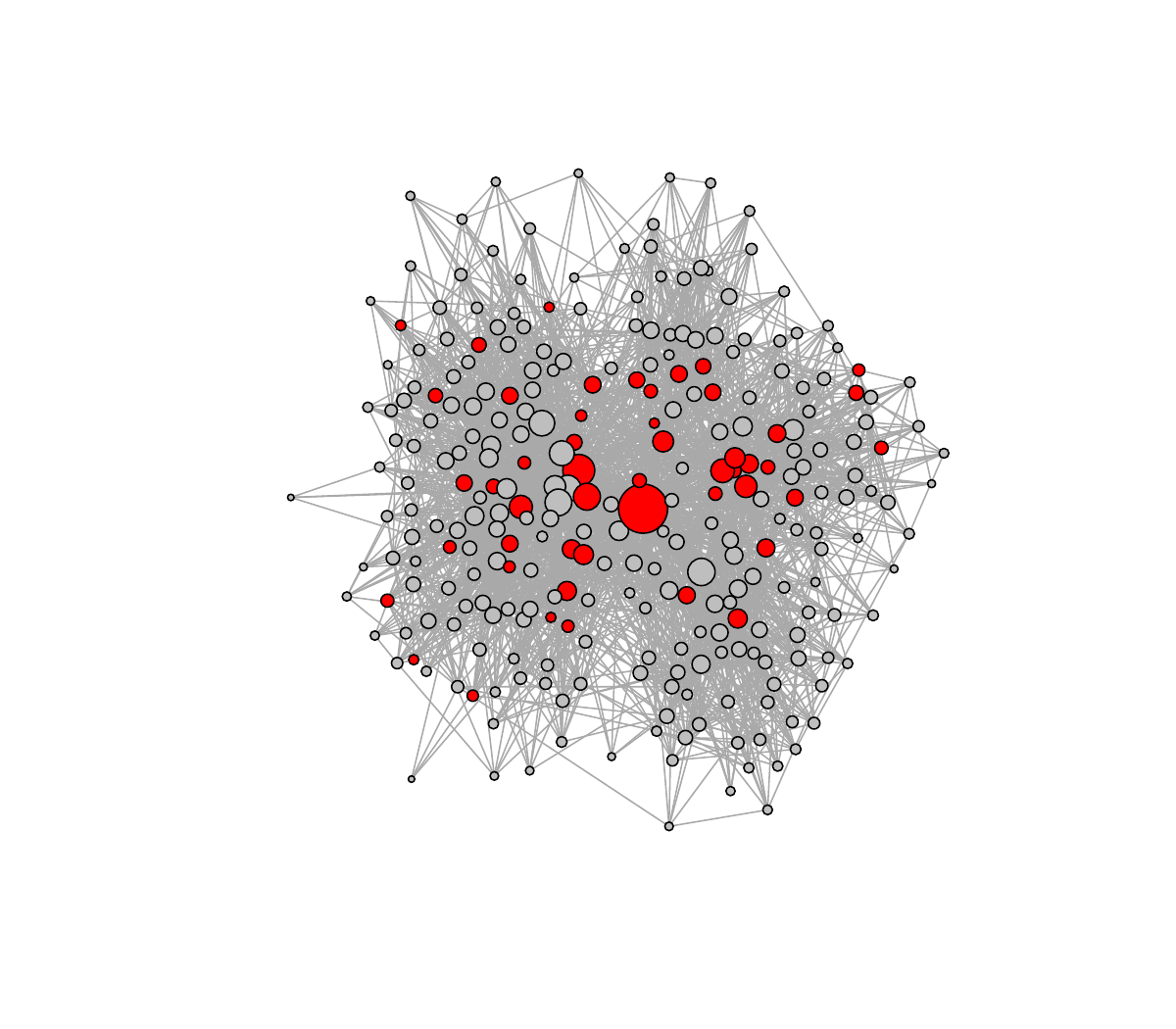}} &
   \raisebox{-0.5\height}{\includegraphics[scale=0.2]{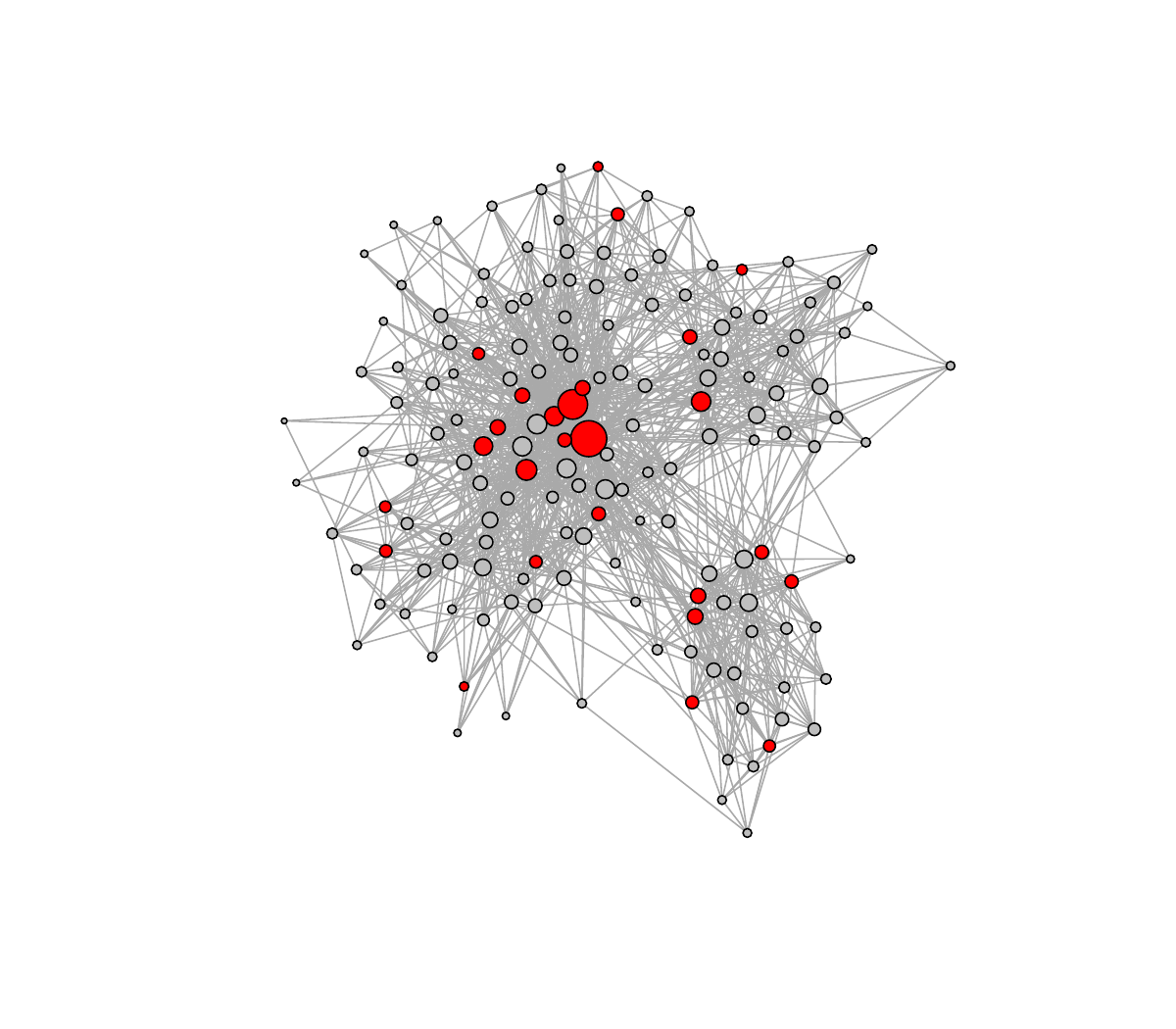}} \\
   \raisebox{-0.5\height}{\includegraphics[scale=0.2]{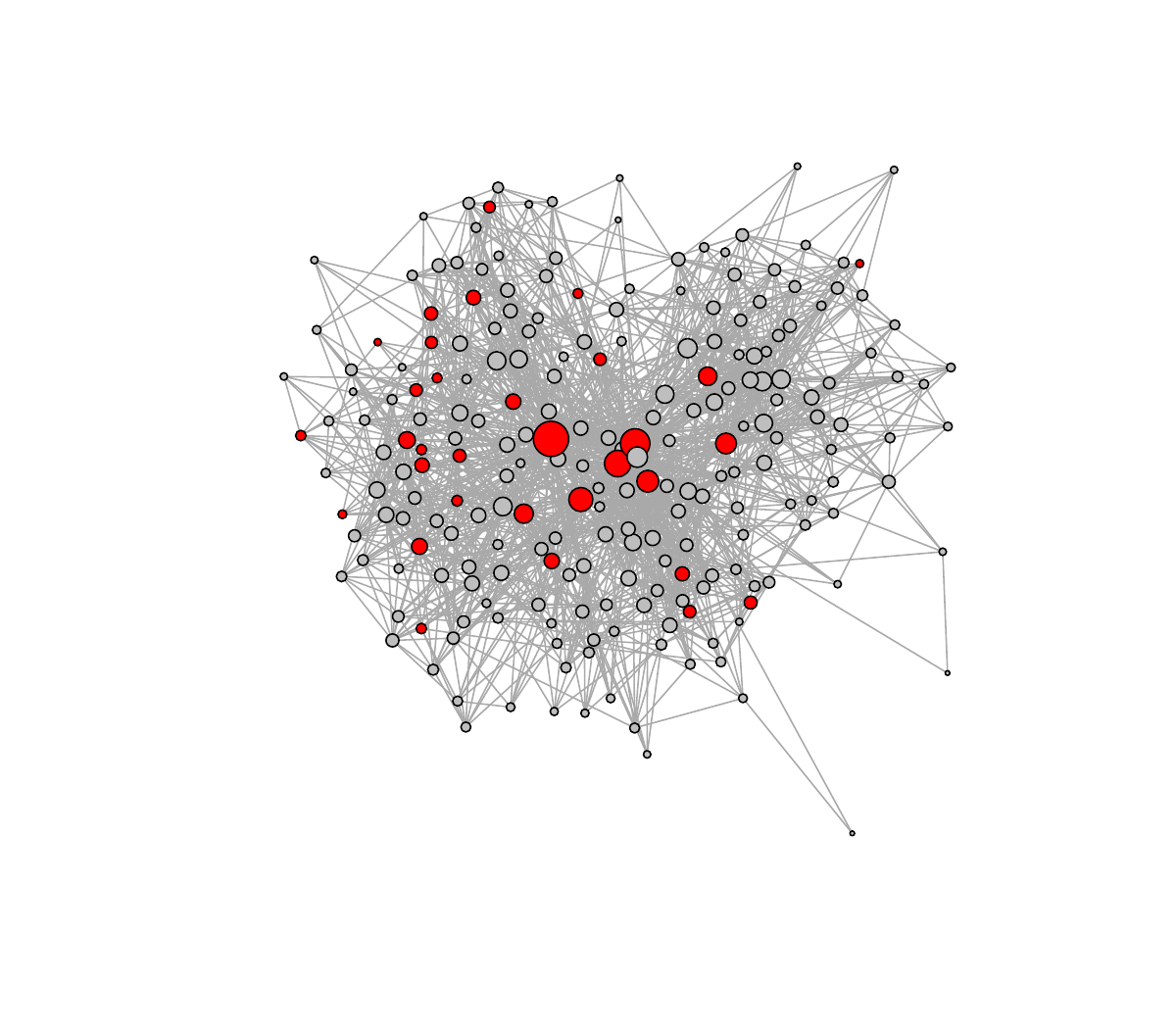}} &
   \raisebox{-0.5\height}{\includegraphics[scale=0.2]{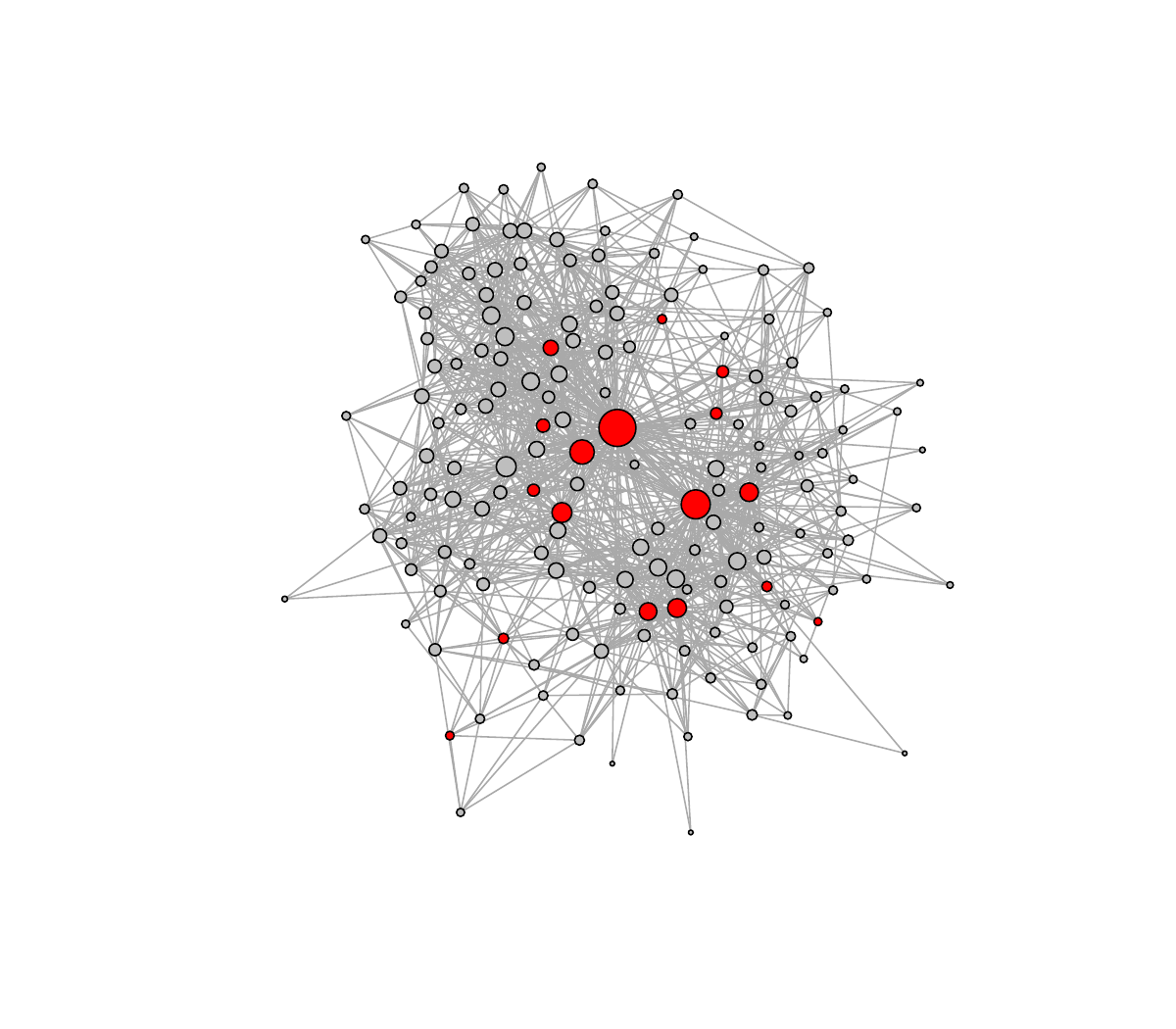}} &
   \raisebox{-0.5\height}{\includegraphics[scale=0.2]{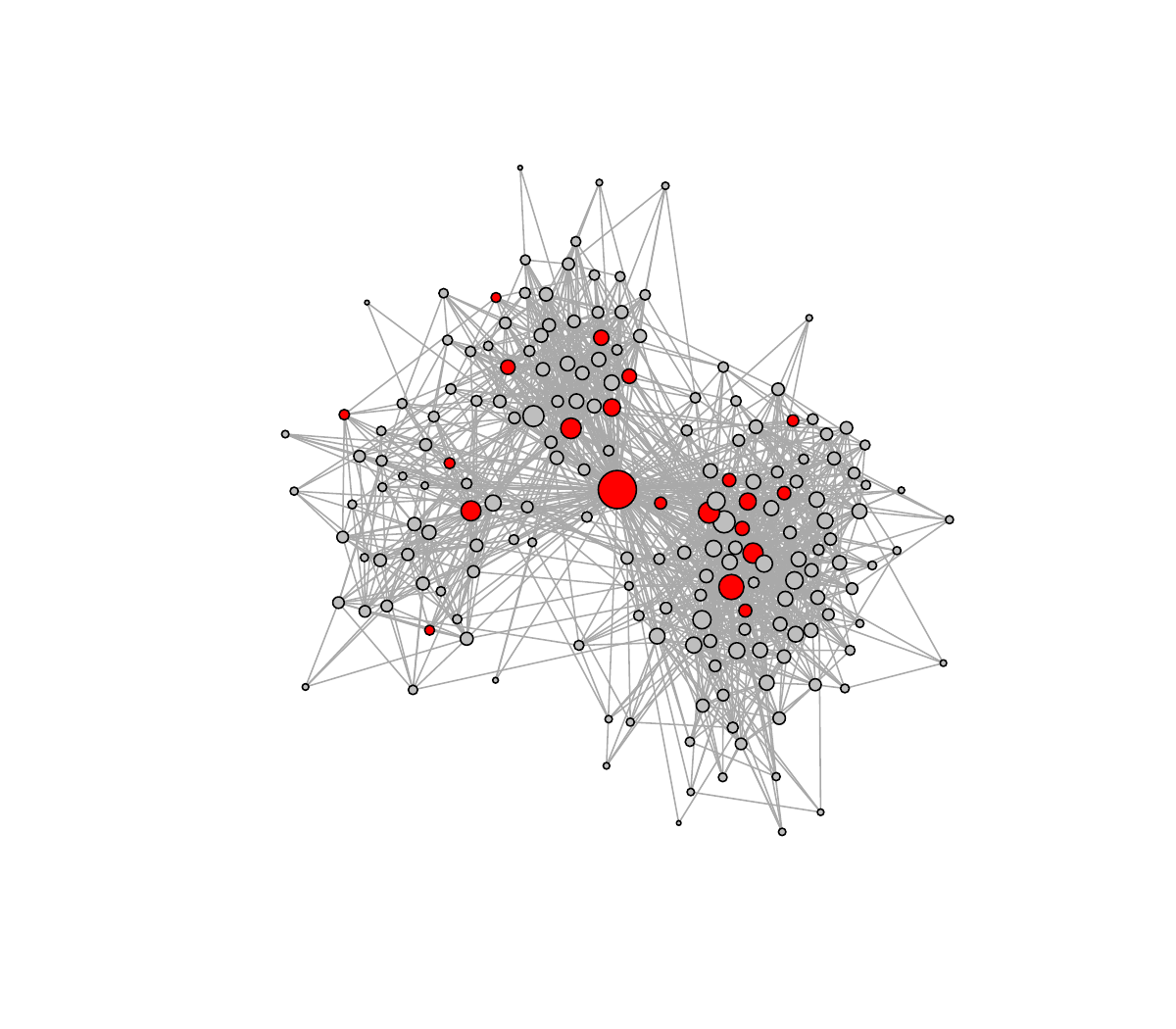}} &
   \raisebox{-0.5\height}{\includegraphics[scale=0.2]{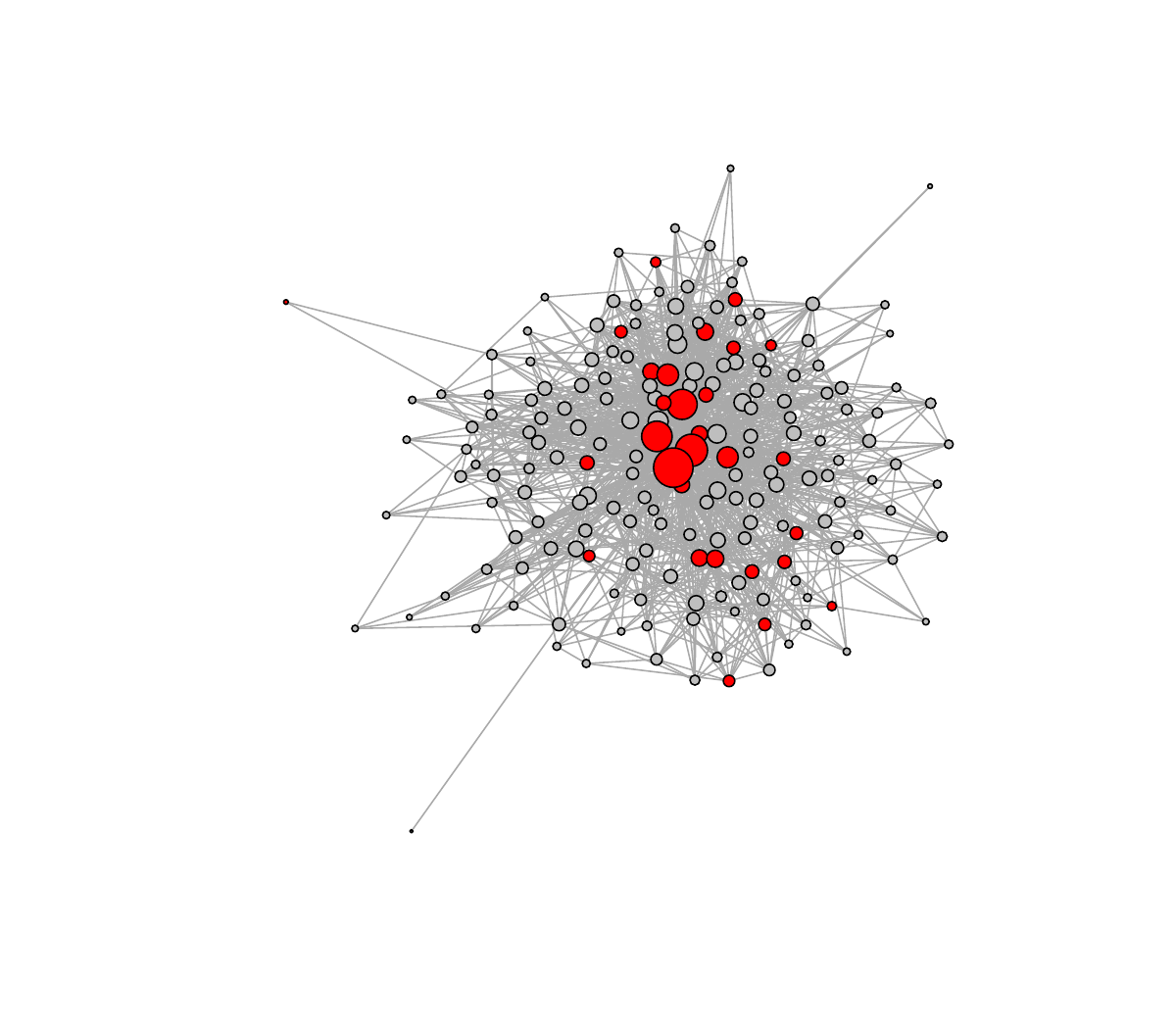}} \\
   \raisebox{-0.5\height}{\includegraphics[scale=0.2]{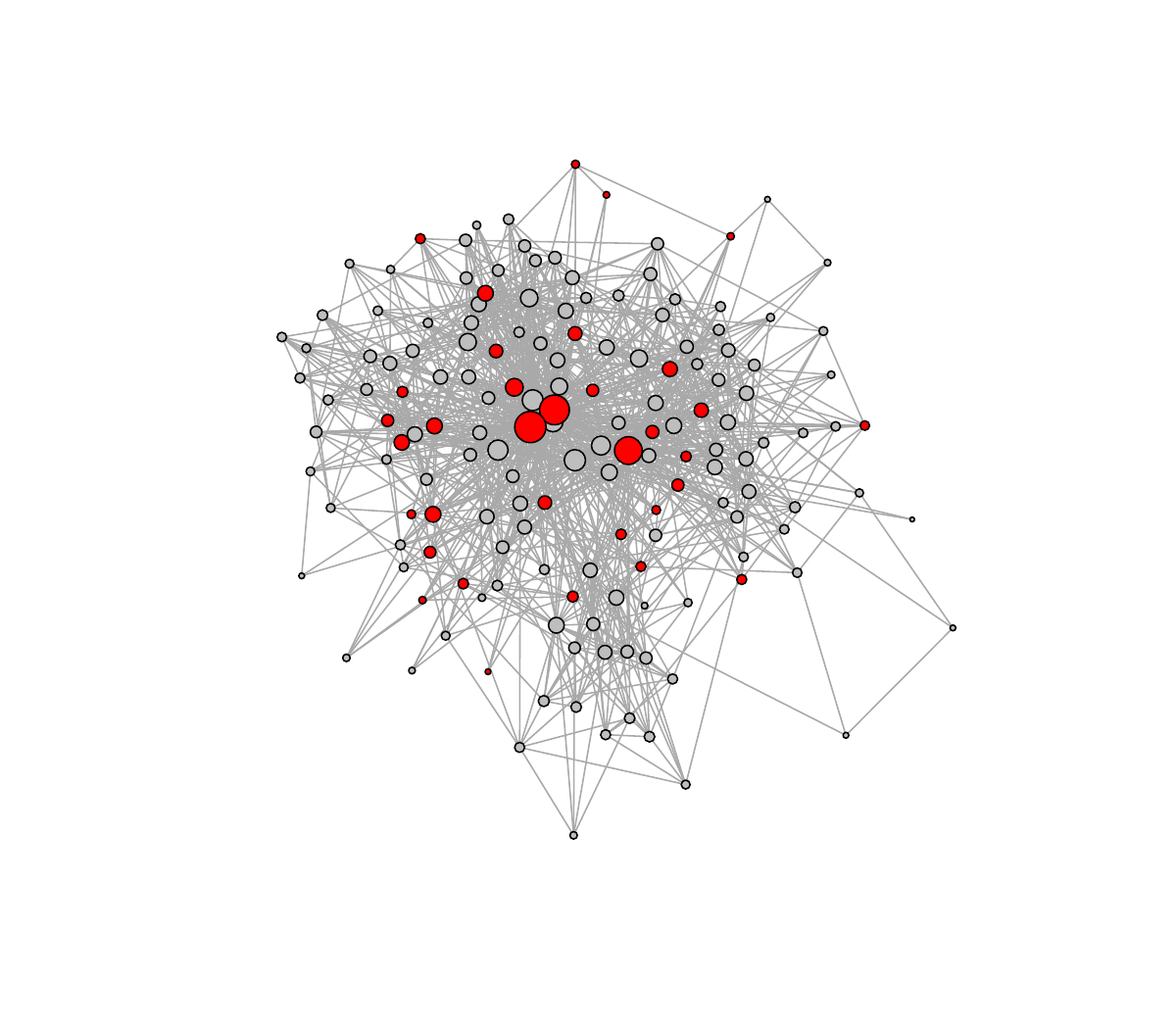}} &
   \raisebox{-0.5\height}{\includegraphics[scale=0.2]{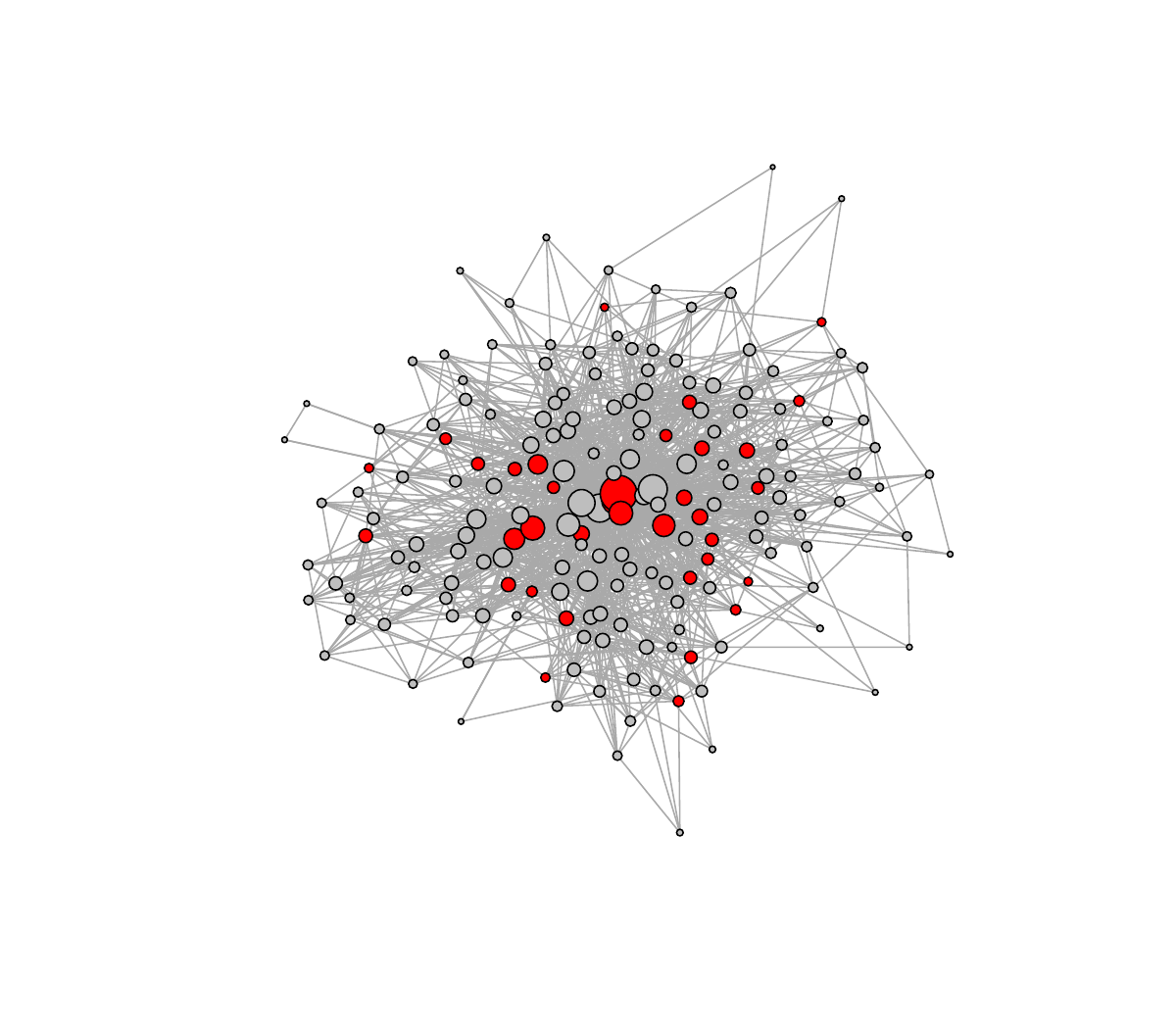}} &
   \raisebox{-0.5\height}{\includegraphics[scale=0.2]{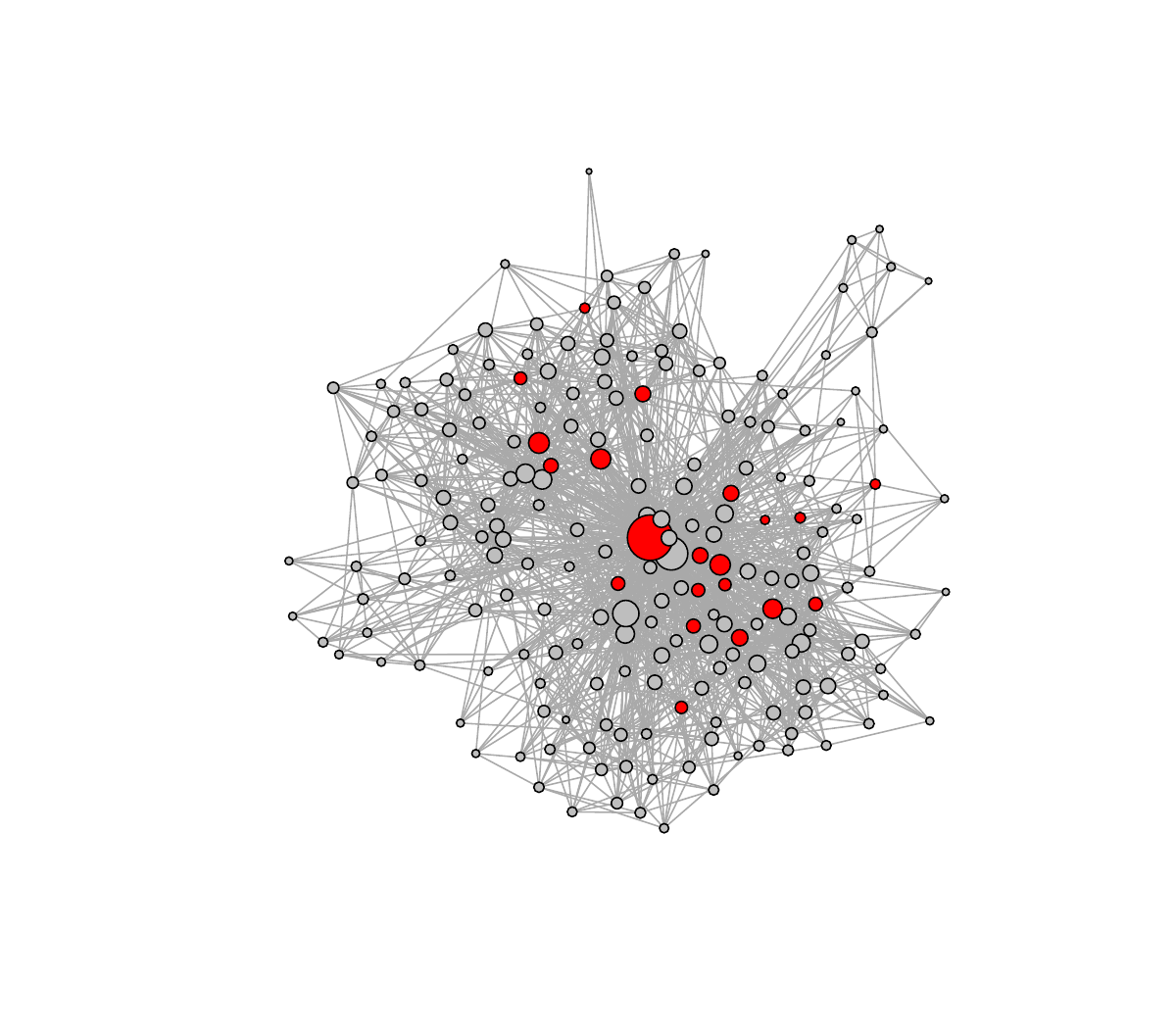}} &
   \raisebox{-0.5\height}{\includegraphics[scale=0.2]{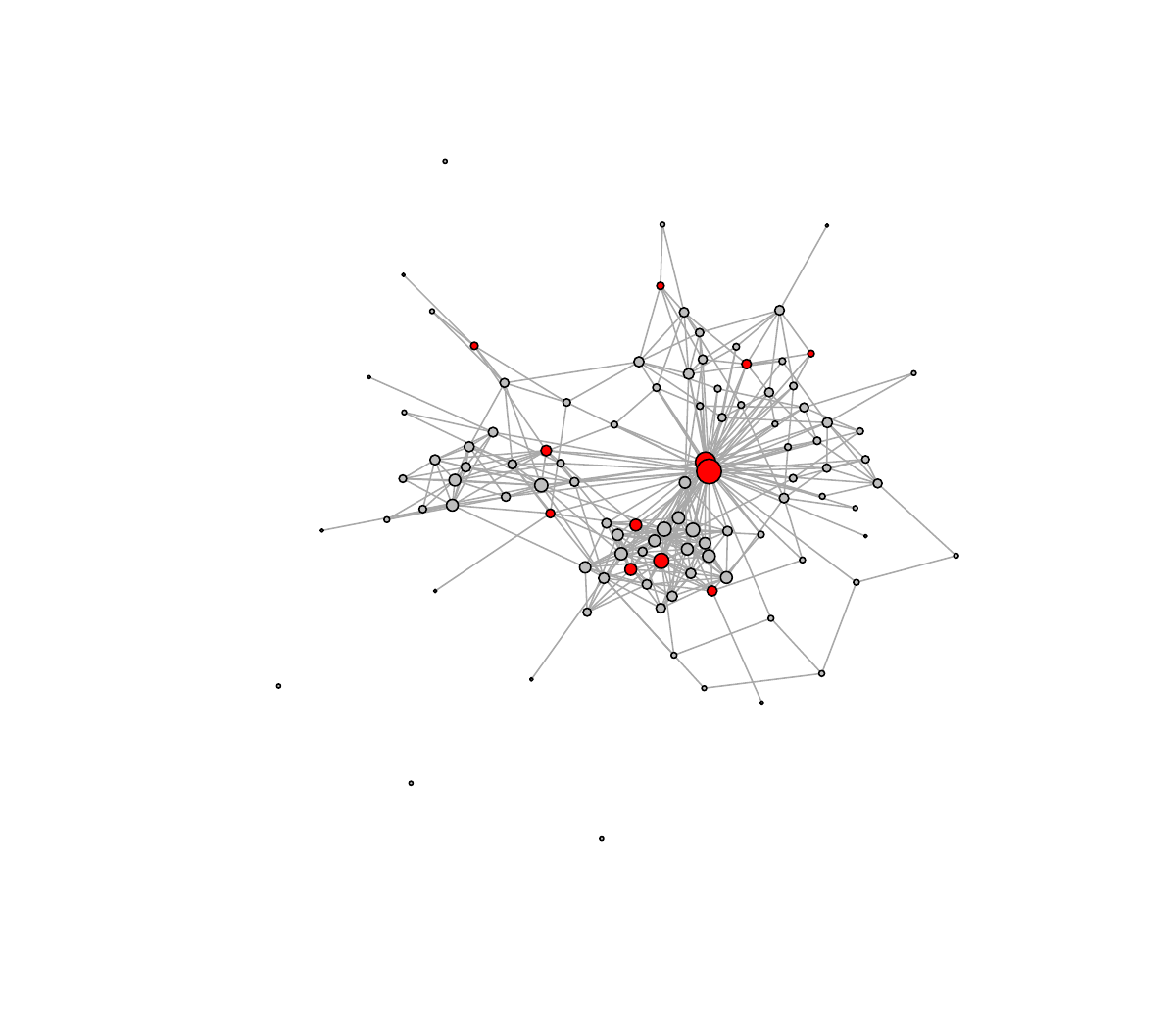}} \\
   \raisebox{-0.5\height}{\includegraphics[scale=0.2]{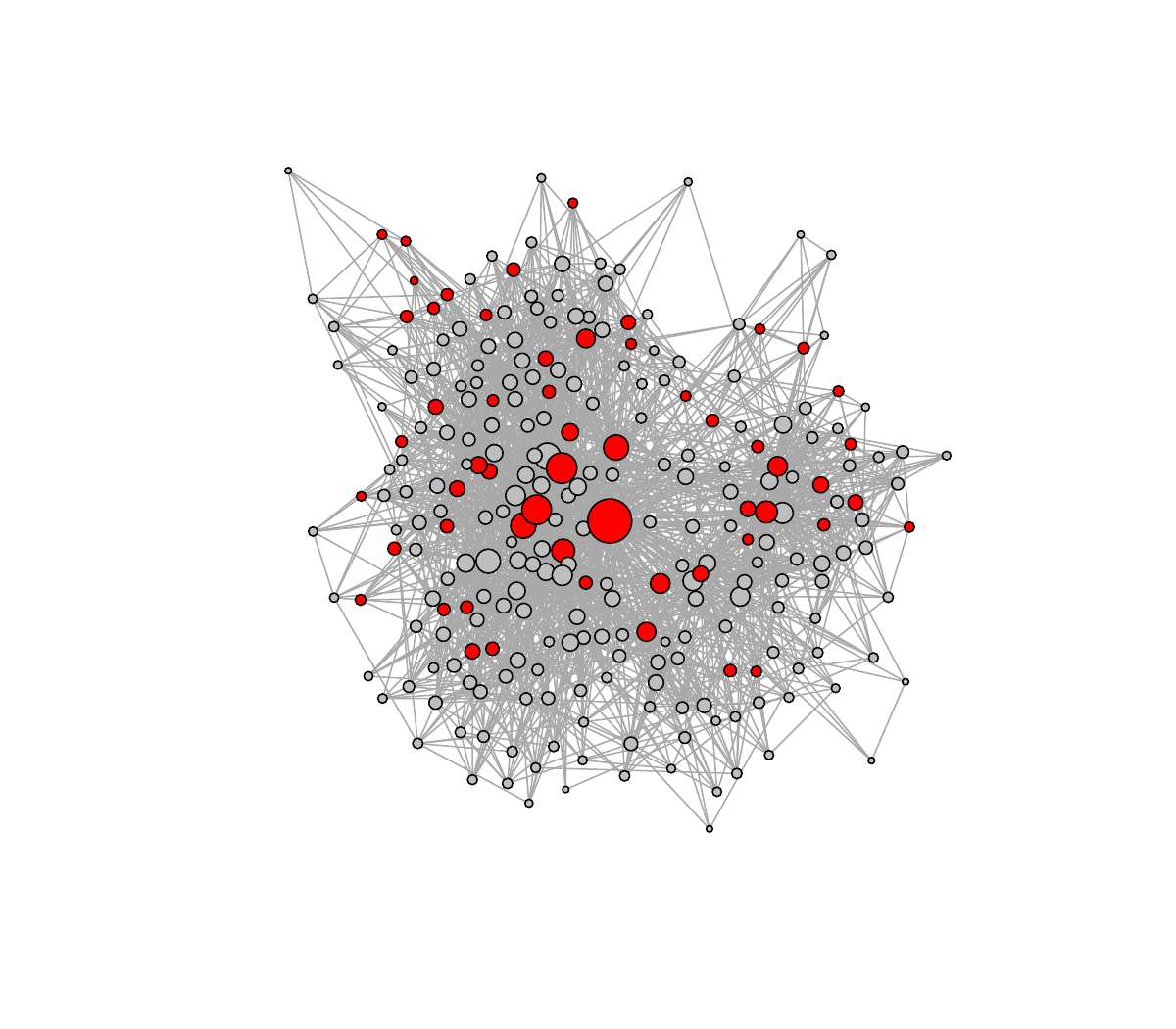}} &
   \raisebox{-0.5\height}{\includegraphics[scale=0.2]{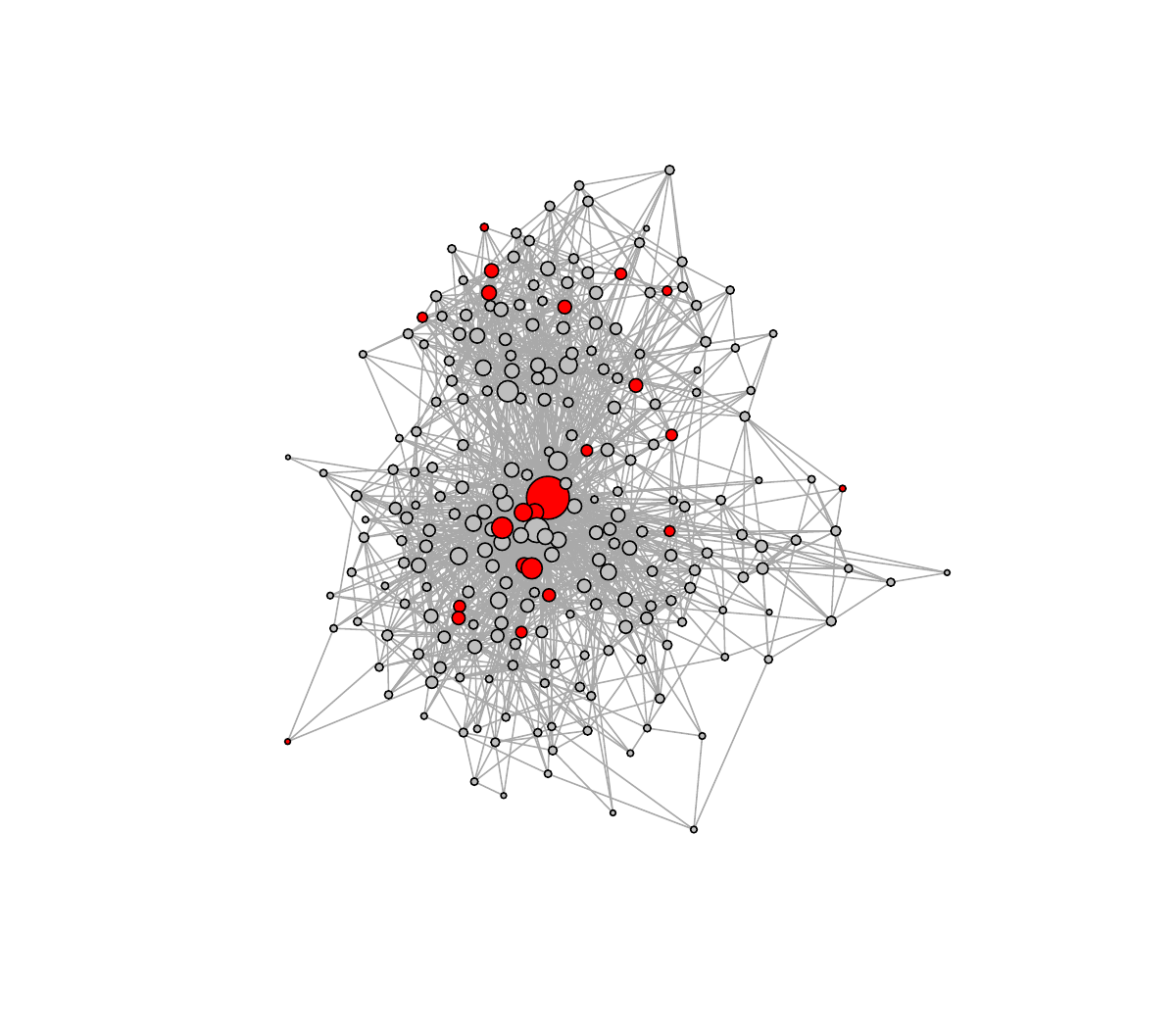}} &
   \raisebox{-0.5\height}{\includegraphics[scale=0.2]{empirical/plot_net/network_leader_9.pdf}} & 
   \raisebox{-0.5\height}{\includegraphics[scale=0.2]{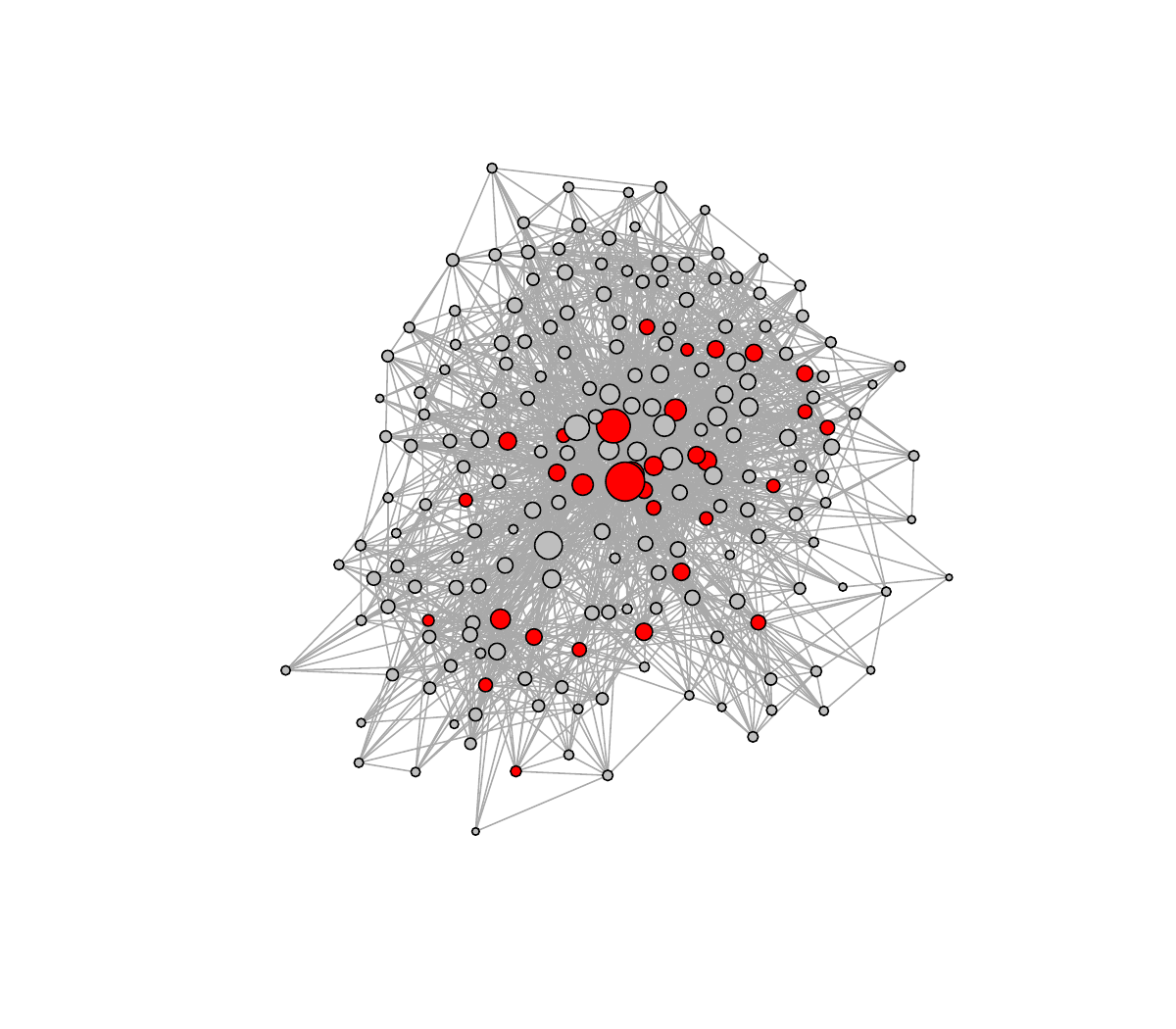}}
   \end{tabular}
   \end{center}
   \caption{Additional Network Visualizations for Villages (Leadership) }\label{ns3}
   \end{figure}
\end{landscape}





\begin{landscape}

\vfill

\begin{table}[h]
  \centering
  \caption{Estimation Results for Network Formation in Each Village}\label{nreg2} 
  \resizebox{\textwidth}{!}{%
  \begin{tabular}{l *{16}{S[table-format=-1.3, table-space-text-post={***}, table-align-text-post=false]}}
  \hline\hline
  {Village Index} & {Index 7} & {Index 9} & {Index 11} & {Index 14}
                  & {Index 29} & {Index 40} & {Index 46} & {Index 47}
                  & {Index 50} & {Index 51} & {Index 55} & {Index 66}
                  & {Index 75} & {Index 77} & {Index 80} & {Index 81} \\ 
  \hline
{Dif in Age} & -0.019*** & -0.014*** & -0.019* & -0.019*** & -0.022*** & -0.015*** & -0.016*** & -0.015*** & -0.014*** & -0.019*** & -0.018*** & -0.031*** & -0.012*** & -0.017*** & -0.017*** & -0.018***
 \\
 & {(0.002)} & {(0.002)} & {(0.012)} & {(0.002)} & {(0.001)} & {(0.002)} & {(0.001)} & {(0.002)} & {(0.002)} & {(0.002)} & {(0.002)} & {(0.002)} & {(0.004)} & {(0.001)} & {(0.001)} & {(0.001)}
 \\
{Dif in Gender} & -0.584*** & -0.738*** & -0.145 & -0.569*** & -0.524*** & -0.483*** & -0.525*** & -0.476*** & -0.684*** & -0.509*** & -0.440*** & -0.447*** & -0.357*** & -0.586*** & -0.480*** & -0.420***
 \\
 & {(0.039)} & {(0.052)} & {(0.177)} & {(0.035)} & {(0.025)} & {(0.041)} & {(0.033)} & {(0.044)} & {(0.042)} & {(0.039)} & {(0.050)} & {(0.038)} & {(0.086)} & {(0.034)} & {(0.031)} & {(0.034)}
 \\
{Dif in Education} & -0.301*** & -0.248*** & -0.206 & -0.172*** & -0.098*** & 0.037 & -0.169*** & -0.087* & -0.275*** & -0.159* & -0.195*** & -0.369*** & 0.029 & -0.192*** & -0.203*** & -0.104***
 \\
 & {(0.053)} & {(0.061)} & {(0.265)} & {(0.053)} & {(0.037)} & {(0.090)} & {(0.062)} & {(0.051)} & {(0.066)} & {(0.087)} & {(0.047)} & {(0.051)} & {(0.083)} & {(0.040)} & {(0.032)} & {(0.036)}
 \\
{Dif in Income} & -0.042** & -0.047** & 0.058 & 0.005 & -0.029** & -0.027 & -0.007 & 0.007 & 0.003 & 0.000 & -0.035* & -0.024 & -0.054 & -0.016 & -0.036*** & -0.029*
 \\
 & {(0.020)} & {(0.021)} & {(0.131)} & {(0.017)} & {(0.013)} & {(0.022)} & {(0.016)} & {(0.019)} & {(0.019)} & {(0.018)} & {(0.020)} & {(0.015)} & {(0.037)} & {(0.015)} & {(0.013)} & {(0.017)}
 \\
{Dif in CaC} & -0.005 & 0.198 & 0.096 & -0.345*** & 0.020 & -0.160 & -0.131 & -0.293** & -0.081 & -0.252** & -0.053 & -0.042 & -0.125 & -0.208* & -0.006 & -0.005
 \\
 & {(0.106)} & {(0.141)} & {(0.489)} & {(0.124)} & {(0.093)} & {(0.172)} & {(0.108)} & {(0.129)} & {(0.122)} & {(0.118)} & {(0.132)} & {(0.135)} & {(0.240)} & {(0.121)} & {(0.080)} & {(0.099)}
 \\
{At Least One Leader} & 0.101 & 0.349*** & 0.658 & 0.691*** & 0.238*** & 0.314*** & 0.300*** & 0.485*** & 0.394*** & 0.193** & 0.141 & 0.264*** & 0.547** & 0.142 & -0.068 & 0.307***
 \\
 & {(0.133)} & {(0.130)} & {(0.469)} & {(0.135)} & {(0.065)} & {(0.111)} & {(0.091)} & {(0.143)} & {(0.111)} & {(0.097)} & {(0.100)} & {(0.098)} & {(0.226)} & {(0.091)} & {(0.067)} & {(0.110)}
 \\
{At Least One Phone} & -0.368*** & -0.216*** & {--} & -0.075 & -0.250*** & -0.177** & -0.208*** & -0.033 & -0.036 & -0.021 & -0.094 & -0.080 & 0.045 & -0.182*** & -0.220*** & -0.293***
 \\
 & {(0.077)} & {(0.083)} & {} & {(0.084)} & {(0.053)} & {(0.078)} & {(0.063)} & {(0.125)} & {(0.075)} & {(0.070)} & {(0.105)} & {(0.074)} & {(0.252)} & {(0.062)} & {(0.055)} & {(0.103)}
 \\
  \hline\hline
  \end{tabular}
  }
  \parbox{0.8\textwidth}{
  \footnotesize{\bf Notes:} ``CaC'' stands for ``care about the community''. 
  \rev{Posterior means of the dyadic coefficients of the AME probit model~(\ref{emp1}) with 
  additive and multiplicative ($R = 3$) latent effects, 500 burn-in and 5,000 scans, thinned every 
  10; posterior standard deviations in parentheses. Village 11 ($N = 28$) is fitted with $R = 1$ 
  and without the phone covariate, which separates its ties.} \\
  \rev{Stars are based on the ratio of the posterior mean to the posterior standard deviation 
  (the $z$-statistic reported by the \texttt{amen} package): *** $|z| > 2.58$, ** $|z| > 1.96$, * $|z| > 1.64$.}}
\end{table}
\end{landscape}

\begin{table}[h]
  \centering
  \caption{Estimation Results for Treatment Assignment Mechanism} 
  \label{reg2} 
  \begin{tabular}{l S[table-format=-1.3, table-space-text-post={***}, table-align-text-post=false]}
  \hline\hline
  {Treatment Variable} & {Individual Adoption of U-Bridge Program ($D_{ig}$)} \\ 
  \hline
  Age        & -0.020** \\
             & {(0.007)} \\
  Female     & -0.876*** \\
             & {(0.236)} \\
  Income     & -0.012 \\
             & {(0.090)} \\
  Education  & 1.805*** \\
             & {(0.235)} \\
  Leader     & 0.626** \\
             & {(0.241)} \\
  Caring for Community & -0.890* \\
                 & {(0.457)} \\
  Having a Phone  & 1.230*** \\
             & {(0.304)} \\
  Attending Meetings    & 2.282*** \\
                        & {(0.213)}\\
  Village Fixed Effects & $\checkmark$ \\
  Observations             & {3,018} \\
  \hline\hline
  \end{tabular}

  \parbox{0.8\textwidth}{
  \footnotesize{\bf Notes:} Robust standard errors clustered at the village level are in parentheses. \\
  *** $p<0.01$, ** $p<0.05$, * $p<0.1$.}
\end{table}

\subsection{Estimated marginal distribution of the number of adopter friends based on the GPS}

\rev{Since $\mathbbm{E}[r(\{ig\}, t; \bm{X}_g)] = Pr(T_{ig} = t)$, the estimated score averaged over 
residents should reproduce the marginal distribution of the number of adopter friends ($Z_{ig}$). 
This checks the fit of the formation and adoption models.} \rev{Figure~\ref{eme} plots the implied marginal distribution of $Z_{ig}$ under the 
common-exposure and the treatment-diffusion specifications against its empirical distribution. 
Both reproduce the share of residents with one adopter friend but place more mass at zero 
($0.42$ and $0.45$ against $0.33$) and less at three or more ($0.12$ in both against $0.19$); see 
Section~\ref{sec:expo} for the interpretation.}

\begin{figure}[!ht]
  \begin{center}
  \begin{tabular}{c}
   \raisebox{-0.5\height}{\includegraphics[width=0.75\textwidth]{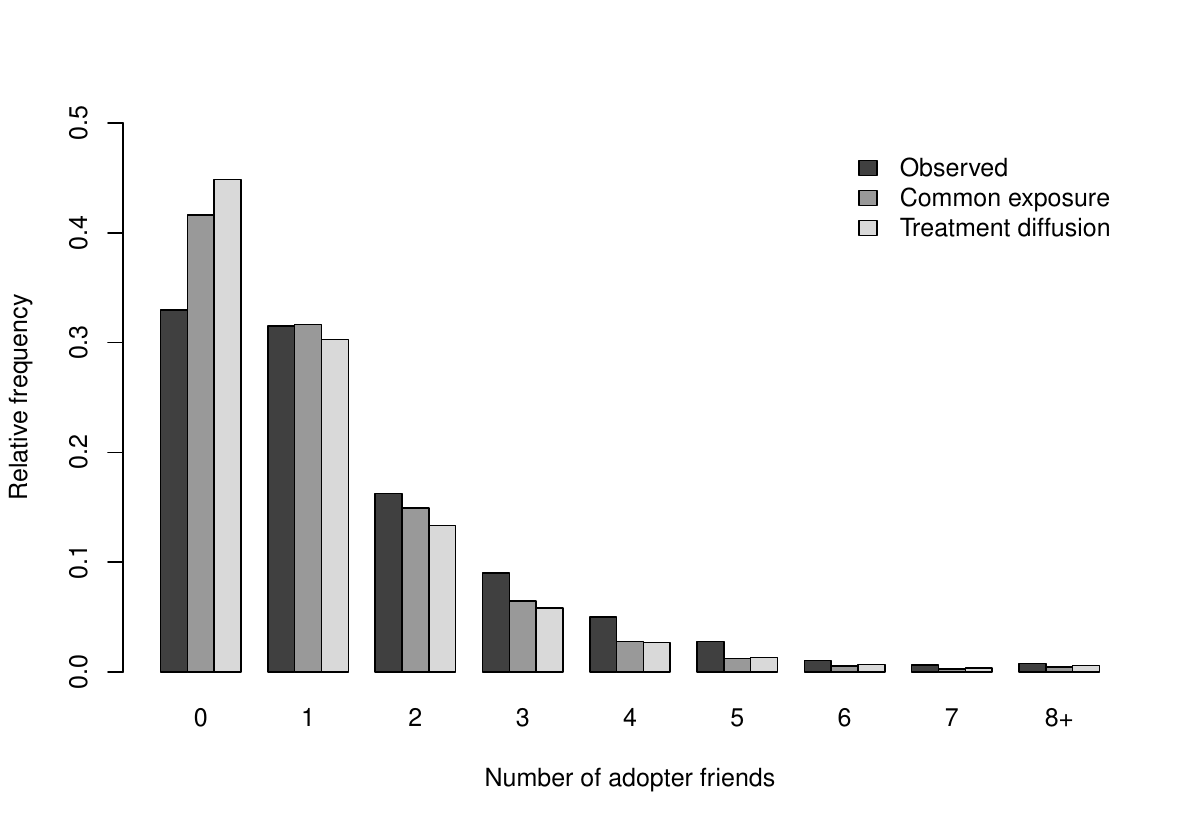}}
   \end{tabular}
   \end{center}
   \caption{Marginal Distribution of the Number of Adopter Friends}\label{eme}
   \footnotesize\textbf{Note:} \rev{Relative frequencies of the number of adopter friends: observed, 
   and implied by 500 joint draws per village of the network and adoption vector from the fitted 
   models under common exposure and under treatment diffusion (Appendix~\ref{app:rob}).}
   \end{figure}

\begin{revblock}
\subsection{Estimates, inference and sensitivity}\label{app:appest}

Table~\ref{tab:trim} reports the trimming of Section~\ref{ea} and the covariate means of kept and 
excluded residents. Table~\ref{tab:appest} reports the estimates behind Figure~\ref{te}; 
Table~\ref{tab:appinf} the few-cluster checks; Table~\ref{tab:appsens} the latent-augmented 
estimates of Section~\ref{sec:sens}; Table~\ref{tab:bal} covariate balance with the weights and 
samples the estimators use.

\begin{table}[h]
\centering
\caption{Trimming for the spillover effects: residents kept and excluded, and their covariate means}\label{tab:trim}
\footnotesize
\setlength{\tabcolsep}{4pt}
\begin{tabular}{l r r *{8}{S[table-format=2.2]}}
\hline\hline
Effect & Kept & Excluded & {Age} & {Female} & {Income} & {Educ.} & {Leader} & {CaC} & {Phone} & {Attend} \\
\hline
$\tau^{s}(1)$, kept     & 2803 &     & 36.87 & 0.61 & 2.54 & 0.20 & 0.13 & 0.28 & 0.57 & 0.06 \\
$\tau^{s}(1)$, excluded &      & 215 & 42.15 & 0.10 & 2.94 & 0.66 & 0.35 & 0.28 & 0.81 & 0.41 \\
$\tau^{s}(2)$, kept     & 2132 &     & 36.91 & 0.56 & 2.59 & 0.23 & 0.14 & 0.29 & 0.61 & 0.07 \\
$\tau^{s}(2)$, excluded &      & 886 & 38.05 & 0.61 & 2.51 & 0.24 & 0.15 & 0.27 & 0.55 & 0.12 \\
\hline\hline
\end{tabular}
\parbox{0.9\textwidth}{\footnotesize{\bf Notes:} A resident is kept for $\tau^{s}(z)$ if both 
$r(\{ig\}, (0,z); \bm{X}_g)$ and $r(\{ig\}, (0,0); \bm{X}_g)$ lie in $[0.05, 0.95]$. The direct effects 
on adopters use odds weights and no trimming (17 adopters with one adopter friend and 934 
comparison residents; 28 and 463 for two adopter friends; 130 and 1,893 for at least one). 
Population direct effects by AIPW with the same rule would keep 168 and 150 residents 
(estimates $0.13$ ($0.09$) and $-0.05$ ($0.15$)).}
\end{table}

\begin{table}[h]
\centering
\caption{Estimated spillover effects and direct effects on adopters}\label{tab:appest}
\begin{tabular}{l l S[table-format=-1.3] S[table-format=1.3] S[table-format=-1.3] S[table-format=-1.3] r}
\hline\hline
Effect & Estimator & {Estimate} & {SE} & {95\% lower} & {95\% upper} & {$N$} \\
\hline
$\tau^{s}(1)$ & AIPW-JM    & 0.034 & 0.016 &  0.003 & 0.065 & 2803 \\
            & AIPW-COND  & 0.018 & 0.020 & -0.020 & 0.057 & 2502 \\
            & IPW-JM     & 0.023 & 0.027 & -0.029 & 0.076 & 2803 \\
            & Regression & 0.049 & 0.007 &  0.036 & 0.062 & 3018 \\
$\tau^{s}(2)$ & AIPW-JM    & 0.051 & 0.034 & -0.015 & 0.117 & 2132 \\
            & AIPW-COND  & 0.063 & 0.044 & -0.023 & 0.149 & 1611 \\
            & IPW-JM     & 0.008 & 0.045 & -0.079 & 0.096 & 2132 \\
            & Regression & 0.098 & 0.013 &  0.071 & 0.124 & 3018 \\
$\tau^{s}(\text{any})$ & AIPW-JM & 0.052 & 0.016 & 0.021 & 0.084 & 2799 \\
\hline
$\tau^{e}(1)$ & AIPW-ATT   & 0.363 & 0.094 &  0.179 & 0.547 & 17 \\
                & IPW-ATT    & 0.417 & 0.080 &  0.260 & 0.574 & 17 \\
                & Regression & 0.280 & 0.075 &  0.134 & 0.426 & 17 \\
$\tau^{e}(2)$ & AIPW-ATT   & 0.188 & 0.109 & -0.025 & 0.401 & 28 \\
                & IPW-ATT    & 0.232 & 0.156 & -0.073 & 0.538 & 28 \\
                & Regression & 0.251 & 0.070 &  0.114 & 0.388 & 28 \\
$\tau^{e}(\text{any})$ & AIPW-ATT & 0.271 & 0.069 & 0.135 & 0.407 & 130 \\
                          & IPW-ATT  & 0.247 & 0.085 & 0.080 & 0.415 & 130 \\
\hline\hline
\end{tabular}
\parbox{0.92\textwidth}{\footnotesize{\bf Notes:} Village-clustered standard errors and normal 
critical values. For the spillovers, $N$ is the number of residents after trimming (for AIPW-COND, 
trimming on the network-conditional score, which also excludes residents whose realized degree does not 
permit a level); for the direct effects, $N$ is the number of adopters at the level. ``Regression'' 
is the coefficient contrast of the outcome model of Table~\ref{reg}, column (3), with contextual 
covariates, which coincides with regression imputation under that linear model; its standard 
error is the delta-method standard error with village-clustered covariance. 
$\tau^{s}(\text{any})$ and $\tau^{e}(\text{any})$ use the exposure ``at least one adopter 
friend'' among residents with at least one tie.}
\end{table}

\begin{table}[h]
\centering
\caption{Inference with 16 villages for the main estimates}\label{tab:appinf}
\footnotesize
\setlength{\tabcolsep}{4.5pt}
\begin{tabular}{l S[table-format=1.3] S[table-format=1.3] c c S[table-format=1.3] S[table-format=1.3]}
\hline\hline
 & & {Cluster} & \multicolumn{2}{c}{95\% interval} & {Bootstrap} & {Within-village} \\
\cline{4-5}
Effect & {Estimate} & {SE} & {normal} & {$t_{15}$} & {SE} & {SE} \\
\hline
$\tau^{s}(1)$ & 0.034 & 0.016 & $[0.003, 0.065]$  & $[0.000, 0.068]$  & 0.017 & 0.023 \\
$\tau^{s}(2)$ & 0.051 & 0.034 & $[-0.015, 0.117]$ & $[-0.021, 0.123]$ & 0.036 & 0.037 \\
$\tau^{e}(1)$ & 0.363 & 0.094 & $[0.179, 0.547]$  & $[0.163, 0.563]$  & 0.097 & {--} \\
$\tau^{e}(2)$ & 0.188 & 0.109 & $[-0.025, 0.401]$ & $[-0.044, 0.420]$ & 0.122 & {--} \\
\hline\hline
\end{tabular}
\parbox{0.95\textwidth}{\footnotesize{\bf Notes:} Bootstrap SE: wild cluster bootstrap with 
Rademacher weights by village applied to the outcome-model residuals; the outcome model is 
re-estimated on each bootstrap outcome, the estimator recomputed with the scores held fixed, and 
the standard error is the standard deviation over 1,999 draws (the bootstrap distribution is 
centred at the regression contrast rather than at the estimate, so only its spread is used; the 
corresponding interval is the estimate $\pm 1.96$ times this SE). Within-village SE: 
estimator~(\ref{hac2}) applied village by village with $b_i$ centred within village (heuristic 
here, Remark~\ref{rem:few}; not defined for the effects on the exposed).}
\end{table}

\begin{table}[h]
\centering
\caption{Latent-augmented specification: AME additive and multiplicative effects added to the adoption and outcome models}\label{tab:appsens}
\begin{tabular}{l *{4}{S[table-format=1.3]}}
\hline\hline
 & {$\tau^{s}(1)$} & {$\tau^{s}(2)$} & {$\tau^{e}(1)$} & {$\tau^{e}(2)$} \\
\hline
Main specification      & 0.034 & 0.051 & 0.363 & 0.188 \\
                        & {(0.016)} & {(0.034)} & {(0.094)} & {(0.109)} \\
Latent-augmented        & 0.010 & 0.028 & 0.320 & 0.168 \\
                        & {(0.022)} & {(0.034)} & {(0.104)} & {(0.093)} \\
\hline\hline
\end{tabular}
\parbox{0.8\textwidth}{\footnotesize{\bf Notes:} Posterior means of $(\alpha_{ig}, \tilde{U}_{ig})$ 
from the AME model added as covariates to the Firth logit for adoption and to the outcome model; 
scores recomputed with the new adoption probabilities. Village-clustered standard errors in 
parentheses.}
\end{table}

\begin{table}[h]
\centering
\caption{Covariate balance with the weights and samples used by the estimators: standardized mean differences}\label{tab:bal}
\footnotesize
\setlength{\tabcolsep}{3.5pt}
\begin{tabular}{l *{10}{S[table-format=-1.2]}}
\hline\hline
 & \multicolumn{2}{c}{$(0,1)$ vs $(0,0)$} & \multicolumn{2}{c}{$(0,2)$ vs $(0,0)$} & \multicolumn{2}{c}{$(1,1)$ vs $(0,1)$} & \multicolumn{2}{c}{$(1,2)$ vs $(0,2)$} & \multicolumn{2}{c}{$(1,\text{any})$ vs $(0,\text{any})$} \\
 & {Raw} & {Wtd} & {Raw} & {Wtd} & {Raw} & {Wtd} & {Raw} & {Wtd} & {Raw} & {Wtd} \\
\hline
Age              &  0.02 & -0.01 &  0.02 & -0.07 & -0.30 & -0.04 & -0.32 & -0.04 & -0.14 & -0.04 \\
Female           & -0.05 &  0.11 &  0.03 &  0.22 & -0.60 & -0.12 & -0.36 &  0.15 & -0.75 &  0.04 \\
Income           & -0.10 & -0.18 & -0.13 & -0.21 &  0.64 &  0.20 &  0.43 &  0.20 &  0.39 &  0.01 \\
Education        &  0.02 & -0.11 &  0.00 & -0.16 &  1.75 &  0.46 &  1.50 &  0.23 &  1.24 & -0.07 \\
Leader           &  0.00 &  0.01 &  0.19 &  0.17 &  0.20 &  0.11 & -0.03 & -0.35 &  0.30 & -0.07 \\
Care about community & -0.05 & -0.13 & -0.07 & -0.15 &  0.17 &  0.19 & -0.20 & -0.06 & -0.08 & -0.02 \\
Has phone        &  0.07 & -0.04 &  0.04 & -0.08 &  0.58 & -0.12 &  0.67 &  0.01 &  0.71 & -0.01 \\
Attended meeting &  0.08 &  0.06 &  0.14 &  0.15 &  1.40 &  0.55 &  0.88 & -0.20 &  1.16 & -0.13 \\
\hline\hline
\end{tabular}
\parbox{0.95\textwidth}{\footnotesize{\bf Notes:} Spillover contrasts: residents at $(0,z)$ 
versus $(0,0)$ within the trimmed sample of $\tau^{s}(z)$, weighted by $1/\hat{r}$ at the respective 
level. Direct contrasts: adopters at $(1,z)$ versus non-adopters at $(0,z)$, the latter weighted by 
the odds $\hat{r}(1,z) / \hat{r}(0,z)$, as in the estimator of the effect on the exposed.}
\end{table}

\subsection{Robustness checks and sensitivity grid}\label{app:rob}

\begin{table}[h]
\centering
\caption{Robustness: treatment-diffusion factorization}\label{tab:rob}
\begin{tabular}{l S[table-format=1.3] S[table-format=1.3] S[table-format=-1.3] S[table-format=1.3] r}
\hline\hline
Effect & {Estimate} & {SE} & {95\% lower} & {95\% upper} & {$N$} \\
\hline
$\tau^{s}(1)$     & 0.025 & 0.014 & -0.003 & 0.053 & 2880 \\
$\tau^{s}(2)$     & 0.072 & 0.037 & -0.001 & 0.145 & 2237 \\
$\tau^{e}(1)$ & 0.372 & 0.090 &  0.197 & 0.548 & 17 \\
$\tau^{e}(2)$ & 0.186 & 0.107 & -0.023 & 0.395 & 28 \\
\hline\hline
\end{tabular}
\parbox{0.95\textwidth}{\footnotesize{\bf Notes:} Adoption model: Firth logit with the number 
of adopter friends as a regressor (maximum pseudo-likelihood for an autologistic model; 
coefficient $0.35$, standard error $0.057$, profile-likelihood interval $[0.24, 0.47]$). Joint 
draws per village: network from the plug-in link probabilities, adoption initialized at 
independent draws and updated by twenty Gibbs sweeps (the mean number of adopter friends is 
stable after ten sweeps, although in several villages it differs from the observed mean, since 
the joint law implied by pseudo-likelihood full conditionals need not reproduce the observed 
adoption rates); 2,000 draws per village; scores at every level as Monte Carlo 
frequencies (with a score of $0.05$ the Monte Carlo error is about ten percent of the weight). 
Same trimming rules and estimators as Table~\ref{tab:appest}.}
\end{table}

\begin{table}[h]
\centering
\caption{Hypothetical confounder of adoption: shift in the estimates of the direct effects on adopters}\label{tab:sensgrid}
\begin{tabular}{S[table-format=1.1] S[table-format=1.2] *{2}{S[table-format=-1.3]}}
\hline\hline
{$\gamma_d$} & {$\gamma_y$} & {$\tau^{e}(1)$} & {$\tau^{e}(2)$} \\
\hline
0.5 & 0.10 & -0.015 & -0.013 \\
1.0 & 0.10 & -0.029 & -0.027 \\
2.0 & 0.10 & -0.041 & -0.042 \\
0.5 & 0.25 & -0.028 & -0.029 \\
1.0 & 0.25 & -0.057 & -0.063 \\
2.0 & 0.25 & -0.097 & -0.105 \\
0.5 & 0.50 & -0.064 & -0.051 \\
1.0 & 0.50 & -0.112 & -0.120 \\
2.0 & 0.50 & -0.205 & -0.207 \\
\hline\hline
\end{tabular}
\parbox{0.85\textwidth}{\footnotesize{\bf Notes:} A binary $U_{ig}$ with $Pr(U_{ig} = 1) = 1/2$ 
lowers the adoption log-odds by $\gamma_d$ and raises participation by $\gamma_y$ (the opposite 
signs give the same magnitudes). $U_{ig}$ is drawn from its conditional distribution given 
observed adoption and covariates under the fitted adoption model, $\gamma_y U_{ig}$ is added to 
the outcome, and the estimates are recomputed ignoring $U_{ig}$; entries are the average change 
over 50 draws. The spillover estimates change by less than $0.003$ because the realized network 
is held fixed. For reference, the adoption-model coefficients of secondary education and meeting 
attendance are $1.8$ and $2.3$.}
\end{table}
\end{revblock}

\end{document}